\pdfoutput=1
\documentclass{article}
\newcommand{\releaseurl}{\url{https://github.com/cohere-ai/rcp-ndcg}}

\let\ORIGaddcontentsline\addcontentsline
\usepackage{iclr2027_conference}
\let\addcontentsline\ORIGaddcontentsline
\usepackage[utf8]{inputenc}
\usepackage[T1]{fontenc}
\usepackage{times}

\usepackage{amsmath,amssymb,amsfonts,mathtools}
\usepackage{amsthm}
\usepackage{aliascnt}
\usepackage{graphicx}
\usepackage{subcaption}
\usepackage{wrapfig}
\usepackage{float}
\usepackage{booktabs}
\usepackage{multirow}
\usepackage{multicol}
\usepackage{tabularx}
\usepackage{colortbl}
\usepackage{arydshln}

\usepackage{microtype}
\usepackage{nicefrac}
\usepackage{url}
\usepackage[usenames,dvipsnames]{xcolor}
\usepackage[inline]{enumitem}
\usepackage{hyperref}
\usepackage[capitalize,noabbrev]{cleveref}

\usepackage{algorithm}
\usepackage{algorithmic}

\usepackage{tocloft}
\usepackage{etoc}
\usepackage[toc,page,header]{appendix}
\usepackage{etoolbox}
\usepackage{scalerel}
\usepackage{placeins}
\usepackage{dsfont}
\usepackage{leftindex}
\usepackage{overpic}
\usepackage{cancel}
\usepackage[most]{tcolorbox}
\usepackage{mdframed}

\usepackage{tikz}
\usetikzlibrary{arrows.meta,bending,positioning,3d}
\usepackage{tikz-3dplot}

\newif\ificlr
\iclrfalse

\theoremstyle{plain}

\newaliascnt{proposition}{theorem}
\newtheorem{proposition}[proposition]{Proposition}
\aliascntresetthe{proposition}
\newaliascnt{lemma}{theorem}
\newtheorem{lemma}[lemma]{Lemma}
\aliascntresetthe{lemma}
\newaliascnt{corollary}{theorem}

\aliascntresetthe{corollary}

\theoremstyle{definition}
\newaliascnt{definition}{theorem}

\aliascntresetthe{definition}
\newaliascnt{assumption}{theorem}

\aliascntresetthe{assumption}

\theoremstyle{remark}
\newaliascnt{remark}{theorem}

\aliascntresetthe{remark}
\newaliascnt{fact}{theorem}

\aliascntresetthe{fact}

\crefname{appendix}{Appendix}{Appendices}
\Crefname{appendix}{Appendix}{Appendices}
\crefname{theorem}{Theorem}{Theorems}
\Crefname{theorem}{Theorem}{Theorems}
\crefname{proposition}{Proposition}{Propositions}
\Crefname{proposition}{Proposition}{Propositions}
\crefname{corollary}{Corollary}{Corollaries}
\Crefname{corollary}{Corollary}{Corollaries}
\crefname{lemma}{Lemma}{Lemmas}
\Crefname{lemma}{Lemma}{Lemmas}
\crefname{definition}{Definition}{Definitions}
\Crefname{definition}{Definition}{Definitions}
\crefname{assumption}{Assumption}{Assumptions}
\Crefname{assumption}{Assumption}{Assumptions}
\crefname{remark}{Remark}{Remarks}
\Crefname{remark}{Remark}{Remarks}

\crefname{fact}{Fact}{Facts}
\Crefname{fact}{Fact}{Facts}

\definecolor{DarkSlateGrey}{HTML}{335c67}
\definecolor{SteelBlue}{HTML}{457b9d}
\definecolor{DeepTeal}{HTML}{66827a}
\definecolor{BrownRed}{HTML}{9e2a2b}
\definecolor{CinnamonWood}{HTML}{bf6535}
\definecolor{ForestGreen}{HTML}{2d6a4f}
\definecolor{DrySage}{HTML}{b3bb95}
\definecolor{VanillaCustard}{HTML}{fff3b0}
\definecolor{NightBordeaux}{HTML}{540b0e}
\definecolor{Amethyst}{HTML}{7b2d8e}
\definecolor{HoneyBronze}{HTML}{e8b45b}
\definecolor{Jasmine}{HTML}{f0c977}
\definecolor{HoneyBronzeAlt}{HTML}{e09f3e}

\colorlet{doccolor}{SteelBlue}
\colorlet{querycolor}{Amethyst}
\colorlet{critcolor}{CinnamonWood}
\newcommand{\dpar}[1]{{\color{doccolor}#1}}
\newcommand{\qpar}[1]{{\color{querycolor}#1}}
\newcommand{\cpar}[1]{{\color{critcolor}#1}}

\definecolor{iccvblue}{rgb}{0.21,0.49,0.74}
\definecolor{cvprblue}{rgb}{0.21,0.49,0.74}
\colorlet{mygreen}{ForestGreen}
\colorlet{myblue}{DarkSlateGrey}
\colorlet{myyellow}{HoneyBronzeAlt}
\colorlet{myred}{BrownRed}
\colorlet{mydarkred}{NightBordeaux}
\colorlet{mywhite}{VanillaCustard}

\hypersetup{
    colorlinks,
    linkcolor={red!50!black},
    citecolor={blue!50!black},
    urlcolor={blue!80!black}
}
\hypersetup{
    pdftitle={Rubric-Calibrated Preferences: Cross-Query Calibration of LLM Judgments via Item Response Theory},
    pdfkeywords={retrieval evaluation, reranking, nDCG, LLM judges, item response theory, calibration}
}

\hypersetup{pdfauthor={Fabian David Schmidt, Donato Crisostomi, Carlos Lassance, Nils Reimers}, pdfsubject={arXiv preprint}}

\newcommand{\reranker}[1]{\baseline{#1}}

\newcommand{\rcp}{\textup{RCP}}
\newcommand{\rcpndcg}{\mbox{\textup{RCP-nDCG}}}

\makeatletter
\newcommand{\setval}[2]{\expandafter\gdef\csname val@#1\endcsname{#2}}
\newcommand{\val}[1]{\ifcsname val@#1\endcsname\csname val@#1\endcsname\else\textcolor{red}{\textbf{??#1??}}\PackageWarning{vals}{Undefined value #1}\fi}
\makeatother

\newcommand{\dataset}[1]{\textsc{#1}}
\newcommand{\model}[1]{\mbox{\textsc{#1}}}

\makeatletter
\newcommand{\rr@name}[2]{\expandafter\def\csname rrname@#1\endcsname{#2}}
\rr@name{CTXL-RR-1B}{Contextual v2 1B}\rr@name{CTXL-RR-2B}{Contextual v2 2B}\rr@name{CTXL-RR-6B}{Contextual v2 6B}
\rr@name{Jina-RR-v3}{Jina Reranker v3}
\rr@name{Qwen3-RR-0.6B}{Qwen3-Reranker-0.6B}\rr@name{Qwen3-RR-4B}{Qwen3-Reranker-4B}\rr@name{Qwen3-RR-8B}{Qwen3-Reranker-8B}
\rr@name{RR4-Pro}{Rerank 4 Pro}\rr@name{RR4-Fast}{Rerank 4 Fast}
\rr@name{Voyage2.5}{Voyage 2.5}\rr@name{Voyage2.5-lt}{Voyage 2.5 Lite}
\rr@name{zerank-1-sm}{zerank-1-small}\rr@name{zerank-1}{zerank-1}\rr@name{zerank-2}{zerank-2}
\newcommand{\baseline}[1]{\ifcsname rrname@#1\endcsname\csname rrname@#1\endcsname\else#1\fi}
\makeatother

\makeatletter
\def\adl@drawiv#1#2#3{%
\hskip.5\tabcolsep
\xleaders#3{#2.5\@tempdimb #1{1}#2.5\@tempdimb}%
#2\z@ plus1fil minus1fil\relax
\hskip.5\tabcolsep}
\newcommand{\cdashlinelr}[1]{%
\noalign{\vskip\aboverulesep
\global\let\@dashdrawstore\adl@draw
\global\let\adl@draw\adl@drawiv}
\cdashline{#1}
\noalign{\global\let\adl@draw\@dashdrawstore
\vskip\belowrulesep}}
\makeatother

\usetikzlibrary{positioning, arrows.meta, fit, backgrounds, calc, shadows}

\colorlet{blue1}{SteelBlue}
\definecolor{blue1bg}{HTML}{e8eff4}
\colorlet{red1}{CinnamonWood}
\definecolor{red1bg}{HTML}{f8ece0}
\colorlet{red1pill}{BrownRed}
\colorlet{dark1}{DarkSlateGrey}
\definecolor{gray1}{HTML}{6c757d}
\definecolor{grayL}{HTML}{e9ecef}

\usepackage{tcolorbox}
\usepackage{fvextra}
\tcbuselibrary{breakable, skins}

\newtcolorbox{promptboxclean}[1][]{
  breakable,
  enhanced,
  colback=white,
  colframe=black!20,
  boxrule=0.5pt,
  arc=1mm,
  left=1.5mm,
  right=1.5mm,
  top=1mm,
  bottom=1mm,
  title={#1},
  fonttitle=\bfseries,
}

\setval{a1a.auc.ndcg}{0.521}
\setval{a1a.auc.ndcg.bright}{0.501}
\setval{a1a.auc.ndcg.flips}{0.346}
\setval{a1a.auc.ndcg.hi}{0.576}
\setval{a1a.auc.ndcg.lo}{0.467}
\setval{a1a.auc.ndcg.n}{854}
\setval{a1a.auc.ndcg.nanomteb}{0.565}
\setval{a1a.auc.ndcg.nn}{0.525}
\setval{a1a.auc.ndcg.vidore}{0.503}
\setval{a1a.auc.theta}{0.810}
\setval{a1a.auc.theta.bright}{0.825}
\setval{a1a.auc.theta.bright.n}{289}
\setval{a1a.auc.theta.flips}{0.763}
\setval{a1a.auc.theta.hi}{0.845}
\setval{a1a.auc.theta.lo}{0.776}
\setval{a1a.auc.theta.n}{854}
\setval{a1a.auc.theta.nanomteb}{0.793}
\setval{a1a.auc.theta.nanomteb.n}{284}
\setval{a1a.auc.theta.nn}{0.828}
\setval{a1a.auc.theta.nn.n}{402}
\setval{a1a.auc.theta.vidore}{0.821}
\setval{a1a.auc.theta.vidore.n}{281}
\setval{a1a.band.ndcg.5}{53.1}
\setval{a1a.band.rcp.1}{47.2}
\setval{a1a.band.rcp.1.hi}{63.8}
\setval{a1a.band.rcp.1.lo}{30.7}
\setval{a1a.band.rcp.1.n}{36}
\setval{a1a.band.rcp.5}{96.8}
\setval{a1a.ctrl.ceiling.qrel-perfect}{98.3}
\setval{a1a.ctrl.ceiling.rcp-perfect}{87.5}
\setval{a1a.ctrl.gainmap.binary-ge2.d14}{75.2}
\setval{a1a.ctrl.gainmap.binary-ge4.d14}{64.4}
\setval{a1a.ctrl.query-blind-loqo}{65.7}
\setval{a1a.ctrl.withinpair.split.constant}{70.8}
\setval{a1a.ctrl.withinpair.split.ndcg}{50.4}
\setval{a1a.ctrl.withinpair.split.rcp}{77.0}
\setval{a1a.d14.bright}{74.6}
\setval{a1a.d14.bright.hi}{85.2}
\setval{a1a.d14.bright.k}{50}
\setval{a1a.d14.bright.lo}{64.1}
\setval{a1a.d14.bright.n}{67}
\setval{a1a.d14.hotpotqa.k}{3}
\setval{a1a.d14.hotpotqa.n}{6}
\setval{a1a.d14.nanomteb}{68.3}
\setval{a1a.d14.nanomteb.hi}{80.3}
\setval{a1a.d14.nanomteb.k}{41}
\setval{a1a.d14.nanomteb.lo}{56.4}
\setval{a1a.d14.nanomteb.n}{60}
\setval{a1a.d14.ndcg-seen}{70.5}
\setval{a1a.d14.ndcg-seen.hi}{77.4}
\setval{a1a.d14.ndcg-seen.k}{122}
\setval{a1a.d14.ndcg-seen.lo}{63.7}
\setval{a1a.d14.ndcg-seen.n}{173}
\setval{a1a.d14.nn}{68.1}
\setval{a1a.d14.nn.hi}{77.5}
\setval{a1a.d14.nn.k}{64}
\setval{a1a.d14.nn.lo}{58.7}
\setval{a1a.d14.nn.n}{94}
\setval{a1a.d14.vidore}{74.1}
\setval{a1a.d14.vidore.hi}{85.5}
\setval{a1a.d14.vidore.k}{43}
\setval{a1a.d14.vidore.lo}{62.8}
\setval{a1a.d14.vidore.n}{58}
\setval{a1a.d14.zerank}{76.9}
\setval{a1a.d17.low.ndcg-agrees.k}{8}
\setval{a1a.d17.low.ndcg-agrees.n}{13}
\setval{a1a.d17.low.ndcg-disagrees.k}{8}
\setval{a1a.d17.low.ndcg-disagrees.n}{22}
\setval{a1a.d17.marginal.median-abs-dndcg}{0.226}
\setval{a1a.decisive}{74.8}
\setval{a1a.decisive.hi}{82.3}
\setval{a1a.decisive.k}{98}
\setval{a1a.decisive.lo}{67.3}
\setval{a1a.decisive.n}{131}
\setval{a1a.desc.annotators}{46}
\setval{a1a.desc.annotators.bright}{39}
\setval{a1a.desc.annotators.nanomteb}{33}
\setval{a1a.desc.annotators.vidore}{35}
\setval{a1a.desc.collection-end-date}{2026-08-25}
\setval{a1a.desc.collection-start-date}{2026-07-20}
\setval{a1a.desc.contests}{311}
\setval{a1a.desc.contests.bright}{103}
\setval{a1a.desc.contests.nanomteb}{108}
\setval{a1a.desc.contests.vidore}{100}
\setval{a1a.desc.datasets}{24}
\setval{a1a.desc.datasets.bright}{7}
\setval{a1a.desc.datasets.nanomteb}{12}
\setval{a1a.desc.datasets.vidore}{5}
\setval{a1a.desc.decided-flips}{156}
\setval{a1a.desc.decided-ndcg-tie}{27}
\setval{a1a.desc.decided-theta-tie}{8}
\setval{a1a.desc.distinct-docs}{2{,}268}
\setval{a1a.desc.grades}{7{,}080}
\setval{a1a.desc.ndcg-blind}{14}
\setval{a1a.desc.nn-contests}{149}
\setval{a1a.desc.pairs}{42}
\setval{a1a.desc.queries}{292}
\setval{a1a.desc.queries.bright}{95}
\setval{a1a.desc.queries.nanomteb}{97}
\setval{a1a.desc.queries.vidore}{100}
\setval{a1a.desc.rerankers}{10}
\setval{a1a.desc.reviews}{933}
\setval{a1a.desc.reviews.bright}{309}
\setval{a1a.desc.reviews.nanomteb}{324}
\setval{a1a.desc.reviews.vidore}{300}
\setval{a1a.desc.verdict-a}{139}
\setval{a1a.desc.verdict-b}{150}
\setval{a1a.desc.verdict-tie}{22}
\setval{a1a.desc.zerank-share}{52.1}
\setval{a1a.h2h}{69.9}
\setval{a1a.h2h.hi}{77.1}
\setval{a1a.h2h.k}{109}
\setval{a1a.h2h.lo}{62.7}
\setval{a1a.h2h.n}{156}
\setval{a1a.h2h.nn}{65.4}
\setval{a1a.h2h.nn.hi}{75.9}
\setval{a1a.h2h.nn.k}{51}
\setval{a1a.h2h.nn.lo}{54.9}
\setval{a1a.h2h.nn.n}{78}
\setval{a1a.hit-ndcg}{51.9}
\setval{a1a.hit-theta}{77.9}
\setval{a1a.hit-theta.k}{219}
\setval{a1a.hit-theta.n}{281}
\setval{a1a.len.human-prefers-longer}{56.7}
\setval{a1a.len.logit.b-loglen}{-0.02}
\setval{a1a.len.logit.b-loglen.hi}{0.23}
\setval{a1a.len.logit.b-loglen.lo}{-0.28}
\setval{a1a.len.logit.b-ndcg}{0.27}
\setval{a1a.len.logit.b-theta}{1.87}
\setval{a1a.loao.d14.max}{74.4}
\setval{a1a.loao.d14.min}{70.6}
\setval{a1a.loao.d14.min.n}{46}
\setval{a1a.loao.h2h.max}{71.8}
\setval{a1a.loao.h2h.min}{68.4}
\setval{a1a.loao.h2h.min.n}{46}
\setval{a1a.marginal}{44.0}
\setval{a1a.marginal.hi}{63.9}
\setval{a1a.marginal.k}{11}
\setval{a1a.marginal.lo}{24.1}
\setval{a1a.marginal.n}{25}
\setval{a1a.primary}{72.4}
\setval{a1a.primary.hi}{78.9}
\setval{a1a.primary.k}{134}
\setval{a1a.primary.lo}{66.0}
\setval{a1a.primary.n}{185}
\setval{a1a.qstd.n}{254}
\setval{a1a.qstd.ndcg.d51}{7.7}
\setval{a1a.qstd.ndcg.d51.hi}{24.9}
\setval{a1a.qstd.ndcg.d51.lo}{-9.6}
\setval{a1a.qstd.rcp.d51}{45.2}
\setval{a1a.qstd.rcp.d51.hi}{60.4}
\setval{a1a.qstd.rcp.d51.lo}{29.9}
\setval{a1a.qstd.share-disagree}{54.0}
\setval{a1a.sweep.p02.above}{82.4}
\setval{a1a.sweep.p02.below}{47.2}
\setval{a1b.abl.union.all.n-decided}{289}
\setval{a1b.agree.human-human.all}{74.2}
\setval{a1b.agree.human-human.bright}{73.8}
\setval{a1b.agree.human-human.nano}{77.2}
\setval{a1b.agree.human-human.vidore}{71.5}
\setval{a1b.agree.judge-human.all}{78.4}
\setval{a1b.agree.judge-human.bright}{77.9}
\setval{a1b.agree.judge-human.nano}{82.3}
\setval{a1b.agree.judge-human.vidore}{74.9}
\setval{a1b.agree.qrel-human.all}{65.0}
\setval{a1b.agree.qrel-human.bright}{65.0}
\setval{a1b.agree.qrel-human.nano}{64.5}
\setval{a1b.agree.qrel-human.vidore}{65.5}
\setval{a1b.auc.ge2.all.gain}{0.910}
\setval{a1b.auc.ge2.all.qrel}{0.651}
\setval{a1b.auc.ge2.bright.gain}{0.889}
\setval{a1b.auc.ge2.bright.qrel}{0.596}
\setval{a1b.auc.ge2.nano.gain}{0.939}
\setval{a1b.auc.ge2.nano.qrel}{0.670}
\setval{a1b.auc.ge2.vidore.gain}{0.881}
\setval{a1b.auc.ge2.vidore.qrel}{0.664}
\setval{a1b.auc.within-q0.ge2.all.gain}{0.914}
\setval{a1b.auc.within-q0.ge2.all.gain.ci}{[0.894, 0.933]}
\setval{a1b.auc.within-q0.ge2.bright.gain}{0.881}
\setval{a1b.auc.within-q0.ge2.nano.gain}{0.946}
\setval{a1b.auc.within-q0.ge2.vidore.gain}{0.907}
\setval{a1b.cal.auc.ge2.mean.all.cal}{0.906}
\setval{a1b.cal.auc.ge2.mean.bright.cal}{0.889}
\setval{a1b.cal.auc.ge2.mean.bright.raw}{0.856}
\setval{a1b.cal.auc.ge2.mean.bright.z}{0.833}
\setval{a1b.cal.auc.ge2.mean.nano.cal}{0.939}
\setval{a1b.cal.auc.ge2.mean.nano.raw}{0.829}
\setval{a1b.cal.auc.ge2.mean.nano.z}{0.756}
\setval{a1b.cal.auc.ge2.mean.vidore.cal}{0.881}
\setval{a1b.cal.auc.ge2.mean.vidore.raw}{0.831}
\setval{a1b.cal.auc.ge2.mean.vidore.z}{0.810}
\setval{a1b.cal.qlevel.all.cal}{0.795}
\setval{a1b.cal.qlevel.all.raw}{0.538}
\setval{a1b.cal.rho.mean.bright.cal}{0.753}
\setval{a1b.cal.rho.mean.bright.d-cal-raw}{0.072}
\setval{a1b.cal.rho.mean.bright.d-cal-raw.ci}{[0.026, 0.120]}
\setval{a1b.cal.rho.mean.bright.raw}{0.681}
\setval{a1b.cal.rho.mean.bright.z}{0.646}
\setval{a1b.cal.rho.mean.nano.cal}{0.830}
\setval{a1b.cal.rho.mean.nano.d-cal-raw}{0.190}
\setval{a1b.cal.rho.mean.nano.d-cal-raw.ci}{[0.122, 0.259]}
\setval{a1b.cal.rho.mean.nano.raw}{0.641}
\setval{a1b.cal.rho.mean.nano.z}{0.515}
\setval{a1b.cal.rho.mean.vidore.cal}{0.760}
\setval{a1b.cal.rho.mean.vidore.d-cal-raw}{0.067}
\setval{a1b.cal.rho.mean.vidore.d-cal-raw.ci}{[0.005, 0.139]}
\setval{a1b.cal.rho.mean.vidore.raw}{0.693}
\setval{a1b.cal.rho.mean.vidore.z}{0.648}
\setval{a1b.ece.ge1.all}{0.167}
\setval{a1b.ece.ge2.all}{0.054}
\setval{a1b.ece.ge2.all.ci}{[0.042, 0.071]}
\setval{a1b.ece.ge3.all}{0.187}
\setval{a1b.mis.q0-ge2.all}{27.8}
\setval{a1b.mis.q0-ge2.bright}{22.8}
\setval{a1b.mis.q0-ge2.nano}{33.1}
\setval{a1b.mis.q0-ge2.vidore}{28.9}
\setval{a1b.mis.q1-lt2.all}{37.5}
\setval{a1b.mis.q1-lt2.bright}{52.8}
\setval{a1b.mis.q1-lt2.nano}{25.7}
\setval{a1b.mis.q1-lt2.vidore}{37.5}
\setval{a1b.n.docs.all}{2{,}268}
\setval{a1b.n.docs.bright}{819}
\setval{a1b.n.docs.nano}{723}
\setval{a1b.n.docs.vidore}{726}
\setval{a1b.op.all.qrel.fpr}{17.3}
\setval{a1b.op.all.qrel.tpr}{47.5}
\setval{a1b.op.all.rcp.fpr}{7.3}
\setval{a1b.op.all.rcp.tpr}{64.1}
\setval{a1b.op.bright.qrel.fpr}{14.3}
\setval{a1b.op.bright.qrel.tpr}{33.6}
\setval{a1b.op.bright.rcp.fpr}{5.6}
\setval{a1b.op.bright.rcp.tpr}{56.6}
\setval{a1b.op.nano.qrel.fpr}{13.3}
\setval{a1b.op.nano.qrel.tpr}{47.2}
\setval{a1b.op.nano.rcp.fpr}{4.5}
\setval{a1b.op.nano.rcp.tpr}{58.0}
\setval{a1b.op.vidore.qrel.fpr}{25.2}
\setval{a1b.op.vidore.qrel.tpr}{58.0}
\setval{a1b.op.vidore.rcp.fpr}{14.5}
\setval{a1b.op.vidore.rcp.tpr}{72.8}
\setval{a1b.pair.q0rel-over-q1irr.within.all}{0.791}
\setval{a1b.par.balanced.all.sb}{0.773}
\setval{a1b.par.full.all.cal}{0.792}
\setval{a1b.qrel.rate.all}{28.7}
\setval{a1b.qrel.rate.bright}{19.7}
\setval{a1b.qrel.rate.nano}{28.5}
\setval{a1b.qrel.rate.vidore}{39.0}
\setval{a1b.samp.pool-concordant.ge2.all.gain}{0.932}
\setval{a1b.samp.pool-concordant.ge2.all.qrel}{0.722}
\setval{a1b.samp.pool-concordant.ge2.bright.gain}{0.916}
\setval{a1b.samp.pool-concordant.ge2.bright.qrel}{0.672}
\setval{a1b.samp.pool-concordant.ge2.nano.gain}{0.966}
\setval{a1b.samp.pool-concordant.ge2.nano.qrel}{0.749}
\setval{a1b.samp.pool-concordant.ge2.vidore.gain}{0.888}
\setval{a1b.samp.pool-concordant.ge2.vidore.qrel}{0.721}
\setval{a1b.samp.pool-signflip.ge2.all.gain}{0.900}
\setval{a1b.samp.pool-signflip.ge2.all.qrel}{0.622}
\setval{a1b.samp.pool-signflip.ge2.bright.gain}{0.878}
\setval{a1b.samp.pool-signflip.ge2.bright.qrel}{0.567}
\setval{a1b.samp.pool-signflip.ge2.nano.gain}{0.924}
\setval{a1b.samp.pool-signflip.ge2.nano.qrel}{0.635}
\setval{a1b.samp.pool-signflip.ge2.vidore.gain}{0.877}
\setval{a1b.samp.pool-signflip.ge2.vidore.qrel}{0.641}
\setval{a2.agree.pool.dl19.rcpa.atnistlin}{14/14}
\setval{a2.agree.pool.dl19.rcpb.atnistlin}{14/14}
\setval{a2.agree.pool.dl20.rcpa.atnistlin}{37/37}
\setval{a2.agree.pool.dl20.rcpb.atnistlin}{37/37}
\setval{a2.agree.top150.dl19.rcpa.atnistlin}{11/13}
\setval{a2.agree.top150.dl19.rcpb.atnistlin}{11/13}
\setval{a2.agree.top150.dl20.rcpa.atnistlin}{36/36}
\setval{a2.agree.top150.dl20.rcpb.atnistlin}{36/36}
\setval{a2.budget.qwen.pooled.placementsperdoc}{3.93}
\setval{a2.budget.qwen.pooled.windowsperdoc}{12.27}
\setval{a2.budget.qwendl19.placementsperdoc}{3.92}
\setval{a2.budget.qwendl19.windowsperdoc}{11.27}
\setval{a2.budget.qwendl20.placementsperdoc}{3.94}
\setval{a2.budget.qwendl20.windowsperdoc}{13.08}
\setval{a2.budget.stab.b03.docjudgmentshare}{35.0}
\setval{a2.budget.stab.b05.docjudgmentshare}{50.2}
\setval{a2.calraw.dl19.rho.diffrawa}{0.031}
\setval{a2.calraw.dl20.rho.diffrawa}{0.054}
\setval{a2.conc.dl19.oss}{0.859}
\setval{a2.conc.dl19.overall}{0.874}
\setval{a2.conc.dl20.oss}{0.822}
\setval{a2.conc.dl20.overall}{0.828}
\setval{a2.condensed.pool.pooled.rcpa.lin.rerankers.shareunjudged}{0.538}
\setval{a2.condensed.pool.pooled.rcpb.exp.rerankers.shareunjudged}{0.348}
\setval{a2.cost.imputedshareusd.2p00}{34.7}
\setval{a2.cost.reduced.b03.2p00.usdperquery}{1.33}
\setval{a2.cost.reduced.b05.2p00.usdperquery}{1.87}
\setval{a2.cost.usd.2p00}{360.91}
\setval{a2.cost.usdperquery.2p00}{3.72}
\setval{a2.dec.pool.dl19.count}{56}
\setval{a2.dec.pool.dl19.nistexp}{18}
\setval{a2.dec.pool.dl19.nistlin}{14}
\setval{a2.dec.pool.dl19.rcpa}{60}
\setval{a2.dec.pool.dl19.rcpb}{63}
\setval{a2.dec.pool.dl20.count}{64}
\setval{a2.dec.pool.dl20.nistexp}{36}
\setval{a2.dec.pool.dl20.nistlin}{37}
\setval{a2.dec.pool.dl20.rcpa}{70}
\setval{a2.dec.pool.dl20.rcpb}{70}
\setval{a2.dec.pool.pooled.nistlin}{39}
\setval{a2.dec.top150.dl19.count}{55}
\setval{a2.dec.top150.dl19.nistexp}{14}
\setval{a2.dec.top150.dl19.nistlin}{13}
\setval{a2.dec.top150.dl19.rcpa}{60}
\setval{a2.dec.top150.dl19.rcpb}{61}
\setval{a2.dec.top150.dl20.count}{64}
\setval{a2.dec.top150.dl20.nistexp}{35}
\setval{a2.dec.top150.dl20.nistlin}{36}
\setval{a2.dec.top150.dl20.rcpa}{68}
\setval{a2.dec.top150.dl20.rcpb}{69}
\setval{a2.fail.pool.dl19.nistlin.anypct}{7.0}
\setval{a2.fail.pool.dl20.nistlin.anypct}{5.6}
\setval{a2.kendall.pool.dl19.rcpa.nistlin.n14}{0.209}
\setval{a2.kendall.pool.dl19.rcpa.nistlin.n18}{0.477}
\setval{a2.kendall.pool.dl20.rcpa.nistlin.n14}{0.604}
\setval{a2.rho.dl19.ceiling}{0.768}
\setval{a2.rho.dl19.floor}{0.131}
\setval{a2.rho.dl19.pooled}{0.649}
\setval{a2.rho.dl19.share}{84.4}
\setval{a2.rho.dl20.ceiling}{0.743}
\setval{a2.rho.dl20.floor}{0.169}
\setval{a2.rho.dl20.pooled}{0.595}
\setval{a2.rho.dl20.share}{80.2}
\setval{a2.secondjudge.perdoc.ndocs}{27{,}875}
\setval{a2.secondjudge.perdoc.rho.ac}{0.934}
\setval{a2.secondjudge.pool.kendall18avsc}{0.987}
\setval{a2.setup.dl19.judged}{9{,}260}
\setval{a2.setup.dl19.queries}{43}
\setval{a2.setup.dl20.judged}{11{,}386}
\setval{a2.setup.dl20.queries}{54}
\setval{a2.unjudged.pool.nonzeroabove.max}{5.4}
\setval{a2.unjudged.pool.nonzeroabove.min}{3.7}
\setval{a2.unjudged.pool.nonzeroabove.n}{9}
\setval{a2.unjudged.pool.zerank2}{8.9}
\setval{a2.zerank.pool.dl19.zerank2.rank18.nistlin}{14}
\setval{a2.zerank.pool.dl20.zerank2.rank18.nistlin}{9}
\setval{a3.auc.qrel-ge1.diff.rcp-minus-raw-bt}{0.025}
\setval{a3.auc.qrel-ge1.diff.rcp-minus-raw-bt.hi}{0.031}
\setval{a3.auc.qrel-ge1.diff.rcp-minus-raw-bt.lo}{0.019}
\setval{a3.auc.qrel-ge1.raw-bt}{0.892}
\setval{a3.auc.qrel-ge1.rcp-gain}{0.917}
\setval{a3.conc.g1g0}{0.871}
\setval{a3.conc.g2g0}{0.941}
\setval{a3.diag.brier.2pl}{0.0225}
\setval{a3.diag.brier.logreg}{0.0426}
\setval{a3.diag.ece.2pl}{0.013}
\setval{a3.diag.ece.bt-only}{0.332}
\setval{a3.diag.ece.logreg}{0.012}
\setval{a3.diag.pqbrier.2pl-better}{94.9}
\setval{a3.diag.pqbrier.2pl-better.hi}{95.7}
\setval{a3.diag.pqbrier.2pl-better.lo}{94.0}
\setval{a3.diag.q3.meanabs}{0.112}
\setval{a3.diag.q3.pairs-gt-02}{2}
\setval{a3.fail.any}{23.1}
\setval{a3.n.questions}{2{,}419}
\setval{a3.n.rows}{14{,}514}
\setval{a3.n.stageb.placements.perdoc}{6.67}
\setval{a3.rank.lin.zerank1}{5}
\setval{a3.rank.lin.zerank2}{4}
\setval{a3.rank.rcp.zerank1}{2}
\setval{a3.rank.rcp.zerank2}{1}
\setval{a3.sens.lin}{66.3}
\setval{a3.sens.rcp}{88.9}
\setval{a3.tau.lin-rcp}{0.714}
\setval{a3.xlang.rho}{0.962}
\setval{a3.xlang.rho.hi}{0.964}
\setval{a3.xlang.rho.lo}{0.961}
\setval{a4.bright.oss.alpha.mean}{-5.58}
\setval{a4.bright.oss.beta.c1}{-2.68}
\setval{a4.bright.oss.beta.c2}{-1.63}
\setval{a4.bright.oss.beta.c3}{0.31}
\setval{a4.bright.oss.beta.c4}{1.50}
\setval{a4.bright.oss.beta.c5}{2.50}
\setval{a4.bright.oss.gamma.c1}{1.15}
\setval{a4.bright.oss.gamma.c2}{1.13}
\setval{a4.bright.oss.gamma.c3}{0.73}
\setval{a4.bright.oss.gamma.c4}{1.13}
\setval{a4.bright.oss.gamma.c5}{0.86}
\setval{a4.bright.oss.q3.meanabs}{0.111}
\setval{a4.bright.oss.q3.pairs-gt02}{2}
\setval{a4.bright.oss.tau.p5}{1.04}
\setval{a4.bright.oss.tau.p95}{3.17}
\setval{a4.bright.oss.tau.ratio-p95-p5}{3.1}
\setval{a4.bright.queries.calibrated}{1{,}383}
\setval{a4.bright.queries.evaluated}{1{,}366}
\setval{a4.bright.queries.rerankers}{1{,}384}
\setval{a4.bright.qwen.alpha.mean}{-5.22}
\setval{a4.bright.qwen.beta.c1}{-3.22}
\setval{a4.bright.qwen.beta.c2}{-1.83}
\setval{a4.bright.qwen.beta.c3}{0.70}
\setval{a4.bright.qwen.beta.c4}{1.85}
\setval{a4.bright.qwen.beta.c5}{2.50}
\setval{a4.bright.qwen.gamma.c1}{1.11}
\setval{a4.bright.qwen.gamma.c2}{1.03}
\setval{a4.bright.qwen.gamma.c3}{0.73}
\setval{a4.bright.qwen.gamma.c4}{1.27}
\setval{a4.bright.qwen.gamma.c5}{0.86}
\setval{a4.bright.qwen.q3.meanabs}{0.108}
\setval{a4.bright.qwen.q3.pairs-gt02}{2}
\setval{a4.bright.qwen.tau.p5}{0.88}
\setval{a4.bright.qwen.tau.p95}{2.66}
\setval{a4.bright.qwen.tau.ratio-p95-p5}{3.0}
\setval{a4.nano.fm14.any}{45.6}
\setval{a4.nano.gainfn.oss.rho-w-linminmax}{0.859}
\setval{a4.nano.oss.alpha.mean}{-5.94}
\setval{a4.nano.oss.beta.c1}{-2.14}
\setval{a4.nano.oss.beta.c2}{-1.46}
\setval{a4.nano.oss.beta.c3}{-0.25}
\setval{a4.nano.oss.beta.c4}{0.91}
\setval{a4.nano.oss.beta.c5}{2.95}
\setval{a4.nano.oss.gamma.c1}{1.25}
\setval{a4.nano.oss.gamma.c2}{1.15}
\setval{a4.nano.oss.gamma.c3}{0.84}
\setval{a4.nano.oss.gamma.c4}{0.97}
\setval{a4.nano.oss.gamma.c5}{0.79}
\setval{a4.nano.oss.q3.meanabs}{0.151}
\setval{a4.nano.oss.q3.pairs-gt02}{2}
\setval{a4.nano.oss.tau.p5}{0.95}
\setval{a4.nano.oss.tau.p95}{3.18}
\setval{a4.nano.oss.tau.ratio-p95-p5}{3.3}
\setval{a4.nano.qwen.alpha.mean}{-5.64}
\setval{a4.nano.qwen.band.high}{5.1}
\setval{a4.nano.qwen.band.low}{79.3}
\setval{a4.nano.qwen.band.mid}{15.6}
\setval{a4.nano.qwen.beta.c1}{-2.57}
\setval{a4.nano.qwen.beta.c2}{-2.00}
\setval{a4.nano.qwen.beta.c3}{-0.25}
\setval{a4.nano.qwen.beta.c4}{0.71}
\setval{a4.nano.qwen.beta.c5}{4.12}
\setval{a4.nano.qwen.g-beta1}{0.22}
\setval{a4.nano.qwen.g-beta4}{0.69}
\setval{a4.nano.qwen.gamma.c1}{1.17}
\setval{a4.nano.qwen.gamma.c2}{1.05}
\setval{a4.nano.qwen.gamma.c3}{0.94}
\setval{a4.nano.qwen.gamma.c4}{1.16}
\setval{a4.nano.qwen.gamma.c5}{0.69}
\setval{a4.nano.qwen.q3.meanabs}{0.123}
\setval{a4.nano.qwen.q3.pairs-gt02}{3}
\setval{a4.nano.qwen.tau.p5}{0.86}
\setval{a4.nano.qwen.tau.p95}{2.72}
\setval{a4.nano.qwen.tau.ratio-p95-p5}{3.2}
\setval{a4.nano.sens.count-qwen}{60.7}
\setval{a4.nano.sens.gain-linminmax-qwen}{68.6}
\setval{a4.nano.sens.ndcg}{34.1}
\setval{a4.nano.sens.ratio-rcp-ndcg}{1.9}
\setval{a4.nano.sens.rcp-oss}{60.4}
\setval{a4.nano.sens.rcp-qwen}{63.5}
\setval{a4.nano.xjudge.mean-rho}{0.959}
\setval{a4.nano.xjudge.min-rho}{0.916}
\setval{a4.table2.final.zerank-1.mean}{87.1}
\setval{a4.table2.final.zerank-2.mean}{87.2}
\setval{a4.table2.final.zerank-ahead}{11}
\setval{a4.table2.final.zerank-gap}{5.3}
\setval{a56.meta.s1.nano.cases}{264}
\setval{a56.old.s1.elig.differ}{264}
\setval{a56.old.s1.elig.total}{649}
\setval{a56.old.s2.human.rate}{79.4}
\setval{a56.old.s2.human.rate.hi}{85.5}
\setval{a56.old.s2.human.rate.lo}{71.5}
\setval{a56.old.s2.human.rate.n}{126}
\setval{a56.old.s2.len.human.when-theta-longer}{89.4}
\setval{a56.old.s2.len.human.when-theta-shorter}{53.3}
\setval{a56.old.s2.len.human.when-theta-shorter.hi}{69.8}
\setval{a56.old.s2.len.human.when-theta-shorter.lo}{36.1}
\setval{a56.robust.budget.b0p3-k2.dl19.kendall}{0.969}
\setval{a56.robust.budget.b0p3-k2.dl19.max-shift}{0.026}
\setval{a56.robust.budget.b0p5-k2.dl19.kendall}{0.987}
\setval{a56.robust.budget.b0p5-k2.dl19.max-shift}{0.009}
\setval{a56.robust.factorial.cells}{80}
\setval{a56.robust.factorial.decisive-pairs}{76}
\setval{a56.robust.factorial.family.construct-validity.kendall-max}{0.956}
\setval{a56.robust.factorial.family.construct-validity.kendall-min}{0.868}
\setval{a56.robust.factorial.family.construct-validity.reversals}{1}
\setval{a56.robust.factorial.family.count.kendall-max}{0.978}
\setval{a56.robust.factorial.family.count.kendall-min}{0.890}
\setval{a56.robust.factorial.family.count.reversals}{0}
\setval{a56.robust.factorial.family.difficulty-placement.kendall-max}{1.000}
\setval{a56.robust.factorial.family.difficulty-placement.kendall-min}{0.714}
\setval{a56.robust.factorial.family.difficulty-placement.reversals}{19}
\setval{a56.robust.factorial.family.ordering.kendall-max}{1.000}
\setval{a56.robust.factorial.family.ordering.kendall-min}{0.934}
\setval{a56.robust.factorial.family.ordering.reversals}{0}
\setval{a56.robust.factorial.family.redundancy.kendall-max}{1.000}
\setval{a56.robust.factorial.family.redundancy.kendall-min}{0.934}
\setval{a56.robust.factorial.family.redundancy.reversals}{0}
\setval{a56.robust.factorial.family.reference.kendall-max}{0.978}
\setval{a56.robust.factorial.family.reference.kendall-min}{0.956}
\setval{a56.robust.factorial.family.reference.reversals}{0}
\setval{a56.robust.factorial.family.wording.kendall-max}{1.000}
\setval{a56.robust.factorial.family.wording.kendall-min}{0.934}
\setval{a56.robust.factorial.family.wording.reversals}{0}
\setval{a56.robust.factorial.opportunities}{6{,}004}
\setval{a56.robust.factorial.pairs}{91}
\setval{a56.robust.factorial.relevance.arms}{59}
\setval{a56.robust.factorial.relevance.kendall-min}{0.890}
\setval{a56.robust.factorial.reversals}{20}
\setval{a56.robust.factorial.reversals.r01}{17}
\setval{a56.robust.factorial.reversals.r02}{2}
\setval{a56.robust.factorial.reversals.r11}{1}
\setval{a56.robust.glm.accuracy}{69.6}
\setval{a56.robust.glm.adjudicable}{148}
\setval{a56.robust.glm.contested}{582}
\setval{a56.robust.glm.qwen-accuracy}{60.1}
\setval{a56.robust.lattice.k12345.eff-oracle}{0.407}
\setval{a56.robust.lattice.k1345.eff-oracle}{0.446}
\setval{a56.robust.lattice.k4.eff-oracle}{0.297}
\setval{a56.robust.lattice.k5.eff-oracle}{0.073}
\setval{a56.robust.synth.beir-binary.boundary-noise}{1.00}
\setval{a56.robust.synth.oracle-graded.boundary-noise}{0.50}
\setval{a56.robust.synth.queries}{50}
\setval{a56.robust.synth.replications}{200}
\setval{a56.xs.bright397.beta-monotone}{yes}
\setval{a56.xs.bright397.eff.oracle}{0.70}
\setval{a56.xs.bright397.eff.system}{0.47}
\setval{a56.xs.bright397.gain.region.ge-b4}{0.9}
\setval{a56.xs.bright397.gain.region.lt-b1}{73.1}
\setval{a56.xs.bright397.gamma-argmax}{C4}
\setval{a56.xs.bright397.gamma-argmax-margin}{0.160}
\setval{a56.xs.bright397.tau.p5}{0.88}
\setval{a56.xs.bright397.tau.p95}{2.66}
\setval{a56.xs.bright397.tau.p95-over-p5}{3.0}
\setval{a56.xs.brightoss.beta-monotone}{yes}
\setval{a56.xs.brightoss.eff.oracle}{0.75}
\setval{a56.xs.brightoss.eff.system}{0.45}
\setval{a56.xs.brightoss.gain.region.ge-b4}{1.0}
\setval{a56.xs.brightoss.gain.region.lt-b1}{83.2}
\setval{a56.xs.brightoss.gamma-argmax}{C1}
\setval{a56.xs.brightoss.gamma-argmax-margin}{0.025}
\setval{a56.xs.brightoss.tau.p5}{1.04}
\setval{a56.xs.brightoss.tau.p95}{3.17}
\setval{a56.xs.brightoss.tau.p95-over-p5}{3.1}
\setval{a56.xs.cala.beta-monotone}{no}
\setval{a56.xs.cala.beta.c5}{7.27}
\setval{a56.xs.cala.dl.eff.oracle}{0.35}
\setval{a56.xs.cala.dl.eff.system}{0.51}
\setval{a56.xs.cala.dl.gain.region.ge-b4}{14.8}
\setval{a56.xs.cala.dl.gain.region.lt-b1}{58.0}
\setval{a56.xs.cala.dl.tau.p5}{1.25}
\setval{a56.xs.cala.dl.tau.p95}{3.35}
\setval{a56.xs.cala.dl.tau.p95-over-p5}{2.7}
\setval{a56.xs.cala.gamma-argmax}{C4}
\setval{a56.xs.cala.gamma-argmax-margin}{0.108}
\setval{a56.xs.calc.beta-monotone}{no}
\setval{a56.xs.calc.eff.oracle}{0.45}
\setval{a56.xs.calc.eff.system}{0.55}
\setval{a56.xs.calc.gain.region.ge-b4}{11.1}
\setval{a56.xs.calc.gain.region.lt-b1}{67.5}
\setval{a56.xs.calc.gamma-argmax}{C4}
\setval{a56.xs.calc.gamma-argmax-margin}{0.070}
\setval{a56.xs.calc.tau.p5}{1.45}
\setval{a56.xs.calc.tau.p95}{3.65}
\setval{a56.xs.calc.tau.p95-over-p5}{2.5}
\setval{a56.xs.nano397.beta-monotone}{yes}
\setval{a56.xs.nano397.eff.oracle}{0.66}
\setval{a56.xs.nano397.eff.system}{0.47}
\setval{a56.xs.nano397.gain.region.ge-b4}{5.1}
\setval{a56.xs.nano397.gain.region.lt-b1}{79.3}
\setval{a56.xs.nano397.gamma-argmax}{C1}
\setval{a56.xs.nano397.gamma-argmax-margin}{0.012}
\setval{a56.xs.nano397.tau.p5}{0.86}
\setval{a56.xs.nano397.tau.p95}{2.72}
\setval{a56.xs.nano397.tau.p95-over-p5}{3.2}
\setval{a56.xs.nanooss.beta-monotone}{yes}
\setval{a56.xs.nanooss.eff.oracle}{0.52}
\setval{a56.xs.nanooss.eff.system}{0.33}
\setval{a56.xs.nanooss.gain.region.ge-b4}{3.8}
\setval{a56.xs.nanooss.gain.region.lt-b1}{86.7}
\setval{a56.xs.nanooss.gamma-argmax}{C1}
\setval{a56.xs.nanooss.gamma-argmax-margin}{0.101}
\setval{a56.xs.nanooss.tau.p5}{0.95}
\setval{a56.xs.nanooss.tau.p95}{3.18}
\setval{a56.xs.nanooss.tau.p95-over-p5}{3.3}
\setval{a56.xs.vidore397.beta-monotone}{yes}
\setval{a56.xs.vidore397.gamma-argmax}{C4}
\setval{a56.xs.vidore397.gamma-argmax-margin}{0.270}
\setval{a56.xs.vidore397.native.eff.oracle}{0.74}
\setval{a56.xs.vidore397.native.eff.system}{0.60}
\setval{a56.xs.vidore397.native.gain.region.ge-b4}{2.2}
\setval{a56.xs.vidore397.native.gain.region.lt-b1}{63.9}
\setval{a56.xs.vidore397.native.tau.p5}{0.77}
\setval{a56.xs.vidore397.native.tau.p95}{2.29}
\setval{a56.xs.vidore397.native.tau.p95-over-p5}{3.0}
\setval{b10.t1.bright.cells-n}{1{,}092}
\setval{b10.t1.dl.cells-n}{182}
\setval{b10.t1.nano.cells-n}{1{,}183}
\setval{b12.tt.bright.rcp-h0.dchg}{12}
\setval{b12.tt.nanomteb.rcp-h0.dchg}{54}
\setval{b12.tt.trecdl.rcp-h0.dchg}{4}
\setval{b12.zp.human.all.p0.g-rcp}{2.9}
\setval{b12.zp.human.all.p0.gap-count}{-3.2}
\setval{b12.zp.human.all.p0.gap-count.ci}{[-5.0, -1.7]}
\setval{b12.zp.human.all.p0.gap-rcp}{-0.2}
\setval{b12.zp.human.all.p0.gap-rcp.ci}{[-1.9, +1.3]}
\setval{b12.zp.human.all.p0.ge2}{3.2}
\setval{b12.zp.trecdl.allfail.share}{45.8}
\setval{b16.nl.1a.cons3.k10.count-id.delta}{+0.051}
\setval{b16.nl.1a.cons3.k10.count-id.winrate}{72.1}
\setval{b16.nl.1a.cons3.k10.count-id.winrate.ci}{[58.1, 83.7]}
\setval{b16.nl.1a.cover.count-id}{90.2}
\setval{b16.nl.1a.cover.rcp}{90.9}
\setval{b16.nl.1a.nist.k10.count-id.delta}{+0.025}
\setval{b16.nl.1a.nist.k10.count-id.winrate}{48.8}
\setval{b16.nl.1a.nist.k10.count-id.winrate.ci}{[33.3, 63.4]}
\setval{b16.nl.1a.p1.k10.count-id.delta}{+0.076}
\setval{b16.nl.1a.p1.k10.count-id.winrate}{71.4}
\setval{b16.nl.1a.p1.k10.count-id.winrate.ci}{[57.1, 83.7]}
\setval{b16.nl.1a.p2.k10.count-id.delta}{+0.059}
\setval{b16.nl.1a.p2.k10.count-id.winrate}{61.9}
\setval{b16.nl.1a.p2.k10.count-id.winrate.ci}{[46.5, 76.2]}
\setval{b16.nl.mp.dl19.tie.cons3.sys-R}{14}
\setval{b16.nl.mp.dl19.tie.cons3.winrate-R}{72.2}
\setval{b16.nl.mp.dl19.tie.nist.sys-R}{5}
\setval{b16.nl.mp.dl19.tie.nist.winrate-R}{62.4}
\setval{b16.nl.mp.dl19.tie.p1.sys-R}{14}
\setval{b16.nl.mp.dl19.tie.p1.winrate-R}{75.1}
\setval{b16.nl.mp.dl19.tie.p2.sys-R}{8}
\setval{b16.nl.mp.dl19.tie.p2.winrate-R}{64.5}
\setval{b16.nl.sat.dl.pathC.s0.tie5C}{15.9}
\setval{b16.nl.sat.dl.pathC.s5.tie5C}{66.8}
\setval{b16.nl.sat.nanomteb.pathC.s5.tie5C}{16.5}
\setval{b16.nl.sat.vidore.pathC.s5.tie5C}{7.8}
\setval{b16.nl.sat.vidore.pathC.s5.tie5R}{1.0}
\setval{b17.excl.reviews}{14}
\setval{b17.with.primary}{72.3}
\setval{b17.without.primary}{72.4}
\setval{b19.mtie.consensus.pooled.both-decisive.n-decided}{112}
\setval{b19.mtie.consensus.pooled.both-decisive.rcp-share}{66.1}
\setval{b19.mtie.consensus.pooled.both-decisive.rcp-share.ci}{[57.1, 75.2]}
\setval{b19.mtie.consensus.pooled.count-only.n-decided}{73}
\setval{b19.mtie.consensus.pooled.count-only.rcp-share}{52.1}
\setval{b19.mtie.consensus.pooled.count-only.rcp-share.ci}{[40.0, 63.8]}
\setval{b19.mtie.consensus.pooled.rcp-only.n-decided}{55}
\setval{b19.mtie.consensus.pooled.rcp-only.rcp-share}{72.7}
\setval{b19.mtie.consensus.pooled.rcp-only.rcp-share.ci}{[60.0, 84.3]}
\setval{b19.mtie.tiezone.bright.percentile}{11.0}
\setval{b19.mtie.tiezone.nanomteb.percentile}{23.5}
\setval{b19.mtie.tiezone.rcp}{0.02}
\setval{b19.mtie.tiezone.trecdl.percentile}{38.6}
\setval{b19.mtie.tiezone.vidore.percentile}{22.2}
\setval{b19.restr.excl-arguana.s1.consensus.pooled.rcp-share}{65.2}
\setval{b19.restr.excl-arguana.s2.consensus.pooled.rcp-share}{87.4}
\setval{b19.restr.excl-arguana.s3.consensus.pooled.rcp-share}{82.2}
\setval{b19.s1.consensus.pooled.n-cases}{480}
\setval{b19.s1.consensus.pooled.n-decided}{317}
\setval{b19.s1.consensus.pooled.rcp-share}{63.1}
\setval{b19.s1.consensus.pooled.rcp-share.ci}{[57.3, 68.5]}
\setval{b19.s1.deepseek-4-1-flash.pooled.n-decided}{354}
\setval{b19.s1.deepseek-4-1-flash.pooled.rcp-share}{60.5}
\setval{b19.s1.glm-5-3-flash.pooled.n-decided}{297}
\setval{b19.s1.glm-5-3-flash.pooled.rcp-share}{61.3}
\setval{b19.s1.kimi-k3.pooled.n-decided}{323}
\setval{b19.s1.kimi-k3.pooled.rcp-share}{66.3}
\setval{b19.s1ctl.consensus.pooled.n-cases}{120}
\setval{b19.s1ctl.consensus.pooled.n-decided}{76}
\setval{b19.s1ctl.consensus.pooled.rcp-share}{75.0}
\setval{b19.s1ctl.consensus.pooled.rcp-share.ci}{[65.3, 84.3]}
\setval{b19.s1ctl.deepseek-4-1-flash.pooled.n-decided}{92}
\setval{b19.s1ctl.deepseek-4-1-flash.pooled.rcp-share}{75.0}
\setval{b19.s1ctl.glm-5-3-flash.pooled.n-decided}{79}
\setval{b19.s1ctl.glm-5-3-flash.pooled.rcp-share}{70.9}
\setval{b19.s1ctl.kimi-k3.pooled.n-decided}{79}
\setval{b19.s1ctl.kimi-k3.pooled.rcp-share}{77.2}
\setval{b19.s2.consensus.pooled.n-cases}{447}
\setval{b19.s2.consensus.pooled.n-decided}{333}
\setval{b19.s2.consensus.pooled.rcp-share}{86.8}
\setval{b19.s2.consensus.pooled.rcp-share.ci}{[83.1, 90.3]}
\setval{b19.s2.deepseek-4-1-flash.bright.n-decided}{82}
\setval{b19.s2.deepseek-4-1-flash.bright.rcp-share}{84.1}
\setval{b19.s2.deepseek-4-1-flash.bright.rcp-share.ci}{[75.9, 92.0]}
\setval{b19.s2.deepseek-4-1-flash.nanomteb.n-decided}{74}
\setval{b19.s2.deepseek-4-1-flash.nanomteb.rcp-share}{86.5}
\setval{b19.s2.deepseek-4-1-flash.nanomteb.rcp-share.ci}{[78.4, 93.4]}
\setval{b19.s2.deepseek-4-1-flash.pooled.n-decided}{352}
\setval{b19.s2.deepseek-4-1-flash.pooled.rcp-share}{86.1}
\setval{b19.s2.deepseek-4-1-flash.trecdl.n-decided}{86}
\setval{b19.s2.deepseek-4-1-flash.trecdl.rcp-share}{94.2}
\setval{b19.s2.deepseek-4-1-flash.trecdl.rcp-share.ci}{[88.6, 98.8]}
\setval{b19.s2.deepseek-4-1-flash.vidore.n-decided}{110}
\setval{b19.s2.deepseek-4-1-flash.vidore.rcp-share}{80.9}
\setval{b19.s2.deepseek-4-1-flash.vidore.rcp-share.ci}{[73.1, 87.9]}
\setval{b19.s2.glm-5-3-flash.bright.n-decided}{73}
\setval{b19.s2.glm-5-3-flash.bright.rcp-share}{80.8}
\setval{b19.s2.glm-5-3-flash.bright.rcp-share.ci}{[71.2, 89.3]}
\setval{b19.s2.glm-5-3-flash.nanomteb.n-decided}{58}
\setval{b19.s2.glm-5-3-flash.nanomteb.rcp-share}{86.2}
\setval{b19.s2.glm-5-3-flash.nanomteb.rcp-share.ci}{[77.2, 94.5]}
\setval{b19.s2.glm-5-3-flash.pooled.n-decided}{306}
\setval{b19.s2.glm-5-3-flash.pooled.rcp-share}{85.3}
\setval{b19.s2.glm-5-3-flash.trecdl.n-decided}{78}
\setval{b19.s2.glm-5-3-flash.trecdl.rcp-share}{93.6}
\setval{b19.s2.glm-5-3-flash.trecdl.rcp-share.ci}{[87.8, 98.7]}
\setval{b19.s2.glm-5-3-flash.vidore.n-decided}{97}
\setval{b19.s2.glm-5-3-flash.vidore.rcp-share}{81.4}
\setval{b19.s2.glm-5-3-flash.vidore.rcp-share.ci}{[73.2, 89.1]}
\setval{b19.s2.kimi-k3.bright.n-decided}{80}
\setval{b19.s2.kimi-k3.bright.rcp-share}{81.2}
\setval{b19.s2.kimi-k3.bright.rcp-share.ci}{[72.5, 89.3]}
\setval{b19.s2.kimi-k3.nanomteb.n-decided}{70}
\setval{b19.s2.kimi-k3.nanomteb.rcp-share}{85.7}
\setval{b19.s2.kimi-k3.nanomteb.rcp-share.ci}{[76.8, 93.2]}
\setval{b19.s2.kimi-k3.pooled.n-decided}{340}
\setval{b19.s2.kimi-k3.pooled.rcp-share}{86.2}
\setval{b19.s2.kimi-k3.trecdl.n-decided}{82}
\setval{b19.s2.kimi-k3.trecdl.rcp-share}{98.8}
\setval{b19.s2.kimi-k3.trecdl.rcp-share.ci}{[96.2, 100.0]}
\setval{b19.s2.kimi-k3.vidore.n-decided}{108}
\setval{b19.s2.kimi-k3.vidore.rcp-share}{80.6}
\setval{b19.s2.kimi-k3.vidore.rcp-share.ci}{[72.8, 87.7]}
\setval{b19.s3.consensus.pooled.n-cases}{100}
\setval{b19.s3.consensus.pooled.n-decided}{50}
\setval{b19.s3.consensus.pooled.rcp-share}{84.0}
\setval{b19.s3.consensus.pooled.rcp-share.ci}{[73.5, 93.5]}
\setval{b19.s3.deepseek-4-1-flash.pooled.n-decided}{71}
\setval{b19.s3.deepseek-4-1-flash.pooled.rcp-share}{76.1}
\setval{b19.s3.glm-5-3-flash.pooled.n-decided}{45}
\setval{b19.s3.glm-5-3-flash.pooled.rcp-share}{86.7}
\setval{b19.s3.kimi-k3.pooled.n-decided}{52}
\setval{b19.s3.kimi-k3.pooled.rcp-share}{80.8}
\setval{b2.bin.all.r4.either-missing.ac1}{0.846}
\setval{b2.bin.all.r4.either-missing.kappa}{0.411}
\setval{b2.bin.all.r4.either-missing.prev}{11.8}
\setval{b2.desc.bright.grades}{2{,}580}
\setval{b2.desc.nano.grades}{2{,}322}
\setval{b2.desc.vidore.grades}{2{,}178}
\setval{b2.grade.all.alpha-interval}{0.571}
\setval{b2.grade.all.alpha-interval.hi}{0.604}
\setval{b2.grade.all.alpha-interval.lo}{0.535}
\setval{b2.grade.all.g5.ac2-quad}{0.630}
\setval{b2.grade.all.g5.agree-within-1}{77.8}
\setval{b2.grade.all.icc1}{0.572}
\setval{b2.grade.all.rel-panel}{0.802}
\setval{b2.grade.all.rel-panel.hi}{0.823}
\setval{b2.grade.all.rel-panel.lo}{0.778}
\setval{b2.loao.alpha-interval.max}{0.581}
\setval{b2.loao.alpha-interval.min}{0.554}
\setval{b2.loao.verdict-ab-agree.max}{73.8}
\setval{b2.loao.verdict-ab-agree.min}{70.2}
\setval{b2.verdict.all.all.ab.agree}{72.1}
\setval{b2.verdict.all.all.ab.kappa}{0.421}
\setval{b2.verdict.all.all.ab.kappa.hi}{0.505}
\setval{b2.verdict.all.all.ab.kappa.lo}{0.335}
\setval{b20.affected.queries}{446}
\setval{b20.new.a4.bright.fm14.any}{30.7}
\setval{b20.new.a4.bright.xjudge.mean-rho}{0.988}
\setval{b20.new.draft.app.g.lb.bright.count.ctxlrr1b}{41.6}
\setval{b20.new.draft.app.g.lb.bright.count.ctxlrr1b.rank}{12}
\setval{b20.new.draft.app.g.lb.bright.count.ctxlrr2b}{40.1}
\setval{b20.new.draft.app.g.lb.bright.count.ctxlrr2b.rank}{13}
\setval{b20.new.draft.app.g.lb.bright.count.ctxlrr6b}{39.8}
\setval{b20.new.draft.app.g.lb.bright.count.ctxlrr6b.rank}{14}
\setval{b20.new.draft.app.g.lb.bright.count.jinarrv3}{52.5}
\setval{b20.new.draft.app.g.lb.bright.count.jinarrv3.rank}{11}
\setval{b20.new.draft.app.g.lb.bright.count.qwen3rr06b}{53.4}
\setval{b20.new.draft.app.g.lb.bright.count.qwen3rr06b.rank}{10}
\setval{b20.new.draft.app.g.lb.bright.count.qwen3rr4b}{66.2}
\setval{b20.new.draft.app.g.lb.bright.count.qwen3rr4b.rank}{6}
\setval{b20.new.draft.app.g.lb.bright.count.qwen3rr8b}{64.4}
\setval{b20.new.draft.app.g.lb.bright.count.qwen3rr8b.rank}{8}
\setval{b20.new.draft.app.g.lb.bright.count.rr4fast}{64.8}
\setval{b20.new.draft.app.g.lb.bright.count.rr4fast.rank}{7}
\setval{b20.new.draft.app.g.lb.bright.count.rr4pro}{68.0}
\setval{b20.new.draft.app.g.lb.bright.count.rr4pro.rank}{5}
\setval{b20.new.draft.app.g.lb.bright.count.voyage25}{70.1}
\setval{b20.new.draft.app.g.lb.bright.count.voyage25.rank}{4}
\setval{b20.new.draft.app.g.lb.bright.count.voyage25lt}{62.1}
\setval{b20.new.draft.app.g.lb.bright.count.voyage25lt.rank}{9}
\setval{b20.new.draft.app.g.lb.bright.count.zerank1}{77.5}
\setval{b20.new.draft.app.g.lb.bright.count.zerank1.rank}{2}
\setval{b20.new.draft.app.g.lb.bright.count.zerank1sm}{72.4}
\setval{b20.new.draft.app.g.lb.bright.count.zerank1sm.rank}{3}
\setval{b20.new.draft.app.g.lb.bright.count.zerank2}{79.5}
\setval{b20.new.draft.app.g.lb.bright.count.zerank2.rank}{1}
\setval{b20.new.draft.app.g.lb.bright.ndcg.ctxlrr1b}{17.9}
\setval{b20.new.draft.app.g.lb.bright.ndcg.ctxlrr1b.rank}{13}
\setval{b20.new.draft.app.g.lb.bright.ndcg.ctxlrr2b}{17.3}
\setval{b20.new.draft.app.g.lb.bright.ndcg.ctxlrr2b.rank}{14}
\setval{b20.new.draft.app.g.lb.bright.ndcg.ctxlrr6b}{18.1}
\setval{b20.new.draft.app.g.lb.bright.ndcg.ctxlrr6b.rank}{12}
\setval{b20.new.draft.app.g.lb.bright.ndcg.jinarrv3}{24.0}
\setval{b20.new.draft.app.g.lb.bright.ndcg.jinarrv3.rank}{10}
\setval{b20.new.draft.app.g.lb.bright.ndcg.qwen3rr06b}{18.7}
\setval{b20.new.draft.app.g.lb.bright.ndcg.qwen3rr06b.rank}{11}
\setval{b20.new.draft.app.g.lb.bright.ndcg.qwen3rr4b}{28.4}
\setval{b20.new.draft.app.g.lb.bright.ndcg.qwen3rr4b.rank}{8}
\setval{b20.new.draft.app.g.lb.bright.ndcg.qwen3rr8b}{27.9}
\setval{b20.new.draft.app.g.lb.bright.ndcg.qwen3rr8b.rank}{9}
\setval{b20.new.draft.app.g.lb.bright.ndcg.rr4fast}{30.1}
\setval{b20.new.draft.app.g.lb.bright.ndcg.rr4fast.rank}{6}
\setval{b20.new.draft.app.g.lb.bright.ndcg.rr4pro}{33.5}
\setval{b20.new.draft.app.g.lb.bright.ndcg.rr4pro.rank}{4}
\setval{b20.new.draft.app.g.lb.bright.ndcg.voyage25}{33.9}
\setval{b20.new.draft.app.g.lb.bright.ndcg.voyage25.rank}{3}
\setval{b20.new.draft.app.g.lb.bright.ndcg.voyage25lt}{28.5}
\setval{b20.new.draft.app.g.lb.bright.ndcg.voyage25lt.rank}{7}
\setval{b20.new.draft.app.g.lb.bright.ndcg.zerank1}{38.3}
\setval{b20.new.draft.app.g.lb.bright.ndcg.zerank1.rank}{1}
\setval{b20.new.draft.app.g.lb.bright.ndcg.zerank1sm}{32.9}
\setval{b20.new.draft.app.g.lb.bright.ndcg.zerank1sm.rank}{5}
\setval{b20.new.draft.app.g.lb.bright.ndcg.zerank2}{37.8}
\setval{b20.new.draft.app.g.lb.bright.ndcg.zerank2.rank}{2}
\setval{b20.new.draft.app.g.lb.bright.oss.ctxlrr1b}{37.8}
\setval{b20.new.draft.app.g.lb.bright.oss.ctxlrr1b.rank}{12}
\setval{b20.new.draft.app.g.lb.bright.oss.ctxlrr2b}{37.5}
\setval{b20.new.draft.app.g.lb.bright.oss.ctxlrr2b.rank}{13}
\setval{b20.new.draft.app.g.lb.bright.oss.ctxlrr6b}{37.3}
\setval{b20.new.draft.app.g.lb.bright.oss.ctxlrr6b.rank}{14}
\setval{b20.new.draft.app.g.lb.bright.oss.jinarrv3}{52.0}
\setval{b20.new.draft.app.g.lb.bright.oss.jinarrv3.rank}{11}
\setval{b20.new.draft.app.g.lb.bright.oss.qwen3rr06b}{52.2}
\setval{b20.new.draft.app.g.lb.bright.oss.qwen3rr06b.rank}{10}
\setval{b20.new.draft.app.g.lb.bright.oss.qwen3rr4b}{66.2}
\setval{b20.new.draft.app.g.lb.bright.oss.qwen3rr4b.rank}{6}
\setval{b20.new.draft.app.g.lb.bright.oss.qwen3rr8b}{65.1}
\setval{b20.new.draft.app.g.lb.bright.oss.qwen3rr8b.rank}{7}
\setval{b20.new.draft.app.g.lb.bright.oss.rr4fast}{65.0}
\setval{b20.new.draft.app.g.lb.bright.oss.rr4fast.rank}{8}
\setval{b20.new.draft.app.g.lb.bright.oss.rr4pro}{70.5}
\setval{b20.new.draft.app.g.lb.bright.oss.rr4pro.rank}{5}
\setval{b20.new.draft.app.g.lb.bright.oss.voyage25}{71.0}
\setval{b20.new.draft.app.g.lb.bright.oss.voyage25.rank}{4}
\setval{b20.new.draft.app.g.lb.bright.oss.voyage25lt}{62.0}
\setval{b20.new.draft.app.g.lb.bright.oss.voyage25lt.rank}{9}
\setval{b20.new.draft.app.g.lb.bright.oss.zerank1}{79.3}
\setval{b20.new.draft.app.g.lb.bright.oss.zerank1.rank}{2}
\setval{b20.new.draft.app.g.lb.bright.oss.zerank1sm}{73.4}
\setval{b20.new.draft.app.g.lb.bright.oss.zerank1sm.rank}{3}
\setval{b20.new.draft.app.g.lb.bright.oss.zerank2}{80.8}
\setval{b20.new.draft.app.g.lb.bright.oss.zerank2.rank}{1}
\setval{b20.new.draft.app.g.lb.bright.rcp.ctxlrr1b}{42.8}
\setval{b20.new.draft.app.g.lb.bright.rcp.ctxlrr1b.rank}{12}
\setval{b20.new.draft.app.g.lb.bright.rcp.ctxlrr2b}{41.7}
\setval{b20.new.draft.app.g.lb.bright.rcp.ctxlrr2b.rank}{13}
\setval{b20.new.draft.app.g.lb.bright.rcp.ctxlrr6b}{41.5}
\setval{b20.new.draft.app.g.lb.bright.rcp.ctxlrr6b.rank}{14}
\setval{b20.new.draft.app.g.lb.bright.rcp.jinarrv3}{54.4}
\setval{b20.new.draft.app.g.lb.bright.rcp.jinarrv3.rank}{11}
\setval{b20.new.draft.app.g.lb.bright.rcp.qwen3rr06b}{54.7}
\setval{b20.new.draft.app.g.lb.bright.rcp.qwen3rr06b.rank}{10}
\setval{b20.new.draft.app.g.lb.bright.rcp.qwen3rr4b}{68.1}
\setval{b20.new.draft.app.g.lb.bright.rcp.qwen3rr4b.rank}{6}
\setval{b20.new.draft.app.g.lb.bright.rcp.qwen3rr8b}{66.5}
\setval{b20.new.draft.app.g.lb.bright.rcp.qwen3rr8b.rank}{7}
\setval{b20.new.draft.app.g.lb.bright.rcp.rr4fast}{66.5}
\setval{b20.new.draft.app.g.lb.bright.rcp.rr4fast.rank}{8}
\setval{b20.new.draft.app.g.lb.bright.rcp.rr4pro}{71.0}
\setval{b20.new.draft.app.g.lb.bright.rcp.rr4pro.rank}{5}
\setval{b20.new.draft.app.g.lb.bright.rcp.voyage25}{72.1}
\setval{b20.new.draft.app.g.lb.bright.rcp.voyage25.rank}{4}
\setval{b20.new.draft.app.g.lb.bright.rcp.voyage25lt}{64.1}
\setval{b20.new.draft.app.g.lb.bright.rcp.voyage25lt.rank}{9}
\setval{b20.new.draft.app.g.lb.bright.rcp.zerank1}{79.7}
\setval{b20.new.draft.app.g.lb.bright.rcp.zerank1.rank}{2}
\setval{b20.new.draft.app.g.lb.bright.rcp.zerank1sm}{73.9}
\setval{b20.new.draft.app.g.lb.bright.rcp.zerank1sm.rank}{3}
\setval{b20.new.draft.app.g.lb.bright.rcp.zerank2}{81.6}
\setval{b20.new.draft.app.g.lb.bright.rcp.zerank2.rank}{1}
\setval{b20.new.draft.app.hi.b14.all.top1diff}{10}
\setval{b20.prev.qwen.stageb-share}{7.1}
\setval{b21.human.arguana.contests-own}{6}
\setval{b21.human.arguana.own-grade-mean}{0.28}
\setval{b21.human.arguana.own-reviews}{18}
\setval{b7.trec.count-kall.calls-total}{113}
\setval{b8.human.docs.all.tielev.acc.tournament}{0.680}
\setval{b8.human.docs.all.tielev.acc.tournament.hi}{0.752}
\setval{b8.human.docs.all.tielev.acc.tournament.lo}{0.608}
\setval{b8.human.docs.all.tielev.n}{256}
\setval{b8.nist.qwen.pooled.lev.budget.mean-pool-docs}{287.4}
\setval{b8.nist.qwen.pooled.tielev.acc.tournament}{0.671}
\setval{b8.nist.qwen.pooled.tielev.acc.tournament.hi}{0.725}
\setval{b8.nist.qwen.pooled.tielev.acc.tournament.lo}{0.620}
\setval{b8.nist.qwen.pooled.tielev.n}{26{,}730}
\setval{b8.vidore.tielev.acc.tournament}{0.660}
\setval{b8.vidore.tielev.acc.tournament.hi}{0.691}
\setval{b8.vidore.tielev.acc.tournament.lo}{0.630}
\setval{b8.vidore.tielev.n}{16{,}187}
\setval{b9.nist.dl19.tie.all.n-pairs}{25{,}194}
\setval{b9.nist.dl19.tie.all.nist-vs-parry-both}{0.700}
\setval{b9.nist.dl19.tie.all.tournament.vs-nist}{0.671}
\setval{b9.nist.dl19.tie.all.tournament.vs-parry1}{0.757}
\setval{draft.2.poolsize}{150}
\setval{draft.2.siglevel}{5}
\setval{draft.3.nano.stagea.windowsperdoc}{14.4}
\setval{draft.3.nano.stageb.placementsperdoc}{6.7}
\setval{draft.3.whycal.opposite}{7}
\setval{draft.4.dl.stagea-calls.dl19}{325}
\setval{draft.4.dl.stagea-calls.dl20}{375}
\setval{draft.4.nano.queries}{649}
\setval{draft.4.pay}{35}
\setval{draft.5.dl.pairs}{91}
\setval{draft.5.margin.decisive}{0.02}
\setval{draft.5.sameann.d-cal-raw}{0.047}
\setval{draft.5.sameann.d-cal-raw.hi}{0.061}
\setval{draft.5.sameann.d-cal-raw.lo}{0.035}
\setval{draft.5.study1.human}{77.9}
\setval{draft.5.study1.human.n}{416}
\setval{draft.5.synth.rcp.type1.max}{6.4}
\setval{draft.5.synth.rcp.type1.min}{3.8}
\setval{draft.5.xlang.count}{0.815}
\setval{draft.5.xlang.count3}{0.888}
\setval{draft.app.ab.boot.b}{2{,}000}
\setval{draft.app.ab.pairs}{91}
\setval{draft.app.ab.pigeonhole}{25}
\setval{draft.app.ckl.bright.stagea.windowsperdoc}{14.4}
\setval{draft.app.ckl.bright.stageb.placementsperdoc}{6.7}
\setval{draft.app.ckl.bt.maxiter}{250}
\setval{draft.app.ckl.cal.maxiter}{500}
\setval{draft.app.ckl.metajudge.doccap}{1{,}500}
\setval{draft.app.ckl.nvfp4.rev}{\texttt {0368c1b}}
\setval{draft.app.ckl.params.ctxl1}{1.3}
\setval{draft.app.ckl.params.ctxl2}{2.6}
\setval{draft.app.ckl.params.ctxl6}{6.8}
\setval{draft.app.ckl.params.jina}{0.6}
\setval{draft.app.ckl.params.qwen3rr06}{0.6}
\setval{draft.app.ckl.params.qwen3rr4}{4.0}
\setval{draft.app.ckl.params.qwen3rr8}{8.2}
\setval{draft.app.ckl.params.zr1}{4.0}
\setval{draft.app.ckl.params.zr1s}{1.7}
\setval{draft.app.ckl.params.zr2}{4.0}
\setval{draft.app.ckl.parry.commit}{\texttt {12969fb}}
\setval{draft.app.ckl.rrf-k}{60}
\setval{draft.app.ckl.seed}{42}
\setval{draft.app.ckl.stagea.adaptive-calls}{56}
\setval{draft.app.ckl.stagea.calls}{216}
\setval{draft.app.ckl.stagea.coverage-windows}{80}
\setval{draft.app.ckl.stagea.overlap}{0.3}
\setval{draft.app.ckl.stageb.random-windows}{50}
\setval{draft.app.ckl.stageb.windows}{100}
\setval{draft.app.ckl.temperature}{1.0}
\setval{draft.app.d.band.e4}{0.20}
\setval{draft.app.d.gain.g1}{0.25}
\setval{draft.app.d.gain.g2}{0.45}
\setval{draft.app.d.gain.g2p5}{0.575}
\setval{draft.app.d.gain.g3}{0.70}
\setval{draft.app.d.gain.g3p5}{0.85}
\setval{draft.app.d.gain.g4}{1}
\setval{draft.app.d.strata.concordant}{90}
\setval{draft.app.d.strata.signflip}{221}
\setval{draft.app.d.tie-margin}{0.005}
\setval{draft.app.ef.dl.calfit.extra}{500}
\setval{draft.app.ef.dl.calfit.queries}{597}
\setval{draft.app.ef.dl.pooled.rcpa.agree.k}{36}
\setval{draft.app.ef.dl.pooled.rcpa.contra}{1}
\setval{draft.app.ef.dl20.kendall18}{0.725}
\setval{draft.app.ef.dl20.rho.diffrawa.hi}{0.088}
\setval{draft.app.ef.dl20.rho.diffrawa.lo}{0.015}
\setval{draft.app.ef.vd.xlang.cnt3.diff}{0.074}
\setval{draft.app.ef.vd.xlang.cnt3.diff.hi}{0.078}
\setval{draft.app.ef.vd.xlang.cnt3.diff.lo}{0.070}
\setval{draft.app.g.lb.nano.count.ctxlrr1b}{81.9}
\setval{draft.app.g.lb.nano.count.ctxlrr1b.rank}{12}
\setval{draft.app.g.lb.nano.count.ctxlrr2b}{82.9}
\setval{draft.app.g.lb.nano.count.ctxlrr2b.rank}{10}
\setval{draft.app.g.lb.nano.count.ctxlrr6b}{83.2}
\setval{draft.app.g.lb.nano.count.ctxlrr6b.rank}{8}
\setval{draft.app.g.lb.nano.count.jinarrv3}{70.7}
\setval{draft.app.g.lb.nano.count.jinarrv3.rank}{14}
\setval{draft.app.g.lb.nano.count.qwen3rr06b}{79.8}
\setval{draft.app.g.lb.nano.count.qwen3rr06b.rank}{13}
\setval{draft.app.g.lb.nano.count.qwen3rr4b}{83.4}
\setval{draft.app.g.lb.nano.count.qwen3rr4b.rank}{7}
\setval{draft.app.g.lb.nano.count.qwen3rr8b}{83.0}
\setval{draft.app.g.lb.nano.count.qwen3rr8b.rank}{9}
\setval{draft.app.g.lb.nano.count.rr4fast}{82.4}
\setval{draft.app.g.lb.nano.count.rr4fast.rank}{11}
\setval{draft.app.g.lb.nano.count.rr4pro}{84.2}
\setval{draft.app.g.lb.nano.count.rr4pro.rank}{5}
\setval{draft.app.g.lb.nano.count.voyage25}{85.9}
\setval{draft.app.g.lb.nano.count.voyage25.rank}{4}
\setval{draft.app.g.lb.nano.count.voyage25lt}{83.9}
\setval{draft.app.g.lb.nano.count.voyage25lt.rank}{6}
\setval{draft.app.g.lb.nano.count.zerank1}{88.5}
\setval{draft.app.g.lb.nano.count.zerank1.rank}{2}
\setval{draft.app.g.lb.nano.count.zerank1sm}{86.2}
\setval{draft.app.g.lb.nano.count.zerank1sm.rank}{3}
\setval{draft.app.g.lb.nano.count.zerank2}{88.8}
\setval{draft.app.g.lb.nano.count.zerank2.rank}{1}
\setval{draft.app.g.lb.nano.ndcg.ctxlrr1b}{70.0}
\setval{draft.app.g.lb.nano.ndcg.ctxlrr1b.rank}{9}
\setval{draft.app.g.lb.nano.ndcg.ctxlrr2b}{72.2}
\setval{draft.app.g.lb.nano.ndcg.ctxlrr2b.rank}{6}
\setval{draft.app.g.lb.nano.ndcg.ctxlrr6b}{74.5}
\setval{draft.app.g.lb.nano.ndcg.ctxlrr6b.rank}{2}
\setval{draft.app.g.lb.nano.ndcg.jinarrv3}{69.3}
\setval{draft.app.g.lb.nano.ndcg.jinarrv3.rank}{11}
\setval{draft.app.g.lb.nano.ndcg.qwen3rr06b}{69.9}
\setval{draft.app.g.lb.nano.ndcg.qwen3rr06b.rank}{10}
\setval{draft.app.g.lb.nano.ndcg.qwen3rr4b}{74.1}
\setval{draft.app.g.lb.nano.ndcg.qwen3rr4b.rank}{3}
\setval{draft.app.g.lb.nano.ndcg.qwen3rr8b}{74.7}
\setval{draft.app.g.lb.nano.ndcg.qwen3rr8b.rank}{1}
\setval{draft.app.g.lb.nano.ndcg.rr4fast}{71.7}
\setval{draft.app.g.lb.nano.ndcg.rr4fast.rank}{8}
\setval{draft.app.g.lb.nano.ndcg.rr4pro}{73.3}
\setval{draft.app.g.lb.nano.ndcg.rr4pro.rank}{4}
\setval{draft.app.g.lb.nano.ndcg.voyage25}{72.4}
\setval{draft.app.g.lb.nano.ndcg.voyage25.rank}{5}
\setval{draft.app.g.lb.nano.ndcg.voyage25lt}{71.8}
\setval{draft.app.g.lb.nano.ndcg.voyage25lt.rank}{7}
\setval{draft.app.g.lb.nano.ndcg.zerank1}{67.0}
\setval{draft.app.g.lb.nano.ndcg.zerank1.rank}{13}
\setval{draft.app.g.lb.nano.ndcg.zerank1sm}{67.0}
\setval{draft.app.g.lb.nano.ndcg.zerank1sm.rank}{12}
\setval{draft.app.g.lb.nano.ndcg.zerank2}{66.5}
\setval{draft.app.g.lb.nano.ndcg.zerank2.rank}{14}
\setval{draft.app.g.lb.nano.oss.ctxlrr1b}{80.7}
\setval{draft.app.g.lb.nano.oss.ctxlrr1b.rank}{12}
\setval{draft.app.g.lb.nano.oss.ctxlrr2b}{82.6}
\setval{draft.app.g.lb.nano.oss.ctxlrr2b.rank}{8}
\setval{draft.app.g.lb.nano.oss.ctxlrr6b}{82.8}
\setval{draft.app.g.lb.nano.oss.ctxlrr6b.rank}{7}
\setval{draft.app.g.lb.nano.oss.jinarrv3}{71.6}
\setval{draft.app.g.lb.nano.oss.jinarrv3.rank}{14}
\setval{draft.app.g.lb.nano.oss.qwen3rr06b}{77.6}
\setval{draft.app.g.lb.nano.oss.qwen3rr06b.rank}{13}
\setval{draft.app.g.lb.nano.oss.qwen3rr4b}{82.2}
\setval{draft.app.g.lb.nano.oss.qwen3rr4b.rank}{10}
\setval{draft.app.g.lb.nano.oss.qwen3rr8b}{82.5}
\setval{draft.app.g.lb.nano.oss.qwen3rr8b.rank}{9}
\setval{draft.app.g.lb.nano.oss.rr4fast}{81.6}
\setval{draft.app.g.lb.nano.oss.rr4fast.rank}{11}
\setval{draft.app.g.lb.nano.oss.rr4pro}{83.7}
\setval{draft.app.g.lb.nano.oss.rr4pro.rank}{5}
\setval{draft.app.g.lb.nano.oss.voyage25}{85.6}
\setval{draft.app.g.lb.nano.oss.voyage25.rank}{3}
\setval{draft.app.g.lb.nano.oss.voyage25lt}{83.0}
\setval{draft.app.g.lb.nano.oss.voyage25lt.rank}{6}
\setval{draft.app.g.lb.nano.oss.zerank1}{87.1}
\setval{draft.app.g.lb.nano.oss.zerank1.rank}{1}
\setval{draft.app.g.lb.nano.oss.zerank1sm}{85.0}
\setval{draft.app.g.lb.nano.oss.zerank1sm.rank}{4}
\setval{draft.app.g.lb.nano.oss.zerank2}{87.0}
\setval{draft.app.g.lb.nano.oss.zerank2.rank}{2}
\setval{draft.app.g.lb.nano.rcp.ctxlrr1b}{80.8}
\setval{draft.app.g.lb.nano.rcp.ctxlrr1b.rank}{12}
\setval{draft.app.g.lb.nano.rcp.ctxlrr2b}{82.3}
\setval{draft.app.g.lb.nano.rcp.ctxlrr2b.rank}{10}
\setval{draft.app.g.lb.nano.rcp.ctxlrr6b}{83.1}
\setval{draft.app.g.lb.nano.rcp.ctxlrr6b.rank}{6}
\setval{draft.app.g.lb.nano.rcp.jinarrv3}{70.0}
\setval{draft.app.g.lb.nano.rcp.jinarrv3.rank}{14}
\setval{draft.app.g.lb.nano.rcp.qwen3rr06b}{78.3}
\setval{draft.app.g.lb.nano.rcp.qwen3rr06b.rank}{13}
\setval{draft.app.g.lb.nano.rcp.qwen3rr4b}{82.4}
\setval{draft.app.g.lb.nano.rcp.qwen3rr4b.rank}{8}
\setval{draft.app.g.lb.nano.rcp.qwen3rr8b}{82.4}
\setval{draft.app.g.lb.nano.rcp.qwen3rr8b.rank}{9}
\setval{draft.app.g.lb.nano.rcp.rr4fast}{81.1}
\setval{draft.app.g.lb.nano.rcp.rr4fast.rank}{11}
\setval{draft.app.g.lb.nano.rcp.rr4pro}{83.4}
\setval{draft.app.g.lb.nano.rcp.rr4pro.rank}{5}
\setval{draft.app.g.lb.nano.rcp.voyage25}{85.1}
\setval{draft.app.g.lb.nano.rcp.voyage25.rank}{3}
\setval{draft.app.g.lb.nano.rcp.voyage25lt}{82.5}
\setval{draft.app.g.lb.nano.rcp.voyage25lt.rank}{7}
\setval{draft.app.g.lb.nano.rcp.zerank1}{87.1}
\setval{draft.app.g.lb.nano.rcp.zerank1.rank}{2}
\setval{draft.app.g.lb.nano.rcp.zerank1sm}{84.7}
\setval{draft.app.g.lb.nano.rcp.zerank1sm.rank}{4}
\setval{draft.app.g.lb.nano.rcp.zerank2}{87.2}
\setval{draft.app.g.lb.nano.rcp.zerank2.rank}{1}
\setval{draft.app.g.lodo.bright.oss.onepl}{2.4}
\setval{draft.app.g.lodo.bright.oss.red}{49.6}
\setval{draft.app.g.lodo.bright.qwen.onepl}{2.0}
\setval{draft.app.g.lodo.bright.qwen.red}{46.5}
\setval{draft.app.g.lodo.nano.oss.onepl}{1.3}
\setval{draft.app.g.lodo.nano.oss.red}{64.1}
\setval{draft.app.g.lodo.nano.qwen.onepl}{1.1}
\setval{draft.app.g.lodo.nano.qwen.red}{60.6}
\setval{draft.app.g.s.nano.nq.arguana}{50}
\setval{draft.app.g.t.nano.ndcg.ctxlrr6b.mean}{74.5}
\setval{draft.app.g.t.nano.ndcg.qwen3rr8b.mean}{74.7}
\setval{draft.app.g.t.nano.ndcg.zerank1sm.mean}{67.0}
\setval{draft.app.g.t.nano.ndcg.zerank2.mean}{66.5}
\setval{draft.app.hi.b14.all.ndatasets}{35}
\setval{draft.app.jm.b3.disagree.n}{760}
\setval{draft.app.jm.b3.disagree.rcp}{82.9}
\setval{draft.app.jm.b3.disagree.rcp.hi}{86.2}
\setval{draft.app.jm.b3.disagree.rcp.lo}{79.1}
\setval{draft.app.jm.b3.zerank.M}{$-$7.7}
\setval{draft.app.jm.b3.zerank.M.hi}{$-$1.8}
\setval{draft.app.jm.b3.zerank.M.lo}{$-$13.1}
\setval{draft.app.jm.s1.hotpotqa.human.k}{0}
\setval{draft.app.jm.s1.hotpotqa.human.n}{12}
\setval{draft.app.jm.s1.human.hi}{81.6}
\setval{draft.app.jm.s1.human.lo}{73.7}
\setval{draft.app.jm.s2.human.controls}{20}
\setval{draft.app.jm.s2.human.signflips}{80}
\setval{draft.ff.chance}{53}
\setval{draft.ff.chance.exact}{52.7}
\setval{draft.tab.verdicts.decisive-margin}{0.02}
\setval{n1.counts.pooled.g2}{4{,}167}
\setval{n1.counts.pooled.suspects}{1{,}043}
\setval{n1.fs.auc.pooled}{0.746 [0.719, 0.774]}
\setval{n1.nz.rho.pooled}{0.955 [0.773, 0.982]}
\setval{n1.parry.g3.ge2_pmean.ref}{46.7}
\setval{n1.parry.g3.ge2_pmean.ref.n}{488}
\setval{n1.parry.g3.ge2_pmean.sus}{3.3}
\setval{n1.parry.g3.ge2_pmean.sus.n}{209}
\setval{n1.parry.ge2_pmean.ref}{40.1}
\setval{n1.parry.ge2_pmean.sus}{2.7}
\setval{n1.parry.mean.pmean.ref}{1.493}
\setval{n1.parry.mean.pmean.sus}{0.413}
\setval{n1.rr.auc.pooled}{0.879 [0.865, 0.894]}
\setval{n1.rr.auc01.pooled}{0.295 [0.270, 0.321]}
\setval{n1.rr.rho.pooled}{0.631 [0.433, 0.776]}
\setval{n10.s2.bright.aff.maj}{70.3}
\setval{n10.s2.bright.aff.self-both}{175}
\setval{n10.s2.bright.aff.self-only-rcp}{9}
\setval{n10.s2.bright.aff.signflips}{187}
\setval{n10.s2.bright.all.ctl.maj}{86.3}
\setval{n10.s2.bright.all.maj}{74.4}
\setval{n10.s2.bright.all.maj.hi}{77.9}
\setval{n10.s2.bright.all.maj.k}{467}
\setval{n10.s2.bright.all.maj.lo}{70.7}
\setval{n10.s2.bright.all.maj.n}{628}
\setval{n10.s2.bright.all.unan}{80.2}
\setval{n10.s2.bright.unaff.maj}{75.5}
\setval{n10.s2.nano.controls}{200}
\setval{n10.s2.nano.ctl.maj}{88.5}
\setval{n10.s2.nano.ctl.maj.k}{154}
\setval{n10.s2.nano.ctl.maj.n}{174}
\setval{n10.s2.nano.deepseek}{74.0}
\setval{n10.s2.nano.ds.count}{13}
\setval{n10.s2.nano.ds.hotpotqa.maj.k}{4}
\setval{n10.s2.nano.ds.hotpotqa.maj.n}{14}
\setval{n10.s2.nano.ds.maj-favors-rcp}{12}
\setval{n10.s2.nano.glm}{78.6}
\setval{n10.s2.nano.len.maj.rcp-longer}{76.2}
\setval{n10.s2.nano.len.maj.rcp-shorter}{81.1}
\setval{n10.s2.nano.len.rcp-longer-share}{61.5}
\setval{n10.s2.nano.maj}{78.4}
\setval{n10.s2.nano.maj.hi}{82.3}
\setval{n10.s2.nano.maj.k}{514}
\setval{n10.s2.nano.maj.lo}{74.1}
\setval{n10.s2.nano.maj.n}{656}
\setval{n10.s2.nano.pos.glm.rcp-first}{79.0}
\setval{n10.s2.nano.pos.glm.rcp-second}{69.3}
\setval{n10.s2.nano.signflips}{800}
\setval{n10.s2.nano.unan}{82.5}
\setval{n10.s2.nano.unan.n}{458}
\setval{n11.stagea.rho.max}{0.971}
\setval{n11.stagea.rho.min}{0.968}
\setval{n14.rubric.docs}{20}
\setval{n14.rubric.docs.total}{300}
\setval{n14.rubric.docshare}{0.3}
\setval{n14.rubric.queries}{5}
\setval{n14.rubric.queries.total}{15}
\setval{n15.human.rcp-ideal-fallback}{27}
\setval{n15.human.target-contests}{100}
\setval{n15.nano.scidocs.emptypos}{7}
\setval{n17.ties.cala.dl19.decisive}{25}
\setval{n17.ties.cala.dl20.decisive}{41}
\setval{n18.ties.maxdelta}{0.4}
\setval{n2.doc.conc.nist_pmean}{0.762 [0.702, 0.809]}
\setval{n2.doc.conc.theta.nist}{0.817 [0.747, 0.860]}
\setval{n2.doc.conc.theta.pmean}{0.844 [0.819, 0.866]}
\setval{n2.sys.dec.pmean}{25}
\setval{n20.nano.arguana.selfdocs}{21}
\setval{n7.gap.dl19}{0.048}
\setval{n7.gap.dl19.hi}{0.102}
\setval{n7.gap.dl19.lo}{-0.006}
\setval{n7.gap.dl20}{0.030}
\setval{n7.gap.dl20.hi}{0.071}
\setval{n7.gap.dl20.lo}{-0.012}
\setval{n7.hs.zr.all}{69.3}
\setval{n7.hs.zr.all.hi}{76.6}
\setval{n7.hs.zr.all.lo}{61.9}
\setval{n7.hs.zr.all.n}{153}
\setval{n7.hs.zr.rcpfav}{78.8}
\setval{n7.hs.zr.rcpfav.hi}{88.8}
\setval{n7.hs.zr.rcpfav.lo}{68.8}
\setval{n7.hs.zr.rcpfav.n}{66}
\setval{n7.hs.zr.rcpfav02}{85.3}
\setval{n7.hs.zr.rcpfav02.hi}{92.2}
\setval{n7.hs.zr.rcpfav02.lo}{78.4}
\setval{n7.hs.zr.rcpfav02.n}{102}
\setval{n7.rank.cond.dl19}{9}
\setval{n7.rank.cond.dl20}{5}
\setval{n7.rank.raw.dl19}{14}
\setval{n7.rank.raw.dl20}{9}
\setval{n8.rank.ndcg.zerank1}{13}
\setval{n8.rank.ndcg.zerank2}{14}
\setval{n8.rank.rcp.ctxlrr6b}{6}
\setval{n8.rank.rcp.qwen3rr8b}{9}
\setval{n8.rank.rcp.zerank1}{2}
\setval{n8.rank.rcp.zerank2}{1}
\setval{n9.s1.bright12.maj}{87.8}
\setval{n9.s1.bright12.maj.hi}{89.9}
\setval{n9.s1.bright12.maj.lo}{85.5}
\setval{n9.s1.bright12.maj.n}{843}
\setval{n9.s1.bright12.unan}{89.6}
\setval{n9.s1.bright12.unan.n}{694}
\setval{n9.s1.brightfix.cases}{373}
\setval{n9.s1.brightfix.identical}{72}
\setval{n9.s1.brightfix.maj}{85.8}
\setval{n9.s1.brightfix.maj.n}{344}
\setval{n9.s1.brightfix.nonquery}{11}
\setval{n9.s1.brightfix.orig}{456}
\setval{n9.s1.brightoss.unaff.maj}{84.0}
\setval{n9.s1.brightoss.unaff.maj.n}{514}
\setval{n9.s1.human.unan}{89.6}
\setval{n9.s1.human.unan.hi}{93.7}
\setval{n9.s1.human.unan.k}{120}
\setval{n9.s1.human.unan.lo}{83.2}
\setval{n9.s1.human.unan.n}{134}
\setval{n9.s1.nano.cases}{264}
\setval{n9.s1.nano.ds.count}{13}
\setval{n9.s1.nano.ds.favor}{12}
\setval{n9.s1.nano.hotpot.maj.n}{8}
\setval{n9.s1.nano.maj}{77.9}
\setval{n9.s1.nano.maj.hi}{83.3}
\setval{n9.s1.nano.maj.k}{173}
\setval{n9.s1.nano.maj.lo}{72.1}
\setval{n9.s1.nano.maj.n}{222}
\setval{n9.s1.nano.unan}{82.6}
\setval{n9.s1.nano.unan.n}{167}
\setval{n9.s1.nanooss.maj}{72.2}
\setval{n9.s1.nanooss.maj.n}{216}
\setval{nr.a2.b21.main-wo-own}{72.0}
\setval{nr.a2.b21.main-wo-own.k}{131}
\setval{nr.a2.b21.main-wo-own.n}{182}
\setval{nr.a2.gainmap.binary-ge2.main}{\val {a1a.ctrl.gainmap.binary-ge2.d14}}
\setval{nr.a2.gainmap.binary-ge4.main}{\val {a1a.ctrl.gainmap.binary-ge4.d14}}
\setval{nr.a2.main.bright}{\val {a1a.d14.bright}}
\setval{nr.a2.main.bright.hi}{\val {a1a.d14.bright.hi}}
\setval{nr.a2.main.bright.k}{\val {a1a.d14.bright.k}}
\setval{nr.a2.main.bright.lo}{\val {a1a.d14.bright.lo}}
\setval{nr.a2.main.bright.n}{\val {a1a.d14.bright.n}}
\setval{nr.a2.main.loao.max}{\val {a1a.loao.d14.max}}
\setval{nr.a2.main.loao.min}{\val {a1a.loao.d14.min}}
\setval{nr.a2.main.loao.runs}{\val {a1a.loao.d14.min.n}}
\setval{nr.a2.main.nanomteb}{\val {a1a.d14.nanomteb}}
\setval{nr.a2.main.nanomteb.hi}{\val {a1a.d14.nanomteb.hi}}
\setval{nr.a2.main.nanomteb.k}{\val {a1a.d14.nanomteb.k}}
\setval{nr.a2.main.nanomteb.lo}{\val {a1a.d14.nanomteb.lo}}
\setval{nr.a2.main.nanomteb.n}{\val {a1a.d14.nanomteb.n}}
\setval{nr.a2.main.ndcg-seen}{\val {a1a.d14.ndcg-seen}}
\setval{nr.a2.main.ndcg-seen.hi}{\val {a1a.d14.ndcg-seen.hi}}
\setval{nr.a2.main.ndcg-seen.k}{\val {a1a.d14.ndcg-seen.k}}
\setval{nr.a2.main.ndcg-seen.lo}{\val {a1a.d14.ndcg-seen.lo}}
\setval{nr.a2.main.ndcg-seen.n}{\val {a1a.d14.ndcg-seen.n}}
\setval{nr.a2.main.nn}{\val {a1a.d14.nn}}
\setval{nr.a2.main.nn.hi}{\val {a1a.d14.nn.hi}}
\setval{nr.a2.main.nn.k}{\val {a1a.d14.nn.k}}
\setval{nr.a2.main.nn.lo}{\val {a1a.d14.nn.lo}}
\setval{nr.a2.main.nn.n}{\val {a1a.d14.nn.n}}
\setval{nr.a2.main.vidore}{\val {a1a.d14.vidore}}
\setval{nr.a2.main.vidore.hi}{\val {a1a.d14.vidore.hi}}
\setval{nr.a2.main.vidore.k}{\val {a1a.d14.vidore.k}}
\setval{nr.a2.main.vidore.lo}{\val {a1a.d14.vidore.lo}}
\setval{nr.a2.main.vidore.n}{\val {a1a.d14.vidore.n}}
\setval{nr.a2.neartie.other.k}{\val {a1a.d17.low.ndcg-disagrees.k}}
\setval{nr.a2.neartie.other.median-abs-dndcg}{\val {a1a.d17.marginal.median-abs-dndcg}}
\setval{nr.a2.neartie.other.n}{\val {a1a.d17.low.ndcg-disagrees.n}}
\setval{nr.a2.neartie.same.k}{\val {a1a.d17.low.ndcg-agrees.k}}
\setval{nr.a2.neartie.same.n}{\val {a1a.d17.low.ndcg-agrees.n}}
\setval{nr.a4.within.bright.rest-mean}{0.46}
\setval{nr.a4.within.bright.zer-mean}{0.63}
\setval{nr.a4.within.nano.rest-mean}{0.63}
\setval{nr.a4.within.nano.zer-mean}{0.79}
\setval{nr.a5.b19.cases}{1{,}147}
\setval{nr.a5.synth.ndcg.type1.max}{6.4}
\setval{nr.a5.synth.ndcg.type1.min}{3.4}
\setval{nr.a6.b19.textcap}{2{,}000}
\setval{nr.figB.cal.attr.bt}{33.2}
\setval{nr.figB.cal.attr.calibration}{40.6}
\setval{nr.figB.cal.attr.criterion}{7.6}
\setval{nr.figB.cal.bss.2pl}{0.814}
\setval{nr.figB.cal.ece.2pl}{0.017}
\setval{nr.figB.cal.ece.bt-only}{0.371}
\setval{nr.figB.cal.ece.logreg}{0.007}
\setval{nr.figB.cal.n-triples}{486{,}750}
\setval{nr.figB.cal.oss.bss.2pl}{0.789}
\setval{nr.figB.cal.pq.2pl-better}{97.8}
\setval{nr.integ.main.hotpotqa.k}{\val {a1a.d14.hotpotqa.k}}
\setval{nr.integ.main.hotpotqa.n}{\val {a1a.d14.hotpotqa.n}}
\setval{nr.m2.cost.pooldocs-perquery}{287}
\setval{nr.m3.judge.oss.active}{5.1}
\setval{nr.m3.judge.qwen.active}{17}
\setval{nr.m3.q3.rawr.nano.qwen}{0.60}
\setval{nr.m4.bright.sens.rcp-gt-ndcg.datasets}{12}
\setval{nr.m4.nano.sens.rcp-gt-ndcg.datasets}{11}

\title{Rubric-Calibrated Preferences:\texorpdfstring{\\}{ } Cross-Query Calibration of LLM Judgments\texorpdfstring{\\}{ } via Item Response Theory}

\iclrfinalcopy
\author{%
  Fabian David Schmidt$^{*}$ \quad Donato Crisostomi$^{*\dagger}$ \quad Carlos Lassance \quad Nils Reimers \\
    Cohere \qquad {\small $^{*}$Equal contribution. \quad $^{\dagger}$Work done during an internship at Cohere.}%
}

\begin{document}

\maketitle
\lhead{Preprint}

\begin{abstract}
  Rerankers decide which documents users and LLMs see, yet their standard metric, nDCG, relies on human relevance labels that are costly, sparse, noisy, and discretely graded. As rerankers approach each other in quality, nDCG on these labels therefore increasingly fails to separate them. LLM judges could supply dense labels. Relative judgments within one query tell even close candidates apart, yet their scores share no scale across queries. Absolute grades share one scale but are too coarse to distinguish documents of similar relevance. In this work, we propose Rubric-Calibrated Preferences (\rcp{}), which combine both kinds of judgment. A listwise Bradley--Terry tournament orders each query's documents, and a rubric of yes/no criteria of increasing stringency provides an absolute standard. Item Response Theory (IRT), which scores test-takers based on their answers to common questions, then uses the shared criteria to put all queries' tournament scores on one scale. \rcp{}'s retrieval metric, \rcpndcg{}, replaces nDCG's discrete labels with the resulting calibrated relevance probabilities. Against blind grades from \val{a1a.desc.annotators} external annotators, calibration raises the correlation between a query's mean score and its mean human grade from \val{a1b.cal.qlevel.all.raw} to \val{a1b.cal.qlevel.all.cal}. The probabilities rank a useful document above a non-useful one with probability \val{a1b.auc.ge2.all.gain} (AUC, chance 0.5), versus \val{a1b.auc.ge2.all.qrel} for the benchmark labels. When the annotators' grades prefer one of two rerankers and exactly one metric agrees, that metric is \rcpndcg{} in \val{a1a.primary}\% of \val{a1a.primary.n} comparisons (chance about \val{draft.ff.chance}\%). On TREC-DL, \rcpndcg{} sides with NIST assessors' grades on every reranker pair that these grades separate significantly. \rcpndcg{} also resolves many of nDCG's ties and separates \val{a4.nano.sens.ratio-rcp-ndcg} times as many reranker pairs on NanoBEIR. Rubric calibration thus turns relative LLM judgments into dense relevance labels that are comparable across queries and agree with human judgment.

\end{abstract}

\etocdepthtag.toc{main}
\section{Introduction}
\label{sec:intro}
\begin{wrapfigure}[16]{r}{0.5\textwidth}
  \vspace{-0.75cm}
  \centering
  \includegraphics[width=0.5\textwidth]{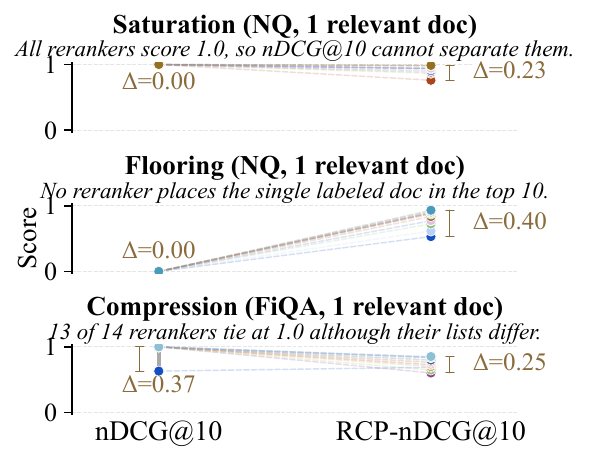}
  \vspace{-0.6cm}
  \caption{On three queries, nDCG@10 saturates, floors, or compresses 14 rerankers, and \rcpndcg{}@10 spreads them.}
  \label{fig:ndcg-failure-modes}
\end{wrapfigure}
Reranking is central to modern retrieval. First-stage retrievers maximize recall, but rerankers select the top documents \citep[\textit{inter alia}]{nogueira2019bert, khattab2020colbert}. Yet the standard reranking metric, nDCG@$k$, is increasingly inadequate. Per query, it saturates when all rerankers score 1, floors when all score 0, and compresses them when half or more of their pairs tie (\cref{fig:ndcg-failure-modes}). For the 14 rerankers we evaluate, one of these failure modes occurs on \val{a4.nano.fm14.any}\% of NanoBEIR and \val{b20.new.a4.bright.fm14.any}\% of BRIGHT queries. These failures stem from human relevance labels (qrels) that are costly, sparse, noisy, and discretely graded \citep{zobel1998reliable, buckley2004retrieval, ban2026completing}, and worsen as rerankers improve \citep{arabzadeh2022shallow}.

Large language models (LLMs) offer a promising alternative~\citep{faggioli2023perspectives}, generating relevance judgments at scale and at a quality competitive with human assessments~\citep{thomas2024searcher, rahmani2025judging}. However, relative judgments such as pair- or listwise comparisons~\citep{sun2023chatgptsearch, yoonacurank} produce fine-grained within-query rankings but no common scale across queries, which rules out graded labels, shared relevance thresholds, and cross-query statistics. Absolute judgments~\citep{zhuang2024beyond, long2025precise} enable cross-query comparison, but their coarse grades leave many documents of a query tied. What is missing is a method that combines the resolution of relative judgments with the cross-query comparability of absolute ones. Our research question is whether a fixed rubric of yes/no criteria can place LLM preferences made within one query on a scale shared across queries. We also test whether the resulting scores agree with human judgment. We study this question in reranking evaluation, where independent human references exist and such scores could replace sparse, discrete, and noisy qrels.

We propose Rubric-Calibrated Preferences (\rcp{}; \cref{fig:pipeline}), a calibration method grounded in Item Response Theory (IRT). IRT scores test-takers on one scale based on their answers to shared questions~\citep{lord1980applications}. In \rcp{}, documents take the test, and the criteria of a rubric are the questions. In Stage~A, a listwise Bradley--Terry tournament~\citep{bradley1952rank} over windows of documents produces per-document scores for a fine-grained within-query ordering. In Stage~B, each document is evaluated against five binary criteria of increasing stringency, yielding an absolute, interpretable relevance signal. An IRT model~\citep{birnbaum1968latent} then uses the shared criteria to place the tournament scores of all queries on one scale. Stage~A sets the order within a query; the rubric sets the scale across queries. In retrieval, \rcp{} labels a benchmark's candidate pool once, and \rcpndcg{} then scores any reranker with the resulting dense gains instead of discrete labels.

\begin{figure*}[t]
  \centering
  \includegraphics[width=\textwidth]{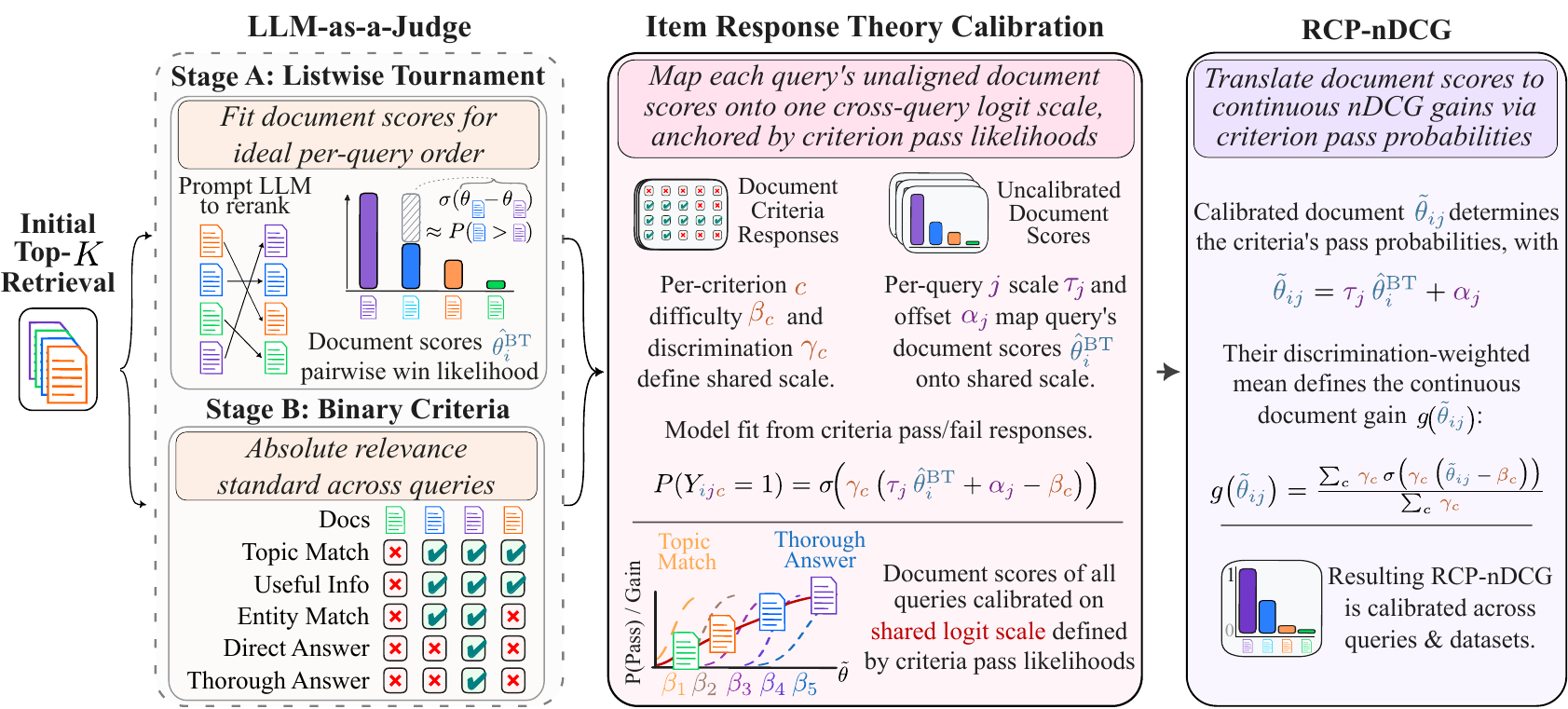}
  \caption{\textup{RCP(-nDCG)}. For each candidate document, \textbf{Stage~A} produces per-query Bradley--Terry scores, and \textbf{Stage~B} records criterion answers. An Item Response Theory \textbf{Calibration} model places the scores of all queries on one scale, which yields the continuous gains of \rcpndcg{} (\cref{sec:method}).}
  \label{fig:pipeline}
  \vspace{-0.6cm}
\end{figure*}

Our contributions are fourfold. \mbox{\textbf{(1)}}~We propose \rcp{}, a calibration method that uses IRT to place within-query LLM preferences on a scale shared across queries. \mbox{\textbf{(2)}}~We derive \rcpndcg{}, a continuous retrieval metric whose gains are calibrated relevance probabilities rather than discrete labels. We evaluate it on 13 NanoBEIR datasets~\citep{zetaalpha2024nanobeir,muennighoff2023mteb}, 12 reasoning-intensive BRIGHT tasks~\citep{su2024bright}, and the eight corpora of ViDoRe v3~\citep{vidore3}, with 14 rerankers from six model families. On NanoBEIR, \rcpndcg{} nearly doubles the share of reranker pairs that a paired $t$-test separates significantly, from \val{a4.nano.sens.ndcg}\% for nDCG to \val{a4.nano.sens.rcp-qwen}\% (\cref{sec:ranking-results}). \mbox{\textbf{(3)}}~We test whether these separations are correct against two human references. First, \val{a1a.desc.annotators} external annotators graded the top-five lists of reranker pairs without seeing labels, scores, or which list a document came from. Against their grades, the AUC\footnote{AUC: the probability that a label source ranks a random document that annotators rate useful above a random one they rate not useful (0.5 is chance, 1 is perfect).} of the \rcp{} gains is \val{a1b.auc.ge2.all.gain}, versus \val{a1b.auc.ge2.all.qrel} for the qrels (\cref{sec:blind-eval}). Where exactly one of nDCG and \rcpndcg{} favors the annotators' preferred reranker, it is \rcpndcg{} in \val{a1a.primary}\% of \val{a1a.primary.n} contests (chance about \val{draft.ff.chance}\%, as nDCG ties more often; \cref{sec:external}). Second, the official grades of the TREC Deep Learning track (\mbox{TREC-DL})~\citep{craswell2020overview, craswell2021overview} come from NIST assessors. With either of two LLM judges, \rcpndcg{} agrees with NIST on every reranker pair that these grades separate significantly. \mbox{\textbf{(4)}}~We release the code and data at \releaseurl{}. Without further judge calls, the calibrated labels score any retrieval model that reorders our candidate pools. The released LLM judgments and the \val{a1a.desc.grades} human grades can benchmark future LLM relevance judges.

\section{Background}
\label{sec:background}
\paragraph{Notation and nDCG.}
We index queries $q_{\qpar{j}}$ by $\qpar{j}$, a query's $K$ candidate documents $d_{\dpar{i}}$ by $\dpar{i}$, and rubric criteria by $\cpar{c} = 1, \dots, \cpar{C}$. For a ranking $\pi$ with $d_{\pi(r)}$ at rank $r$, the discounted cumulative gain (DCG) and its normalized form, nDCG~\citep{jarvelin2002cumulated}, at cutoff $k$ are
\begin{equation}
  \mathrm{DCG@}k(\pi) = \sum\nolimits_{r=1}^{k} G\bigl(d_{\pi(r)}\bigr) / \log_2(r+1),
  \quad
  \mathrm{nDCG@}k(\pi) = \mathrm{DCG@}k(\pi) / \mathrm{DCG@}k(\pi^*).
  \label{eq:ndcg}
\end{equation}
The ideal ranking $\pi^*$ sorts all of the query's documents by decreasing gain. Benchmarks use qrels as the gains $G(d)$, give unlabeled documents gain 0, and average nDCG@$k$ over queries.

\paragraph{Bradley--Terry model.}
The Bradley--Terry (BT) model~\citep{bradley1952rank} gives the probability that document $\dpar{i}$ is preferred to document $\dpar{i'}$ as $P(\dpar{i} \succ \dpar{i'}) = \sigma(\dpar{\theta_i} - \dpar{\theta_{i'}})$, where $\dpar{\theta_i}$ and $\dpar{\theta_{i'}}$ are latent scores and $\sigma(x) = 1/(1+e^{-x})$.
Within a comparison set $\mathcal{I}$, $\dpar{\theta_i}$ is identified only up to an additive constant. Scores from different sets $\mathcal{I} \neq \mathcal{I}'$, such as two queries, are therefore not comparable. Their spread can also differ, because a judge may be more decisive in one set than in another.

\paragraph{Item Response Theory.}
In the two-parameter logistic (2PL) model of IRT~\citep{birnbaum1968latent}, a respondent with ability $\dpar{\theta_i}$, here document $\dpar{i}$, answers an item, here criterion $\cpar{c}$, positively with probability $P(Y_{\dpar{i}\cpar{c}}=1) = \sigma\bigl(\cpar{\gamma_c}(\dpar{\theta_i}-\cpar{\beta_c})\bigr)$. The difficulty $\cpar{\beta_c}$ is the ability at which criterion $\cpar{c}$ is passed with probability one half, and the discrimination $\cpar{\gamma_c}>0$ sets how steeply this pass probability rises with ability. Because the documents of all queries share the criterion parameters, equal abilities imply equal pass probabilities. All abilities thus share one scale. With only $\cpar{C}$ yes/no criteria, many documents receive identical answers and ability estimates. IRT alone thus cannot order them finely.

\section{Methodology}
\label{sec:method}
\subsection{Stage A: Listwise Bradley--Terry Tournament}
\label{sec:stage-a}

Given a query $q_{\qpar{j}}$ and its candidate set $\mathcal{D}_{\qpar{j}} = \{d_1, \dots, d_K\}$, Stage~A repeats listwise LLM reranking~\citep{sun2023chatgptsearch} over windows of $w = 10$ documents. For each window, the LLM judge assigns every document a relevance score $s_{\dpar{i}} \in [-5,\,+5]$, which the prompt describes as a logit scale (\cref{app:prompts}). Only the $w(w-1)/2$ soft preferences $p_{\dpar{ii'}} = \sigma(s_{\dpar{i}} - s_{\dpar{i'}})$ are kept. BT scores $\dpar{\hat{\theta}_i^{\mathrm{BT}}}$ are fit to these preferences~\citep{pipitone2025zelo} by minimizing the pairwise cross-entropy between $p_{\dpar{ii'}}$ and the predicted preferences $\sigma(\dpar{\hat{\theta}_i^{\mathrm{BT}}} - \dpar{\hat{\theta}_{i'}^{\mathrm{BT}}})$. A three-phase window schedule focuses on uncertain ranks and keeps mid-ranked documents away from dominant ones, which would lower their scores (\cref{app:method-stagea}). Stage~A ranks the documents by $\dpar{\hat{\theta}_i^{\mathrm{BT}}}$.

\subsection{Stage B: Pointwise Criterion Evaluation}
\label{sec:stage-b}

The LLM judge reads $w$ documents per window but rates each document on its own against an absolute standard (\cref{app:prompts}). Windows are drawn first at random, then among documents of similar estimated quality. Each document appears in $n_{\dpar{i}\qpar{j}}$ Stage~B windows, and each appearance yields binary responses $Y_{\dpar{i}\qpar{j}\cpar{c}} \in \{0,1\}$ recording whether $d_{\dpar{i}}$ passes criterion $\cpar{c}$ of \cref{tab:criteria} for query $q_{\qpar{j}}$.

\begin{table}[t]
  \centering
  \caption{Stage~B criteria with discrimination $\cpar{\gamma_c}$ and difficulty $\cpar{\beta_c}$ on NanoBEIR (\model{Qwen3.5-397B}).}
  \label{tab:criteria}
  \footnotesize
  \setlength{\tabcolsep}{4pt}
  \begin{tabularx}{\linewidth}{@{}cl X rr@{}}
    \toprule
    \textbf{ID} & \textbf{Criterion} & \textbf{Diagnostic Question} & $\cpar{\gamma_c}$ & $\cpar{\beta_c}$ \\
    \midrule
    C1 & Topical Relevance    & \emph{Is its topic clearly related to the query's information need?} & $\val{a4.nano.qwen.gamma.c1}$ & $\val{a4.nano.qwen.beta.c1}$ \\
    C2 & Information Utility  & \emph{Does it contain specific, useful information to address the query?} & $\val{a4.nano.qwen.gamma.c2}$ & $\val{a4.nano.qwen.beta.c2}$ \\
    C3 & Entity/Detail Match  & \emph{Does it name and discuss the particular entity or detail in the query?} & $\val{a4.nano.qwen.gamma.c3}$ & $\val{a4.nano.qwen.beta.c3}$ \\
    C4 & Direct Answer        & \emph{Does it explicitly and directly answer the query's primary question?} & $\val{a4.nano.qwen.gamma.c4}$ & $\val{a4.nano.qwen.beta.c4}$ \\
    C5 & Thorough Treatment   & \emph{Does it treat the topic in depth rather than superficially?} & $\val{a4.nano.qwen.gamma.c5}$ & $\val{a4.nano.qwen.beta.c5}$ \\
    \bottomrule
  \end{tabularx}
\end{table}

\paragraph{Rubric design.}
The five criteria target general-purpose retrieval, and other applications may need other criteria. In an ablation on TREC-DL, we test the shipped rubric and fifteen variants (\cref{app:robustness-factorial,app:robustness-lattice}). No significantly separated reranker pair reverses unless the rubric ignores the query or puts all criteria at one end of the scale.

\subsection{Calibration: 2PL IRT Merge}
\label{sec:calibration-merge}

Since BT scores and 2PL abilities both lie on logit scales, we map BT scores affinely to 2PL abilities:
\begin{equation}
  P(Y_{\dpar{i}\qpar{j}\cpar{c}} = 1) = \sigma\!\Big(\cpar{\gamma_c} \,\bigl(\qpar{\tau_j} \,\dpar{\hat{\theta}_i^{\mathrm{BT}}} + \qpar{\alpha_j} - \cpar{\beta_c}\bigr)\Big).
  \label{eq:2pl}
\end{equation}
Each criterion's discrimination $\cpar{\gamma_c}$ and difficulty $\cpar{\beta_c}$ are shared across queries. The offset $\qpar{\alpha_j}$ places query $j$'s BT scores on the shared scale, fixing BT's free additive constant (\cref{sec:background}). The scale factor $\qpar{\tau_j} > 0$ sets their spread. The same BT gap can stand for a large relevance difference on one query and a small one on another. Scaling by $\qpar{\tau_j}$ gives a one-logit difference in ability the same relevance meaning on every query (\cref{sec:diagnostics}). The Stage~B answers of the query's documents estimate $\qpar{\tau_j}$ and $\qpar{\alpha_j}$. Shifting or stretching all abilities, with matching changes to $\cpar{\beta_c}$ and $\cpar{\gamma_c}$, leaves every probability in \cref{eq:2pl} unchanged. We therefore fix the unit with $\sum_{\cpar{c}} \cpar{\gamma_c} = \cpar{C}$ and the origin with $\operatorname{mean}_{\cpar{c}} \cpar{\beta_c} = 0$ (\cref{app:theory-identification}). With $\dpar{\hat{\theta}_i^{\mathrm{BT}}}$ held fixed, we fit all other parameters by maximizing their posterior under \cref{eq:2pl} with weak priors, once per suite and judge (\cref{app:method-fit}). The calibrated ability of document $d_{\dpar{i}}$ on query $q_{\qpar{j}}$ is
\begin{equation}
  \dpar{\tilde{\theta}_{ij}} = \qpar{\tau_j} \, \dpar{\hat{\theta}_i^{\mathrm{BT}}} + \qpar{\alpha_j}.
  \label{eq:calibrated-theta}
\end{equation}

\begin{proposition}[Order preservation]\label{prop:order}
  For any rubric and fit, sorting a query's documents by $\dpar{\tilde{\theta}_{ij}}$, or by any strictly increasing function of it, reproduces their Stage~A order (\cref{app:theory-orderscale}).
\end{proposition}

Calibration serves to make scores comparable across queries: a document with $\dpar{\tilde{\theta}_{ij}} = \cpar{\beta_c}$ passes criterion $\cpar{c}$ with probability one half on every query. Relevance labels, thresholds, and pooled statistics require this comparability (\cref{sec:calibration}).

\subsection{\texorpdfstring{\rcpndcg{}}{RCP-nDCG}: A Continuous Relevance Metric}
\label{sec:theta-ndcg}

\rcpndcg{} replaces the qrel gains of \cref{eq:ndcg} with a continuous gain:
\begin{equation}
  g(\dpar{\tilde{\theta}_{ij}}) = \sum\nolimits_{\cpar{c}} \cpar{\gamma_c}\,\sigma\bigl(\cpar{\gamma_c}\,(\dpar{\tilde{\theta}_{ij}} - \cpar{\beta_c})\bigr) \Big/ \sum\nolimits_{\cpar{c}} \cpar{\gamma_c}.
  \label{eq:theta-gain}
\end{equation}
Weighting each criterion's pass probability by its discrimination makes $g$ rise fastest where the criteria measure most precisely: its slope is proportional to the rubric's Fisher information. For any rubric, $g$ rises smoothly from 0 to 1 (\cref{lem:theory-gain}). \rcpndcg{}@$k$ is nDCG@$k$ with gain $G(d_{\dpar{i}}) = g(\dpar{\tilde{\theta}_{ij}})$ instead of qrels, and its ideal ranking is the Stage~A order (\cref{prop:order}). On NanoBEIR, $g$ is \val{a4.nano.qwen.g-beta1} at C1's difficulty and \val{a4.nano.qwen.g-beta4} at C4's. In \cref{fig:iccs}, \val{a4.nano.qwen.band.low}\% of pool documents lie below C1's, \val{a4.nano.qwen.band.mid}\% in between, and \val{a4.nano.qwen.band.high}\% above C4's. The gain thus varies most above C1's.

\paragraph{Count-nDCG.}
Let $d_{\dpar{i}}$ pass criterion $\cpar{c}$ in $S_{\dpar{i}\qpar{j}\cpar{c}}$ of its $n_{\dpar{i}\qpar{j}}$ Stage~B windows for query $q_{\qpar{j}}$. Our baseline Count-nDCG uses the pass share as gain, like rubric-based judges~\citep{farzi2025criteria},
\begin{equation}
  G_{\mathrm{Count}}(d_{\dpar{i}}) = \sum\nolimits_{\cpar{c}} S_{\dpar{i}\qpar{j}\cpar{c}} \big/ \bigl(\cpar{C}\, n_{\dpar{i}\qpar{j}}\bigr).
  \label{eq:count-gain}
\end{equation}
Documents with equal pass shares therefore tie. With equal discriminations and the ability implied by the Stage~B answers alone, \rcpndcg{} reduces to Count-nDCG (\cref{prop:theory-count-identity}).

\section{Experiments}
\label{sec:experiments}
Our experiments ask three questions. Is the 2PL well fitted, and are \rcp{}'s scores calibrated across queries (\cref{sec:diagnostics,sec:calibration})? Do they label documents better than the qrels (\cref{sec:blind-eval})? Does \rcpndcg{} compare rerankers better than nDCG on the benchmarks' qrels, which we call qrel-nDCG (\cref{sec:ranking-results,sec:external})?

\paragraph{Datasets, pool, and rerankers.}
We evaluate on 13 NanoBEIR datasets~\citep{zetaalpha2024nanobeir,muennighoff2023mteb} (\val{draft.4.nano.queries} queries) and 12 BRIGHT tasks~\citep{su2024bright} (\val{a4.bright.queries.evaluated} queries), two English-only suites with binary qrels. As in BRIGHT's official evaluation, we remove each query's excluded documents, and on NanoArguAna the query's own argument, from every scored ranking (\cref{app:nano-bright-scope}). We also use TREC Deep Learning 2019 and 2020 (TREC-DL; \val{a2.setup.dl19.queries} and \val{a2.setup.dl20.queries} queries), whose MS~MARCO passages NIST assessors graded from 0 to 3~\citep{nguyen2016msmarco,craswell2020overview,craswell2021overview}. The eight corpora of ViDoRe v3~\citep{vidore3} (\val{a3.n.questions} questions) hold pages of enterprise, governmental, and educational documents that human annotators graded on a three-point scale. The judge and the rerankers read each page as the benchmark's OCR text (\cref{app:vidore}).
We produce the top-\val{draft.2.poolsize} candidates per query ($K = \val{draft.2.poolsize}$) via Reciprocal Rank Fusion~\citep{cormack2009rrf} of three retrievers: \baseline{Cohere Embed~v4}, \baseline{Octen-Embedding-8B}, and \baseline{BM25}. Positive qrels that the retrievers miss replace the lowest-ranked candidates, to isolate reranking from first-stage recall. On TREC-DL, each pool joins these candidates with every passage NIST graded for the query (\val{nr.m2.cost.pooldocs-perquery} passages on average). The rerankers reorder the whole pool (\cref{app:trecdl-protocol}).
We evaluate 14 rerankers from six families (Cohere, Contextual~AI, Qwen3 Embedding, Voyage~AI, Jina~AI, ZeroEntropy; \cref{app:method-rerankers}).

\begin{wrapfigure}[22]{r}{0.49\textwidth}
  \vspace{-0.4cm}
  \centering
  \includegraphics{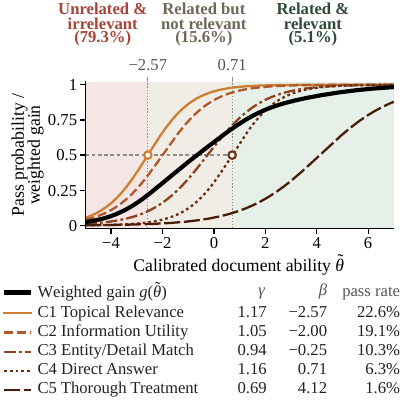}
  \vspace{-0.4cm}
  \caption{Criterion pass probabilities and gain $g$ (NanoBEIR, \model{Qwen3.5-397B}). The legend gives discrimination $\cpar{\gamma}$, difficulty $\cpar{\beta}$, and pass rate. Shading splits at C1's and C4's difficulties.}
  \label{fig:iccs}
\end{wrapfigure}
\paragraph{Judges.}
Both stages use open-weight LLMs, chosen for cost and reproducibility. The primary judge is \model{Qwen3.5-397B}-NVFP4 (\val{nr.m3.judge.qwen.active}B active)~\citep{nvidia2026qwen35nvfp4}. As a second judge, \model{gpt-oss-120b} (\val{nr.m3.judge.oss.active}B active)~\citep{openai2025gptoss} repeats both stages on NanoBEIR, BRIGHT, and TREC-DL (\cref{app:nano-bright-judge,app:trecdl-judge2}). On TREC-DL, the smaller \model{Qwen3.6-27B}~\citep{qwen2026qwen36} replaces the primary judge (\cref{app:trecdl-protocol}).

\paragraph{Human annotations.}
We collected blind relevance grades from \val{a1a.desc.annotators} paid external contractors, each with at least a BSc in the annotated subject (\cref{app:human}). They graded documents from \val{a1a.desc.datasets.nanomteb} NanoBEIR datasets, \val{a1a.desc.datasets.bright} BRIGHT tasks, and the five English corpora of ViDoRe v3. Each contest pairs a query with two rerankers' top-five lists. Annotators grade the union of both lists in random order, without seeing list membership, qrels, or scores. Per document, they answer four questions (Topical, Useful, Answers, Complete), and the most stringent yes sets the grade from 1 to 4, or 0 if none. Answers and Complete also allow Either, which gives a half grade. Most contests are disagreements, on which \rcpndcg{}@10 and qrel-nDCG@10 favor different rerankers over the full rankings, and the rest serve as concordant controls (\cref{app:human-design}). We analyze \val{a1a.desc.contests} contests on \val{a1a.desc.queries} distinct queries, each contest reviewed by three different annotators. Some queries enter two contests with different reranker pairs. A document's label is its mean grade from at least three annotators (\val{a1a.desc.grades} grades, \val{a1a.desc.distinct-docs} documents; \cref{app:human-reliability}). For each annotator, we score both rerankers' top-five lists by nDCG@5, with gains from that annotator's grades (\cref{app:human-instrument}). The reranker with the higher score is the annotator's preferred one. A score gap below \val{draft.app.d.tie-margin} is a tie, which the annotator's rankings of equally graded documents can break. The contest's verdict is the reranker that more of the three annotators prefer. In \val{a1a.desc.verdict-tie} contests, both are preferred equally often (one each and a tie, or three ties), leaving no verdict.

\subsection{2PL Diagnostics}
\label{sec:diagnostics}

We first test how well the 2PL fits the judge's answers to the criteria, since calibration places all queries on one scale only if this fit holds. On NanoBEIR with \model{Qwen3.5-397B}, difficulty increases from C1 to C5. C1 to C4 have near-uniform discriminations, separating documents about equally sharply (\cref{tab:criteria,fig:iccs}). The per-query scale $\qpar{\tau_j}$ varies about threefold, from \val{a4.nano.qwen.tau.p5} to \val{a4.nano.qwen.tau.p95} between its 5th and 95th percentiles, and similarly for \model{gpt-oss-120b} and on BRIGHT (\cref{app:nano-bright-params}). The 2PL's independence assumption is met: once the calibrated score is accounted for, two criteria's answers correlate only weakly (mean absolute correlation \val{a4.nano.qwen.q3.meanabs}, against \val{nr.m3.q3.rawr.nano.qwen} for raw answers; \cref{app:calibration-diagnostics}). Criterion parameters also transfer across datasets. When fitted on all other datasets, they cut the cross-entropy of a held-out dataset's answers by \val{draft.app.g.lodo.nano.qwen.red}\% against per-criterion pass rates (\cref{app:nano-bright-params}). The 2PL thus fits the judge's answers well enough to place all queries on one scale.

\subsection{Calibration Quality}
\label{sec:calibration}

\begin{figure}[t]
  \centering
  \includegraphics[width=\textwidth]{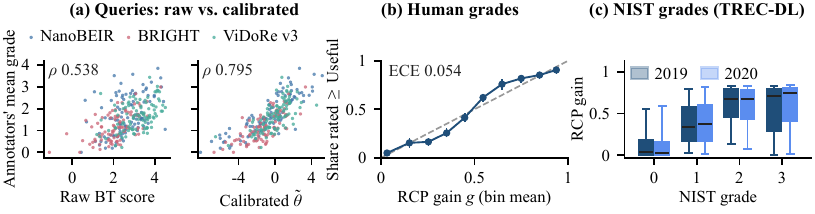}
  \caption{Calibration of \rcp{} against human and NIST grades. (a)~Per human-study query: annotators' mean grade against mean raw BT score (left) and mean calibrated score (right), with Spearman $\rho$. (b)~Documents binned by gain $g$: share of their individual grades at Useful or better against the bin's mean gain (diagonal: perfect calibration). (c)~\rcp{} gain of TREC-DL passages by NIST grade.}
  \label{fig:extcal}
  \vspace{-0.25cm}
\end{figure}

We then test the calibrated scores against three references: the judge's own answers, our human grades, and NIST's grades.

\paragraph{Judge's own passes.}
As \rcpndcg{}'s gains derive from predicted pass probabilities, we check them against the judge's \val{nr.figB.cal.n-triples} (query, document, criterion) answers on NanoBEIR (\cref{app:calibration-link}). The Expected Calibration Error (ECE) averages the gap between predicted probability and observed pass rate over weighted bins (0 is perfect). A sigmoid of the raw BT score, ignoring each query's scale, is miscalibrated (ECE\,=\,\val{nr.figB.cal.ece.bt-only}), while the 2PL is well calibrated (ECE\,=\,\val{nr.figB.cal.ece.2pl}).

\paragraph{Human grades.}
Calibration also aligns the gain with blind human grades (\cref{fig:extcal}b). If $g$ is the probability that one annotator rates a document Useful or better, about 30\% of the grades of documents with $g \approx 0.3$ should be Useful or better. \Cref{fig:extcal}b bins documents by $g$ and plots this share against each bin's mean gain. Against our \val{a1a.desc.grades} human grades, the ECE is \val{a1b.ece.ge2.all} (95\% CI \val{a1b.ece.ge2.all.ci}; \cref{app:calibration-human}). Calibration keeps each query's order and changes only how documents of different queries compare. Across queries, the Spearman correlation of a query's mean score with its annotators' mean grade rises from \val{a1b.cal.qlevel.all.raw} for raw BT to \val{a1b.cal.qlevel.all.cal} after calibration (\cref{fig:extcal}a).

\paragraph{NIST grades.}
On TREC-DL, NIST's graded judgments share no annotations with our study. Within a query, the share of differently graded passage pairs that \rcp{} orders as NIST does is \val{a2.conc.dl19.overall} (2019) and \val{a2.conc.dl20.overall} (2020) (0.5 is chance). This share reflects Stage~A alone (\cref{prop:order}). Over the graded passages of all queries, the Spearman correlation of the calibrated score $\dpar{\tilde{\theta}}$ with the NIST grade is \val{a2.rho.dl19.pooled} and \val{a2.rho.dl20.pooled}. With each query's scores in NIST's grade order, these correlations would be \val{a2.rho.dl19.ceiling} and \val{a2.rho.dl20.ceiling}. The observed values thus reach \val{a2.rho.dl19.share}\% and \val{a2.rho.dl20.share}\% of this ceiling. In 2019, the spread of gains at grade~3 (\cref{fig:extcal}c) comes partly from positives that re-assessors rate low (\cref{app:trecdl-audit}).

\paragraph{Calibration across queries.} In every suite, calibration raises the Spearman correlation between a document's score and its mean grade, pooled over all queries, from \val{a1b.cal.rho.mean.nano.raw}--\val{a1b.cal.rho.mean.vidore.raw} for raw BT to \val{a1b.cal.rho.mean.bright.cal}--\val{a1b.cal.rho.mean.nano.cal} (95\% CIs of the gains above zero; \cref{app:calibration-link}). The calibrated score tracks that grade as closely as a second, independent panel of annotators (Spearman \val{a1b.par.full.all.cal} against \val{a1b.par.balanced.all.sb}; \cref{app:human-labels}).

\subsection{Blind Evaluation of the Labels}
\label{sec:blind-eval}

A metric is only as good as its relevance labels. We therefore test \rcp{}'s labels against the qrels, first on whole sets of documents, then on single documents. On \val{n9.s1.nano.cases} NanoBEIR queries, the set of positive qrel documents differs from the equally large set with the highest \rcp{} gains. By majority, three blind LLM judges from other model families prefer the \rcp{} set on \val{n9.s1.nano.maj}\% of the \val{n9.s1.nano.maj.n} queries they decide, and on \val{n9.s1.bright12.maj}\% on BRIGHT (\cref{app:meta-study1}). Where two IR experts, judging blind, agree (\val{n9.s1.human.unan.n} queries), they prefer the \rcp{} set in \val{n9.s1.human.unan}\%.

The gains also label single documents better than the qrels. Our human study grades the \val{a1b.n.docs.all} documents that two rerankers place in their top five, blind to qrels and scores. We draw one of these documents that annotators rate useful (mean grade at least 2) and one they do not. The \rcp{} gain ranks the useful one higher with probability \val{a1b.auc.ge2.all.gain} (AUC, ties counting half, chance 0.5). The binary qrels do so with probability \val{a1b.auc.ge2.all.qrel}, because they give about half of such pairs the same label. Of these documents, annotators rate \val{a1b.mis.q0-ge2.all}\% of those without a positive qrel as useful and \val{a1b.mis.q1-lt2.all}\% of those with one as not useful (\cref{app:human-labels}). The qrels thus disagree with the annotators in both directions.

\subsection{\texorpdfstring{\rcpndcg{}}{RCP-nDCG} Ranking Results}
\label{sec:ranking-results}

\begin{wrapfigure}[20]{r}{0.52\textwidth}
  \vspace{-0.5cm}
  \centering
  \includegraphics{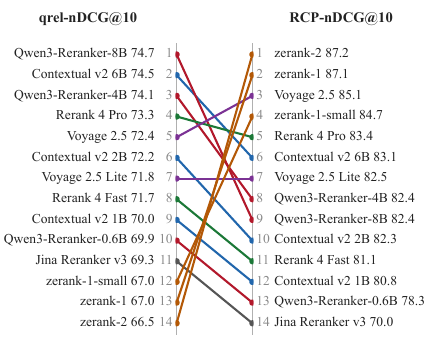}
  \vspace{-0.5cm}
  \caption{14 rerankers on NanoBEIR under qrel-nDCG@10 (left) and \rcpndcg{}@10 (right, \model{Qwen3.5-397B}), with means over datasets, colored by model family. Per-dataset values: \cref{tab:cal-ndcg-qwen-nanomteb}.}
  \label{fig:rankshift}
\end{wrapfigure}
\paragraph{Reranker comparison.}
The largest change from qrel-nDCG@10 to \rcpndcg{}@10 on NanoBEIR is the rise of the ZeroEntropy rerankers (\cref{fig:rankshift}). Under qrel-nDCG@10, \baseline{Qwen3-RR-8B} and \baseline{CTXL-RR-6B} lead with \val{draft.app.g.t.nano.ndcg.qwen3rr8b.mean} and \val{draft.app.g.t.nano.ndcg.ctxlrr6b.mean}, and the three ZeroEntropy models trail (\val{draft.app.g.t.nano.ndcg.zerank2.mean}--\val{draft.app.g.t.nano.ndcg.zerank1sm.mean}). Under \rcpndcg{}@10, \baseline{zerank-2} and \baseline{zerank-1} move from ranks \val{n8.rank.ndcg.zerank2} and \val{n8.rank.ndcg.zerank1} to ranks \val{n8.rank.rcp.zerank2} and \val{n8.rank.rcp.zerank1} (\val{a4.table2.final.zerank-2.mean}, \val{a4.table2.final.zerank-1.mean}), and the two former leaders fall to ranks \val{n8.rank.rcp.qwen3rr8b} and \val{n8.rank.rcp.ctxlrr6b}.

\paragraph{Does the judge favor LLM-trained rerankers?}
The zerank rerankers train on pairwise preferences of an LLM ensemble~\citep{pipitone2025zelo}. A metric built on an LLM judge could thus reward this training signal rather than quality. On TREC-DL, where NIST's grades test this concern, \rcpndcg{}@10 (judge \model{Qwen3.6-27B}) ranks \baseline{zerank-2} first of 18 systems (14 rerankers and four first-stage rankings) in both years. Under qrel-nDCG@10, which uses the NIST grades, \baseline{zerank-2} ranks \val{a2.zerank.pool.dl19.zerank2.rank18.nistlin}th in 2019 and \val{a2.zerank.pool.dl20.zerank2.rank18.nistlin}th in 2020, but its gap to each year's leader is not significant (the 95\% CI includes zero). Part of this rank deficit arises because NIST graded only passages retrieved by the systems submitted to the 2019 and 2020 tracks. Our pools add candidates from current retrievers, and nDCG counts every ungraded passage as irrelevant. Such passages fill \val{a2.unjudged.pool.zerank2}\% of \baseline{zerank-2}'s top-10 slots, against \val{a2.unjudged.pool.nonzeroabove.min}--\val{a2.unjudged.pool.nonzeroabove.max}\% for the \val{a2.unjudged.pool.nonzeroabove.n} systems from other families that NIST ranks above it with both years pooled. Our human study on NanoBEIR, BRIGHT, and ViDoRe v3 grades without any LLM, and annotators never see system names. In \val{n7.hs.zr.rcpfav.n} contests with a verdict, \rcpndcg{} favors the zerank reranker and qrel-nDCG the other. The annotators prefer zerank in \val{n7.hs.zr.rcpfav}\% of them (\cref{app:trecdl-zerank}). NIST's grades thus do not contradict zerank's rise, which our annotators support.

\paragraph{\texorpdfstring{\rcpndcg{}}{RCP-nDCG} is more discriminative.}
A metric separates two rerankers on a dataset if a paired $t$-test on their per-query scores rejects equal means at level 0.05. Its sensitivity~\citep{clarke2020compatibility} is the share of the \val{a56.robust.factorial.pairs} pairs of our 14 rerankers that it separates, averaged over a suite's datasets.\footnote{If no two rerankers differed, the test would still separate about \val{draft.2.siglevel}\% of pairs by chance.} On NanoBEIR, sensitivity rises from \val{a4.nano.sens.ndcg}\% under qrel-nDCG@10 to \val{a4.nano.sens.rcp-qwen}\% under \rcpndcg{}@10 (\cref{fig:ranking}a). \rcpndcg{}@10 also separates more pairs on BRIGHT, on ViDoRe v3, and in both TREC-DL years, where qrel-nDCG@10 uses the NIST grades. The increase holds on most NanoBEIR datasets, on every BRIGHT task, and with \model{gpt-oss-120b} as judge (\cref{app:nano-bright-sensitivity,app:nano-bright-judge}). Sensitivity measures separation, not correctness, which \cref{sec:external} tests.

\begin{figure}[t]
  \centering
  \includegraphics[width=\textwidth]{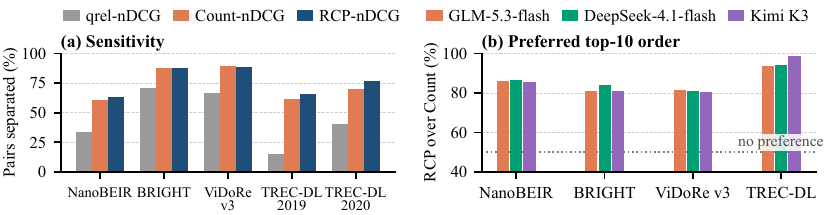}
  \caption{(a)~Sensitivity: share of the \val{a56.robust.factorial.pairs} reranker pairs that each metric separates (paired $t$-test), mean over datasets. (b)~Share of queries on which each judge prefers \rcpndcg{}'s ideal top-10 over the rubric-only Count-nDCG's, among queries where it prefers either list (50\%: no preference).}
  \label{fig:ranking}\label{tab:crosssuite}\label{fig:count-ceiling}
  \vspace{-0.25cm}
\end{figure}
\paragraph{\texorpdfstring{\rcpndcg{}}{RCP-nDCG} vs.\ Count-nDCG.}
Count-nDCG is the rubric-only variant: a document's gain is its share of passed criteria in Stage~B, without the tournament (\cref{eq:count-gain}). Count-nDCG separates about as many reranker pairs as \rcpndcg{} (\cref{fig:ranking}a). The rubric thus supplies most of the added separation (\cref{app:theory-resolution}). Within a query, however, the tournament orders documents that Count-nDCG ties. Each of three blind LLM judges from three model families prefers \rcpndcg{}'s ideal top-10 list over Count-nDCG's on every suite, in \val{b19.s2.kimi-k3.vidore.rcp-share}--\val{b19.s2.kimi-k3.trecdl.rcp-share}\% of queries with a preference (\cref{fig:ranking}b; \cref{sec:count}). Both components are thus needed. The criteria supply the graded signal that the qrels lack (\cref{prop:invisible-swaps}), and the tournament resolves finer differences (\cref{prop:criterion-classes}).

\subsection{Blind External Preference Validation}
\label{sec:human-pref-validation}\label{sec:external}

\begin{figure}[t]
  \centering
  \includegraphics[width=0.48\textwidth]{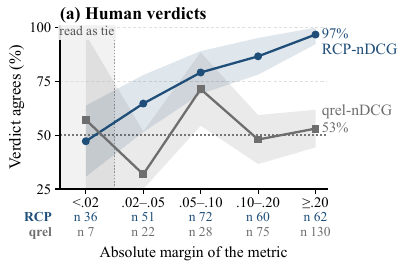}\hfill
  \includegraphics[width=0.48\textwidth]{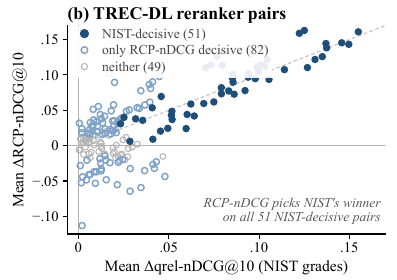}
  \caption{(a)~Share of contests (two rerankers' top-five lists) whose majority verdict sides with a metric's winner, by that metric's absolute margin between the lists (95\% CIs, $n$ per band, chance 50\%). Left of the dotted line: margins below \val{draft.5.margin.decisive}, read as ties. (b)~TREC-DL reranker pairs of both years (judge \model{Qwen3.6-27B}), oriented so that NIST favors the first.}
  \label{fig:dose}
  \vspace{-0.25cm}
\end{figure}
Finally, we test whether \rcpndcg{} prefers the same rerankers as human judges.

\paragraph{Human verdicts.} The annotators' verdicts side with \rcpndcg{}@5 more often than with qrel-nDCG@5. In \val{a1a.primary.n} contests, exactly one metric favors the verdict. That metric is \rcpndcg{} in \val{a1a.primary}\% of them (\val{a1a.primary.k} of \val{a1a.primary.n}; 95\% CI [\val{a1a.primary.lo}, \val{a1a.primary.hi}]\%). Chance is about \val{draft.ff.chance}\%, not 50\%, because qrel-nDCG ties more often and a contest in which one metric ties can count only for the other. The share stays above chance on each suite alone (\cref{app:human-verdicts}) and without the contests of the zerank rerankers, which train on LLM preferences (\val{a1a.d14.nn}\%, against \val{a1a.d14.zerank}\% on their contests).

\paragraph{Metric margins.} A metric's margin is its score difference between the two rerankers. In the human study, only \rcpndcg{}'s margins predict the annotators' verdicts (\cref{fig:dose}a). A metric that tracks the annotators should match the verdict more often as its margin grows. When a metric's margin is \val{draft.app.d.band.e4} or more, the verdict matches its winner in \val{a1a.band.rcp.5}\% of contests for \rcpndcg{} and in \val{a1a.band.ndcg.5}\% for qrel-nDCG. Below a margin of \val{draft.5.margin.decisive}, it matches \rcpndcg{}'s winner in only \val{a1a.sweep.p02.below}\%, near chance. We therefore read such margins as ties.

\paragraph{NIST grades.} \rcpndcg{}@10 agrees with NIST on every reranker pair that NIST's grades separate (\cref{fig:dose}b; \cref{app:trecdl-pairs}). We call a pair NIST-decisive when qrel-nDCG@10 with the NIST grades as linear gains separates it (paired $t$-test). Only \val{a2.dec.pool.dl19.nistlin} of the \val{draft.5.dl.pairs} pairs are NIST-decisive in 2019 and \val{a2.dec.pool.dl20.nistlin} in 2020. \rcpndcg{}@10 favors NIST's winner on each of them with either TREC-DL judge. \rcpndcg{}@10 also separates many more pairs, \val{a2.dec.pool.dl19.rcpa} in 2019 and \val{a2.dec.pool.dl20.rcpa} in 2020 with \model{Qwen3.6-27B}.

In sum, both human references favor the rerankers that \rcpndcg{} favors. When exactly one metric agrees with our annotators, it is \rcpndcg{} in \val{a1a.primary}\% of contests. Only its margins predict their verdicts. It agrees with NIST on every NIST-decisive pair while separating many more pairs. \rcpndcg{} thus compares rerankers better than qrel-nDCG.

\section{Related Work}
\label{sec:related}
\paragraph{Item Response Theory.}
IRT places respondents and test items on one latent scale \citep{rasch1980probabilistic,birnbaum1968latent,samejima1969graded,lord1980applications}. In NLP, it builds evaluation scales, compresses benchmark suites, and guides model search \citep{lalor2016irt,polotinybenchmarks,mencattini2025merge}. These works treat models as respondents and benchmark questions as items. \rcp{} instead treats documents as respondents and rubric criteria as items, which links the queries.

\paragraph{LLMs in retrieval.}
LLMs rerank through listwise and pairwise prompting \citep{sun2023chatgptsearch,qin2024prp}. As judges, they remain sensitive to bias and prompting \citep{liu2023geval,wang2024fair,rahmani2025judging}, and their use as gold labels remains debated \citep{thomas2024searcher,soboroff2025dontuse}. \citet{yan2024consolidating} adjust pointwise LLM labels to respect pairwise preferences. Concurrent work rescales each query's BT scores to a fixed range \citep{georgiev2026sabermath}. \rcp{} instead keeps the tournament order and sets each query's scale and offset from shared criteria (\cref{app:related}).

\section{Conclusion}
\label{sec:conclusion}
Sparse, discrete human labels make nDCG saturate or tie close rerankers. LLM judgments could replace these labels, but relative judgments share no scale across queries, and absolute ones are too coarse. We introduced \rcp{} and its retrieval metric, \rcpndcg{}. \rcp{} fuses relative and rubric-based LLM judgments through Item Response Theory. Against \val{a1a.desc.annotators} annotators' grades, \rcp{}'s gains reach an AUC of \val{a1b.auc.ge2.all.gain}, versus \val{a1b.auc.ge2.all.qrel} for the qrels. When exactly one of \rcpndcg{} and qrel-nDCG agrees with the annotators, it is \rcpndcg{} in \val{a1a.primary}\% of \val{a1a.primary.n} contests (chance about \val{draft.ff.chance}\%). With either TREC-DL judge, \rcpndcg{} favors NIST's winner on every NIST-decisive pair (\cref{sec:external}).

Each component of \rcp{} has a separate role. The criteria define an absolute, interpretable relevance scale and carry most of \rcpndcg{}'s added separation of reranker pairs over qrel-nDCG (\cref{sec:ranking-results}). The IRT merge places all queries on one scale and raises the Spearman correlation of a query's mean score with its annotators' mean grade from \val{a1b.cal.qlevel.all.raw} to \val{a1b.cal.qlevel.all.cal} (\cref{sec:calibration}). The tournament resolves differences too fine for the rubric alone. Three blind LLM judges prefer \rcpndcg{}'s ideal top-10 over the rubric-only Count-nDCG's on every suite, in \val{b19.s2.kimi-k3.vidore.rcp-share}--\val{b19.s2.kimi-k3.trecdl.rcp-share}\% of the queries with a preference (\cref{fig:ranking}b). We therefore suggest that retrieval benchmarks label their candidate pools once with \rcp{} and report \rcpndcg{} next to qrel-nDCG.

\paragraph{Limitations.}
\rcp{}'s gains and our annotators' grades value each document alone and ignore redundancy between documents, as on \dataset{HotpotQA}. Both measure perceived relevance, not factual correctness. The annotators' questions resemble the rubric by design. Relevance is hard to judge consistently, even for NIST assessors~\citep{voorhees2000variations}. Two annotators grade a document within one step of each other in \val{b2.grade.all.g5.agree-within-1}\% of cases. LLM judges may share the evaluated rerankers' preferences. Models fit on \rcp{} signal are best tested on other labels (\cref{app:limitations}).

\paragraph{Beyond retrieval.}
Arenas and LLM verifiers also pool judgments made one prompt at a time~\citep{zheng2023mtbenchjudge,chiang2024chatbot}, and a rubric shared across prompts could calibrate them in the same way. A multidimensional IRT model could split the calibrated score $\dpar{\tilde\theta}$ into dimensions such as correctness and safety, each with its own criteria, and compare candidates by Pareto dominance. As rubric-based RL becomes central to LLM post-training, the calibrated gain could also serve as a continuous reward that rises fastest where the rubric measures most precisely (\cref{lem:theory-gain}).

\section*{AI use statement}
The LLM judges are the object of this study. \model{Qwen3.5-397B}, \model{Qwen3.6-27B}, and \model{gpt-oss-120b} produce the relevance judgments that our method calibrates (\cref{app:method-judges}). Other LLMs, namely \model{GLM-5.3-flash}, \model{DeepSeek-4.1-flash}, \model{Kimi~K3}, and \model{GLM-5.2}, serve only as blind comparison judges and auxiliary checks and never produce a gain (\cref{app:method-judges}). In this work, we used generative AI tools to implement methods, to clean and reformat datasets, and to support data analysis. These tools also helped formulate the mathematical claims of \cref{app:theory} and write their proofs. We discussed our conceptual framework, the research methodology, and the experiments with these tools to improve them. A coding agent also edited the wording of the rubric criteria during rubric development (\cref{app:prompts}). We have not used generative AI tools to devise the conceptual framework, to generate synthetic data sets, to propose or refine hypotheses, or to interpret results. Translation is not applicable to this work. Additionally, we used generative AI tools to create or edit software code and figures, to draft and edit parts of the paper for readability, and to identify and validate references. No LLM annotation replaced the human annotation of the human study. We have reviewed all AI-assisted work and take responsibility for the final content of this work, including text, claims, and artifacts produced with the aid of generative AI.

\section*{Ethics statement}
The human study used \val{a1a.desc.annotators} paid external contractors, who received \$\val{draft.4.pay} per hour and gave informed consent. The study had no formal review by an institutional review board. In the released data, one-way hashes replace the annotators' identities, and free-text comments are omitted (\cref{app:assets}). All datasets are public and used under their stated licenses or terms of use (\cref{app:assets}). The LLM judges may share biases with other LLM-based evaluation systems (\cref{sec:limitations}).

\section*{Reproducibility statement}
The method, the calibration fit, the evaluation protocol, and the judge budgets are specified in \cref{app:method}. \Cref{app:prompts} prints the Stage~A and Stage~B prompts and summarizes the prompts of the comparison judges. The window schedules use a fixed seed (\cref{app:method}).  We release the code of both stages, the calibration, the metric, the prompts, and all judged scores at \releaseurl{}, as \cref{app:assets} details.

\newenvironment{ack}{\section*{Acknowledgments and Disclosure of Funding}}{\par}
\begin{ack}
We thank the annotators of our human study and the annotation team that handled their recruitment, instructions, and quality control. We thank Marius Mosbach for feedback on a draft of the manuscript. Cohere funded this work, including the annotation. All authors were affiliated with Cohere during this work. Cohere sells Rerank~4~Pro and Rerank~4~Fast, two of our 14 rerankers, and Embed~v4, one of our three first-stage retrievers.
\end{ack}

\bibliographystyle{iclr2027_conference}
\bibliography{references}

@article{nogueira2019bert,
  title   = {Passage Re-ranking with {BERT}},
  author  = {Nogueira, Rodrigo and Cho, Kyunghyun},
  journal = {arXiv preprint arXiv:1901.04085},
  year    = {2019},
  url     = {https://arxiv.org/abs/1901.04085}
}

@inproceedings{khattab2020colbert,
  title     = {{ColBERT}: Efficient and Effective Passage Search via Contextualized Late Interaction over {BERT}},
  author    = {Khattab, Omar and Zaharia, Matei},
  booktitle = {Proceedings of the 43rd International ACM SIGIR Conference on Research and Development in Information Retrieval},
  pages     = {39--48},
  year      = {2020},
  publisher = {ACM},
  doi       = {10.1145/3397271.3401075}
}

@inproceedings{sun2023chatgptsearch,
  title     = {Is {ChatGPT} Good at Search? {Investigating} Large Language Models as Re-Ranking Agents},
  author    = {Sun, Weiwei and Yan, Lingyong and Ma, Xinyu and Wang, Shuaiqiang and Ren, Pengjie and Chen, Zhumin and Yin, Dawei and Ren, Zhaochun},
  booktitle = {Proceedings of the 2023 Conference on Empirical Methods in Natural Language Processing},
  pages     = {14918--14937},
  year      = {2023},
  address   = {Singapore},
  publisher = {Association for Computational Linguistics},
  doi       = {10.18653/v1/2023.emnlp-main.923}
}

@inproceedings{qin2024prp,
  title     = {Large Language Models are Effective Text Rankers with Pairwise Ranking Prompting},
  author    = {Qin, Zhen and Jagerman, Rolf and Hui, Kai and Zhuang, Honglei and Wu, Junru and Yan, Le and Shen, Jiaming and Liu, Tianqi and Liu, Jialu and Metzler, Donald and Wang, Xuanhui and Bendersky, Michael},
  booktitle = {Findings of the Association for Computational Linguistics: NAACL 2024},
  pages     = {1504--1518},
  year      = {2024},
  address   = {Mexico City, Mexico},
  publisher = {Association for Computational Linguistics},
  doi       = {10.18653/v1/2024.findings-naacl.97}
}

@inproceedings{zhuang2024beyond,
  title     = {Beyond Yes and No: Improving Zero-Shot {LLM} Rankers via Scoring Fine-Grained Relevance Labels},
  author    = {Zhuang, Honglei and Qin, Zhen and Hui, Kai and Wu, Junru and Yan, Le and Wang, Xuanhui and Bendersky, Michael},
  booktitle = {Proceedings of the 2024 Conference of the North American Chapter of the Association for Computational Linguistics: Human Language Technologies (Volume 2: Short Papers)},
  pages     = {358--370},
  year      = {2024},
  address   = {Mexico City, Mexico},
  publisher = {Association for Computational Linguistics},
  doi       = {10.18653/v1/2024.naacl-short.31}
}

@inproceedings{long2025precise,
  title     = {Precise Zero-Shot Pointwise Ranking with {LLMs} through Post-Aggregated Global Context Information},
  author    = {Long, Kehan and Li, Shasha and Xu, Chen and Tang, Jintao and Wang, Ting},
  booktitle = {Proceedings of the 48th International ACM SIGIR Conference on Research and Development in Information Retrieval},
  pages     = {2384--2394},
  year      = {2025},
  publisher = {ACM},
  doi       = {10.1145/3726302.3730061}
}

@inproceedings{yoonacurank,
  title     = {{AcuRank}: Uncertainty-Aware Adaptive Computation for Listwise Reranking},
  author    = {Yoon, Soyoung and Kim, Gyuwan and Cho, Gyu-Hwung and Hwang, Seung-won},
  booktitle = {Advances in Neural Information Processing Systems},
  volume    = {38},
  year      = {2025},
  doi       = {10.52202/085713-0890},
  pages     = {29944--29973}
}

@article{zhang2025qwen3embedding,
  title   = {{Qwen3} Embedding: Advancing Text Embedding and Reranking Through Foundation Models},
  author  = {Zhang, Yanzhao and Li, Mingxin and Long, Dingkun and Zhang, Xin and Lin, Huan and Yang, Baosong and Xie, Pengjun and Yang, An and Liu, Dayiheng and Lin, Junyang and Huang, Fei and Zhou, Jingren},
  journal = {arXiv preprint arXiv:2506.05176},
  year    = {2025},
  url     = {https://arxiv.org/abs/2506.05176}
}

@article{wang2025jinarerankerv3,
  title   = {{jina-reranker-v3}: Last but Not Late Interaction for Listwise Document Reranking},
  author  = {Wang, Feng and Li, Yuqing and Xiao, Han},
  journal = {arXiv preprint arXiv:2509.25085},
  year    = {2025},
  url     = {https://arxiv.org/abs/2509.25085}
}

@article{pipitone2025zelo,
  title   = {{zELO}: {ELO}-inspired Training Method for Rerankers and Embedding Models},
  author  = {Pipitone, Nicholas and Houir Alami, Ghita and Avadhanam, Advaith and Kaminskyi, Anton and Khoo, Ashley},
  journal = {arXiv preprint arXiv:2509.12541},
  year    = {2025},
  url     = {https://arxiv.org/abs/2509.12541}
}

@misc{contextualai2025rerankv2,
  title        = {Open-Sourcing Reranker v2},
  author       = {{Contextual AI}},
  year         = {2025},
  month        = aug,
  howpublished = {Contextual AI blog},
  url          = {https://contextual.ai/blog/rerank-v2},
  note         = {Published August 27, 2025. Models: \texttt{ctxl-rerank-v2-instruct-multilingual-\{1b,2b,6b\}}}
}

@misc{voyageai2025rerank25,
  title        = {{rerank-2.5 and rerank-2.5-lite: instruction-following rerankers}},
  author       = {{Voyage AI}},
  year         = {2025},
  month        = aug,
  howpublished = {Voyage AI blog},
  url          = {https://blog.voyageai.com/2025/08/11/rerank-2-5/},
  note         = {Published August 11, 2025}
}

@misc{cohere2025rerank4,
  title        = {Introducing {Rerank 4}: {Cohere}'s most powerful reranker yet},
  author       = {{Cohere}},
  year         = {2025},
  month        = dec,
  howpublished = {Cohere blog},
  url          = {https://cohere.com/blog/rerank-4},
  note         = {Published December 11, 2025. Model identifiers \texttt{rerank-v4.0-pro} and \texttt{rerank-v4.0-fast} as listed in the Cohere release notes, \url{https://docs.cohere.com/changelog/rerank-v4.0}}
}

@misc{openai2025gptoss,
  title        = {{gpt-oss-120b \& gpt-oss-20b Model Card}},
  author       = {{OpenAI}},
  year         = {2025},
  howpublished = {arXiv preprint arXiv:2508.10925},
  url          = {https://arxiv.org/abs/2508.10925}
}

@misc{nvidia2026qwen35nvfp4,
  title        = {{Qwen3.5-397B-A17B-NVFP4}},
  author       = {{NVIDIA}},
  year         = {2026},
  month        = feb,
  howpublished = {Hugging Face model card},
  url          = {https://huggingface.co/nvidia/Qwen3.5-397B-A17B-NVFP4},
  note         = {Released February 17, 2026. NVFP4-quantized version of Qwen3.5-397B-A17B}
}

@misc{qwen2026qwen35,
  title        = {{Qwen3.5-397B-A17B}},
  author       = {{Qwen Team}},
  year         = {2026},
  month        = feb,
  howpublished = {Hugging Face model card},
  url          = {https://huggingface.co/Qwen/Qwen3.5-397B-A17B}
}

@misc{qwen2026qwen36,
  title        = {{Qwen3.6-27B}},
  author       = {{Qwen Team}},
  year         = {2026},
  month        = apr,
  howpublished = {Hugging Face model card},
  url          = {https://huggingface.co/Qwen/Qwen3.6-27B},
  note         = {FP8 release: \url{https://huggingface.co/Qwen/Qwen3.6-27B-FP8}}
}

@inproceedings{faggioli2023perspectives,
  title     = {Perspectives on Large Language Models for Relevance Judgment},
  author    = {Faggioli, Guglielmo and Dietz, Laura and Clarke, Charles L. A. and Demartini, Gianluca and Hagen, Matthias and Hauff, Claudia and Kando, Noriko and Kanoulas, Evangelos and Potthast, Martin and Stein, Benno and Wachsmuth, Henning},
  booktitle = {Proceedings of the 2023 ACM SIGIR International Conference on Theory of Information Retrieval},
  pages     = {39--50},
  year      = {2023},
  address   = {Taipei, Taiwan},
  publisher = {ACM},
  doi       = {10.1145/3578337.3605136}
}

@inproceedings{thomas2024searcher,
  title     = {Large Language Models can Accurately Predict Searcher Preferences},
  author    = {Thomas, Paul and Spielman, Seth and Craswell, Nick and Mitra, Bhaskar},
  booktitle = {Proceedings of the 47th International ACM SIGIR Conference on Research and Development in Information Retrieval},
  pages     = {1930--1940},
  year      = {2024},
  publisher = {ACM},
  doi       = {10.1145/3626772.3657707}
}

@article{rahmani2025judging,
  title   = {Judging the Judges: A Collection of {LLM}-Generated Relevance Judgements},
  author  = {Rahmani, Hossein A. and Siro, Clemencia and Aliannejadi, Mohammad and Craswell, Nick and Clarke, Charles L. A. and Faggioli, Guglielmo and Mitra, Bhaskar and Thomas, Paul and Yilmaz, Emine},
  journal = {arXiv preprint arXiv:2502.13908},
  year    = {2025},
  url     = {https://arxiv.org/abs/2502.13908}
}

@article{soboroff2025dontuse,
  title   = {Don't Use {LLMs} to Make Relevance Judgments},
  author  = {Soboroff, Ian},
  journal = {Information Retrieval Research},
  volume  = {1},
  number  = {1},
  pages   = {29--46},
  year    = {2025},
  doi     = {10.54195/irrj.19625}
}

@inproceedings{ban2026completing,
  title     = {Completing Missing Annotation: Multi-Agent Debate for Accurate and Scalable Relevant Assessment for {IR} Benchmarks},
  author    = {Ban, Minjeong and Choi, Jeonghwan and Min, Hyangsuk and Kim, Nicole Hee-Yeon and Kim, Minseok and Lee, Jae-Gil and Song, Hwanjun},
  booktitle = {The Fourteenth International Conference on Learning Representations},
  year      = {2026},
  url       = {https://proceedings.iclr.cc/paper_files/paper/2026/hash/31f527006510a66f6accf141483efa4c-Abstract-Conference.html}
}

@inproceedings{liu2023geval,
  title     = {{G-Eval}: {NLG} Evaluation using {GPT}-4 with Better Human Alignment},
  author    = {Liu, Yang and Iter, Dan and Xu, Yichong and Wang, Shuohang and Xu, Ruochen and Zhu, Chenguang},
  booktitle = {Proceedings of the 2023 Conference on Empirical Methods in Natural Language Processing},
  pages     = {2511--2522},
  year      = {2023},
  address   = {Singapore},
  publisher = {Association for Computational Linguistics},
  doi       = {10.18653/v1/2023.emnlp-main.153}
}

@inproceedings{zheng2023mtbenchjudge,
  title     = {Judging {LLM}-as-a-Judge with {MT-Bench} and {Chatbot Arena}},
  author    = {Zheng, Lianmin and Chiang, Wei-Lin and Sheng, Ying and Zhuang, Siyuan and Wu, Zhanghao and Zhuang, Yonghao and Lin, Zi and Li, Zhuohan and Li, Dacheng and Xing, Eric P. and Zhang, Hao and Gonzalez, Joseph E. and Stoica, Ion},
  booktitle = {Advances in Neural Information Processing Systems (Datasets and Benchmarks Track)},
  volume    = {36},
  pages     = {46595--46623},
  year      = {2023},
  doi       = {10.52202/075280-2020}
}

@inproceedings{wang2024fair,
  title     = {Large Language Models are not Fair Evaluators},
  author    = {Wang, Peiyi and Li, Lei and Chen, Liang and Cai, Zefan and Zhu, Dawei and Lin, Binghuai and Cao, Yunbo and Kong, Lingpeng and Liu, Qi and Liu, Tianyu and Sui, Zhifang},
  booktitle = {Proceedings of the 62nd Annual Meeting of the Association for Computational Linguistics (Volume 1: Long Papers)},
  pages     = {9440--9450},
  year      = {2024},
  address   = {Bangkok, Thailand},
  publisher = {Association for Computational Linguistics},
  doi       = {10.18653/v1/2024.acl-long.511}
}

@inproceedings{chiang2024chatbot,
  title     = {{Chatbot Arena}: An Open Platform for Evaluating {LLMs} by Human Preference},
  author    = {Chiang, Wei-Lin and Zheng, Lianmin and Sheng, Ying and Angelopoulos, Anastasios Nikolas and Li, Tianle and Li, Dacheng and Zhu, Banghua and Zhang, Hao and Jordan, Michael and Gonzalez, Joseph E. and Stoica, Ion},
  booktitle = {Proceedings of the 41st International Conference on Machine Learning},
  series    = {Proceedings of Machine Learning Research},
  volume    = {235},
  pages     = {8359--8388},
  year      = {2024},
  publisher = {PMLR},
  url       = {https://proceedings.mlr.press/v235/chiang24b.html}
}

@inproceedings{ye2024flask,
  title     = {{FLASK}: Fine-grained Language Model Evaluation based on Alignment Skill Sets},
  author    = {Ye, Seonghyeon and Kim, Doyoung and Kim, Sungdong and Hwang, Hyeonbin and Kim, Seungone and Jo, Yongrae and Thorne, James and Kim, Juho and Seo, Minjoon},
  booktitle = {The Twelfth International Conference on Learning Representations},
  year      = {2024},
  url       = {https://proceedings.iclr.cc/paper_files/paper/2024/hash/f41b4a6b202adcd8e150a9d4f124d8f6-Abstract-Conference.html}
}

@inproceedings{liu2024calibrating,
  title     = {Calibrating {LLM}-Based Evaluator},
  author    = {Liu, Yuxuan and Yang, Tianchi and Huang, Shaohan and Zhang, Zihan and Huang, Haizhen and Wei, Furu and Deng, Weiwei and Sun, Feng and Zhang, Qi},
  booktitle = {Proceedings of the 2024 Joint International Conference on Computational Linguistics, Language Resources and Evaluation (LREC-COLING 2024)},
  pages     = {2638--2656},
  year      = {2024},
  address   = {Torino, Italia},
  publisher = {ELRA and ICCL},
  url       = {https://aclanthology.org/2024.lrec-main.237/}
}

@inproceedings{farzi2024pencils,
  title     = {Pencils Down! {Automatic} Rubric-based Evaluation of Retrieve/Generate Systems},
  author    = {Farzi, Naghmeh and Dietz, Laura},
  booktitle = {Proceedings of the 2024 ACM SIGIR International Conference on Theory of Information Retrieval},
  pages     = {175--184},
  year      = {2024},
  publisher = {ACM},
  doi       = {10.1145/3664190.3672511}
}

@inproceedings{farzi2025criteria,
  title     = {Criteria-Based {LLM} Relevance Judgments},
  author    = {Farzi, Naghmeh and Dietz, Laura},
  booktitle = {Proceedings of the 2025 International ACM SIGIR Conference on Innovative Concepts and Theories in Information Retrieval (ICTIR)},
  pages     = {254--263},
  year      = {2025},
  publisher = {ACM},
  doi       = {10.1145/3731120.3744591}
}

@inproceedings{hashemi2024llmrubric,
  title     = {{LLM-Rubric}: A Multidimensional, Calibrated Approach to Automated Evaluation of Natural Language Texts},
  author    = {Hashemi, Helia and Eisner, Jason and Rosset, Corby and Van Durme, Benjamin and Kedzie, Chris},
  booktitle = {Proceedings of the 62nd Annual Meeting of the Association for Computational Linguistics (Volume 1: Long Papers)},
  pages     = {13806--13834},
  year      = {2024},
  address   = {Bangkok, Thailand},
  publisher = {Association for Computational Linguistics},
  doi       = {10.18653/v1/2024.acl-long.745}
}

@inproceedings{lee2025checkeval,
  title     = {{CheckEval}: A Reliable {LLM}-as-a-Judge Framework for Evaluating Text Generation Using Checklists},
  author    = {Lee, Yukyung and Kim, JoongHoon and Kim, Jaehee and Cho, Hyowon and Kang, Jaewook and Kang, Pilsung and Kim, Najoung},
  booktitle = {Proceedings of the 2025 Conference on Empirical Methods in Natural Language Processing},
  year      = {2025},
  address   = {Suzhou, China},
  publisher = {Association for Computational Linguistics},
  doi       = {10.18653/v1/2025.emnlp-main.796},
  pages     = {15771--15798}
}

@inproceedings{yan2024consolidating,
  title     = {Consolidating Ranking and Relevance Predictions of Large Language Models through Post-Processing},
  author    = {Yan, Le and Qin, Zhen and Zhuang, Honglei and Jagerman, Rolf and Wang, Xuanhui and Bendersky, Michael and Oosterhuis, Harrie},
  booktitle = {Proceedings of the 2024 Conference on Empirical Methods in Natural Language Processing},
  pages     = {410--423},
  year      = {2024},
  address   = {Miami, Florida, USA},
  publisher = {Association for Computational Linguistics},
  doi       = {10.18653/v1/2024.emnlp-main.25}
}

@article{georgiev2026sabermath,
  title   = {{SABER-Math}: Automated Benchmark for Information Retrieval Evaluation in Mathematics},
  author  = {Georgiev, Nikolay and Drencheva, Maria and Ibragimova, Kseniia and Petrov, Ivo and Dimitrov, Dimitar I. and Vechev, Martin},
  journal = {arXiv preprint arXiv:2606.29894},
  year    = {2026},
  note    = {Accepted at EMNLP 2026},
  url     = {https://arxiv.org/abs/2606.29894}
}

@inproceedings{takehi2025lara,
  title     = {{LLM}-Assisted Relevance Assessments: When Should We Ask {LLMs} for Help?},
  author    = {Takehi, Rikiya and Voorhees, Ellen M. and Sakai, Tetsuya and Soboroff, Ian},
  booktitle = {Proceedings of the 48th International ACM SIGIR Conference on Research and Development in Information Retrieval},
  pages     = {95--105},
  year      = {2025},
  publisher = {ACM},
  doi       = {10.1145/3726302.3729916}
}

@article{jarvelin2002cumulated,
  title   = {Cumulated Gain-Based Evaluation of {IR} Techniques},
  author  = {J{\"a}rvelin, Kalervo and Kek{\"a}l{\"a}inen, Jaana},
  journal = {ACM Transactions on Information Systems},
  volume  = {20},
  number  = {4},
  pages   = {422--446},
  year    = {2002},
  doi     = {10.1145/582415.582418}
}

@article{ndcg,
  title   = {Cumulated Gain-Based Evaluation of {IR} Techniques},
  author  = {J{\"a}rvelin, Kalervo and Kek{\"a}l{\"a}inen, Jaana},
  journal = {ACM Transactions on Information Systems},
  volume  = {20},
  number  = {4},
  pages   = {422--446},
  year    = {2002},
  doi     = {10.1145/582415.582418}
}

@inproceedings{zobel1998reliable,
  title     = {How Reliable Are the Results of Large-Scale Information Retrieval Experiments?},
  author    = {Zobel, Justin},
  booktitle = {Proceedings of the 21st Annual International ACM SIGIR Conference on Research and Development in Information Retrieval},
  pages     = {307--314},
  year      = {1998},
  publisher = {ACM},
  doi       = {10.1145/290941.291014}
}

@inproceedings{buckley2004retrieval,
  title     = {Retrieval Evaluation with Incomplete Information},
  author    = {Buckley, Chris and Voorhees, Ellen M.},
  booktitle = {Proceedings of the 27th Annual International ACM SIGIR Conference on Research and Development in Information Retrieval},
  pages     = {25--32},
  year      = {2004},
  publisher = {ACM},
  doi       = {10.1145/1008992.1009000}
}

@inproceedings{sakai2006bootstrap,
  title     = {Evaluating Evaluation Metrics Based on the Bootstrap},
  author    = {Sakai, Tetsuya},
  booktitle = {Proceedings of the 29th Annual International ACM SIGIR Conference on Research and Development in Information Retrieval},
  pages     = {525--532},
  year      = {2006},
  publisher = {ACM},
  doi       = {10.1145/1148170.1148261}
}

@inproceedings{sakai2007alternatives,
  title     = {Alternatives to {Bpref}},
  author    = {Sakai, Tetsuya},
  booktitle = {Proceedings of the 30th Annual International ACM SIGIR Conference on Research and Development in Information Retrieval},
  pages     = {71--78},
  year      = {2007},
  publisher = {ACM},
  doi       = {10.1145/1277741.1277756}
}

@article{clarke2020compatibility,
  title   = {Assessing Top-$k$ Preferences},
  author  = {Clarke, Charles L. A. and Vtyurina, Alexandra and Smucker, Mark D.},
  journal = {ACM Transactions on Information Systems},
  volume  = {39},
  number  = {3},
  year    = {2021},
  doi     = {10.1145/3451161},
  pages     = {1--21}
}

@article{arabzadeh2022shallow,
  title   = {Shallow Pooling for Sparse Labels},
  author  = {Arabzadeh, Negar and Vtyurina, Alexandra and Yan, Xinyi and Clarke, Charles L. A.},
  journal = {Information Retrieval Journal},
  volume  = {25},
  number  = {4},
  pages   = {365--385},
  year    = {2022},
  doi     = {10.1007/s10791-022-09411-0}
}

@inproceedings{nguyen2016msmarco,
  title     = {{MS MARCO}: A Human Generated {MAchine Reading COmprehension} Dataset},
  author    = {Nguyen, Tri and Rosenberg, Mir and Song, Xia and Gao, Jianfeng and Tiwary, Saurabh and Majumder, Rangan and Deng, Li},
  booktitle = {Proceedings of the Workshop on Cognitive Computation: Integrating Neural and Symbolic Approaches 2016, co-located with the 30th Annual Conference on Neural Information Processing Systems ({NIPS} 2016)},
  series    = {CEUR Workshop Proceedings},
  volume    = {1773},
  address   = {Barcelona, Spain},
  publisher = {CEUR-WS.org},
  year      = {2016},
  url       = {https://ceur-ws.org/Vol-1773/CoCoNIPS_2016_paper9.pdf}
}

@inproceedings{craswell2020overview,
  title     = {Overview of the {TREC} 2019 Deep Learning Track},
  author    = {Craswell, Nick and Mitra, Bhaskar and Yilmaz, Emine and Campos, Daniel and Voorhees, Ellen M.},
  booktitle = {Proceedings of the Twenty-Eighth Text REtrieval Conference ({TREC} 2019)},
  series    = {NIST Special Publication},
  volume    = {500-331},
  publisher = {National Institute of Standards and Technology},
  year      = {2020},
  url       = {https://trec.nist.gov/pubs/trec28/papers/OVERVIEW.DL.pdf},
  note      = {arXiv:2003.07820}
}

@inproceedings{craswell2021overview,
  title     = {Overview of the {TREC} 2020 Deep Learning Track},
  author    = {Craswell, Nick and Mitra, Bhaskar and Yilmaz, Emine and Campos, Daniel},
  booktitle = {Proceedings of the Twenty-Ninth Text REtrieval Conference ({TREC} 2020)},
  series    = {NIST Special Publication},
  volume    = {1266},
  publisher = {National Institute of Standards and Technology},
  year      = {2021},
  url       = {https://trec.nist.gov/pubs/trec29/papers/OVERVIEW.DL.pdf},
  note      = {arXiv:2102.07662}
}

@article{voorhees2000variations,
  title   = {Variations in Relevance Judgments and the Measurement of Retrieval Effectiveness},
  author  = {Voorhees, Ellen M.},
  journal = {Information Processing \& Management},
  volume  = {36},
  number  = {5},
  pages   = {697--716},
  year    = {2000},
  doi     = {10.1016/S0306-4573(00)00010-8}
}

@inproceedings{parry2025reannotation,
  title     = {Variations in Relevance Judgments and the Shelf Life of Test Collections},
  author    = {Parry, Andrew and Fr{\"o}be, Maik and Scells, Harrisen and Schlatt, Ferdinand and Faggioli, Guglielmo and Zerhoudi, Saber and MacAvaney, Sean and Yang, Eugene},
  booktitle = {Proceedings of the 48th International ACM SIGIR Conference on Research and Development in Information Retrieval},
  pages     = {3387--3397},
  year      = {2025},
  publisher = {ACM},
  doi       = {10.1145/3726302.3730308}
}

@inproceedings{whitehill2009whose,
  title     = {Whose Vote Should Count More: Optimal Integration of Labels from Labelers of Unknown Expertise},
  author    = {Whitehill, Jacob and Ruvolo, Paul and Wu, Tingfan and Bergsma, Jacob and Movellan, Javier},
  booktitle = {Advances in Neural Information Processing Systems},
  volume    = {22},
  pages     = {2035--2043},
  year      = {2009},
  url       = {https://proceedings.neurips.cc/paper_files/paper/2009/hash/f899139df5e1059396431415e770c6dd-Abstract.html}
}

@book{krippendorff2018content,
  title     = {Content Analysis: An Introduction to Its Methodology},
  author    = {Krippendorff, Klaus},
  edition   = {Fourth},
  year      = {2018},
  publisher = {SAGE Publications},
  address   = {Thousand Oaks, CA},
  isbn      = {978-1-5063-9566-1}
}

@article{koo2016guideline,
  title   = {A Guideline of Selecting and Reporting Intraclass Correlation Coefficients for Reliability Research},
  author  = {Koo, Terry K. and Li, Mae Y.},
  journal = {Journal of Chiropractic Medicine},
  volume  = {15},
  number  = {2},
  pages   = {155--163},
  year    = {2016},
  doi     = {10.1016/j.jcm.2016.02.012}
}

@book{rasch1980probabilistic,
  title     = {Probabilistic Models for Some Intelligence and Attainment Tests},
  author    = {Rasch, Georg},
  edition   = {Expanded},
  year      = {1980},
  publisher = {University of Chicago Press},
  address   = {Chicago}
}

@incollection{birnbaum1968latent,
  title     = {Some Latent Trait Models and Their Use in Inferring an Examinee's Ability},
  author    = {Birnbaum, Allan},
  booktitle = {Statistical Theories of Mental Test Scores},
  editor    = {Lord, Frederic M. and Novick, Melvin R.},
  pages     = {397--479},
  year      = {1968},
  publisher = {Addison-Wesley},
  address   = {Reading, MA}
}

@article{samejima1969graded,
  title   = {Estimation of Latent Ability Using a Response Pattern of Graded Scores},
  author  = {Samejima, Fumiko},
  journal = {Psychometrika},
  volume  = {34},
  number  = {S1},
  pages   = {1--97},
  year    = {1969},
  note    = {Psychometrika Monograph Supplement No.~17},
  doi     = {10.1007/BF03372160}
}

@book{lord1980applications,
  title     = {Applications of Item Response Theory to Practical Testing Problems},
  author    = {Lord, Frederic M.},
  year      = {1980},
  publisher = {Lawrence Erlbaum Associates},
  address   = {Hillsdale, NJ}
}

@inproceedings{lalor2016irt,
  title     = {Building an Evaluation Scale using Item Response Theory},
  author    = {Lalor, John P. and Wu, Hao and Yu, Hong},
  booktitle = {Proceedings of the 2016 Conference on Empirical Methods in Natural Language Processing},
  pages     = {648--657},
  year      = {2016},
  address   = {Austin, Texas},
  publisher = {Association for Computational Linguistics},
  doi       = {10.18653/v1/D16-1062}
}

@inproceedings{polotinybenchmarks,
  title     = {{tinyBenchmarks}: Evaluating {LLMs} with Fewer Examples},
  author    = {Maia Polo, Felipe and Weber, Lucas and Choshen, Leshem and Sun, Yuekai and Xu, Gongjun and Yurochkin, Mikhail},
  booktitle = {Proceedings of the 41st International Conference on Machine Learning},
  series    = {Proceedings of Machine Learning Research},
  volume    = {235},
  pages     = {34303--34326},
  year      = {2024},
  publisher = {PMLR},
  url       = {https://proceedings.mlr.press/v235/maia-polo24a.html}
}

@inproceedings{mencattini2025merge,
  title     = {{MERGE}$^3$: Efficient Evolutionary Merging on Consumer-grade {GPUs}},
  author    = {Mencattini, Tommaso and Minut, Robert Adrian and Crisostomi, Donato and Santilli, Andrea and Rodol{\`a}, Emanuele},
  booktitle = {Proceedings of the 42nd International Conference on Machine Learning},
  series    = {Proceedings of Machine Learning Research},
  volume    = {267},
  pages     = {43694--43715},
  year      = {2025},
  publisher = {PMLR},
  url       = {https://proceedings.mlr.press/v267/mencattini25a.html}
}

@article{bradley1952rank,
  title   = {Rank Analysis of Incomplete Block Designs: {I}. {The} Method of Paired Comparisons},
  author  = {Bradley, Ralph Allan and Terry, Milton E.},
  journal = {Biometrika},
  volume  = {39},
  number  = {3/4},
  pages   = {324--345},
  year    = {1952},
  doi     = {10.1093/biomet/39.3-4.324}
}

@article{yen1984effects,
  title   = {Effects of Local Item Dependence on the Fit and Equating Performance of the Three-Parameter Logistic Model},
  author  = {Yen, Wendy M.},
  journal = {Applied Psychological Measurement},
  volume  = {8},
  number  = {2},
  pages   = {125--145},
  year    = {1984},
  doi     = {10.1177/014662168400800201}
}

@article{land2026auditing,
  title   = {Auditing {LLM} Benchmarks with Item Response Theory},
  author  = {Land, Sander and Bikel, Daniel M.},
  journal = {arXiv preprint arXiv:2605.30504},
  year    = {2026},
  note    = {Accepted at EMNLP 2026},
  url     = {https://arxiv.org/abs/2605.30504}
}

@article{robertson2009prf,
  title   = {The Probabilistic Relevance Framework: {BM25} and Beyond},
  author  = {Robertson, Stephen and Zaragoza, Hugo},
  journal = {Foundations and Trends in Information Retrieval},
  volume  = {3},
  number  = {4},
  pages   = {333--389},
  year    = {2009},
  doi     = {10.1561/1500000019}
}

@inproceedings{cormack2009rrf,
  title     = {Reciprocal Rank Fusion Outperforms {Condorcet} and Individual Rank Learning Methods},
  author    = {Cormack, Gordon V. and Clarke, Charles L. A. and Buettcher, Stefan},
  booktitle = {Proceedings of the 32nd International ACM SIGIR Conference on Research and Development in Information Retrieval},
  pages     = {758--759},
  year      = {2009},
  publisher = {ACM},
  doi       = {10.1145/1571941.1572114}
}

@inproceedings{thakur2021beir,
  title     = {{BEIR}: A Heterogeneous Benchmark for Zero-shot Evaluation of Information Retrieval Models},
  author    = {Thakur, Nandan and Reimers, Nils and R{\"u}ckl{\'e}, Andreas and Srivastava, Abhishek and Gurevych, Iryna},
  booktitle = {Proceedings of the Neural Information Processing Systems Track on Datasets and Benchmarks},
  volume    = {1},
  year      = {2021},
  url       = {https://datasets-benchmarks-proceedings.neurips.cc/paper/2021/hash/65b9eea6e1cc6bb9f0cd2a47751a186f-Abstract-round2.html}
}

@inproceedings{muennighoff2023mteb,
  title     = {{MTEB}: Massive Text Embedding Benchmark},
  author    = {Muennighoff, Niklas and Tazi, Nouamane and Magne, Lo{\"\i}c and Reimers, Nils},
  booktitle = {Proceedings of the 17th Conference of the European Chapter of the Association for Computational Linguistics},
  pages     = {2014--2037},
  year      = {2023},
  address   = {Dubrovnik, Croatia},
  publisher = {Association for Computational Linguistics},
  doi       = {10.18653/v1/2023.eacl-main.148}
}

@misc{zetaalpha2024nanobeir,
  author       = {{Zeta Alpha}},
  title        = {{NanoBEIR}},
  year         = {2024},
  howpublished = {\url{https://huggingface.co/collections/zeta-alpha-ai/nanobeir-66e1a0af21dfd93e620cd9f6}}
}

@inproceedings{su2024bright,
  title     = {{BRIGHT}: A Realistic and Challenging Benchmark for Reasoning-Intensive Retrieval},
  author    = {Su, Hongjin and Yen, Howard and Xia, Mengzhou and Shi, Weijia and Muennighoff, Niklas and Wang, Han-yu and Liu, Haisu and Shi, Quan and Siegel, Zachary S. and Tang, Michael and Sun, Ruoxi and Yoon, Jinsung and Arik, Sercan O. and Chen, Danqi and Yu, Tao},
  booktitle = {The Thirteenth International Conference on Learning Representations},
  year      = {2025},
  url       = {https://proceedings.iclr.cc/paper_files/paper/2025/hash/7a0f8055c838df8e62329a76c7c6403d-Abstract-Conference.html}
}

@inproceedings{vidore3,
  title     = {{ViDoRe V3}: A Comprehensive Evaluation of Retrieval Augmented Generation in Complex Real-World Scenarios},
  author    = {Loison, Ant{\'o}nio and Mac{\'e}, Quentin and Edy, Antoine and Xing, Victor and Balough, Tom and Moreira, Gabriel de Souza P. and Liu, Bo and Faysse, Manuel and Hudelot, C{\'e}line and Viaud, Gautier},
  booktitle = {Proceedings of the 64th Annual Meeting of the Association for Computational Linguistics (Volume 1: Long Papers)},
  pages     = {16570--16600},
  year      = {2026},
  address   = {San Diego, California, United States},
  publisher = {Association for Computational Linguistics},
  doi       = {10.18653/v1/2026.acl-long.755}
}

@article{shrout1979intraclass,
  title   = {Intraclass Correlations: Uses in Assessing Rater Reliability},
  author  = {Shrout, Patrick E. and Fleiss, Joseph L.},
  journal = {Psychological Bulletin},
  volume  = {86},
  number  = {2},
  pages   = {420--428},
  year    = {1979},
  doi     = {10.1037/0033-2909.86.2.420}
}

\newpage
\begin{appendices}
  \crefalias{section}{appendix}\crefalias{subsection}{appendix}
  \etocdepthtag.toc{appendix}
  \raggedbottom
{%
  \etocsettagdepth{main}{none}%
  \etocsettagdepth{appendix}{section}%
  \etocsettocstyle{\vspace{\baselineskip}}{}%
  \tableofcontents
}
\paragraph{Reader's guide.}
The table maps each main-text section to the appendix sections that hold its details and robustness checks. After the limitations, definitions, and method details, the appendix follows the order of \cref{sec:experiments}.
\begin{center}
\small
\renewcommand{\arraystretch}{1.12}
\begin{tabular}{@{}p{0.12\textwidth}p{0.55\textwidth}p{0.27\textwidth}@{}}
\toprule
Main text & Claim or component & Appendix \\
\midrule
\cref{sec:method} & Protocols, calibration fit, judges, budgets, and cost; prompts; proofs & \cref{app:method}, \cref{app:prompts}, \cref{app:theory} \\
\cref{sec:experiments} & Human study: design, what annotators saw, annotators, reliability & \cref{app:human} \\
\cref{sec:diagnostics} & Criterion parameters, per-query $\qpar{\tau_j}$ and $\qpar{\alpha_j}$, transfer, label-free checks & \cref{app:nano-bright-params}, \cref{app:calibration-diagnostics} \\
\cref{sec:calibration} & Gain against human grades; calibrated scores against simpler maps; NIST placement; NIST positives that \rcp{} scores lowest & \cref{app:calibration-human}, \cref{app:trecdl-placement}, \cref{app:trecdl-audit} \\
\cref{sec:blind-eval} & Study 1 (qrel sets against \rcp{} sets); gain against human labels & \cref{app:meta-study1}, \cref{app:human-labels} \\
\cref{sec:ranking-results} & Failure modes, sensitivity and scores per dataset, second judges, zerank, robustness, Count-nDCG & \cref{app:nano-bright-failure}, \cref{app:nano-bright-sensitivity}, \cref{app:nano-bright-tables}, \cref{app:nano-bright-judge}, \cref{app:trecdl-zerank}, \cref{app:robustness-factorial}, \cref{sec:count} \\
\cref{sec:external} & Human verdicts and their controls, NIST-decisive pairs & \cref{app:human-verdicts}, \cref{app:trecdl-pairs} \\
\cref{sec:related} & Extended related work & \cref{app:related} \\
\cref{sec:conclusion} & Limitations; \dataset{HotpotQA} & \cref{app:limitations}, \cref{app:meta-hotpotqa} \\
Appendix only & Study 2 (reranker lists judged by LLMs) & \cref{app:meta-study2} \\
\bottomrule
\end{tabular}
\end{center}

\newpage

\section{Limitations}
\label{app:limitations}
\label{sec:limitations}

\paragraph{What the metric measures.}
Since \rcp{}'s gains and our annotators' grades value each document on its own, they do not penalize redundancy. On \dataset{HotpotQA}, whose questions need a primary and a secondary fact, a redundant document on the primary fact can outrank the document with the secondary one (\cref{app:meta-hotpotqa}). Pointwise gains also cannot credit a document that becomes relevant only once another document is known. The judge and the annotators both measure perceived relevance, not factual correctness (\cref{app:human-instrument}). When the \rcpndcg{}@5 scores of two lists differ by less than \val{draft.5.margin.decisive}, annotators side with \rcpndcg{} only at chance. We therefore read such differences as ties (\cref{app:human-verdicts}).

\paragraph{The evidence.}
Our document-level and contest-level human results share annotations. The NIST grades share none with either. Relevance is hard to judge consistently, even for NIST assessors~\citep{voorhees2000variations}. Two annotators of the same document grade it within one step of each other in \val{b2.grade.all.g5.agree-within-1}\% of cases. Gwet's AC2 with quadratic weights, which credits near misses on the ordinal scale, is \val{b2.grade.all.g5.ac2-quad}, and Krippendorff's $\alpha$ with interval distances is \val{b2.grade.all.alpha-interval}. Both are 0 at chance and 1 at perfect agreement. We therefore label documents by the mean of at least three grades. The annotators' questions resemble the rubric by design (\cref{app:human-instrument}). The external LLM judges of \cref{app:meta,sec:count} may share systematic preferences with our judge~\citep{soboroff2025dontuse}. The NIST check replaces our primary judge with the newer and smaller \model{Qwen3.6-27B} and adds \model{gpt-oss-120b}. With either judge, \rcpndcg{} favors NIST's winner on every NIST-decisive pair (\cref{app:trecdl-pairs}). On \reranker{zerank-2}, which learns from LLM preferences, unjudged passages explain only part of its lower rank under the NIST grades (\cref{sec:ranking-results}).

\paragraph{Scope and cost.}
We refined the criteria on \val{n14.rubric.queries.total} NanoBEIR queries with \val{n14.rubric.docs} documents each, about \val{n14.rubric.docshare}\% of its pooled documents (\cref{app:prompts}). All other data are held out. The criterion parameters also transfer to datasets left out of the fit (\cref{app:nano-bright-params}). Scores are comparable only within one calibration fit. The fit also ignores the uncertainty of the Stage~A scores. Documents outside the judged pool need new judgments and a refit (\cref{app:method-cost}). We validate \rcp{} as an evaluation method. Models trained on its feedback are best tested on other labels, since \rcpndcg{} would reward their own training signal. This concern may be smaller than it seems. Rerankers are far smaller than the LLM judge and spend far less computation per document. Such rerankers are thus unlikely to reach the judge's quality. Future work can test whether \rcpndcg{} remains informative for models trained on its feedback.

\FloatBarrier
\section{Definitions}\label{app:notation}

\subsection{Metrics and ties}\label{app:notation-metrics}\label{tab:notation-metrics}
This appendix defines the metrics, statistics, and intervals that the other appendices use. All four compared metrics are nDCG (\cref{eq:ndcg}) and differ only in the gain. nDCG on the qrels (qrel-nDCG) uses the qrel grades as gains.
Count-nDCG uses the pass share of \cref{eq:count-gain}. Raw-BT-nDCG applies the gain $g$ of \cref{eq:theta-gain} to the uncalibrated BT scores, centered at zero within each query. \rcpndcg{} uses the calibrated gain of \cref{eq:theta-gain}.
Reranker scores can tie. On ViDoRe v3, each document in a tied group receives the group's mean gain. On NanoBEIR and BRIGHT, we break tied scores by document identifier (the trec\_eval convention). Giving tied documents their group's mean gain changes the sensitivity of qrel-nDCG and \rcpndcg{} (\cref{app:notation-stats}) by at most \val{n18.ties.maxdelta} points. On TREC-DL, tied scores keep their order in the judge pool, which starts with the fused first-stage ranking. On NanoBEIR and BRIGHT, the failure-mode shares instead keep the first-stage order of tied documents (\cref{app:nano-bright-failure}).

\subsection{Statistics}\label{app:notation-stats}\label{sec:background-sensitivity}
A metric separates two rerankers on a dataset when a paired $t$-test on their per-query scores rejects equal means at the \val{draft.2.siglevel}\% level. Sensitivity is the share of the \val{draft.app.ab.pairs} reranker pairs of a dataset that a metric separates, averaged over the datasets of a suite. If no two rerankers differed, about \val{draft.2.siglevel}\% of pairs would separate by chance.
The AUC is the share of pooled document pairs, one rated Useful or better and one not, that a score orders correctly, with ties counting half (0.5 is chance, 1 is perfect). Within-query concordance is the share of same-query pairs with different grades that the score orders like the grades (0.5 is chance).

The expected calibration error (ECE) is the size-weighted mean gap between predicted probabilities and observed frequencies in bins of the prediction (0 is perfect). The Brier skill score is the share of a constant prediction's squared error that the pass probabilities remove (1 is perfect, 0 is no better than the constant).

ICC(1,1) is the expected correlation between two single grades of one document~\citep{shrout1979intraclass,koo2016guideline}. ICC(1,$k$) is the expected correlation between the mean grades of two independent panels of $k$ annotators. The rank-based split-half reliability correlates the mean grades of two random half-panels and steps the correlation up to the full panel with the Spearman--Brown formula. Krippendorff's $\alpha$ with the interval metric is one minus the ratio of observed to chance-expected squared grade differences (0 is chance, 1 is perfect agreement)~\citep{krippendorff2018content}. On our grades it is \val{b2.grade.all.alpha-interval}, nearly equal to ICC(1,1) at \val{b2.grade.all.icc1}. Fleiss' $\kappa$ compares the observed agreement of annotators with the agreement that their answer rates produce by chance (0 is chance, 1 is perfect agreement). Gwet's AC1 also corrects for chance but depends less on how common each answer is.

A contest's verdict is the list that more annotators pick. A contest is decided when its verdict names a list that exactly one of the two metrics also names. The main measure is the share of decided contests whose verdict sides with \rcpndcg{}. The chance level of the main measure is \val{draft.ff.chance.exact}\%, not one half, because qrel-nDCG ties more often (\cref{app:human-verdicts}). A contest in which one metric ties can count only for the other. A metric's margin is its score difference between list~A and list~B. A metric ties when the absolute value of its margin is below \val{draft.app.d.tie-margin}. The margin AUC of a metric is the probability that a random annotator pick of list~A has a larger A-minus-B margin than a random pick of list~B (0.5 means no information).

\subsection{Intervals}\label{app:notation-intervals}\label{tab:notation-intervals}
Unless the text or a caption says otherwise, intervals are 95\% query-clustered percentile bootstraps with \val{draft.app.ab.boot.b} resamples and a fixed seed, stratified by suite, corpus, or year. Rates over human contests, such as the main measure, instead use an analytic query-clustered interval. The margin AUC uses a delete-one-query jackknife. All methods treat the query as the unit, because documents and contests of one query share its pool.

\FloatBarrier
\section{Method Details}\label{app:method}\label{app:setup-details}

\subsection{Protocol, pools, and rerankers}\label{app:method-pools}\label{app:method-rerankers}
This appendix details the candidate pools, the judging schedules, the calibration fit, the judges and their cost, and the setup of each suite. \Cref{tab:app-hyper} collects the settings of the full protocol, the Stage~A and Stage~B schedules that NanoBEIR and BRIGHT use (\cref{app:method-stagea}).

\begin{table}[htb]
  \centering
  \caption{Settings of the full protocol, which NanoBEIR and BRIGHT use under both judges. \Cref{app:trecdl} describes the reduced protocol of TREC-DL.}
  \label{tab:app-hyper}
  \footnotesize
  \begin{tabularx}{\linewidth}{@{}Xl@{}}
    \toprule
    Setting & Value \\
    \midrule
    Pool size $K$ (top of the fused first-stage ranking) & \val{draft.2.poolsize} \\
    First-stage fusion & RRF over three retrievers, $\kappa = \val{draft.app.ckl.rrf-k}$ \\
    Stage~A: $w = 10$, score range $[-5,+5]$, coverage windows per query (random, then stratified, each reversed) & \val{draft.app.ckl.stagea.coverage-windows} \\
    Stage~A: adaptive calls per query; discount of covered boundaries after each pick & \val{draft.app.ckl.stagea.adaptive-calls} ($7$ batches of $8$ windows); \val{draft.app.ckl.stagea.overlap} \\
    Stage~A: BT fit & ridge $10^{-4}$, L-BFGS, at most \val{draft.app.ckl.bt.maxiter} iterations \\
    Stage~B: windows per query, each of $w=10$ documents on $C=5$ criteria, no reversed copies & \val{draft.app.ckl.stageb.windows} \\
    Calibration: hard constraints & $\sum_{\cpar{c}} \cpar{\gamma_c} = C$, $\operatorname{mean}_{\cpar{c}} \cpar{\beta_c} = 0$ \\
    Calibration: priors on $\qpar{\tau_j}$ and $\qpar{\alpha_j}$ & $\mathcal{N}(1, 1)$ and $\mathcal{N}(0, 2^2)$ \\
    Calibration: ridge on criterion parameters; fit & $10^{-4}$; L-BFGS, at most \val{draft.app.ckl.cal.maxiter} iterations \\
    Temperature, schedule seed, metric cutoff & \val{draft.app.ckl.temperature}, \val{draft.app.ckl.seed}, $10$ \\
    \bottomrule
  \end{tabularx}
\end{table}

Each query's candidate pool holds the top $K = \val{draft.2.poolsize}$ documents of a reciprocal rank fusion (RRF) of three first-stage retrievers. The three retrievers are the dense \baseline{Cohere Embed~v4}, \baseline{BM25}~\citep{robertson2009prf} as implemented in \texttt{bm25s}, and the dense \baseline{Octen-Embedding-8B}. Fusion scores each document $d$ by $\sum_{m} 1/\bigl(\kappa + \mathrm{rank}_m(d)\bigr)$ over the three rankings $m$~\citep{cormack2009rrf}. On NanoBEIR, BRIGHT, and ViDoRe v3, documents that the qrels mark relevant but the retrievers miss replace the lowest-ranked pool documents. The pool size $K$ thus stays fixed. On TREC-DL, each pool also joins every passage that NIST assessors judged for the query. These pools are therefore larger (\cref{app:trecdl}). The pools of a few NanoBEIR and BRIGHT queries differ from this construction (\cref{app:nano-bright}).

We evaluate the 14 rerankers of \cref{tab:app-rerankers}. The human study compares pairs drawn from \val{a1a.desc.rerankers} of them (\cref{app:human}). The three ZeroEntropy rerankers are trained on preferences from an ensemble of LLMs~\citep{pipitone2025zelo}. A metric built on an LLM judge could favor them. We therefore also report the human results without them.

\begin{table}[t]
  \centering
  \caption{The 14 rerankers. Families, identifier prefix, and license of the public weights: Contextual~AI~\citep{contextualai2025rerankv2} (\nolinkurl{ContextualAI/}, CC BY-NC-SA 4.0), Jina~AI~\citep{wang2025jinarerankerv3} (\nolinkurl{jinaai/}, CC BY-NC 4.0), Qwen~\citep{zhang2025qwen3embedding} (\nolinkurl{Qwen/}, Apache 2.0), Cohere~\citep{cohere2025rerank4} (API only), Voyage~AI~\citep{voyageai2025rerank25} (API only), ZeroEntropy~\citep{pipitone2025zelo} (\nolinkurl{zeroentropy/}, Apache 2.0).}
  \label{tab:app-rerankers}
  \footnotesize
  \begin{tabularx}{\linewidth}{@{}l>{\raggedright\arraybackslash}Xr@{}}
    \toprule
    Name & Public identifier & Params \\
    \midrule
    \baseline{CTXL-RR-1B} & \nolinkurl{ctxl-rerank-v2-instruct-multilingual-1b} & \val{draft.app.ckl.params.ctxl1} \\
    \baseline{CTXL-RR-2B} & \nolinkurl{ctxl-rerank-v2-instruct-multilingual-2b} & \val{draft.app.ckl.params.ctxl2} \\
    \baseline{CTXL-RR-6B} & \nolinkurl{ctxl-rerank-v2-instruct-multilingual-6b} & \val{draft.app.ckl.params.ctxl6} \\
    \baseline{Jina-RR-v3} & \nolinkurl{jina-reranker-v3} & \val{draft.app.ckl.params.jina} \\
    \baseline{Qwen3-RR-0.6B} & \nolinkurl{Qwen3-Reranker-0.6B} & \val{draft.app.ckl.params.qwen3rr06} \\
    \baseline{Qwen3-RR-4B} & \nolinkurl{Qwen3-Reranker-4B} & \val{draft.app.ckl.params.qwen3rr4} \\
    \baseline{Qwen3-RR-8B} & \nolinkurl{Qwen3-Reranker-8B} & \val{draft.app.ckl.params.qwen3rr8} \\
    \baseline{RR4-Fast} & \nolinkurl{rerank-v4.0-fast} & -- \\
    \baseline{RR4-Pro} & \nolinkurl{rerank-v4.0-pro} & -- \\
    \baseline{Voyage2.5-lt} & \nolinkurl{rerank-2.5-lite} & -- \\
    \baseline{Voyage2.5} & \nolinkurl{rerank-2.5} & -- \\
    \baseline{zerank-1-sm} & \nolinkurl{zerank-1-small} & \val{draft.app.ckl.params.zr1s} \\
    \baseline{zerank-1} & \nolinkurl{zerank-1} & \val{draft.app.ckl.params.zr1} \\
    \baseline{zerank-2} & \nolinkurl{zerank-2} & \val{draft.app.ckl.params.zr2} \\
    \bottomrule
  \end{tabularx}
\end{table}

\subsection{Stage A and Stage B schedules}\label{app:method-stagea}\label{app:tournament-phases}\label{app:method-stageb}\label{app:rubric-design}
Stage~A schedules its windows in three phases, first random, then stratified, then adaptive. Random windows give every document broad coverage. Stratified windows then group documents of similar estimated quality. Both phases show each window a second time in reverse order.

The adaptive phase adds \val{draft.app.ckl.stagea.adaptive-calls} calls per query in 7 batches of 8 windows and refits the BT model before each batch. Let $p_r = \sigma\bigl(\hat{\theta}_{(r)} - \hat{\theta}_{(r+1)}\bigr)$ be the predicted preference between the documents at ranks $r$ and $r+1$ of the current order. The phase values the boundary between them by
\begin{equation}
  v_r = p_r\,(1 - p_r) \cdot \frac{1}{\log_2(r+1)} \cdot \frac{1}{1 + n_r},
  \label{eq:app-adaptive}
\end{equation}
where $n_r$ counts the windows that have already held both documents. The first factor peaks when the order of the two documents is a coin flip. The second factor is the rank discount of nDCG. The third favors pairs that the judge has rarely compared. A batch greedily selects the 8 windows of $w$ consecutive documents with the largest sum of boundary values. After each pick, the values of the boundaries it covers shrink by the factor of \cref{tab:app-hyper}. Adaptive windows are shown once, in the current estimated order.

The $w$ scores of a window fix only $w-1$ independent differences. Each of its $\binom{w}{2}$ pairs therefore has weight $2/w$. Each BT fit minimizes the weighted mean cross-entropy between soft and predicted preferences plus the ridge penalty of \cref{tab:app-hyper}. If a response's scores cannot be parsed, its ranking enters the fit as hard preferences with weight 1.

Stage~B also shows $w$ documents per call and asks about each on its own (\cref{sec:stage-b}). Under the full protocol and on ViDoRe v3, each query receives \val{draft.app.ckl.stageb.windows} windows, whose first \val{draft.app.ckl.stageb.random-windows} are balanced random windows. Their sampling is inspired by balanced incomplete block designs. Each window takes the documents shown least often so far, with random tie-breaks. All documents thus appear about equally often and meet many different partners. A Rasch model~\citep{rasch1980probabilistic}, the one-parameter case of IRT, then gives each document a preliminary estimate from these responses. The remaining windows group documents with similar estimates. The calibration treats each of a document's responses as one observation (\cref{app:method-fit}). We call the three Stage~A phases together with this Stage~B schedule the full protocol, which NanoBEIR and BRIGHT use (\cref{tab:app-hyper}). ViDoRe v3 uses its Stage~B schedule (\cref{app:vidore-setup}).

\subsection{Calibration fit}\label{app:method-fit}
Within a query, the ability is a positive affine function of the BT score. This affine map keeps the query's order and the model linear on the logit scale (\cref{eq:2pl}). We set $\qpar{\tau_j} = \operatorname{softplus}(t_j)$ with an unconstrained $t_j$, so that no negative scale can reverse the tournament order. The discriminations are $\cpar{\gamma_c} = C\,e^{u_c} / \sum_{c'} e^{u_{c'}}$ and the difficulties are $\cpar{\beta_c} = v_c - \operatorname{mean}_{c'} v_{c'}$. Let $N$ be the number of Stage~B responses, each of which answers all $C$ criteria. The fit minimizes
\begin{equation}
  \frac{1}{NC} \biggl[ \sum_{n=1}^{N} \sum_{c=1}^{C} \ell_{nc} + \sum_{\qpar{j}} \Bigl( \frac{(\qpar{\tau_j} - 1)^2}{2\sigma_\tau^2} + \frac{\qpar{\alpha_j^2}}{2\sigma_\alpha^2} \Bigr) \biggr] + \frac{\lambda}{2} \sum_{c} \bigl( u_c^2 + v_c^2 \bigr),
  \label{eq:app-objective}
\end{equation}
where $\ell_{nc}$ is the binary cross-entropy of response $n$ on criterion $c$ under \cref{eq:2pl} and the hyperparameters take the values of \cref{tab:app-hyper}. Up to the factor $1/(NC)$, the objective is a negative log posterior under the priors of \cref{tab:app-hyper}, with the ridge acting as a Gaussian prior on $u_c$ and $v_c$. The fit is therefore the maximum a posteriori (MAP) estimate. The priors keep the estimates finite when all documents of a query fail every criterion or all pass every criterion.

We hold $\dpar{\hat{\theta}^{\mathrm{BT}}}$ fixed from Stage~A and do not propagate their uncertainty (\cref{sec:limitations}). The fit starts from $\qpar{\tau_j} = 1$, $\qpar{\alpha_j} = 0$, and equal criterion parameters. On ViDoRe v3, the fit covers all six language versions of every question. On TREC-DL, the main fit pools both years with further MS~MARCO queries without NIST grades (\cref{app:trecdl-protocol}).

\subsection{Judges, budgets, and cost}\label{app:method-judges}\label{app:method-budget}\label{app:method-runtime}\label{app:runtime-estimates}\label{app:method-cost}\label{app:trecdl-cost}\label{app:robustness-budget}

\paragraph{Judges and serving.}
We serve the open-weight judges with SGLang, each on one node of eight GPUs. The primary judge, \model{Qwen3.5-397B}~\citep{qwen2026qwen35}, judges NanoBEIR, BRIGHT, and ViDoRe v3 in its 4-bit NVFP4 release~\citep{nvidia2026qwen35nvfp4} at public revision \val{draft.app.ckl.nvfp4.rev} on eight B200 GPUs. The second judge, \model{gpt-oss-120b}~\citep{openai2025gptoss}, runs in its released precision (MXFP4 mixture-of-experts weights) on H100 GPUs and repeats both stages on NanoBEIR, BRIGHT, and TREC-DL. On TREC-DL, the smaller dense \model{Qwen3.6-27B}~\citep{qwen2026qwen36} judges in its 8-bit (FP8) release under the reduced protocol of \cref{app:trecdl-protocol}. No other LLM produces a gain. The other LLMs serve as blind comparison judges (\cref{app:meta,app:robustness-judges,app:prompts-meta}) and as a TREC-DL arbiter (\cref{app:trecdl-judge2}).

\paragraph{Budgets and cost.}
\Cref{tab:app-budget} lists how often the judge sees each document. A reversed copy counts as a separate call. A placement is one appearance of a document in a Stage~B window.

\begin{table}[t]
  \centering
  \caption{Judge budgets per suite. Calls and Stage~B windows are counted per query. Windows and placements per document are means over pool documents. \model{gpt-oss-120b} uses the settings of \model{Qwen3.5-397B} on NanoBEIR and BRIGHT.}
  \label{tab:app-budget}
  \footnotesize
  \begin{tabular}{@{}llrcrrr@{}}
    \toprule
    & & \multicolumn{3}{c}{Stage~A} & \multicolumn{2}{c}{Stage~B} \\
    \cmidrule(lr){3-5} \cmidrule(l){6-7}
    Suite & Judge & Calls & Reversed & Windows/doc & Windows & Placements/doc \\
    \midrule
    NanoBEIR & \model{Qwen3.5-397B} & \val{draft.app.ckl.stagea.calls} & yes & \val{draft.3.nano.stagea.windowsperdoc} & \val{draft.app.ckl.stageb.windows} & \val{draft.3.nano.stageb.placementsperdoc} \\
    BRIGHT & \model{Qwen3.5-397B} & \val{draft.app.ckl.stagea.calls} & yes & \val{draft.app.ckl.bright.stagea.windowsperdoc}$^{a}$ & \val{draft.app.ckl.stageb.windows} & \val{draft.app.ckl.bright.stageb.placementsperdoc}$^{a}$ \\
    ViDoRe v3 & \model{Qwen3.5-397B} & --$^{b}$ & --$^{b}$ & --$^{b}$ & \val{draft.app.ckl.stageb.windows} & \val{a3.n.stageb.placements.perdoc} \\
    TREC-DL 2019 & \model{Qwen3.6-27B} & \val{draft.4.dl.stagea-calls.dl19} & no & \val{a2.budget.qwendl19.windowsperdoc} & \val{b7.trec.count-kall.calls-total}$^{c}$ & \val{a2.budget.qwendl19.placementsperdoc} \\
    TREC-DL 2020 & \model{Qwen3.6-27B} & \val{draft.4.dl.stagea-calls.dl20} & no & \val{a2.budget.qwendl20.windowsperdoc} & \val{b7.trec.count-kall.calls-total}$^{c}$ & \val{a2.budget.qwendl20.placementsperdoc} \\
    \bottomrule
  \end{tabular}

  \smallskip
  \parbox{\linewidth}{\footnotesize $^{a}$Computed from the configuration. The released window logs give the same value. $^{b}$Not recorded in the stored run artifacts. $^{c}$Mean over the queries of both years, whose Stage~B windows are pooled.}
\end{table}

Each pool is judged once and then scores any reranking of it without further judge calls. A document outside the judged pool needs new windows and a refit of the 2PL model, which can shift earlier calibrated scores. The cost of \rcp{} is thus paid per pool, as for human qrels, not per reranker.

We price the logged tokens at public API rates. The TREC-DL run of \model{Qwen3.6-27B} costs \$\val{a2.cost.usd.2p00} in total and \$\val{a2.cost.usdperquery.2p00} per query. This total includes imputed reasoning tokens for the 2020 queries, which make up \val{a2.cost.imputedshareusd.2p00}\% of it. We also replay the logged judgments at two reduced budgets, which keep \val{a2.budget.stab.b05.docjudgmentshare}\% and \val{a2.budget.stab.b03.docjudgmentshare}\% of the document judgments. These replays cost \$\val{a2.cost.reduced.b05.2p00.usdperquery} and \$\val{a2.cost.reduced.b03.2p00.usdperquery} per query. Each replay draws five random subsamples and refits both stages on each. The replays and the full leaderboard they are compared with fit the calibration on the TREC-DL queries alone. The replays score the 18 TREC-DL systems, namely the 14 rerankers, the three retrievers, and their fusion. We report TREC-DL 2019, whose changes exceed those of 2020. With half of the judgments, the largest change in a system's mean \rcpndcg{}@10 is \val{a56.robust.budget.b0p5-k2.dl19.max-shift}, averaged over subsamples. The Kendall rank correlation between the reduced and the full leaderboard is \val{a56.robust.budget.b0p5-k2.dl19.kendall}, where 1 means the same order. With \val{a2.budget.stab.b03.docjudgmentshare}\% of the judgments, these values are \val{a56.robust.budget.b0p3-k2.dl19.max-shift} and \val{a56.robust.budget.b0p3-k2.dl19.kendall}. Neither budget reverses a system pair that the full budget separates. This check therefore compares \rcpndcg{} only with itself, not with NIST.

\subsection{NanoBEIR and BRIGHT}\label{app:nano-bright}\label{app:nano-bright-scope}
All \rcp{} scores on NanoBEIR and BRIGHT come from \model{Qwen3.5-397B}, except where we name \model{gpt-oss-120b}. \Cref{app:calibration-gain} compares gain functions on these suites. \Cref{app:notation-metrics} describes how we break tied reranker scores, including in the failure-mode shares (\cref{app:nano-bright-failure}).

\paragraph{NanoBEIR.} NanoBEIR~\citep{zetaalpha2024nanobeir,muennighoff2023mteb} holds \val{draft.4.nano.queries} queries in 13 datasets. The rerankers reorder only each query's judged pool of \val{draft.2.poolsize} documents. This also holds for the four NFCorpus queries with more positive qrels than $K$. On \val{n15.nano.scidocs.emptypos} NanoSCIDOCS queries, one positive document has empty text and is missing from the pool. qrel-nDCG still counts it in its ideal ranking. \rcpndcg{} gives gain zero to a document without a calibrated score, and its ideal ranking draws on the judged pool (\cref{eq:theta-gain}). On \val{n20.nano.arguana.selfdocs} of the \val{draft.app.g.s.nano.nq.arguana} NanoArguAna queries, the pool holds the query argument itself. In one of them, the argument is a copy filed under another identifier. BEIR's evaluator and MTEB's ArguAna task drop a retrieved document whose identifier equals the query's. We remove these documents, and the one copy, from every ranking and from the ideal rankings of \rcpndcg{} and Count-nDCG. The qrels never mark them relevant. The calibration keeps them (0.02\% of the Stage~B placements), and the human study and the LLM-judge studies use the lists as shown.

\paragraph{BRIGHT.} BRIGHT has \val{a4.bright.queries.rerankers} queries in our reranker runs and \val{a4.bright.queries.calibrated} in the calibration fit. The two query sets differ only in the two TheoremQA tasks. All metrics use the \val{a4.bright.queries.evaluated} queries in both sets, except qrel-nDCG's sensitivity, which uses all TheoremQA-Theorems queries as in \cref{tab:crosssuite}.

As in BRIGHT's official evaluation~\citep{su2024bright}, we remove each query's listed excluded documents from every ranking we score. The removal affects \val{b20.affected.queries} of the \val{a4.bright.queries.evaluated} evaluated queries. None of the removed documents is qrel-relevant. The judge's pools still contained the excluded documents, \val{b20.prev.qwen.stageb-share}\% of the Stage~B observations of each fit. We keep the fitted parameters and residual checks (\cref{tab:nb-params,tab:app-diagnostics}). The human study is unaffected (\cref{app:human-design}). A few follow-up analyses still use rankings that contain the excluded documents.

\subsection{ViDoRe v3}\label{app:vidore}\label{app:vidore-setup}
ViDoRe v3~\citep{vidore3} grades retrieval from pages of enterprise, governmental, and educational documents, which every system reads as the benchmark's OCR text, not as an image. The benchmark has eight corpora, five English and three French. Each question exists in six language versions with one shared set of relevance labels. We ran both stages of \rcp{} on all \val{a3.n.rows} question rows and fitted one calibration over them. The unit of analysis is the native question (\val{a3.n.questions} questions, \cref{app:vidore-xlang}). The five translated versions pose the same question in another language over the same pages. Retrieval with them is therefore cross-lingual. Since they also share the question's relevance labels, scoring them as separate queries would count each question six times in the paired tests. MTEB's ViDoRe v3 tasks instead average one score per language subset and thus mix monolingual and cross-lingual retrieval. Each question row has its own candidate pool of \val{draft.2.poolsize} pages, built as on NanoBEIR (\cref{app:method-pools}). After a pre-filter by a vision-language model, ViDoRe's human assessors graded the remaining pages on a three-point scale~\citep{vidore3}. Relevant pages thus have grade 1 or 2. On ViDoRe v3, qrel-nDCG uses these grades as linear gains. \model{Qwen3.5-397B} judges ViDoRe v3, with \val{a3.n.stageb.placements.perdoc} Stage~B placements per page on average (\cref{tab:app-budget}). Pages with tied reranker scores receive their group's mean gain. The human study drew its contests from the English corpora (\cref{app:human}).

Per corpus, qrel-nDCG separates \val{a3.sens.lin}\% of the \val{a56.robust.factorial.pairs} reranker pairs on average, and \rcpndcg{} separates \val{a3.sens.rcp}\%. The Kendall correlation between the two leaderboards, where 1 means identical orders, is \val{a3.tau.lin-rcp}. qrel-nDCG@10 saturates, floors, or compresses on \val{a3.fail.any}\% of questions (\cref{app:vidore-failure}).

\subsection{TREC-DL 2019 and 2020}\label{app:trecdl}\label{app:trecdl-protocol}
The TREC-DL tracks of 2019 and 2020 give dense graded human labels that share no annotations with our human study. NIST assessors graded MS~MARCO passages from 0 to 3 under their own guidelines, years before our method existed~\citep{nguyen2016msmarco,craswell2020overview,craswell2021overview}.

\paragraph{Queries and pools.}
The two tracks have \val{a2.setup.dl19.queries} and \val{a2.setup.dl20.queries} queries with NIST grades, over \val{a2.setup.dl19.judged} and \val{a2.setup.dl20.judged} graded passages. A passage outside this judgment pool is unjudged and counts as irrelevant. Each pool joins the fused top \val{draft.2.poolsize} passages (\cref{app:method-pools}) with every graded passage, so that every graded passage receives an \rcp{} gain. The mean pool holds \val{b8.nist.qwen.pooled.lev.budget.mean-pool-docs} passages per query.

\paragraph{Panels and metrics.}
The 14 rerankers rank each query's candidates in two panels. The pool panel reorders the full pool. The top-150 panel reorders only the \val{draft.2.poolsize} fused passages. Leaderboards of 18 systems add the three retrievers and their fusion. We write qrel-nDCG@10 for nDCG@10 on the NIST grades with linear gain and an ideal ranking over every graded passage, as in trec\_eval. The exponential gain $2^{\mathrm{grade}}-1$ is a reported variant. \rcpndcg{}@10 and Count-nDCG@10 draw their ideal rankings from the candidate pool. We test each year separately. For each of the \val{draft.5.dl.pairs} reranker pairs, a paired $t$-test compares per-query scores. A metric separates the pair when the two-sided test rejects equal means at the \val{draft.2.siglevel}\% level. A pair is NIST-decisive when qrel-nDCG@10 on the NIST grades separates it. A metric agrees with NIST on such a pair when its mean difference has the same sign as NIST's.

\paragraph{Judge and reduced protocol.}
The main TREC-DL scores of \rcp{} and Count-nDCG use the five criteria of \cref{tab:criteria}. The reduced protocol skips the reversed Stage~A copies of the full protocol and adds fewer adaptive windows per passage. A passage appears in \val{a2.budget.qwen.pooled.windowsperdoc} Stage~A windows on average, against \val{draft.3.nano.stagea.windowsperdoc} on NanoBEIR. Stage~B places each passage in \val{a2.budget.qwen.pooled.placementsperdoc} windows (\cref{tab:app-budget}).

\paragraph{Calibration fits.}
The main calibration CAL-A fits the 2PL model once over \val{draft.app.ef.dl.calfit.queries} queries. The fit covers the queries of both years and \val{draft.app.ef.dl.calfit.extra} MS~MARCO development queries without NIST grades. CAL-B fits the same Stage~A scores and a repeated Stage~B run on the TREC-DL queries alone. CAL-C fits the judgments of the second judge (\cref{app:trecdl-judge2}). In all three fits, C3 is easier to pass than C2. C5 is far more difficult than any other criterion.

\FloatBarrier
\section{Human Study Protocol}\label{app:human}

This appendix describes the design and sample of the human study, what annotators saw, who the annotators were, and how reliably they graded. The study has two layers. The label layer compares document grades with qrels and \rcp{} scores (\cref{app:human-labels}). The verdict layer asks which of two lists annotators prefer (\cref{app:human-verdicts}). Both layers rest on the same annotations and are therefore not independent evidence.

\subsection{Design and sample}\label{app:human-design}\label{app:human-collection}

The study asks whether annotators side with \rcpndcg{} more often than with qrel-nDCG where exactly one of the two metrics names the annotators' choice. A contest is one query with two rerankers. Contests of the same query are correlated. Every interval therefore treats a query as one cluster (\cref{app:notation-intervals}).

The study covers \val{a1a.desc.datasets.nanomteb} of the 13 NanoBEIR datasets, \val{a1a.desc.datasets.bright} of the 12 BRIGHT tasks, and \val{a1a.desc.datasets.vidore} English-language corpora of ViDoRe v3. We exclude NanoSCIDOCS, whose qrels mark the papers that the query paper cites. We also exclude BRIGHT's five coding and theorem tasks (LeetCode, Pony, AoPS, and both TheoremQA tasks), where relevance rests on a shared line of reasoning. The suites contribute \val{a1a.desc.queries.nanomteb}, \val{a1a.desc.queries.bright}, and \val{a1a.desc.queries.vidore} distinct queries. The contests draw on \val{a1a.desc.rerankers} rerankers in \val{a1a.desc.pairs} distinct pairs (\cref{tab:app-human-sample}).

We oversample disagreements, contests on which the two metrics name opposite winners. We split disagreements into decisive and marginal strata by the size of the two metrics' margins, their score differences between the two lists. The draw compares the metrics at depth 10, the usual evaluation depth, on the full rankings. Every verdict analysis compares them at depth 5 on the displayed documents. Annotators grade only each list's top five, since the union of two top-ten lists, up to 20 documents, is too long to grade with care. The sample also includes concordant contests, on which both metrics name the same winner. We targeted \val{n15.human.target-contests} contests per suite. The final sample holds \val{a1a.desc.contests} contests, \val{draft.app.d.strata.signflip} drawn as disagreements and \val{draft.app.d.strata.concordant} as concordant contests. These contests cover \val{a1a.desc.queries} distinct queries, because some queries enter more than one contest with different reranker pairs. Contests with a zerank reranker make up \val{a1a.desc.zerank-share}\% of the sample. We also report results on the \val{a1a.desc.nn-contests} contests without zerank. The first valid review was completed on \val{a1a.desc.collection-start-date} and the last on \val{a1a.desc.collection-end-date}.

\begin{table}[!htbp]
\centering
\footnotesize
\caption{Composition of the human study. A review is one annotator's grading of one contest. Since annotators worked across suites, their counts do not add up.}
\label{tab:app-human-sample}
\begin{tabular*}{\textwidth}{@{}l@{\extracolsep{\fill}}rrrr@{}}
\toprule
& NanoBEIR & BRIGHT & ViDoRe v3 & All \\
\midrule
Datasets or corpora & \val{a1a.desc.datasets.nanomteb} & \val{a1a.desc.datasets.bright} & \val{a1a.desc.datasets.vidore} & \val{a1a.desc.datasets} \\
Queries & \val{a1a.desc.queries.nanomteb} & \val{a1a.desc.queries.bright} & \val{a1a.desc.queries.vidore} & \val{a1a.desc.queries} \\
Contests & \val{a1a.desc.contests.nanomteb} & \val{a1a.desc.contests.bright} & \val{a1a.desc.contests.vidore} & \val{a1a.desc.contests} \\
Reviews & \val{a1a.desc.reviews.nanomteb} & \val{a1a.desc.reviews.bright} & \val{a1a.desc.reviews.vidore} & \val{a1a.desc.reviews} \\
Annotators & \val{a1a.desc.annotators.nanomteb} & \val{a1a.desc.annotators.bright} & \val{a1a.desc.annotators.vidore} & \val{a1a.desc.annotators} \\
Grades & \val{b2.desc.nano.grades} & \val{b2.desc.bright.grades} & \val{b2.desc.vidore.grades} & \val{a1a.desc.grades} \\
Distinct graded documents & \val{a1b.n.docs.nano} & \val{a1b.n.docs.bright} & \val{a1b.n.docs.vidore} & \val{a1b.n.docs.all} \\
\bottomrule
\end{tabular*}
\end{table}
\FloatBarrier

\subsection{What annotators saw}\label{app:human-instrument}

Annotators saw the query, a one-line task description of what counts as relevant, and the union of both rerankers' top five documents in random order. The annotators never saw the rerankers, the two lists, the qrels, or any metric score. For each document they answered the four questions of \cref{tab:app-human-instrument}. Written guidelines asked them to grade only the shown text at face value and to choose Either between adjacent grades where it was offered. The questions echo the wording of four Stage~B criteria (\cref{tab:criteria}), with Topical matching C1, Useful C2, Answers C4, and Complete C5. The match lies in the wording only. We did not measure whether the questions and the criteria are equivalent.

\begin{table}[!htbp]
\centering
\footnotesize
\caption{The annotators' four questions, in the order shown, with the grade that the most stringent yes gives and the gain of that grade in the verdicts. Answers and Complete also accept Either, which gives the half grade in parentheses and a linearly interpolated gain when it directly follows the last yes.}
\label{tab:app-human-instrument}
\begin{tabularx}{\textwidth}{@{}l X l c c@{}}
\toprule
Name & Question & Answers & Grade & Gain \\
\midrule
-- & (every answer is no) & -- & 0 & 0 \\
Topical & \emph{Is this document clearly about the query's specific topic or need, not just the same broad area?} & yes, no & 1 & \val{draft.app.d.gain.g1} \\
Useful & \emph{Can you extract at least one concrete fact, step, or example that specifically addresses this query?} & yes, no & 2 & \val{draft.app.d.gain.g2} \\
Answers & \emph{Does it directly answer the query's main question or need?} & yes, Either, no & 3 (2.5) & \val{draft.app.d.gain.g3} (\val{draft.app.d.gain.g2p5}) \\
Complete & \emph{Does it fully resolve the query's need on its own, with nothing important left to add?} & yes, Either, no & 4 (3.5) & \val{draft.app.d.gain.g4} (\val{draft.app.d.gain.g3p5}) \\
\bottomrule
\end{tabularx}
\end{table}

Annotators answer every question for every document, whatever their earlier answers. The most stringent yes therefore sets the grade. A fixed gain map (\cref{tab:app-human-instrument}) turns grades into gains. Each annotator's grades give both lists an nDCG@5 under one shared ideal ranking. An annotator's pick is the list with the higher nDCG@5. A gap below \val{draft.app.d.tie-margin} is a tie. A Copeland score over the annotator's rankings of equally graded documents can break such a tie. The score counts wins plus half the ties in these pairwise judgments, weighted by each document's effect on the score difference. The contest verdict counts picks of list~A against picks of list~B and ignores ties (\cref{app:notation-stats}). Of the \val{a1a.desc.contests} contests, \val{a1a.desc.verdict-a} have verdict~A, \val{a1a.desc.verdict-b} verdict~B, and \val{a1a.desc.verdict-tie} a tie. \Cref{app:human-controls} recomputes the verdicts under other gain maps.

A metric's margin is its nDCG@5 difference between the two displayed lists, with the ideal DCG@5 taken over the query's whole candidate pool. For \rcpndcg{}, the ideal falls back to the displayed documents in \val{n15.human.rcp-ideal-fallback} contests. Under the same rule, a metric names a winner only when its margin is not a tie. A contest is decided when it has a verdict and exactly one metric names the verdict's list. The main measure is the share of decided contests whose verdict sides with \rcpndcg{}. In \val{a1a.desc.ndcg-blind} contests, no document with a positive qrel appears in either top five. \Cref{tab:app-human-verdicts} also reports the main measure without these contests. On some contests one displayed document is an attention check, an off-topic document without a positive qrel and with a near-zero \rcp{} gain. A grade above 0 on it voids that annotator's review. Both metrics call it irrelevant. Voiding a review therefore cannot favor either metric.

\subsection{Annotators and terms}\label{app:human-annotators}

The annotators were \val{a1a.desc.annotators} external contractors, each with at least a BSc in the subject of the documents they graded. They were not employees of the authors' organization or professional relevance assessors. They were paid \$\val{draft.4.pay} per hour and gave informed consent. The study had no review by an institutional review board. A separate team ran recruitment, quality control, and all contact with the annotators behind an information barrier. The authors wrote the questions and the guidelines. Three annotators graded each contest, \val{a1a.desc.reviews} reviews in all. The separate team could exclude an annotator only on grounds that involve neither metric, such as failed attention checks or low agreement with other annotators. The team excluded one annotator, whose reviews we do not use and who is not among the \val{a1a.desc.annotators} counted above. Re-admitting this annotator's \val{b17.excl.reviews} reviews moves the main measure from \val{b17.without.primary}\% to \val{b17.with.primary}\%. The results on disagreements alone and on disagreements with both margins at least \val{draft.tab.verdicts.decisive-margin} are unchanged (\cref{app:human-verdicts}).
\FloatBarrier

\subsection{Reliability}\label{app:human-reliability}

We measure how consistently annotators grade, first for single grades and then for the panel means and verdicts that our analyses use. Krippendorff's $\alpha$ for one annotator's grade on the 0 to 4 scale, where 0 is chance and 1 is perfect agreement, is \val{b2.grade.all.alpha-interval} (95\% CI [\val{b2.grade.all.alpha-interval.lo}, \val{b2.grade.all.alpha-interval.hi}]). This value is below the level that \citet{krippendorff2018content} treats as reliable. The analyses instead rest on panel means and majority verdicts. For the panel means of $k$ annotators, the reliability ICC(1,$k$) (\cref{app:notation-stats}) is \val{b2.grade.all.rel-panel} (95\% CI [\val{b2.grade.all.rel-panel.lo}, \val{b2.grade.all.rel-panel.hi}]), within the band that \citet{koo2016guideline} call good. Pairs of annotators of one contest agree on a pick in \val{b2.verdict.all.all.ab.agree}\% of cases, where one half is chance. Fleiss' $\kappa$ corrects this share for chance agreement, with 0 at chance and 1 for perfect agreement. On the picks, $\kappa$ is \val{b2.verdict.all.all.ab.kappa} (95\% CI [\val{b2.verdict.all.all.ab.kappa.lo}, \val{b2.verdict.all.all.ab.kappa.hi}]). For a rare answer, $\kappa$ can be low even when raw agreement is high. Among yes or no answers, annotators say yes to Complete in only \val{b2.bin.all.r4.either-missing.prev}\% of cases. On these answers, $\kappa$ is \val{b2.bin.all.r4.either-missing.kappa}. Gwet's AC1, which depends less on prevalence, is \val{b2.bin.all.r4.either-missing.ac1} on the same answers. With any one annotator left out, $\alpha$ stays between \val{b2.loao.alpha-interval.min} and \val{b2.loao.alpha-interval.max} and verdict agreement between \val{b2.loao.verdict-ab-agree.min}\% and \val{b2.loao.verdict-ab-agree.max}\%.

\FloatBarrier
\providecommand{\efv}[1]{\ensuremath{\val{#1}}}%
\section{Diagnostics and Calibration}
\label{app:calibration}
This section details the 2PL fits behind \cref{sec:diagnostics,sec:calibration}: their parameters and transfer, label-free diagnostics, calibration against human and NIST grades, and the shape of the gain.

\subsection{Fitted parameters and transfer across datasets}\label{app:nano-bright-params}

This subsection supports the diagnostic claims of \cref{sec:diagnostics}.
\Cref{tab:nb-params} lists the criterion parameters $\cpar{\beta_c}$ and $\cpar{\gamma_c}$ of the four fits on NanoBEIR and BRIGHT, one per suite and judge, shared across all queries (\cref{sec:calibration-merge}).

\begin{table}[t]
  \centering
  \caption{Parameters of the four fits. (a)~Difficulty $\cpar{\beta_c}$ and discrimination $\cpar{\gamma_c}$ per criterion. (b)~The 5th and 95th percentiles of $\qpar{\tau_j}$, their ratio, and the mean offset $\qpar{\alpha_j}$.}
  \label{tab:nb-params}
  \scriptsize
  \setlength{\tabcolsep}{4.5pt}
  \begin{tabular}{@{}l*{2}{r}@{\hspace{9pt}}*{2}{r}@{\hspace{9pt}}*{2}{r}@{\hspace{9pt}}*{2}{r}@{}}
    \multicolumn{9}{@{}l}{(a) Criterion parameters} \\
    \toprule
    & \multicolumn{4}{c}{NanoBEIR} & \multicolumn{4}{c}{BRIGHT} \\
    \cmidrule(lr){2-5}\cmidrule(l){6-9}
    & \multicolumn{2}{c}{\model{Qwen3.5-397B}} & \multicolumn{2}{c}{\model{gpt-oss-120b}} & \multicolumn{2}{c}{\model{Qwen3.5-397B}} & \multicolumn{2}{c}{\model{gpt-oss-120b}} \\
    \cmidrule(lr){2-3}\cmidrule(lr){4-5}\cmidrule(lr){6-7}\cmidrule(l){8-9}
    & $\cpar{\beta_c}$ & $\cpar{\gamma_c}$ & $\cpar{\beta_c}$ & $\cpar{\gamma_c}$ & $\cpar{\beta_c}$ & $\cpar{\gamma_c}$ & $\cpar{\beta_c}$ & $\cpar{\gamma_c}$ \\
    \midrule
    C1 & $\val{a4.nano.qwen.beta.c1}$ & \val{a4.nano.qwen.gamma.c1} & $\val{a4.nano.oss.beta.c1}$ & \val{a4.nano.oss.gamma.c1} & $\val{a4.bright.qwen.beta.c1}$ & \val{a4.bright.qwen.gamma.c1} & $\val{a4.bright.oss.beta.c1}$ & \val{a4.bright.oss.gamma.c1} \\
    C2 & $\val{a4.nano.qwen.beta.c2}$ & \val{a4.nano.qwen.gamma.c2} & $\val{a4.nano.oss.beta.c2}$ & \val{a4.nano.oss.gamma.c2} & $\val{a4.bright.qwen.beta.c2}$ & \val{a4.bright.qwen.gamma.c2} & $\val{a4.bright.oss.beta.c2}$ & \val{a4.bright.oss.gamma.c2} \\
    C3 & $\val{a4.nano.qwen.beta.c3}$ & \val{a4.nano.qwen.gamma.c3} & $\val{a4.nano.oss.beta.c3}$ & \val{a4.nano.oss.gamma.c3} & $\val{a4.bright.qwen.beta.c3}$ & \val{a4.bright.qwen.gamma.c3} & $\val{a4.bright.oss.beta.c3}$ & \val{a4.bright.oss.gamma.c3} \\
    C4 & $\val{a4.nano.qwen.beta.c4}$ & \val{a4.nano.qwen.gamma.c4} & $\val{a4.nano.oss.beta.c4}$ & \val{a4.nano.oss.gamma.c4} & $\val{a4.bright.qwen.beta.c4}$ & \val{a4.bright.qwen.gamma.c4} & $\val{a4.bright.oss.beta.c4}$ & \val{a4.bright.oss.gamma.c4} \\
    C5 & $\val{a4.nano.qwen.beta.c5}$ & \val{a4.nano.qwen.gamma.c5} & $\val{a4.nano.oss.beta.c5}$ & \val{a4.nano.oss.gamma.c5} & $\val{a4.bright.qwen.beta.c5}$ & \val{a4.bright.qwen.gamma.c5} & $\val{a4.bright.oss.beta.c5}$ & \val{a4.bright.oss.gamma.c5} \\
    \bottomrule
  \end{tabular}\par
  \vspace{3pt}
  \setlength{\tabcolsep}{5pt}
  \begin{tabular}{@{}llrrrr@{}}
    \multicolumn{6}{@{}l}{(b) Per-query parameters} \\
    \toprule
    Suite & Judge & $\qpar{\tau_j}$ 5\% & $\qpar{\tau_j}$ 95\% & 95\%/5\% & mean $\qpar{\alpha_j}$ \\
    \midrule
    NanoBEIR & \model{Qwen3.5-397B} & \val{a4.nano.qwen.tau.p5} & \val{a4.nano.qwen.tau.p95} & \val{a4.nano.qwen.tau.ratio-p95-p5} & $\val{a4.nano.qwen.alpha.mean}$ \\
    NanoBEIR & \model{gpt-oss-120b} & \val{a4.nano.oss.tau.p5} & \val{a4.nano.oss.tau.p95} & \val{a4.nano.oss.tau.ratio-p95-p5} & $\val{a4.nano.oss.alpha.mean}$ \\
    BRIGHT & \model{Qwen3.5-397B} & \val{a4.bright.qwen.tau.p5} & \val{a4.bright.qwen.tau.p95} & \val{a4.bright.qwen.tau.ratio-p95-p5} & $\val{a4.bright.qwen.alpha.mean}$ \\
    BRIGHT & \model{gpt-oss-120b} & \val{a4.bright.oss.tau.p5} & \val{a4.bright.oss.tau.p95} & \val{a4.bright.oss.tau.ratio-p95-p5} & $\val{a4.bright.oss.alpha.mean}$ \\
    \bottomrule
  \end{tabular}
\end{table}

\begin{figure}[t]
  \centering
  \includegraphics[width=0.34\textwidth]{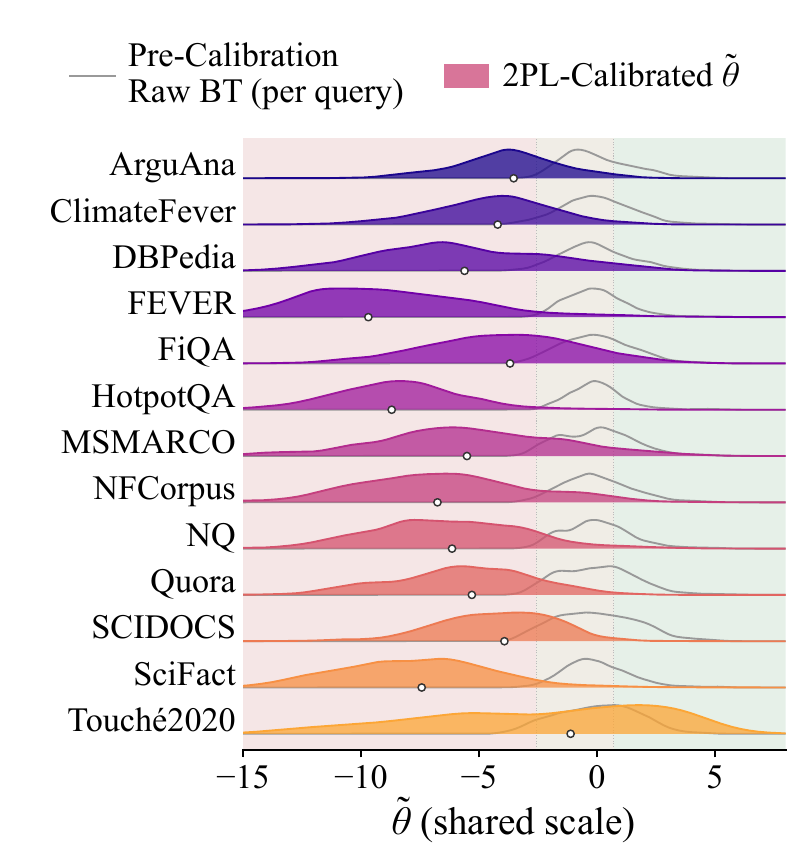}
  \caption{Per NanoBEIR dataset, pool documents' raw BT scores, centered per query (grey), and their calibrated scores (color; dot: median $\qpar{\alpha_j}$).}
  \label{fig:ridges}
\end{figure}

The per-query scale $\qpar{\tau_j}$ and offset $\qpar{\alpha_j}$ calibrate each query's BT scores (\cref{eq:calibrated-theta}). \Cref{fig:ridges} shows the raw and calibrated pool scores. As BT scores are centered within each query, the offset equals the query's mean calibrated pool score. In every fit, this mean lies well below C1's difficulty (\cref{tab:nb-params}).

To test whether the criterion parameters transfer across datasets, we fit them on all datasets but one, freeze them, and fit only $\qpar{\tau_j}$ and $\qpar{\alpha_j}$ on the held-out queries. We score the fit by the binary cross-entropy (BCE) of the held-out Stage~B responses, where lower is better. A baseline predicts each response by its criterion's pass rate on the training datasets. On NanoBEIR with \model{Qwen3.5-397B}, the mean held-out BCE drops by \val{draft.app.g.lodo.nano.qwen.red}\% against the baseline. With \model{gpt-oss-120b}, the reduction on NanoBEIR is \val{draft.app.g.lodo.nano.oss.red}\%. On BRIGHT, it is \val{draft.app.g.lodo.bright.qwen.red}\% with \model{Qwen3.5-397B} and \val{draft.app.g.lodo.bright.oss.red}\% with \model{gpt-oss-120b}. A variant with every discrimination set to 1 and the difficulties kept (a Rasch form) raises the mean held-out BCE relative to the 2PL. The increase is \val{draft.app.g.lodo.nano.qwen.onepl}--\val{draft.app.g.lodo.nano.oss.onepl}\% on NanoBEIR and \val{draft.app.g.lodo.bright.qwen.onepl}--\val{draft.app.g.lodo.bright.oss.onepl}\% on BRIGHT. Each range spans the two judges.

\subsection{Label-free diagnostics}\label{app:calibration-diagnostics}\label{app:vidore-diagnostics}

The shared criterion parameters let the fitted model check a rubric without relevance labels. \Cref{tab:app-diagnostics} reports six checks for each fit, except that the two TREC-DL fits lack $Q_3$. The first column states whether the difficulties rise from C1 to C5, as the rubric intends. The second names the criterion with the largest discrimination and its margin over the runner-up. The last column gives the shares of pool documents below the difficulty of C1 and at or above that of C4. Yen's $Q_3$~\citep{yen1984effects} tests local independence, the 2PL assumption that a document's answers to different criteria are independent given its ability. The statistic correlates the residuals of two criteria over documents, with $|Q_3| > 0.2$ as the warning sign. On NanoBEIR with \model{Qwen3.5-397B}, the raw correlations between criteria average \val{nr.m3.q3.rawr.nano.qwen} in absolute value, whereas the mean $|Q_3|$ is \val{a4.nano.qwen.q3.meanabs}.

\begin{table}[t]
  \centering
  \footnotesize
  \caption{Label-free diagnostics of each fit. Efficiency is the mean test information (\cref{eq:app-info}) at each query's ideal top 10 or at the rerankers' top-10 lists, divided by the peak information. The $|Q_3|$ column gives the mean over the ten criterion pairs and, in parentheses, the number of pairs with $|Q_3| > 0.2$. TREC-DL values cover the evaluated queries of both years, and BRIGHT pools still contained the documents that BRIGHT's official evaluation excludes (\cref{app:nano-bright}).}
  \label{tab:app-diagnostics}
  \scriptsize
  \setlength{\tabcolsep}{3pt}
  \begin{tabular*}{\textwidth}{@{}l@{\extracolsep{\fill}}cccccc@{}}
    \toprule
    & \begin{tabular}[b]{@{}c@{}}$\cpar{\beta_c}$\\rising\end{tabular} & \begin{tabular}[b]{@{}c@{}}Largest\\$\cpar{\gamma_c}$\end{tabular} & \begin{tabular}[b]{@{}c@{}}$\qpar{\tau_j}$ p5--p95\\(ratio)\end{tabular} & \begin{tabular}[b]{@{}c@{}}Efficiency\\ideal top 10 /\\reranker top 10\end{tabular} & \begin{tabular}[b]{@{}c@{}}$|Q_3|$\\(pairs)\end{tabular} & \begin{tabular}[b]{@{}c@{}}Pool (\%)\\$<\cpar{\beta_{\mathrm{C1}}}$ / $\geq\cpar{\beta_{\mathrm{C4}}}$\end{tabular} \\
    \midrule
    \multicolumn{7}{@{}l}{NanoBEIR} \\
    \quad \model{Qwen3.5-397B} & \val{a56.xs.nano397.beta-monotone} & \val{a56.xs.nano397.gamma-argmax} (+\val{a56.xs.nano397.gamma-argmax-margin}) & \val{a56.xs.nano397.tau.p5}--\val{a56.xs.nano397.tau.p95} (\val{a56.xs.nano397.tau.p95-over-p5}) & \val{a56.xs.nano397.eff.oracle} / \val{a56.xs.nano397.eff.system} & \val{a4.nano.qwen.q3.meanabs} (\val{a4.nano.qwen.q3.pairs-gt02}) & \val{a56.xs.nano397.gain.region.lt-b1} / \val{a56.xs.nano397.gain.region.ge-b4} \\
    \quad \model{gpt-oss-120b} & \val{a56.xs.nanooss.beta-monotone} & \val{a56.xs.nanooss.gamma-argmax} (+\val{a56.xs.nanooss.gamma-argmax-margin}) & \val{a56.xs.nanooss.tau.p5}--\val{a56.xs.nanooss.tau.p95} (\val{a56.xs.nanooss.tau.p95-over-p5}) & \val{a56.xs.nanooss.eff.oracle} / \val{a56.xs.nanooss.eff.system} & \val{a4.nano.oss.q3.meanabs} (\val{a4.nano.oss.q3.pairs-gt02}) & \val{a56.xs.nanooss.gain.region.lt-b1} / \val{a56.xs.nanooss.gain.region.ge-b4} \\
    \multicolumn{7}{@{}l}{BRIGHT} \\
    \quad \model{Qwen3.5-397B} & \val{a56.xs.bright397.beta-monotone} & \val{a56.xs.bright397.gamma-argmax} (+\val{a56.xs.bright397.gamma-argmax-margin}) & \val{a56.xs.bright397.tau.p5}--\val{a56.xs.bright397.tau.p95} (\val{a56.xs.bright397.tau.p95-over-p5}) & \val{a56.xs.bright397.eff.oracle} / \val{a56.xs.bright397.eff.system} & \val{a4.bright.qwen.q3.meanabs} (\val{a4.bright.qwen.q3.pairs-gt02}) & \val{a56.xs.bright397.gain.region.lt-b1} / \val{a56.xs.bright397.gain.region.ge-b4} \\
    \quad \model{gpt-oss-120b} & \val{a56.xs.brightoss.beta-monotone} & \val{a56.xs.brightoss.gamma-argmax} (+\val{a56.xs.brightoss.gamma-argmax-margin}) & \val{a56.xs.brightoss.tau.p5}--\val{a56.xs.brightoss.tau.p95} (\val{a56.xs.brightoss.tau.p95-over-p5}) & \val{a56.xs.brightoss.eff.oracle} / \val{a56.xs.brightoss.eff.system} & \val{a4.bright.oss.q3.meanabs} (\val{a4.bright.oss.q3.pairs-gt02}) & \val{a56.xs.brightoss.gain.region.lt-b1} / \val{a56.xs.brightoss.gain.region.ge-b4} \\
    \multicolumn{7}{@{}l}{ViDoRe v3} \\
    \quad \model{Qwen3.5-397B} & \val{a56.xs.vidore397.beta-monotone} & \val{a56.xs.vidore397.gamma-argmax} (+\val{a56.xs.vidore397.gamma-argmax-margin}) & \val{a56.xs.vidore397.native.tau.p5}--\val{a56.xs.vidore397.native.tau.p95} (\val{a56.xs.vidore397.native.tau.p95-over-p5}) & \val{a56.xs.vidore397.native.eff.oracle} / \val{a56.xs.vidore397.native.eff.system} & \val{a3.diag.q3.meanabs} (\val{a3.diag.q3.pairs-gt-02}) & \val{a56.xs.vidore397.native.gain.region.lt-b1} / \val{a56.xs.vidore397.native.gain.region.ge-b4} \\
    \multicolumn{7}{@{}l}{TREC-DL} \\
    \quad \model{Qwen3.6-27B} & \val{a56.xs.cala.beta-monotone} & \val{a56.xs.cala.gamma-argmax} (+\val{a56.xs.cala.gamma-argmax-margin}) & \val{a56.xs.cala.dl.tau.p5}--\val{a56.xs.cala.dl.tau.p95} (\val{a56.xs.cala.dl.tau.p95-over-p5}) & \val{a56.xs.cala.dl.eff.oracle} / \val{a56.xs.cala.dl.eff.system} & -- & \val{a56.xs.cala.dl.gain.region.lt-b1} / \val{a56.xs.cala.dl.gain.region.ge-b4} \\
    \quad \model{gpt-oss-120b} & \val{a56.xs.calc.beta-monotone} & \val{a56.xs.calc.gamma-argmax} (+\val{a56.xs.calc.gamma-argmax-margin}) & \val{a56.xs.calc.tau.p5}--\val{a56.xs.calc.tau.p95} (\val{a56.xs.calc.tau.p95-over-p5}) & \val{a56.xs.calc.eff.oracle} / \val{a56.xs.calc.eff.system} & -- & \val{a56.xs.calc.gain.region.lt-b1} / \val{a56.xs.calc.gain.region.ge-b4} \\
    \bottomrule
  \end{tabular*}
\end{table}

On TREC-DL, the difficulties do not rise in rubric order. C3 is easier to pass than C2, and C5 is far harder, with difficulty \val{a56.xs.cala.beta.c5} under \model{Qwen3.6-27B}. Across the fits, the ratio of the 95th to the 5th percentile of $\qpar{\tau_j}$ lies between \val{a56.xs.calc.tau.p95-over-p5} and \val{a56.xs.nanooss.tau.p95-over-p5}. Since this spread contains estimation noise, \cref{app:calibration-link} tests per-query scaling against human grades. Information efficiency divides the mean test information at a query's top 10 by the peak information of the rubric. A value of 1 means that the rubric is most precise at the top 10. At the ideal top 10, the efficiency is highest on ViDoRe v3 and BRIGHT and lowest on TREC-DL. On TREC-DL, the rubric is most informative below the documents that fill the top 10. On TREC-DL, a search over rubric subsets finds that dropping C2 raises the efficiency (\cref{app:robustness-lattice}). On ViDoRe v3, the flagged $Q_3$ pairs are C1 with C2 and C4 with C5. A rubric can also pass every check and still miss aspects of relevance, which the human and NIST tests of \cref{sec:external} cover.

\subsection{Calibration against human grades and simpler maps}\label{app:calibration-human}\label{app:calibration-link}
This subsection supports the calibration checks of \cref{sec:calibration}, using the human grades of \cref{app:human}. Because nDCG divides by each query's ideal DCG, calibration changes \rcpndcg{} only through the ratios between the nonlinear gains of a query's documents. In our human study, scoring a contest's two lists with this gain on the uncalibrated BT scores reverses the winner in \val{draft.3.whycal.opposite} of the \val{a1b.abl.union.all.n-decided} contests with a verdict.
\paragraph{The gain against single human grades.}
We ask whether the gain $g$ behaves like the probability that one annotator rates a document Useful or better. We measure this with the expected calibration error (ECE) of \cref{sec:calibration} over the study's \val{a1a.desc.grades} single grades. Over all three suites, the ECE is \val{a1b.ece.ge2.all} (95\% CI \val{a1b.ece.ge2.all.ci}), where 0 means perfect calibration. At other thresholds, the ECE is higher, \val{a1b.ece.ge1.all} against Topical or above and \val{a1b.ece.ge3.all} against Answers or above. Of the three thresholds, the gain thus reads best as the chance of a Useful or better rating. Since no human grade enters the fit, this correspondence is an empirical result.

\paragraph{The judge's own answers.}
We check the 2PL's pass probabilities against the judge's own Stage~B answers. As a simpler map, we fit a per-criterion logistic regression on the raw BT score. On NanoBEIR, the logistic regression has an ECE of \val{nr.figB.cal.ece.logreg}, against \val{nr.figB.cal.ece.2pl} for the 2PL. We reran this comparison on the \val{a3.n.questions} native-language questions of ViDoRe v3 with judge \model{Qwen3.5-397B}. On ViDoRe v3, a sigmoid of the raw BT score is miscalibrated, with an ECE of \val{a3.diag.ece.bt-only} over 20 bins. With all queries pooled, the logistic regression and the 2PL are both well calibrated, with ECEs of \val{a3.diag.ece.logreg} and \val{a3.diag.ece.2pl}. We also compare per-query Brier scores, the mean squared error of the predicted pass probabilities on a query's answers (lower is better). On ViDoRe v3, the Brier score is \val{a3.diag.brier.2pl} for the 2PL against \val{a3.diag.brier.logreg} for the logistic regression. The 2PL has the lower Brier score on \val{a3.diag.pqbrier.2pl-better}\% of queries (95\% CI [\val{a3.diag.pqbrier.2pl-better.lo}, \val{a3.diag.pqbrier.2pl-better.hi}]\%). On NanoBEIR, the Brier score is lower for the 2PL than for the logistic regression on \val{nr.figB.cal.pq.2pl-better}\% of queries. Against always predicting the global pass rate, the Brier skill score of the 2PL is \val{nr.figB.cal.bss.2pl}, where 0 means no improvement and 1 is perfect. Of the baseline Brier score, criterion base rates remove \val{nr.figB.cal.attr.criterion}\%, the BT signal \val{nr.figB.cal.attr.bt}\%, and the per-query $\qpar{\tau_j}$ and $\qpar{\alpha_j}$ of the 2PL \val{nr.figB.cal.attr.calibration}\%. With \model{gpt-oss-120b}, the Brier skill score is \val{nr.figB.cal.oss.bss.2pl}.

\paragraph{Human grades across queries.}
\Cref{tab:app-link} compares raw BT scores, per-query $z$-scores, and the calibrated $\dpar{\tilde{\theta}}$ of \rcp{}, pooled over the graded documents of each study suite. Per-query standardization lowers both the Spearman correlation and the AUC below raw BT in every suite. A rescaling that ignores the rubric thus does not supply a scale shared across queries. Calibration raises the Spearman correlation with the documents' mean grade over raw BT in every suite. The increases are \val{a1b.cal.rho.mean.nano.d-cal-raw} on NanoBEIR (95\% CI \val{a1b.cal.rho.mean.nano.d-cal-raw.ci}), \val{a1b.cal.rho.mean.bright.d-cal-raw} on BRIGHT (95\% CI \val{a1b.cal.rho.mean.bright.d-cal-raw.ci}), and \val{a1b.cal.rho.mean.vidore.d-cal-raw} on ViDoRe v3 (95\% CI \val{a1b.cal.rho.mean.vidore.d-cal-raw.ci}). Since annotators differ in leniency, we also compare same-suite pairs of documents from different queries that the same annotator graded differently. Calibration raises the share of these pairs that the score orders like the annotator's grades above that of raw BT. The increase is \val{draft.5.sameann.d-cal-raw} (95\% CI [\val{draft.5.sameann.d-cal-raw.lo}, \val{draft.5.sameann.d-cal-raw.hi}], clustered by annotator).

\begin{table}[t]
  \centering
  \footnotesize
  \caption{Three scores against blind human grades, pooled over each suite's graded documents. Spearman $\rho$ with the mean grade and AUC at Useful (0.5 is chance). $z$(BT) standardizes raw BT scores per query. Judge \model{Qwen3.5-397B}.}
  \label{tab:app-link}
  \scriptsize
  \setlength{\tabcolsep}{4.5pt}
  \begin{tabular}{@{}lcccccc@{}}
    \toprule
    & \multicolumn{3}{c}{Spearman $\rho$ with mean grade} & \multicolumn{3}{c}{AUC at Useful} \\
    \cmidrule(lr){2-4}\cmidrule(l){5-7}
    Score & NanoBEIR & BRIGHT & ViDoRe v3 & NanoBEIR & BRIGHT & ViDoRe v3 \\
    \midrule
    Raw BT $\dpar{\hat{\theta}^{\mathrm{BT}}}$ & \val{a1b.cal.rho.mean.nano.raw} & \val{a1b.cal.rho.mean.bright.raw} & \val{a1b.cal.rho.mean.vidore.raw} & \val{a1b.cal.auc.ge2.mean.nano.raw} & \val{a1b.cal.auc.ge2.mean.bright.raw} & \val{a1b.cal.auc.ge2.mean.vidore.raw} \\
    Per-query $z$(BT) & \val{a1b.cal.rho.mean.nano.z} & \val{a1b.cal.rho.mean.bright.z} & \val{a1b.cal.rho.mean.vidore.z} & \val{a1b.cal.auc.ge2.mean.nano.z} & \val{a1b.cal.auc.ge2.mean.bright.z} & \val{a1b.cal.auc.ge2.mean.vidore.z} \\
    Calibrated $\dpar{\tilde{\theta}}$ (\rcp{}) & \val{a1b.cal.rho.mean.nano.cal} & \val{a1b.cal.rho.mean.bright.cal} & \val{a1b.cal.rho.mean.vidore.cal} & \val{a1b.cal.auc.ge2.mean.nano.cal} & \val{a1b.cal.auc.ge2.mean.bright.cal} & \val{a1b.cal.auc.ge2.mean.vidore.cal} \\
    \bottomrule
  \end{tabular}
\end{table}

\subsection{Placement against NIST grades and ViDoRe qrels}\label{app:trecdl-placement}\label{app:vidore-order}\label{app:vidore-xlang}\label{tab:app-trecdl-placement}

This subsection supports the NIST checks of \cref{sec:calibration}, repeats them on the ViDoRe v3 qrels, and tests whether the tournament's order reproduces across translations.

\paragraph{NIST grades on TREC-DL.}
Within a query, the concordance of \rcp{} with the NIST grades is the share of differently graded passage pairs that \rcp{} orders as NIST does (0.5 is chance). This concordance reflects Stage~A alone (\cref{prop:order}) and is \val{a2.conc.dl19.overall} in 2019 and \val{a2.conc.dl20.overall} in 2020 (\cref{sec:calibration}).

The pooled Spearman correlation of $\dpar{\tilde{\theta}}$ with the NIST grade over all graded passages tests whether scores place passages of different queries on one scale. The ceiling reassigns each query's values of $\dpar{\tilde{\theta}}$ in the order of the NIST grades, and the floor shuffles them at random within each query. The observed correlation is \val{a2.rho.dl19.pooled} in 2019 and \val{a2.rho.dl20.pooled} in 2020, against ceilings of \val{a2.rho.dl19.ceiling} and \val{a2.rho.dl20.ceiling} and floors of \val{a2.rho.dl19.floor} and \val{a2.rho.dl20.floor}.
On graded passages that received at least one Stage~B placement, calibration raises this correlation over raw BT. The increase is $+$\val{a2.calraw.dl19.rho.diffrawa} in 2019, with an interval that includes zero, and $+$\val{a2.calraw.dl20.rho.diffrawa} in 2020 (95\% CI [\val{draft.app.ef.dl20.rho.diffrawa.lo}, \val{draft.app.ef.dl20.rho.diffrawa.hi}]).

\paragraph{ViDoRe qrels.}
The ViDoRe v3 qrels also test the within-question order from Stage~A alone. Of the same-question pairs of a grade-1 page and an unjudged pool page, \rcp{} places the grade-1 page higher in \val{a3.conc.g1g0} (0.5 is chance). For pairs of a grade-2 page and an unjudged page, the share is \val{a3.conc.g2g0}. Placement across questions is the part that calibration can change. With all native questions pooled, the AUC for a qrel grade of at least 1 is \val{a3.auc.qrel-ge1.rcp-gain} for the \rcp{} gain, against \val{a3.auc.qrel-ge1.raw-bt} for raw BT. The improvement over raw BT is $+$\val{a3.auc.qrel-ge1.diff.rcp-minus-raw-bt} (95\% CI [\val{a3.auc.qrel-ge1.diff.rcp-minus-raw-bt.lo}, \val{a3.auc.qrel-ge1.diff.rcp-minus-raw-bt.hi}]). \Cref{app:robustness-count} tests the tournament's order on the ViDoRe v3 pages that Count-nDCG ties.

\paragraph{Reproducibility across translations.}
We reran both stages on each of the six parallel language versions of a ViDoRe v3 question. The tournament's order reproduces across the 15 pairs of versions at a Spearman correlation of \val{a3.xlang.rho} (95\% CI [\val{a3.xlang.rho.lo}, \val{a3.xlang.rho.hi}]), against \val{draft.5.xlang.count} for Count. A pooled Count that uses the Stage~B placements of three language versions has about as many judgments per page as \rcp{}. This pooled Count still reproduces at only \val{draft.5.xlang.count3}, \val{draft.app.ef.vd.xlang.cnt3.diff} below \rcp{} (95\% CI [\val{draft.app.ef.vd.xlang.cnt3.diff.lo}, \val{draft.app.ef.vd.xlang.cnt3.diff.hi}]). The higher reproducibility of the tournament's order is thus not an artifact of the judging budget. A reproducible order can still be wrong, because a judge that misreads a question in every language would reproduce its error.

\subsection{Gain shape}\label{app:calibration-gain}\label{tab:app-gain}

Any increasing map from $\dpar{\tilde{\theta}}$ to $[0, 1]$ keeps each query's order (\cref{prop:order}). Such a map can therefore move a leaderboard only through the ratios between a query's gains. We compare the gain of \cref{eq:theta-gain} with five alternatives. Three replace the discrimination weights: the unweighted mean of the pass probabilities, the logistic $\sigma(\dpar{\tilde{\theta}})$ centered at zero difficulty, and the C4 pass probability. Two linear maps rescale $\dpar{\tilde{\theta}}$, one clipped to $[0, 1]$ and one per-query min-max. The min-max map ignores calibration and equals SABER-Math's per-query rescaling under a linear gain~\citep{georgiev2026sabermath}.

On NanoBEIR and BRIGHT, the 14 rerankers' leaderboards under any alternative gain correlate with the paper's at a Spearman $\rho$ of at least \val{a4.nano.gainfn.oss.rho-w-linminmax}. Leaderboards thus change little, but the share of separated pairs depends more. On NanoBEIR with \model{Qwen3.5-397B}, the per-query min-max gain separates \val{a4.nano.sens.gain-linminmax-qwen}\% of the pairs against \val{a4.nano.sens.rcp-qwen}\% for the paper's gain. The unweighted gain separates about as many pairs. The discrimination weights thus matter little for resolution. Since the min-max gain ignores calibration yet separates more pairs, we read the share of separated pairs as resolution, not validity. We keep the discrimination-weighted gain for its probability reading (\cref{app:calibration-human}) and its test-information slope (\cref{app:theory}).

\FloatBarrier

\FloatBarrier
\section{Agreement with Ground Truth and Human Labels}
\label{app:labels}
This section details the label comparisons of \cref{sec:blind-eval}: single documents against human grades, Study 1 on sets of documents, and the multi-hop queries of HotpotQA.

\subsection{The label layer}\label{app:human-labels}

This subsection supports the external-annotator results of \cref{sec:blind-eval}.
The label layer of the human study (\cref{app:human}) compares each document's mean grade with its qrel and its \rcp{} gain. The unit is a query-document pair. An annotator who graded a pair twice counts once, with the mean of the two grades. \Cref{fig:documents} shows the pooled results, and \cref{tab:app-human-labels} breaks them down by suite.

\begin{figure}[t]
  \centering
  \includegraphics[trim={0 0 274 19},clip,width=122bp]{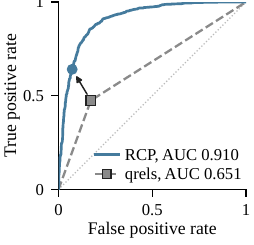}
  \caption{Documents against blind human grades (judge \model{Qwen3.5-397B}): ROC curves for a mean grade of Useful or better over \val{a1b.n.docs.all} documents. The binary qrels give one operating point (square), and the arrow leads to \rcp{} flagging the same share of documents (circle).}
  \label{fig:documents}
\end{figure}

With the three suites pooled, the AUC at Useful (0.5 is chance) is \val{a1b.auc.ge2.all.qrel} for the qrels, and \val{a1b.auc.ge2.all.gain} for the \rcp{} gain. The gain exceeds the qrels in every suite. Among documents with qrel 0 alone, the \rcp{} gain still separates useful from other documents with an AUC of \val{a1b.auc.within-q0.ge2.all.gain} (95\% CI \val{a1b.auc.within-q0.ge2.all.gain.ci}). For same-query pairs of a useful and a non-useful document where the qrels mark only the non-useful one positive, \rcp{} ranks the useful one higher with probability \val{a1b.pair.q0rel-over-q1irr.within.all}. The \rcp{} gain here is the one stored with the study pools. Within each suite, the calibrated score of the final fit reaches the same AUC. The comparison thus does not depend on which fit supplies the scores.

Of the documents with qrel 0, \val{a1b.mis.q0-ge2.all}\% have a mean grade of at least Useful. Of the documents with a positive qrel, \val{a1b.mis.q1-lt2.all}\% fall below Useful. On NanoBEIR and BRIGHT, qrel 0 mostly means unjudged. Since a document on the topic can still correctly carry qrel 0, these shares measure disagreement with the annotators rather than qrel errors. When \rcp{} flags the same share of documents as relevant (\val{a1b.qrel.rate.all}\%), it recovers \val{a1b.op.all.rcp.tpr}\% of the useful documents against \val{a1b.op.all.qrel.tpr}\% for the qrels. The false-positive rates are \val{a1b.op.all.rcp.fpr}\% for \rcp{} against \val{a1b.op.all.qrel.fpr}\% for the qrels.

With the three suites pooled, the Spearman correlation of $\dpar{\tilde{\theta}}$ with the panel's mean grade is \val{a1b.par.full.all.cal}. We compare it with the panel's reliability, the expected correlation between the mean grades of two independent panels of the same size. We estimate it from random half-panels and step the correlation up with the Spearman--Brown formula, which gives \val{a1b.par.balanced.all.sb}. The calibrated score thus tracks the panel's mean grade about as closely as a second panel would.

\label{app:human-crowd}
Inspired by the model-based label audit of \citet{land2026auditing}, we also let the ten study rerankers vote on each pooled document. The vote and the \rcp{} gain disagree on a document when exactly one of them reaches one half. On the \val{draft.app.jm.b3.disagree.n} such documents, annotators side with the gain in \val{draft.app.jm.b3.disagree.rcp}\% (95\% CI [\val{draft.app.jm.b3.disagree.rcp.lo}, \val{draft.app.jm.b3.disagree.rcp.hi}]\%). On documents whose qrel the annotators' grades contradict, the zerank rerankers agree with the qrel less often than the other rerankers (difference \val{draft.app.jm.b3.zerank.M}\,pp, 95\% CI [\val{draft.app.jm.b3.zerank.M.lo}, \val{draft.app.jm.b3.zerank.M.hi}]\,pp).

For each of the four annotation questions, we fit a logistic curve of the annotators' yes answers on the calibrated score $\dpar{\tilde{\theta}}$. The threshold where this curve crosses one half rises from Topical to Complete in every suite. The threshold of Useful lies between the difficulties $\cpar{\beta_c}$ of C2 and C3, that of Answers near C4 (on BRIGHT between C4 and C5), and that of Complete near C5. On BRIGHT, the threshold of Complete lies far above C5.

\begin{table}[!htbp]
\centering
\footnotesize
\caption{Label layer by suite. AUCs at Useful (a mean grade of at least 2; one half is chance). The \rcp{} gain is the one stored with the study pools. The row marked final fit comes from the final runs of \cref{app:nano-bright,app:vidore}. The matched operating point lets \rcp{} flag as many documents as the qrels do. TPR is the share of useful documents flagged, and FPR the share of other documents flagged. Agreement with one annotator uses a gain of one half and a grade of 2 as thresholds. Concordant and disagreement refer to the sampling stratum of a document's contest. ViDoRe v3 qrels count any positive grade as relevant.}
\label{tab:app-human-labels}
\begin{tabular*}{\textwidth}{@{}l@{\extracolsep{\fill}}cccc@{}}
\toprule
& NanoBEIR & BRIGHT & ViDoRe v3 & All \\
\midrule
Graded documents & \val{a1b.n.docs.nano} & \val{a1b.n.docs.bright} & \val{a1b.n.docs.vidore} & \val{a1b.n.docs.all} \\
Documents with a positive qrel (\%) & \val{a1b.qrel.rate.nano} & \val{a1b.qrel.rate.bright} & \val{a1b.qrel.rate.vidore} & \val{a1b.qrel.rate.all} \\
\midrule
AUC: qrels & \val{a1b.auc.ge2.nano.qrel} & \val{a1b.auc.ge2.bright.qrel} & \val{a1b.auc.ge2.vidore.qrel} & \val{a1b.auc.ge2.all.qrel} \\
\quad \rcp{} gain & \val{a1b.auc.ge2.nano.gain} & \val{a1b.auc.ge2.bright.gain} & \val{a1b.auc.ge2.vidore.gain} & \val{a1b.auc.ge2.all.gain} \\
\quad \rcp{} score $\dpar{\tilde{\theta}}$, final fit & \val{a1b.cal.auc.ge2.mean.nano.cal} & \val{a1b.cal.auc.ge2.mean.bright.cal} & \val{a1b.cal.auc.ge2.mean.vidore.cal} & \val{a1b.cal.auc.ge2.mean.all.cal} \\
\quad \rcp{} gain among qrel-0 documents & \val{a1b.auc.within-q0.ge2.nano.gain} & \val{a1b.auc.within-q0.ge2.bright.gain} & \val{a1b.auc.within-q0.ge2.vidore.gain} & \val{a1b.auc.within-q0.ge2.all.gain} \\
\quad qrels\,/\,\rcp{} gain, concordant contests & \val{a1b.samp.pool-concordant.ge2.nano.qrel}\,/\,\val{a1b.samp.pool-concordant.ge2.nano.gain} & \val{a1b.samp.pool-concordant.ge2.bright.qrel}\,/\,\val{a1b.samp.pool-concordant.ge2.bright.gain} & \val{a1b.samp.pool-concordant.ge2.vidore.qrel}\,/\,\val{a1b.samp.pool-concordant.ge2.vidore.gain} & \val{a1b.samp.pool-concordant.ge2.all.qrel}\,/\,\val{a1b.samp.pool-concordant.ge2.all.gain} \\
\quad qrels\,/\,\rcp{} gain, disagreement contests & \val{a1b.samp.pool-signflip.ge2.nano.qrel}\,/\,\val{a1b.samp.pool-signflip.ge2.nano.gain} & \val{a1b.samp.pool-signflip.ge2.bright.qrel}\,/\,\val{a1b.samp.pool-signflip.ge2.bright.gain} & \val{a1b.samp.pool-signflip.ge2.vidore.qrel}\,/\,\val{a1b.samp.pool-signflip.ge2.vidore.gain} & \val{a1b.samp.pool-signflip.ge2.all.qrel}\,/\,\val{a1b.samp.pool-signflip.ge2.all.gain} \\
\midrule
qrel 0, mean grade at least Useful (\%) & \val{a1b.mis.q0-ge2.nano} & \val{a1b.mis.q0-ge2.bright} & \val{a1b.mis.q0-ge2.vidore} & \val{a1b.mis.q0-ge2.all} \\
positive qrel, mean grade below Useful (\%) & \val{a1b.mis.q1-lt2.nano} & \val{a1b.mis.q1-lt2.bright} & \val{a1b.mis.q1-lt2.vidore} & \val{a1b.mis.q1-lt2.all} \\
Matched point, TPR\,/\,FPR (\%): qrels & \val{a1b.op.nano.qrel.tpr}\,/\,\val{a1b.op.nano.qrel.fpr} & \val{a1b.op.bright.qrel.tpr}\,/\,\val{a1b.op.bright.qrel.fpr} & \val{a1b.op.vidore.qrel.tpr}\,/\,\val{a1b.op.vidore.qrel.fpr} & \val{a1b.op.all.qrel.tpr}\,/\,\val{a1b.op.all.qrel.fpr} \\
\quad \rcp{} & \val{a1b.op.nano.rcp.tpr}\,/\,\val{a1b.op.nano.rcp.fpr} & \val{a1b.op.bright.rcp.tpr}\,/\,\val{a1b.op.bright.rcp.fpr} & \val{a1b.op.vidore.rcp.tpr}\,/\,\val{a1b.op.vidore.rcp.fpr} & \val{a1b.op.all.rcp.tpr}\,/\,\val{a1b.op.all.rcp.fpr} \\
\midrule
Agreement with one annotator (\%): judge & \val{a1b.agree.judge-human.nano} & \val{a1b.agree.judge-human.bright} & \val{a1b.agree.judge-human.vidore} & \val{a1b.agree.judge-human.all} \\
\quad qrels & \val{a1b.agree.qrel-human.nano} & \val{a1b.agree.qrel-human.bright} & \val{a1b.agree.qrel-human.vidore} & \val{a1b.agree.qrel-human.all} \\
\quad a second annotator & \val{a1b.agree.human-human.nano} & \val{a1b.agree.human-human.bright} & \val{a1b.agree.human-human.vidore} & \val{a1b.agree.human-human.all} \\
\bottomrule
\end{tabular*}
\end{table}

\subsection{Study 1: qrel sets against \texorpdfstring{\rcp{}}{RCP} sets}\label{app:meta-study1}

This subsection reports Study 1, a comparison of qrel sets with tournament-selected sets.
Study 1 asks whether the qrels or the LLM judge select the better relevant documents for a query. For a NanoBEIR query, the qrel set holds its positive documents. The \rcp{} set holds equally many documents with the highest \rcp{} gain, which by \cref{prop:order} are the tournament's top documents. The cases are the \val{a56.old.s1.elig.differ} NanoBEIR queries, of \val{a56.old.s1.elig.total} in total, that have 1 to 10 positive qrels and whose two sets differ.

Three LLM judges from other model families, \model{GLM-5.3-flash}, \model{DeepSeek-4.1-flash}, and \model{Kimi K3}, compare the two sets blind to their source (\cref{app:meta}). A case counts as decided when at least two judges pick the same set in both display orders. By majority, the judges prefer the \rcp{} set on \val{n9.s1.nano.maj}\% of the \val{n9.s1.nano.maj.n} decided cases (\val{n9.s1.nano.maj.k} of \val{n9.s1.nano.maj.n}; 95\% CI [\val{n9.s1.nano.maj.lo}, \val{n9.s1.nano.maj.hi}]\%). On the \val{n9.s1.nano.unan.n} cases that all three decide the same way, the share is \val{n9.s1.nano.unan}\%. The \rcp{} set wins most decided cases on \val{n9.s1.nano.ds.favor} of the \val{n9.s1.nano.ds.count} datasets. HotpotQA is the exception (\cref{app:meta-hotpotqa}). The result holds for sets selected with \model{gpt-oss-120b} as judge (\val{n9.s1.nanooss.maj}\% of \val{n9.s1.nanooss.maj.n} decided cases).

On BRIGHT, the judges prefer the \rcp{} set on \val{n9.s1.bright12.maj}\% of \val{n9.s1.bright12.maj.n} decided cases over all twelve tasks (95\% CI [\val{n9.s1.bright12.maj.lo}, \val{n9.s1.bright12.maj.hi}]\%). On the \val{n9.s1.bright12.unan.n} cases that all three decide the same way, the share is \val{n9.s1.bright12.unan}\%. On AoPS, LeetCode, and TheoremQA Questions, BRIGHT's evaluation removes some listed documents, such as the query's own problem, from the rankings (\cref{app:nano-bright-scope}). The original sets for these tasks kept them. We therefore rebuilt the \rcp{} sets with these documents removed before taking the top documents. Of the \val{n9.s1.brightfix.orig} original cases on these tasks, \val{n9.s1.brightfix.cases} remain. We dropped \val{n9.s1.brightfix.identical} cases whose two sets became identical. We also dropped \val{n9.s1.brightfix.nonquery} TheoremQA Questions cases whose query is an annotator's note instead of a problem. On the rebuilt cases alone, the share is \val{n9.s1.brightfix.maj}\% of \val{n9.s1.brightfix.maj.n}. For sets selected with \model{gpt-oss-120b}, we report only the nine tasks other than AoPS, LeetCode, and TheoremQA Questions, with \val{n9.s1.brightoss.unaff.maj}\% of \val{n9.s1.brightoss.unaff.maj.n} decided cases.

Two of the authors, both IR researchers, judged all \val{a56.meta.s1.nano.cases} cases blind to which set was which, with ties allowed. As authors, they are not independent of the method, but they never saw which set came from which source. Of the \val{draft.5.study1.human.n} non-tie judgments of these two authors, \val{draft.5.study1.human}\% prefer the \rcp{} set (95\% Wilson CI [\val{draft.app.jm.s1.human.lo}, \val{draft.app.jm.s1.human.hi}]\%). On the \val{n9.s1.human.unan.n} cases where both of them prefer the same set, the \rcp{} set wins \val{n9.s1.human.unan.k} (\val{n9.s1.human.unan}\%, 95\% CI [\val{n9.s1.human.unan.lo}, \val{n9.s1.human.unan.hi}]\%).

Study 1 conditions on queries whose two sets differ and does not compare the label sources on all documents. The external annotators' grades cover all \val{a1b.n.docs.all} documents of the human study, whether or not the qrels and \rcp{} agree on them (\cref{sec:blind-eval,app:human-labels}).

\subsection{Multi-hop queries: HotpotQA}\label{app:meta-hotpotqa}

NanoHotpotQA is the dataset on which Studies~1 and~2 (\cref{app:meta-study1,app:meta-study2}) went most strongly against \rcp{}. Its questions need a primary and a secondary fact, each from its own clearly relevant document, and its qrels mark exactly these two documents. \rcp{} and our annotators judge each document on its own and do not penalize redundancy. A second document on the primary fact can therefore outrank the document with the secondary fact, which the qrel set contains by construction. In Study 1, the judges preferred the qrel set in all \val{n9.s1.nano.hotpot.maj.n} decided HotpotQA cases. The authors who judged Study~1 gave the \rcp{} set \val{draft.app.jm.s1.hotpotqa.human.k} of their \val{draft.app.jm.s1.hotpotqa.human.n} non-tie judgments. In Study 2, the list of \rcpndcg{} won \val{n10.s2.nano.ds.hotpotqa.maj.k} of \val{n10.s2.nano.ds.hotpotqa.maj.n} decided HotpotQA sign-flips.

In the external human study, only \val{nr.integ.main.hotpotqa.n} HotpotQA contests enter the main measure, too few for a claim. In \val{nr.integ.main.hotpotqa.k} of them, the majority verdict sides with \rcpndcg{}. Relevance that arises only from combining documents would need criteria or comparisons that look at several documents at once. \rcp{} has neither, nor does it model diversity or the needs of individual users.

\FloatBarrier
\section{Ranking Results}
\label{app:ranking}
This section details the ranking results of \cref{sec:ranking-results}: failure modes of nDCG@10, sensitivity and scores per dataset, second judges, the zerank rerankers, robustness, and Count-nDCG.

\subsection{Failure modes of nDCG@10}\label{app:nano-bright-failure}\label{app:vidore-failure}
\begin{figure}[t]
  \centering
  \includegraphics[width=0.58\textwidth]{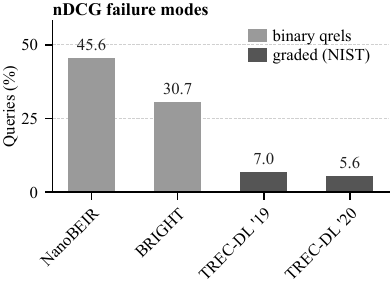}
  \caption{Share of queries on which nDCG@10 saturates, floors, or compresses the rerankers. With the graded NIST labels of TREC-DL, nDCG@10 rarely shows any of these modes.}
  \label{fig:failmodes}
\end{figure}
On a query, nDCG@10 \emph{saturates} when all 14 rerankers score 1 and \emph{floors} when all 14 score 0. The metric \emph{compresses} the rerankers when neither holds and at least half of the \val{a56.robust.factorial.pairs} reranker pairs tie, that is, differ by at most $10^{-6}$. \Cref{fig:failmodes} gives the share of queries with at least one mode on NanoBEIR, BRIGHT, and TREC-DL. At least one mode occurs on \val{a4.nano.fm14.any}\% of the NanoBEIR queries and on \val{b20.new.a4.bright.fm14.any}\% of the BRIGHT queries. On ViDoRe v3, the share is \val{a3.fail.any}\%. With the graded NIST labels of TREC-DL, it is \val{a2.fail.pool.dl19.nistlin.anypct}\% in 2019 and \val{a2.fail.pool.dl20.nistlin.anypct}\% in 2020.

\subsection{Sensitivity per dataset}\label{app:nano-bright-sensitivity}\label{app:vidore-sensitivity}
Per dataset, \rcpndcg{} separates more pairs than qrel-nDCG on \val{nr.m4.nano.sens.rcp-gt-ndcg.datasets} of the 13 NanoBEIR datasets and on all \val{nr.m4.bright.sens.rcp-gt-ndcg.datasets} BRIGHT tasks. The exceptions are ArguAna and HotpotQA, where \rcpndcg{} separates fewer pairs than qrel-nDCG. \Cref{app:meta-hotpotqa} discusses HotpotQA.

\subsection{Reranker scores per dataset}\label{app:nano-bright-tables}
\providecommand{\rcpndcg}{\textup{RCP-nDCG}}
\definecolor{famctxl}{HTML}{2166AC}
\definecolor{famqwen}{HTML}{B2182B}
\definecolor{famcohere}{HTML}{1B7837}
\definecolor{famvoyage}{HTML}{7B3294}
\definecolor{famze}{HTML}{B35806}
\definecolor{famjina}{HTML}{555555}
\providecommand{\crtup}[1]{\,\textcolor{green!50!black}{$\blacktriangle$#1}}
\providecommand{\crtdn}[1]{\,\textcolor{red!60!black}{$\blacktriangledown$#1}}
\providecommand{\crtrot}[1]{\makebox[0pt][l]{\rotatebox{40}{\textsf{#1}}}}
\begin{table}[t]
  \centering
  \caption{nDCG@10 with the benchmark's qrels and \rcpndcg{}@10 on NanoBEIR (judge \model{Qwen3.5-397B}), in \%. Each block is sorted by its mean over the 13 datasets. Arrows give each reranker's change in rank from nDCG@10 to \rcpndcg{}@10. Bold marks the best reranker per column. Colors mark the model family: \textcolor{famcohere}{Cohere}, \textcolor{famctxl}{Contextual~AI}, \textcolor{famqwen}{Qwen3}, \textcolor{famvoyage}{Voyage~AI}, \textcolor{famze}{ZeroEntropy}, and \textcolor{famjina}{Jina~AI}.}
  \label{tab:cal-ndcg-qwen-nanomteb}
  \footnotesize
  \setlength{\tabcolsep}{1.5pt}
  \begin{tabular}{@{}l *{13}{w{c}{2.0em}} w{c}{2.4em}@{\hspace{5pt}}}
    \toprule
     & \multicolumn{13}{c}{\textbf{Datasets}} & \textbf{Mean} \\
    \cmidrule(lr){2-14} \cmidrule(l){15-15}
    \textbf{Reranker} & \crtrot{ArguAna} & \crtrot{ClimateFever} & \crtrot{DBPedia} & \crtrot{FEVER} & \crtrot{FiQA} & \crtrot{HotpotQA} & \crtrot{MSMARCO} & \crtrot{NFCorpus} & \crtrot{NQ} & \crtrot{Quora} & \crtrot{SCIDOCS} & \crtrot{SciFact} & \crtrot{Touch\'e2020} & \\
    \hline
    \rowcolor{gray!15}
     & \multicolumn{14}{c}{\rule{0pt}{2.4ex}\textbf{nDCG@10}} \\
    \hline
    \textcolor{famqwen}{\rule{0pt}{2.3ex}\baseline{Qwen3-RR-8B}} & 86.0 & \textbf{54.1} & 75.0 & 96.2 & 69.0 & 94.3 & 72.3 & 46.7 & 81.6 & 97.1 & 48.9 & 83.6 & 65.9 & \textbf{74.7} \\
    \textcolor{famctxl}{\baseline{CTXL-RR-6B}} & 88.9 & 42.1 & 77.1 & \textbf{98.7} & 68.3 & \textbf{98.0} & 69.6 & 48.6 & \textbf{84.4} & 89.4 & 50.8 & \textbf{90.3} & 62.5 & 74.5 \\
    \textcolor{famqwen}{\baseline{Qwen3-RR-4B}} & 81.5 & 51.7 & 74.6 & 97.3 & 67.1 & 94.3 & 70.4 & 47.4 & 82.1 & 96.0 & 51.5 & 82.9 & \textbf{66.6} & 74.1 \\
    \textcolor{famcohere}{\baseline{RR4-Pro}} & 77.0 & 52.3 & 75.8 & 97.8 & \textbf{70.8} & 94.4 & 67.2 & 47.9 & 81.7 & 96.4 & 46.4 & 81.7 & 63.6 & 73.3 \\
    \textcolor{famvoyage}{\baseline{Voyage2.5}} & \textbf{94.0} & 37.6 & \textbf{77.6} & 96.8 & 67.6 & 95.0 & 70.2 & 46.6 & 80.9 & 93.3 & 44.2 & 82.3 & 55.6 & 72.4 \\
    \textcolor{famctxl}{\baseline{CTXL-RR-2B}} & 76.9 & 41.1 & 75.9 & 95.2 & 65.1 & 96.3 & 68.3 & 48.2 & 80.1 & 93.2 & \textbf{51.7} & 84.4 & 62.4 & 72.2 \\
    \textcolor{famvoyage}{\baseline{Voyage2.5-lt}} & 88.9 & 39.4 & 76.5 & 94.8 & 69.5 & 93.9 & 70.7 & 45.2 & 77.6 & 94.1 & 42.4 & 82.0 & 57.8 & 71.8 \\
    \textcolor{famcohere}{\baseline{RR4-Fast}} & 72.3 & 52.3 & 74.3 & 96.9 & 61.5 & 93.7 & \textbf{72.6} & 46.9 & 79.2 & 95.4 & 44.6 & 80.4 & 61.8 & 71.7 \\
    \textcolor{famctxl}{\baseline{CTXL-RR-1B}} & 68.4 & 40.0 & 73.7 & 96.0 & 62.4 & 95.2 & 68.2 & 46.9 & 73.4 & 93.7 & 47.1 & 83.9 & 61.1 & 70.0 \\
    \textcolor{famqwen}{\baseline{Qwen3-RR-0.6B}} & 75.8 & 47.8 & 70.6 & 96.5 & 58.6 & 92.0 & 67.9 & 46.1 & 75.3 & 93.6 & 44.9 & 80.4 & 58.7 & 69.9 \\
    \textcolor{famjina}{\baseline{Jina-RR-v3}} & 79.5 & 39.6 & 68.0 & 97.7 & 57.9 & 96.3 & 69.5 & \textbf{49.0} & 79.2 & \textbf{97.9} & 34.3 & 78.5 & 52.9 & 69.3 \\
    \textcolor{famze}{\baseline{zerank-1-sm}} & 62.5 & 33.5 & 73.8 & 91.7 & 62.2 & 89.4 & 62.8 & 46.4 & 74.8 & 82.8 & 44.1 & 82.2 & 65.3 & 67.0 \\
    \textcolor{famze}{\baseline{zerank-1}} & 61.0 & 33.7 & 73.6 & 92.3 & 63.6 & 90.2 & 62.0 & 48.1 & 77.9 & 78.6 & 42.9 & 82.9 & 64.3 & 67.0 \\
    \textcolor{famze}{\baseline{zerank-2}} & 56.2 & 36.5 & 74.0 & 92.1 & 62.9 & 90.1 & 59.2 & 47.1 & 77.0 & 76.8 & 45.3 & 83.2 & 64.0 & 66.5 \\
    \hline
    \rowcolor{gray!15}
     & \multicolumn{14}{c}{\rule{0pt}{2.4ex}\textbf{\rcpndcg{}@10}} \\
    \hline
    \textcolor{famze}{\rule{0pt}{2.3ex}\baseline{zerank-2}}\crtup{13} & 55.3 & \textbf{88.9} & \textbf{90.6} & \textbf{93.4} & 90.8 & 89.5 & \textbf{88.3} & \textbf{87.9} & 83.9 & 86.5 & \textbf{90.4} & \textbf{91.1} & \textbf{97.0} & \textbf{87.2} \\
    \textcolor{famze}{\baseline{zerank-1}}\crtup{11} & 56.3 & 87.3 & \textbf{90.6} & 93.0 & \textbf{91.3} & 88.9 & 88.0 & 87.5 & 84.2 & 88.1 & 89.8 & 90.3 & 96.8 & 87.1 \\
    \textcolor{famvoyage}{\baseline{Voyage2.5}}\crtup{2} & 69.3 & 85.4 & 89.9 & 92.8 & 83.1 & \textbf{92.0} & 85.8 & 85.0 & 81.8 & 87.4 & 79.6 & 87.0 & 87.1 & 85.1 \\
    \textcolor{famze}{\baseline{zerank-1-sm}}\crtup{8} & 54.9 & 83.7 & 87.1 & 91.3 & 90.0 & 88.0 & 85.8 & 81.9 & 79.3 & 83.8 & \textbf{90.4} & 88.5 & 96.1 & 84.7 \\
    \textcolor{famcohere}{\baseline{RR4-Pro}}\crtdn{1} & 62.8 & 71.1 & 86.9 & 85.7 & 86.1 & 90.3 & 85.9 & 85.4 & 83.0 & 88.1 & 80.4 & 88.0 & 90.2 & 83.4 \\
    \textcolor{famctxl}{\baseline{CTXL-RR-6B}}\crtdn{4} & \textbf{70.4} & 68.2 & 88.4 & 84.8 & 80.3 & 90.1 & 87.9 & 84.6 & \textbf{85.4} & 78.9 & 84.7 & 84.9 & 91.4 & 83.1 \\
    \textcolor{famvoyage}{\baseline{Voyage2.5-lt}} & 67.0 & 82.4 & 87.7 & 91.7 & 79.7 & 90.3 & 84.3 & 79.2 & 77.7 & 81.1 & 84.7 & 84.4 & 83.0 & 82.5 \\
    \textcolor{famqwen}{\baseline{Qwen3-RR-4B}}\crtdn{5} & 64.7 & 68.7 & 84.5 & 88.2 & 84.2 & 89.7 & 83.3 & 82.1 & 80.9 & 87.6 & 86.4 & 85.4 & 86.1 & 82.4 \\
    \textcolor{famqwen}{\baseline{Qwen3-RR-8B}}\crtdn{8} & 68.2 & 64.7 & 86.6 & 87.7 & 85.3 & 89.7 & 83.6 & 83.1 & 81.3 & \textbf{88.3} & 79.1 & 85.7 & 87.6 & 82.4 \\
    \textcolor{famctxl}{\baseline{CTXL-RR-2B}}\crtdn{4} & 57.3 & 70.7 & 87.5 & 85.4 & 81.1 & 89.2 & 87.2 & 82.1 & 83.7 & 81.2 & 88.8 & 85.2 & 90.6 & 82.3 \\
    \textcolor{famcohere}{\baseline{RR4-Fast}}\crtdn{3} & 61.9 & 64.7 & 83.7 & 85.4 & 83.0 & 89.1 & 83.1 & 79.5 & 77.8 & 85.9 & 84.6 & 86.1 & 89.9 & 81.1 \\
    \textcolor{famctxl}{\baseline{CTXL-RR-1B}}\crtdn{3} & 56.7 & 67.0 & 86.1 & 85.4 & 81.9 & 87.0 & 86.0 & 79.9 & 78.1 & 84.8 & 84.6 & 83.5 & 89.3 & 80.8 \\
    \textcolor{famqwen}{\baseline{Qwen3-RR-0.6B}}\crtdn{3} & 59.4 & 73.7 & 78.3 & 87.3 & 77.0 & 87.0 & 80.2 & 73.5 & 71.6 & 83.9 & 84.2 & 83.3 & 78.7 & 78.3 \\
    \textcolor{famjina}{\baseline{Jina-RR-v3}}\crtdn{3} & 54.4 & 64.9 & 73.3 & 79.4 & 70.0 & 84.4 & 81.6 & 66.0 & 69.1 & 75.1 & 44.2 & 69.8 & 77.4 & 70.0 \\
    \bottomrule
  \end{tabular}
\end{table}

\begin{table}[t]
  \centering
  \caption{nDCG@10 with the benchmark's qrels and \rcpndcg{}@10 on BRIGHT (judge \model{Qwen3.5-397B}), in \%, with the documents that BRIGHT excludes for a query removed. Each block is sorted by its mean over the 12 tasks. Arrows give each reranker's change in rank from nDCG@10 to \rcpndcg{}@10. Bold marks the best reranker per column. Colors mark the model family as in \cref{tab:cal-ndcg-qwen-nanomteb}.}
  \label{tab:cal-ndcg-qwen-bright}
  \footnotesize
  \setlength{\tabcolsep}{1.5pt}
  \begin{tabular}{@{}l *{12}{w{c}{2.0em}} w{c}{2.4em}@{\hspace{5pt}}}
    \toprule
     & \multicolumn{12}{c}{\textbf{Datasets}} & \textbf{Mean} \\
    \cmidrule(lr){2-13} \cmidrule(l){14-14}
    \textbf{Reranker} & \crtrot{AoPS} & \crtrot{Biology} & \crtrot{Earth Sci.} & \crtrot{Economics} & \crtrot{LeetCode} & \crtrot{Pony} & \crtrot{Psychology} & \crtrot{Robotics} & \crtrot{StackOvfl.} & \crtrot{Sust.\ Liv.} & \crtrot{ThQA-Q} & \crtrot{ThQA-T} & \\
    \hline
    \rowcolor{gray!15}
     & \multicolumn{13}{c}{\rule{0pt}{2.4ex}\textbf{nDCG@10}} \\
    \hline
    \textcolor{famze}{\rule{0pt}{2.3ex}\baseline{zerank-1}} & 14.4 & \textbf{54.6} & 51.7 & 33.6 & 21.0 & 38.7 & \textbf{47.6} & 33.2 & 27.6 & 44.0 & 46.2 & \textbf{46.9} & \textbf{38.3} \\
    \textcolor{famze}{\baseline{zerank-2}} & 13.5 & 52.9 & \textbf{52.0} & \textbf{37.4} & 22.1 & 25.2 & 46.5 & \textbf{37.7} & \textbf{31.0} & \textbf{44.9} & 46.3 & 43.6 & 37.8 \\
    \textcolor{famvoyage}{\baseline{Voyage2.5}} & 10.2 & 47.9 & 49.6 & 30.9 & 21.1 & \textbf{41.6} & 42.3 & 30.9 & 26.0 & 38.9 & 37.3 & 30.0 & 33.9 \\
    \textcolor{famcohere}{\baseline{RR4-Pro}} & 12.7 & 50.2 & 49.2 & 31.0 & 29.2 & 3.8 & 44.5 & 27.8 & 24.2 & 35.8 & \textbf{48.1} & 45.6 & 33.5 \\
    \textcolor{famze}{\baseline{zerank-1-sm}} & 10.4 & 46.5 & 45.6 & 29.1 & 19.2 & 24.4 & 41.0 & 25.5 & 26.6 & 40.1 & 45.7 & 40.2 & 32.9 \\
    \textcolor{famcohere}{\baseline{RR4-Fast}} & 10.9 & 44.6 & 45.5 & 30.8 & 20.4 & 3.9 & 38.2 & 23.6 & 23.7 & 34.2 & 46.7 & 38.8 & 30.1 \\
    \textcolor{famvoyage}{\baseline{Voyage2.5-lt}} & 11.1 & 38.4 & 45.3 & 26.6 & 24.9 & 26.2 & 38.0 & 23.5 & 22.2 & 30.8 & 33.3 & 21.9 & 28.5 \\
    \textcolor{famqwen}{\baseline{Qwen3-RR-4B}} & 12.2 & 30.7 & 36.6 & 21.7 & \textbf{40.7} & 18.7 & 27.9 & 23.1 & 26.2 & 23.9 & 42.8 & 36.3 & 28.4 \\
    \textcolor{famqwen}{\baseline{Qwen3-RR-8B}} & 9.4 & 32.7 & 34.4 & 25.6 & 29.4 & 21.2 & 30.7 & 27.4 & 24.4 & 25.4 & 39.7 & 34.1 & 27.9 \\
    \textcolor{famjina}{\baseline{Jina-RR-v3}} & 8.4 & 26.8 & 39.2 & 27.8 & 10.2 & 18.1 & 29.2 & 23.2 & 12.5 & 31.5 & 28.3 & 32.6 & 24.0 \\
    \textcolor{famqwen}{\baseline{Qwen3-RR-0.6B}} & 8.7 & 19.1 & 24.6 & 14.9 & 28.9 & 3.1 & 17.1 & 13.8 & 15.4 & 11.3 & 37.1 & 30.0 & 18.7 \\
    \textcolor{famctxl}{\baseline{CTXL-RR-6B}} & \textbf{17.4} & 10.0 & 16.7 & 18.0 & 15.0 & 29.2 & 20.7 & 12.3 & 5.7 & 11.5 & 38.2 & 22.8 & 18.1 \\
    \textcolor{famctxl}{\baseline{CTXL-RR-1B}} & 11.4 & 10.2 & 24.1 & 18.6 & 21.2 & 16.3 & 19.5 & 13.9 & 10.8 & 12.3 & 37.2 & 18.7 & 17.9 \\
    \textcolor{famctxl}{\baseline{CTXL-RR-2B}} & 14.5 & 11.6 & 23.7 & 16.0 & 14.7 & 12.5 & 23.3 & 9.5 & 6.5 & 12.4 & 34.0 & 28.7 & 17.3 \\
    \hline
    \rowcolor{gray!15}
     & \multicolumn{13}{c}{\rule{0pt}{2.4ex}\textbf{\rcpndcg{}@10}} \\
    \hline
    \textcolor{famze}{\rule{0pt}{2.3ex}\baseline{zerank-2}}\crtup{1} & \textbf{75.5} & \textbf{85.9} & \textbf{86.5} & \textbf{87.8} & 77.0 & \textbf{73.3} & \textbf{82.2} & \textbf{85.2} & \textbf{82.3} & \textbf{84.0} & \textbf{81.3} & 77.8 & \textbf{81.6} \\
    \textcolor{famze}{\baseline{zerank-1}}\crtdn{1} & 70.1 & 84.5 & 85.4 & 86.1 & 75.9 & 68.7 & 81.8 & 82.4 & 80.4 & 83.6 & 80.7 & 76.4 & 79.7 \\
    \textcolor{famze}{\baseline{zerank-1-sm}}\crtup{2} & 71.4 & 76.8 & 81.9 & 80.5 & 70.1 & 52.7 & 76.5 & 77.3 & 74.4 & 77.8 & 76.3 & 70.8 & 73.9 \\
    \textcolor{famvoyage}{\baseline{Voyage2.5}}\crtdn{1} & 54.4 & 78.7 & 83.2 & 78.8 & 69.1 & 63.4 & 79.1 & 78.7 & 72.2 & 79.1 & 69.6 & 59.4 & 72.1 \\
    \textcolor{famcohere}{\baseline{RR4-Pro}}\crtdn{1} & 73.1 & 76.9 & 81.6 & 76.8 & \textbf{78.7} & 15.5 & 76.6 & 77.5 & 67.7 & 72.7 & 75.2 & \textbf{79.1} & 71.0 \\
    \textcolor{famqwen}{\baseline{Qwen3-RR-4B}}\crtup{2} & 68.8 & 68.4 & 74.3 & 65.4 & 70.1 & 63.1 & 61.9 & 70.4 & 72.7 & 60.2 & 72.2 & 70.4 & 68.1 \\
    \textcolor{famqwen}{\baseline{Qwen3-RR-8B}}\crtup{2} & 64.7 & 69.7 & 73.3 & 70.3 & 69.4 & 33.8 & 65.8 & 74.4 & 71.7 & 65.8 & 70.8 & 68.6 & 66.5 \\
    \textcolor{famcohere}{\baseline{RR4-Fast}}\crtdn{2} & 66.8 & 73.9 & 78.5 & 74.3 & 67.2 & 19.3 & 71.3 & 71.7 & 65.7 & 71.9 & 70.1 & 67.2 & 66.5 \\
    \textcolor{famvoyage}{\baseline{Voyage2.5-lt}}\crtdn{2} & 45.3 & 71.4 & 78.5 & 73.3 & 62.4 & 47.1 & 70.8 & 72.5 & 64.3 & 71.4 & 63.0 & 49.1 & 64.1 \\
    \textcolor{famqwen}{\baseline{Qwen3-RR-0.6B}}\crtup{1} & 59.6 & 53.3 & 61.1 & 55.1 & 61.3 & 33.7 & 48.4 & 61.9 & 52.9 & 48.9 & 63.2 & 57.3 & 54.7 \\
    \textcolor{famjina}{\baseline{Jina-RR-v3}}\crtdn{1} & 41.1 & 59.1 & 65.1 & 62.4 & 50.9 & 40.5 & 55.4 & 68.1 & 46.9 & 62.2 & 46.4 & 54.3 & 54.4 \\
    \textcolor{famctxl}{\baseline{CTXL-RR-1B}}\crtup{1} & 47.5 & 41.7 & 48.9 & 48.4 & 32.0 & 30.0 & 50.0 & 49.4 & 31.1 & 38.6 & 51.7 & 44.7 & 42.8 \\
    \textcolor{famctxl}{\baseline{CTXL-RR-2B}}\crtup{1} & 36.5 & 38.2 & 46.5 & 45.6 & 42.3 & 21.1 & 50.0 & 49.0 & 31.8 & 34.9 & 53.8 & 51.1 & 41.7 \\
    \textcolor{famctxl}{\baseline{CTXL-RR-6B}}\crtdn{2} & 38.5 & 40.1 & 45.4 & 45.5 & 36.9 & 30.5 & 50.5 & 47.1 & 27.9 & 32.5 & 52.7 & 50.5 & 41.5 \\
    \bottomrule
  \end{tabular}
\end{table}

\begin{table}[t]
  \centering
  \caption{nDCG@10 with the benchmark's graded qrels as linear gains and \rcpndcg{}@10 on ViDoRe v3 (judge \model{Qwen3.5-397B}), in \%, on the native-language questions. Tied reranker scores receive their group's mean gain. Each block is sorted by its mean over the 8 corpora. Arrows give each reranker's change in rank from nDCG@10 to \rcpndcg{}@10. Bold marks the best reranker per column. Colors mark the model family as in \cref{tab:cal-ndcg-qwen-nanomteb}.}
  \label{tab:cal-ndcg-qwen-vidore}
  \footnotesize
  \setlength{\tabcolsep}{1.5pt}
  \begin{tabular}{@{}l *{8}{w{c}{2.0em}} w{c}{2.4em}@{\hspace{5pt}}}
    \toprule
     & \multicolumn{8}{c}{\textbf{Datasets}} & \textbf{Mean} \\
    \cmidrule(lr){2-9} \cmidrule(l){10-10}
    \textbf{Reranker} & \crtrot{Comp.\ Sci.} & \crtrot{Energy} & \crtrot{Finance (en)} & \crtrot{Finance (fr)} & \crtrot{HR} & \crtrot{Industrial} & \crtrot{Pharma.} & \crtrot{Physics} & \\
    \hline
    \rowcolor{gray!15}
     & \multicolumn{9}{c}{\rule{0pt}{2.4ex}\textbf{nDCG@10}} \\
    \hline
    \textcolor{famvoyage}{\rule{0pt}{2.3ex}\baseline{Voyage2.5}} & \textbf{85.1} & \textbf{74.0} & 75.1 & \textbf{58.5} & \textbf{71.9} & \textbf{62.1} & \textbf{71.6} & 50.3 & \textbf{68.6} \\
    \textcolor{famcohere}{\baseline{RR4-Pro}} & 83.2 & 73.7 & \textbf{75.5} & 57.9 & 70.8 & 61.3 & 71.4 & 51.5 & 68.2 \\
    \textcolor{famvoyage}{\baseline{Voyage2.5-lt}} & 83.8 & 73.2 & 74.6 & 57.3 & 69.6 & \textbf{62.1} & 70.2 & 49.0 & 67.5 \\
    \textcolor{famze}{\baseline{zerank-2}} & 84.2 & 72.4 & 75.1 & 55.8 & 68.9 & 60.7 & 69.8 & 49.0 & 67.0 \\
    \textcolor{famze}{\baseline{zerank-1}} & 84.0 & 71.7 & 75.1 & 55.4 & 68.7 & 60.4 & 69.7 & 48.9 & 66.7 \\
    \textcolor{famcohere}{\baseline{RR4-Fast}} & 81.9 & 70.7 & 74.5 & 53.6 & 68.5 & 59.7 & 69.5 & 48.8 & 65.9 \\
    \textcolor{famctxl}{\baseline{CTXL-RR-6B}} & 74.4 & 71.8 & 73.8 & 55.2 & 69.1 & 60.6 & 68.2 & 51.2 & 65.5 \\
    \textcolor{famqwen}{\baseline{Qwen3-RR-8B}} & 80.8 & 70.6 & 70.3 & 55.5 & 65.5 & 60.0 & 68.6 & \textbf{52.3} & 65.4 \\
    \textcolor{famze}{\baseline{zerank-1-sm}} & 84.3 & 69.5 & 72.8 & 52.6 & 66.8 & 58.7 & 67.5 & 48.5 & 65.1 \\
    \textcolor{famctxl}{\baseline{CTXL-RR-2B}} & 76.7 & 71.6 & 71.7 & 53.2 & 68.1 & 61.1 & 67.5 & 49.0 & 64.9 \\
    \textcolor{famctxl}{\baseline{CTXL-RR-1B}} & 77.7 & 68.9 & 69.1 & 49.7 & 65.4 & 57.4 & 66.4 & 49.2 & 63.0 \\
    \textcolor{famqwen}{\baseline{Qwen3-RR-4B}} & 80.4 & 67.9 & 67.4 & 51.0 & 61.7 & 57.8 & 66.4 & 51.0 & 62.9 \\
    \textcolor{famjina}{\baseline{Jina-RR-v3}} & 74.1 & 65.8 & 63.6 & 41.5 & 57.2 & 57.3 & 67.9 & 45.3 & 59.1 \\
    \textcolor{famqwen}{\baseline{Qwen3-RR-0.6B}} & 74.4 & 53.3 & 57.3 & 34.2 & 45.2 & 49.7 & 60.0 & 47.9 & 52.8 \\
    \hline
    \rowcolor{gray!15}
     & \multicolumn{9}{c}{\rule{0pt}{2.4ex}\textbf{\rcpndcg{}@10}} \\
    \hline
    \textcolor{famze}{\rule{0pt}{2.3ex}\baseline{zerank-2}}\crtup{3} & \textbf{94.2} & \textbf{93.6} & 87.6 & 89.6 & \textbf{92.0} & \textbf{89.4} & \textbf{91.2} & \textbf{93.7} & \textbf{91.4} \\
    \textcolor{famze}{\baseline{zerank-1}}\crtup{3} & 94.1 & 93.3 & \textbf{88.9} & 89.6 & 91.6 & 89.1 & 90.8 & 93.3 & 91.3 \\
    \textcolor{famvoyage}{\baseline{Voyage2.5}}\crtdn{2} & 93.3 & 93.3 & 86.5 & 87.5 & 91.8 & 89.1 & 90.3 & 91.0 & 90.4 \\
    \textcolor{famcohere}{\baseline{RR4-Pro}}\crtdn{2} & 90.8 & 93.0 & 88.6 & \textbf{90.2} & 90.4 & 87.9 & 89.0 & 91.6 & 90.2 \\
    \textcolor{famze}{\baseline{zerank-1-sm}}\crtup{4} & 93.2 & 90.9 & 88.2 & 86.7 & 90.1 & 87.8 & 88.1 & 90.8 & 89.5 \\
    \textcolor{famvoyage}{\baseline{Voyage2.5-lt}}\crtdn{3} & 91.5 & 90.7 & 86.9 & 85.2 & 89.8 & 87.7 & 88.6 & 86.9 & 88.4 \\
    \textcolor{famcohere}{\baseline{RR4-Fast}}\crtdn{1} & 89.8 & 90.2 & 87.2 & 84.7 & 88.4 & 85.1 & 86.7 & 88.2 & 87.5 \\
    \textcolor{famqwen}{\baseline{Qwen3-RR-8B}} & 88.3 & 89.2 & 82.6 & 85.1 & 85.3 & 84.1 & 85.9 & 89.2 & 86.2 \\
    \textcolor{famctxl}{\baseline{CTXL-RR-2B}}\crtup{1} & 82.8 & 89.9 & 85.7 & 83.7 & 87.9 & 85.0 & 85.6 & 88.0 & 86.1 \\
    \textcolor{famctxl}{\baseline{CTXL-RR-6B}}\crtdn{3} & 79.0 & 89.0 & 86.6 & 85.5 & 88.0 & 85.1 & 84.6 & 88.0 & 85.7 \\
    \textcolor{famctxl}{\baseline{CTXL-RR-1B}} & 85.1 & 88.1 & 82.5 & 82.2 & 85.3 & 81.5 & 83.7 & 86.7 & 84.4 \\
    \textcolor{famqwen}{\baseline{Qwen3-RR-4B}} & 87.7 & 86.5 & 79.5 & 80.4 & 82.7 & 81.0 & 84.1 & 88.1 & 83.7 \\
    \textcolor{famqwen}{\baseline{Qwen3-RR-0.6B}}\crtup{1} & 82.5 & 79.8 & 75.1 & 67.2 & 72.5 & 76.4 & 78.3 & 83.2 & 76.9 \\
    \textcolor{famjina}{\baseline{Jina-RR-v3}}\crtdn{1} & 73.2 & 70.5 & 67.9 & 64.3 & 69.7 & 74.3 & 75.9 & 76.9 & 71.6 \\
    \bottomrule
  \end{tabular}
\end{table}

\Cref{tab:cal-ndcg-qwen-nanomteb} lists nDCG@10 on the qrels and \rcpndcg{}@10 for every reranker and NanoBEIR dataset. \Cref{tab:cal-ndcg-qwen-bright,tab:cal-ndcg-qwen-vidore} do the same for the BRIGHT tasks and the ViDoRe v3 corpora. \Cref{tab:nb-leaderboard} adds the suite means on NanoBEIR and BRIGHT under both judges. Since the two metrics differ in level, only their orders of the rerankers are comparable.

\providecommand{\efv}[1]{\ensuremath{\val{#1}}}
\providecommand{\appgname}[2]{\textcolor{#1}{\baseline{#2}}}
\providecommand{\appgfirst}{\rule{0pt}{2.3ex}}
\definecolor{appgCohere}{HTML}{1B7837}
\definecolor{appgCtxl}{HTML}{2166AC}
\definecolor{appgJina}{HTML}{555555}
\definecolor{appgQwen}{HTML}{B2182B}
\definecolor{appgVoyage}{HTML}{7B3294}
\definecolor{appgZE}{HTML}{B35806}

\subsection{Second judges}
\label{app:nano-bright-judge}
\label{app:trecdl-judge2}
\label{app:robustness-consensus}

A metric built on one LLM judge should change little with another. We therefore ran both stages with \model{gpt-oss-120b}~\citep{openai2025gptoss} on NanoBEIR and BRIGHT, and on the TREC-DL pools, where the calibration CAL-C fits its judgments (\cref{app:trecdl-protocol}). \Cref{tab:nb-leaderboard} gives the suite means of qrel-nDCG@10 and Count-nDCG@10, and of \rcpndcg{}@10 under both judges. Within each dataset, we correlate the two judges' \rcpndcg{}@10 over the 14 rerankers, where 1 means identical leaderboards. The mean Spearman correlation over datasets is \val{a4.nano.xjudge.mean-rho} on NanoBEIR and \val{b20.new.a4.bright.xjudge.mean-rho} on BRIGHT. No dataset falls below \val{a4.nano.xjudge.min-rho}. On TREC-DL, the two judges' leaderboards of all 18 systems correlate at Kendall \val{a2.secondjudge.pool.kendall18avsc} with both years pooled. The two judges' calibrated scores over the \val{a2.secondjudge.perdoc.ndocs} pool passages correlate at Spearman \val{a2.secondjudge.perdoc.rho.ac}. Under the second judge, \rcpndcg{}@10 separates \val{a4.nano.sens.rcp-oss}\% of the reranker pairs on NanoBEIR, against \val{a4.nano.sens.rcp-qwen}\% under the primary judge (\cref{app:nano-bright-sensitivity}).

On TREC-DL, the within-query concordance is the share of differently graded passage pairs that \rcp{} orders as NIST does, with 0.5 at chance. The concordance is \val{a2.conc.dl19.oss} (2019) and \val{a2.conc.dl20.oss} (2020) with \model{gpt-oss-120b}, against \val{a2.conc.dl19.overall} and \val{a2.conc.dl20.overall} with \model{Qwen3.6-27B}. A third model, \model{GLM-5.2}, arbitrated the \val{a56.robust.glm.contested} TREC-DL passage pairs on which the two judges disagree most. On the \val{a56.robust.glm.adjudicable} pairs that NIST grades differently, \model{GLM-5.2} orders \val{a56.robust.glm.accuracy}\% as NIST does, against \val{a56.robust.glm.qwen-accuracy}\% for \model{Qwen3.6-27B}. Pooling the two judges' Stage~B answers in one fit does not improve on the better single judge (\cref{app:trecdl-placement}).

\begin{table}[t]
  \centering
  \caption{Cross-suite summary: suite means over datasets (\%), with the rank among the 14 rerankers in parentheses. Columns give nDCG@10 on the qrels (qrel), Count-nDCG@10 (Count) and \rcpndcg{}@10 (\rcp{}) with judge \model{Qwen3.5-397B}, and \rcpndcg{}@10 with the second judge \model{gpt-oss-120b} (\rcp{}$_2$). BRIGHT removes each query's excluded documents, and NanoArguAna the query's own argument (\cref{app:nano-bright,app:nano-bright-scope}). Rows follow \rcpndcg{}@10 on NanoBEIR. Bold marks the best reranker per column.}
  \label{tab:nb-leaderboard}
  \footnotesize
  \setlength{\tabcolsep}{3pt}
  \begin{tabular}{@{}lcccc@{\hspace{8pt}}c@{\hspace{8pt}}cccc@{}}
    \toprule
    & \multicolumn{4}{c}{NanoBEIR} && \multicolumn{4}{c}{BRIGHT} \\
    \cmidrule(lr){2-5}\cmidrule(l){7-10}
    Reranker & qrel & Count & \rcp{} & \rcp{}$_2$ && qrel & Count & \rcp{} & \rcp{}$_2$ \\
    \midrule
\appgname{appgZE}{zerank-2} & \val{draft.app.g.lb.nano.ndcg.zerank2}\,(\val{draft.app.g.lb.nano.ndcg.zerank2.rank}) & \textbf{\val{draft.app.g.lb.nano.count.zerank2}}\,(\val{draft.app.g.lb.nano.count.zerank2.rank}) & \textbf{\val{draft.app.g.lb.nano.rcp.zerank2}}\,(\val{draft.app.g.lb.nano.rcp.zerank2.rank}) & \val{draft.app.g.lb.nano.oss.zerank2}\,(\val{draft.app.g.lb.nano.oss.zerank2.rank}) && \val{b20.new.draft.app.g.lb.bright.ndcg.zerank2}\,(\val{b20.new.draft.app.g.lb.bright.ndcg.zerank2.rank}) & \textbf{\val{b20.new.draft.app.g.lb.bright.count.zerank2}}\,(\val{b20.new.draft.app.g.lb.bright.count.zerank2.rank}) & \textbf{\val{b20.new.draft.app.g.lb.bright.rcp.zerank2}}\,(\val{b20.new.draft.app.g.lb.bright.rcp.zerank2.rank}) & \textbf{\val{b20.new.draft.app.g.lb.bright.oss.zerank2}}\,(\val{b20.new.draft.app.g.lb.bright.oss.zerank2.rank}) \\
\appgname{appgZE}{zerank-1} & \val{draft.app.g.lb.nano.ndcg.zerank1}\,(\val{draft.app.g.lb.nano.ndcg.zerank1.rank}) & \val{draft.app.g.lb.nano.count.zerank1}\,(\val{draft.app.g.lb.nano.count.zerank1.rank}) & \val{draft.app.g.lb.nano.rcp.zerank1}\,(\val{draft.app.g.lb.nano.rcp.zerank1.rank}) & \textbf{\val{draft.app.g.lb.nano.oss.zerank1}}\,(\val{draft.app.g.lb.nano.oss.zerank1.rank}) && \textbf{\val{b20.new.draft.app.g.lb.bright.ndcg.zerank1}}\,(\val{b20.new.draft.app.g.lb.bright.ndcg.zerank1.rank}) & \val{b20.new.draft.app.g.lb.bright.count.zerank1}\,(\val{b20.new.draft.app.g.lb.bright.count.zerank1.rank}) & \val{b20.new.draft.app.g.lb.bright.rcp.zerank1}\,(\val{b20.new.draft.app.g.lb.bright.rcp.zerank1.rank}) & \val{b20.new.draft.app.g.lb.bright.oss.zerank1}\,(\val{b20.new.draft.app.g.lb.bright.oss.zerank1.rank}) \\
\appgname{appgVoyage}{Voyage2.5} & \val{draft.app.g.lb.nano.ndcg.voyage25}\,(\val{draft.app.g.lb.nano.ndcg.voyage25.rank}) & \val{draft.app.g.lb.nano.count.voyage25}\,(\val{draft.app.g.lb.nano.count.voyage25.rank}) & \val{draft.app.g.lb.nano.rcp.voyage25}\,(\val{draft.app.g.lb.nano.rcp.voyage25.rank}) & \val{draft.app.g.lb.nano.oss.voyage25}\,(\val{draft.app.g.lb.nano.oss.voyage25.rank}) && \val{b20.new.draft.app.g.lb.bright.ndcg.voyage25}\,(\val{b20.new.draft.app.g.lb.bright.ndcg.voyage25.rank}) & \val{b20.new.draft.app.g.lb.bright.count.voyage25}\,(\val{b20.new.draft.app.g.lb.bright.count.voyage25.rank}) & \val{b20.new.draft.app.g.lb.bright.rcp.voyage25}\,(\val{b20.new.draft.app.g.lb.bright.rcp.voyage25.rank}) & \val{b20.new.draft.app.g.lb.bright.oss.voyage25}\,(\val{b20.new.draft.app.g.lb.bright.oss.voyage25.rank}) \\
\appgname{appgZE}{zerank-1-sm} & \val{draft.app.g.lb.nano.ndcg.zerank1sm}\,(\val{draft.app.g.lb.nano.ndcg.zerank1sm.rank}) & \val{draft.app.g.lb.nano.count.zerank1sm}\,(\val{draft.app.g.lb.nano.count.zerank1sm.rank}) & \val{draft.app.g.lb.nano.rcp.zerank1sm}\,(\val{draft.app.g.lb.nano.rcp.zerank1sm.rank}) & \val{draft.app.g.lb.nano.oss.zerank1sm}\,(\val{draft.app.g.lb.nano.oss.zerank1sm.rank}) && \val{b20.new.draft.app.g.lb.bright.ndcg.zerank1sm}\,(\val{b20.new.draft.app.g.lb.bright.ndcg.zerank1sm.rank}) & \val{b20.new.draft.app.g.lb.bright.count.zerank1sm}\,(\val{b20.new.draft.app.g.lb.bright.count.zerank1sm.rank}) & \val{b20.new.draft.app.g.lb.bright.rcp.zerank1sm}\,(\val{b20.new.draft.app.g.lb.bright.rcp.zerank1sm.rank}) & \val{b20.new.draft.app.g.lb.bright.oss.zerank1sm}\,(\val{b20.new.draft.app.g.lb.bright.oss.zerank1sm.rank}) \\
\appgname{appgCohere}{RR4-Pro} & \val{draft.app.g.lb.nano.ndcg.rr4pro}\,(\val{draft.app.g.lb.nano.ndcg.rr4pro.rank}) & \val{draft.app.g.lb.nano.count.rr4pro}\,(\val{draft.app.g.lb.nano.count.rr4pro.rank}) & \val{draft.app.g.lb.nano.rcp.rr4pro}\,(\val{draft.app.g.lb.nano.rcp.rr4pro.rank}) & \val{draft.app.g.lb.nano.oss.rr4pro}\,(\val{draft.app.g.lb.nano.oss.rr4pro.rank}) && \val{b20.new.draft.app.g.lb.bright.ndcg.rr4pro}\,(\val{b20.new.draft.app.g.lb.bright.ndcg.rr4pro.rank}) & \val{b20.new.draft.app.g.lb.bright.count.rr4pro}\,(\val{b20.new.draft.app.g.lb.bright.count.rr4pro.rank}) & \val{b20.new.draft.app.g.lb.bright.rcp.rr4pro}\,(\val{b20.new.draft.app.g.lb.bright.rcp.rr4pro.rank}) & \val{b20.new.draft.app.g.lb.bright.oss.rr4pro}\,(\val{b20.new.draft.app.g.lb.bright.oss.rr4pro.rank}) \\
\appgname{appgCtxl}{CTXL-RR-6B} & \val{draft.app.g.lb.nano.ndcg.ctxlrr6b}\,(\val{draft.app.g.lb.nano.ndcg.ctxlrr6b.rank}) & \val{draft.app.g.lb.nano.count.ctxlrr6b}\,(\val{draft.app.g.lb.nano.count.ctxlrr6b.rank}) & \val{draft.app.g.lb.nano.rcp.ctxlrr6b}\,(\val{draft.app.g.lb.nano.rcp.ctxlrr6b.rank}) & \val{draft.app.g.lb.nano.oss.ctxlrr6b}\,(\val{draft.app.g.lb.nano.oss.ctxlrr6b.rank}) && \val{b20.new.draft.app.g.lb.bright.ndcg.ctxlrr6b}\,(\val{b20.new.draft.app.g.lb.bright.ndcg.ctxlrr6b.rank}) & \val{b20.new.draft.app.g.lb.bright.count.ctxlrr6b}\,(\val{b20.new.draft.app.g.lb.bright.count.ctxlrr6b.rank}) & \val{b20.new.draft.app.g.lb.bright.rcp.ctxlrr6b}\,(\val{b20.new.draft.app.g.lb.bright.rcp.ctxlrr6b.rank}) & \val{b20.new.draft.app.g.lb.bright.oss.ctxlrr6b}\,(\val{b20.new.draft.app.g.lb.bright.oss.ctxlrr6b.rank}) \\
\appgname{appgVoyage}{Voyage2.5-lt} & \val{draft.app.g.lb.nano.ndcg.voyage25lt}\,(\val{draft.app.g.lb.nano.ndcg.voyage25lt.rank}) & \val{draft.app.g.lb.nano.count.voyage25lt}\,(\val{draft.app.g.lb.nano.count.voyage25lt.rank}) & \val{draft.app.g.lb.nano.rcp.voyage25lt}\,(\val{draft.app.g.lb.nano.rcp.voyage25lt.rank}) & \val{draft.app.g.lb.nano.oss.voyage25lt}\,(\val{draft.app.g.lb.nano.oss.voyage25lt.rank}) && \val{b20.new.draft.app.g.lb.bright.ndcg.voyage25lt}\,(\val{b20.new.draft.app.g.lb.bright.ndcg.voyage25lt.rank}) & \val{b20.new.draft.app.g.lb.bright.count.voyage25lt}\,(\val{b20.new.draft.app.g.lb.bright.count.voyage25lt.rank}) & \val{b20.new.draft.app.g.lb.bright.rcp.voyage25lt}\,(\val{b20.new.draft.app.g.lb.bright.rcp.voyage25lt.rank}) & \val{b20.new.draft.app.g.lb.bright.oss.voyage25lt}\,(\val{b20.new.draft.app.g.lb.bright.oss.voyage25lt.rank}) \\
\appgname{appgQwen}{Qwen3-RR-4B} & \val{draft.app.g.lb.nano.ndcg.qwen3rr4b}\,(\val{draft.app.g.lb.nano.ndcg.qwen3rr4b.rank}) & \val{draft.app.g.lb.nano.count.qwen3rr4b}\,(\val{draft.app.g.lb.nano.count.qwen3rr4b.rank}) & \val{draft.app.g.lb.nano.rcp.qwen3rr4b}\,(\val{draft.app.g.lb.nano.rcp.qwen3rr4b.rank}) & \val{draft.app.g.lb.nano.oss.qwen3rr4b}\,(\val{draft.app.g.lb.nano.oss.qwen3rr4b.rank}) && \val{b20.new.draft.app.g.lb.bright.ndcg.qwen3rr4b}\,(\val{b20.new.draft.app.g.lb.bright.ndcg.qwen3rr4b.rank}) & \val{b20.new.draft.app.g.lb.bright.count.qwen3rr4b}\,(\val{b20.new.draft.app.g.lb.bright.count.qwen3rr4b.rank}) & \val{b20.new.draft.app.g.lb.bright.rcp.qwen3rr4b}\,(\val{b20.new.draft.app.g.lb.bright.rcp.qwen3rr4b.rank}) & \val{b20.new.draft.app.g.lb.bright.oss.qwen3rr4b}\,(\val{b20.new.draft.app.g.lb.bright.oss.qwen3rr4b.rank}) \\
\appgname{appgQwen}{Qwen3-RR-8B} & \textbf{\val{draft.app.g.lb.nano.ndcg.qwen3rr8b}}\,(\val{draft.app.g.lb.nano.ndcg.qwen3rr8b.rank}) & \val{draft.app.g.lb.nano.count.qwen3rr8b}\,(\val{draft.app.g.lb.nano.count.qwen3rr8b.rank}) & \val{draft.app.g.lb.nano.rcp.qwen3rr8b}\,(\val{draft.app.g.lb.nano.rcp.qwen3rr8b.rank}) & \val{draft.app.g.lb.nano.oss.qwen3rr8b}\,(\val{draft.app.g.lb.nano.oss.qwen3rr8b.rank}) && \val{b20.new.draft.app.g.lb.bright.ndcg.qwen3rr8b}\,(\val{b20.new.draft.app.g.lb.bright.ndcg.qwen3rr8b.rank}) & \val{b20.new.draft.app.g.lb.bright.count.qwen3rr8b}\,(\val{b20.new.draft.app.g.lb.bright.count.qwen3rr8b.rank}) & \val{b20.new.draft.app.g.lb.bright.rcp.qwen3rr8b}\,(\val{b20.new.draft.app.g.lb.bright.rcp.qwen3rr8b.rank}) & \val{b20.new.draft.app.g.lb.bright.oss.qwen3rr8b}\,(\val{b20.new.draft.app.g.lb.bright.oss.qwen3rr8b.rank}) \\
\appgname{appgCtxl}{CTXL-RR-2B} & \val{draft.app.g.lb.nano.ndcg.ctxlrr2b}\,(\val{draft.app.g.lb.nano.ndcg.ctxlrr2b.rank}) & \val{draft.app.g.lb.nano.count.ctxlrr2b}\,(\val{draft.app.g.lb.nano.count.ctxlrr2b.rank}) & \val{draft.app.g.lb.nano.rcp.ctxlrr2b}\,(\val{draft.app.g.lb.nano.rcp.ctxlrr2b.rank}) & \val{draft.app.g.lb.nano.oss.ctxlrr2b}\,(\val{draft.app.g.lb.nano.oss.ctxlrr2b.rank}) && \val{b20.new.draft.app.g.lb.bright.ndcg.ctxlrr2b}\,(\val{b20.new.draft.app.g.lb.bright.ndcg.ctxlrr2b.rank}) & \val{b20.new.draft.app.g.lb.bright.count.ctxlrr2b}\,(\val{b20.new.draft.app.g.lb.bright.count.ctxlrr2b.rank}) & \val{b20.new.draft.app.g.lb.bright.rcp.ctxlrr2b}\,(\val{b20.new.draft.app.g.lb.bright.rcp.ctxlrr2b.rank}) & \val{b20.new.draft.app.g.lb.bright.oss.ctxlrr2b}\,(\val{b20.new.draft.app.g.lb.bright.oss.ctxlrr2b.rank}) \\
\appgname{appgCohere}{RR4-Fast} & \val{draft.app.g.lb.nano.ndcg.rr4fast}\,(\val{draft.app.g.lb.nano.ndcg.rr4fast.rank}) & \val{draft.app.g.lb.nano.count.rr4fast}\,(\val{draft.app.g.lb.nano.count.rr4fast.rank}) & \val{draft.app.g.lb.nano.rcp.rr4fast}\,(\val{draft.app.g.lb.nano.rcp.rr4fast.rank}) & \val{draft.app.g.lb.nano.oss.rr4fast}\,(\val{draft.app.g.lb.nano.oss.rr4fast.rank}) && \val{b20.new.draft.app.g.lb.bright.ndcg.rr4fast}\,(\val{b20.new.draft.app.g.lb.bright.ndcg.rr4fast.rank}) & \val{b20.new.draft.app.g.lb.bright.count.rr4fast}\,(\val{b20.new.draft.app.g.lb.bright.count.rr4fast.rank}) & \val{b20.new.draft.app.g.lb.bright.rcp.rr4fast}\,(\val{b20.new.draft.app.g.lb.bright.rcp.rr4fast.rank}) & \val{b20.new.draft.app.g.lb.bright.oss.rr4fast}\,(\val{b20.new.draft.app.g.lb.bright.oss.rr4fast.rank}) \\
\appgname{appgCtxl}{CTXL-RR-1B} & \val{draft.app.g.lb.nano.ndcg.ctxlrr1b}\,(\val{draft.app.g.lb.nano.ndcg.ctxlrr1b.rank}) & \val{draft.app.g.lb.nano.count.ctxlrr1b}\,(\val{draft.app.g.lb.nano.count.ctxlrr1b.rank}) & \val{draft.app.g.lb.nano.rcp.ctxlrr1b}\,(\val{draft.app.g.lb.nano.rcp.ctxlrr1b.rank}) & \val{draft.app.g.lb.nano.oss.ctxlrr1b}\,(\val{draft.app.g.lb.nano.oss.ctxlrr1b.rank}) && \val{b20.new.draft.app.g.lb.bright.ndcg.ctxlrr1b}\,(\val{b20.new.draft.app.g.lb.bright.ndcg.ctxlrr1b.rank}) & \val{b20.new.draft.app.g.lb.bright.count.ctxlrr1b}\,(\val{b20.new.draft.app.g.lb.bright.count.ctxlrr1b.rank}) & \val{b20.new.draft.app.g.lb.bright.rcp.ctxlrr1b}\,(\val{b20.new.draft.app.g.lb.bright.rcp.ctxlrr1b.rank}) & \val{b20.new.draft.app.g.lb.bright.oss.ctxlrr1b}\,(\val{b20.new.draft.app.g.lb.bright.oss.ctxlrr1b.rank}) \\
\appgname{appgQwen}{Qwen3-RR-0.6B} & \val{draft.app.g.lb.nano.ndcg.qwen3rr06b}\,(\val{draft.app.g.lb.nano.ndcg.qwen3rr06b.rank}) & \val{draft.app.g.lb.nano.count.qwen3rr06b}\,(\val{draft.app.g.lb.nano.count.qwen3rr06b.rank}) & \val{draft.app.g.lb.nano.rcp.qwen3rr06b}\,(\val{draft.app.g.lb.nano.rcp.qwen3rr06b.rank}) & \val{draft.app.g.lb.nano.oss.qwen3rr06b}\,(\val{draft.app.g.lb.nano.oss.qwen3rr06b.rank}) && \val{b20.new.draft.app.g.lb.bright.ndcg.qwen3rr06b}\,(\val{b20.new.draft.app.g.lb.bright.ndcg.qwen3rr06b.rank}) & \val{b20.new.draft.app.g.lb.bright.count.qwen3rr06b}\,(\val{b20.new.draft.app.g.lb.bright.count.qwen3rr06b.rank}) & \val{b20.new.draft.app.g.lb.bright.rcp.qwen3rr06b}\,(\val{b20.new.draft.app.g.lb.bright.rcp.qwen3rr06b.rank}) & \val{b20.new.draft.app.g.lb.bright.oss.qwen3rr06b}\,(\val{b20.new.draft.app.g.lb.bright.oss.qwen3rr06b.rank}) \\
\appgname{appgJina}{Jina-RR-v3} & \val{draft.app.g.lb.nano.ndcg.jinarrv3}\,(\val{draft.app.g.lb.nano.ndcg.jinarrv3.rank}) & \val{draft.app.g.lb.nano.count.jinarrv3}\,(\val{draft.app.g.lb.nano.count.jinarrv3.rank}) & \val{draft.app.g.lb.nano.rcp.jinarrv3}\,(\val{draft.app.g.lb.nano.rcp.jinarrv3.rank}) & \val{draft.app.g.lb.nano.oss.jinarrv3}\,(\val{draft.app.g.lb.nano.oss.jinarrv3.rank}) && \val{b20.new.draft.app.g.lb.bright.ndcg.jinarrv3}\,(\val{b20.new.draft.app.g.lb.bright.ndcg.jinarrv3.rank}) & \val{b20.new.draft.app.g.lb.bright.count.jinarrv3}\,(\val{b20.new.draft.app.g.lb.bright.count.jinarrv3.rank}) & \val{b20.new.draft.app.g.lb.bright.rcp.jinarrv3}\,(\val{b20.new.draft.app.g.lb.bright.rcp.jinarrv3.rank}) & \val{b20.new.draft.app.g.lb.bright.oss.jinarrv3}\,(\val{b20.new.draft.app.g.lb.bright.oss.jinarrv3.rank}) \\
    \bottomrule
  \end{tabular}
\end{table}

\subsection{The zerank rerankers}
\label{app:trecdl-zerank}
\label{app:nano-bright-zerank}
\label{app:vidore-leaderboards}

The three \reranker{zerank} rerankers learn from pairwise preferences of an ensemble of LLMs~\citep{pipitone2025zelo}. A metric built on an LLM judge could therefore favor them (\cref{sec:ranking-results}). On NanoBEIR, qrel-nDCG@10 ranks them last, while \rcpndcg{}@10 ranks \reranker{zerank-2} and \reranker{zerank-1} first and second. The mean \rcpndcg{}@10 of the three zerank rerankers exceeds that of the other 11 rerankers by \val{a4.table2.final.zerank-gap}\,pp over all datasets and on \val{a4.table2.final.zerank-ahead} of the 13 datasets. On BRIGHT, qrel-nDCG@10 already ranks the two first and second (\cref{tab:nb-leaderboard}).

On queries with at least five qrel-positive documents, we correlate each reranker's scores of these documents with their calibrated scores (Spearman, averaged over queries). The three zerank rerankers reach \val{nr.a4.within.nano.zer-mean} on NanoBEIR and \val{nr.a4.within.bright.zer-mean} on BRIGHT, against \val{nr.a4.within.nano.rest-mean} and \val{nr.a4.within.bright.rest-mean} for the other 11 rerankers. A preference shared with the judge would also produce this lead. This result therefore locates the lead but cannot tell whether it reflects better rankings.

On TREC-DL, \rcpndcg{} ranks \reranker{zerank-2} first among the 18 systems in both years, whereas NIST ranks it \val{a2.zerank.pool.dl19.zerank2.rank18.nistlin}th and \val{a2.zerank.pool.dl20.zerank2.rank18.nistlin}th. NIST counts an unjudged passage as irrelevant. Unjudged passages fill \val{a2.unjudged.pool.zerank2}\% of \reranker{zerank-2}'s top-10 slots, against \val{a2.unjudged.pool.nonzeroabove.min}--\val{a2.unjudged.pool.nonzeroabove.max}\% for the \val{a2.unjudged.pool.nonzeroabove.n} systems from other families that NIST ranks above it with both years pooled. Condensed lists~\citep{sakai2007alternatives} drop every unjudged passage before scoring. On condensed lists, the Kendall correlation between the leaderboards of \rcpndcg{} and NIST closes \val{a2.condensed.pool.pooled.rcpa.lin.rerankers.shareunjudged} of its gap to 1 under the linear gain (CAL-A) and \val{a2.condensed.pool.pooled.rcpb.exp.rerankers.shareunjudged} under exponential gain (CAL-B). On condensed lists, \reranker{zerank-2} ranks \val{n7.rank.cond.dl19}th in 2019 and \val{n7.rank.cond.dl20}th in 2020, against \val{n7.rank.raw.dl19}th and \val{n7.rank.raw.dl20}th on the full lists. On the full lists, it trails the NIST leader by \val{n7.gap.dl19} in 2019 (95\% CI [\val{n7.gap.dl19.lo}, \val{n7.gap.dl19.hi}]) and \val{n7.gap.dl20} in 2020 ([\val{n7.gap.dl20.lo}, \val{n7.gap.dl20.hi}]).

In the human study on NanoBEIR, BRIGHT, and ViDoRe v3, annotators never saw a system name. In the \val{n7.hs.zr.all.n} contests that pit a zerank list against another reranker's list and reach a verdict, the verdict picks the zerank list in \val{n7.hs.zr.all}\% (95\% CI [\val{n7.hs.zr.all.lo}, \val{n7.hs.zr.all.hi}]\%). Where \rcpndcg{} picks the zerank list and qrel-nDCG the other list, the verdict picks the zerank list in \val{n7.hs.zr.rcpfav}\% of \val{n7.hs.zr.rcpfav.n} contests ([\val{n7.hs.zr.rcpfav.lo}, \val{n7.hs.zr.rcpfav.hi}]\%). Where \rcpndcg{} picks the zerank list with a margin of at least \val{draft.5.margin.decisive}, the verdict picks it in \val{n7.hs.zr.rcpfav02}\% of \val{n7.hs.zr.rcpfav02.n} contests ([\val{n7.hs.zr.rcpfav02.lo}, \val{n7.hs.zr.rcpfav02.hi}]\%).

On ViDoRe v3, \reranker{zerank-2} and \reranker{zerank-1} rise from qrel ranks \val{a3.rank.lin.zerank2} and \val{a3.rank.lin.zerank1} to \val{a3.rank.rcp.zerank2} and \val{a3.rank.rcp.zerank1}. The rise already appears under Count-nDCG, which scores documents by the rubric alone.

\subsection{Rubric, prompt, and simulation robustness}
\label{app:robustness}
\label{app:robustness-factorial}
\label{app:robustness-lattice}
\label{app:robustness-synthetic}

This subsection tests how the leaderboard of \rcpndcg{} depends on the rubric, the Stage~A prompt, and the criteria, and checks the gain in a simulation.

\paragraph{Rubric and prompt factorial.}
On TREC-DL with \model{Qwen3.6-27B}, we cross five Stage~A prompts with sixteen rubrics and refit the 2PL in each of the \val{a56.robust.factorial.cells} cells (\cref{tab:app-factorial}). The four other Stage~A prompts keep the relevance framing and score anchors of the paper's prompt and replace only its list of scoring criteria. Two paraphrase this list, and two swap in answer-quality or document-quality criteria. Their tournament orders stay close to the shipped one, with mean per-query Spearman correlations of \val{n11.stagea.rho.min} to \val{n11.stagea.rho.max}. The four prompts therefore test the wording of the criteria list rather than a different construct. A pair counts as separated when a paired percentile bootstrap's 95\% interval excludes zero. The reference cell, with the shipped prompt and rubric, separates \val{a56.robust.factorial.decisive-pairs} of the \val{a56.robust.factorial.pairs} pairs. The \val{a56.robust.factorial.relevance.arms} cells whose rubric still measures relevance reverse none of these pairs, and their leaderboards correlate with the reference at Kendall \val{a56.robust.factorial.relevance.kendall-min} or higher. All \val{a56.robust.factorial.reversals} reversals among the \val{a56.robust.factorial.opportunities} opportunities come from rubrics built to fail. Of these, \val{a56.robust.factorial.reversals.r01} come from the all-easy rubric, \val{a56.robust.factorial.reversals.r02} from the all-hard rubric, and \val{a56.robust.factorial.reversals.r11} from a document-quality rubric that ignores the query. With the shipped rubric, the four other prompts also reverse no separated pair.

\begin{table}[t]
  \centering
  \footnotesize
  \caption{Rubric and prompt factorial on TREC-DL (\model{Qwen3.6-27B}). Each family crosses its rubrics with the five Stage~A prompts, except the reference family, which pairs the shipped rubric with the four other prompts. Reversals count the pairs that the reference cell separates and whose sign flips, summed over the family's cells. Kendall is the rank correlation of each cell's 14-reranker leaderboard with the reference.}
  \label{tab:app-factorial}
  \setlength{\tabcolsep}{5pt}
  \begin{tabular}{@{}llccc@{}}
    \toprule
    Family & Rubrics & Cells & Reversals & Kendall range \\
    \midrule
    Criterion wording & three paraphrases of the rubric & 15 & \val{a56.robust.factorial.family.wording.reversals} & \val{a56.robust.factorial.family.wording.kendall-min}--\val{a56.robust.factorial.family.wording.kendall-max} \\
    Criterion count & one, three, or seven criteria & 15 & \val{a56.robust.factorial.family.count.reversals} & \val{a56.robust.factorial.family.count.kendall-min}--\val{a56.robust.factorial.family.count.kendall-max} \\
    Criterion ordering & scrambled order & 5 & \val{a56.robust.factorial.family.ordering.reversals} & \val{a56.robust.factorial.family.ordering.kendall-min}--\val{a56.robust.factorial.family.ordering.kendall-max} \\
    Criterion redundancy & near-duplicate or diverse facets & 10 & \val{a56.robust.factorial.family.redundancy.reversals} & \val{a56.robust.factorial.family.redundancy.kendall-min}--\val{a56.robust.factorial.family.redundancy.kendall-max} \\
    Stage~A criteria wording & shipped rubric, four other prompts & 4 & \val{a56.robust.factorial.family.reference.reversals} & \val{a56.robust.factorial.family.reference.kendall-min}--\val{a56.robust.factorial.family.reference.kendall-max} \\
    Difficulty placement & all easy, all hard, tiled, top-heavy & 20 & \val{a56.robust.factorial.family.difficulty-placement.reversals} & \val{a56.robust.factorial.family.difficulty-placement.kendall-min}--\val{a56.robust.factorial.family.difficulty-placement.kendall-max} \\
    Construct validity & document quality, trivial items & 10 & \val{a56.robust.factorial.family.construct-validity.reversals} & \val{a56.robust.factorial.family.construct-validity.kendall-min}--\val{a56.robust.factorial.family.construct-validity.kendall-max} \\
    \bottomrule
  \end{tabular}
\end{table}

\paragraph{Criterion subsets.}
We refit the TREC-DL model with each non-empty subset of the five criteria. By information efficiency at the ideal top 10 (\cref{app:calibration-diagnostics}), the full rubric ranks third. Its information efficiency is \val{a56.robust.lattice.k12345.eff-oracle}. The best subset drops C2 and reaches \val{a56.robust.lattice.k1345.eff-oracle}, and single criteria reach between \val{a56.robust.lattice.k5.eff-oracle} (C5 alone) and \val{a56.robust.lattice.k4.eff-oracle} (C4 alone). The subset search uses no relevance labels.

\paragraph{Simulation.}
\label{tab:app-synthetic}
A simulation with known truth, \val{a56.robust.synth.replications} replications of \val{a56.robust.synth.queries} queries, tests the gain and the paired $t$-test, not the tournament or the 2PL fit. Over the 16 combinations of pool regime and judge noise, \rcpndcg{} separates \val{draft.5.synth.rcp.type1.min}--\val{draft.5.synth.rcp.type1.max}\% of the truly equal pairs and nDCG \val{nr.a5.synth.ndcg.type1.min}--\val{nr.a5.synth.ndcg.type1.max}\%. Both rates sit near the nominal \val{draft.2.siglevel}\% level of the test. With an accurate judge, \rcpndcg{} detects more true differences than nDCG in every pool regime. Under graded labels, \rcpndcg{} detects fewer true differences than nDCG from a judge noise of \val{a56.robust.synth.oracle-graded.boundary-noise}, although the judge then still ranks documents more accurately than every simulated reranker. Under binary labels, this happens only from a judge noise of \val{a56.robust.synth.beir-binary.boundary-noise}, where the judge ranks documents less accurately than every simulated reranker.

\subsection{\texorpdfstring{\rcpndcg{}}{RCP-nDCG} and Count-nDCG}
\label{sec:count}
\label{app:robustness-count}
\label{app:meta-ideal}
\label{app:trecdl-parry}

Count-nDCG scores documents by the rubric alone (\cref{eq:count-gain}). Count-nDCG separates about as many reranker pairs as \rcpndcg{} (\cref{sec:ranking-results}). Within a query, however, the tournament orders documents that Count-nDCG ties. The tests below concern this order. Except for the per-dataset leaderboards, the values of this subsection come from reranker lists that still contain NanoArguAna's self-documents and BRIGHT's excluded documents (\cref{app:nano-bright,app:nano-bright-scope}).

\paragraph{External judges.}
Three LLM judges from other model families compare two top-10 lists per case (\cref{app:robustness-judges}). The judges are blind to the metrics and see each case twice with swapped labels (\cref{tab:count}). A case counts when at least two of the three external judges name the same list in both label orders. An ideal list holds a query's ten best pool documents under each gain. A minimal pair sets a reranker's top 10 against the same list with its Count-tied documents in the tournament's order. \Cref{tab:app-judge-suite} reports the per-judge, per-suite ideal-list shares, which range from \val{b19.s2.kimi-k3.vidore.rcp-share}\% to \val{b19.s2.kimi-k3.trecdl.rcp-share}\%. The judges favor the tournament's order in \val{b19.s2.consensus.pooled.rcp-share}\% of decided ideal lists and \val{b19.s3.consensus.pooled.rcp-share}\% of minimal pairs. The judges side with the winner of \rcpndcg{} in \val{b19.s1.consensus.pooled.rcp-share}\% of the per-query cases where the two metrics name opposite winners. On controls, where both metrics name the same winner, the judges pick that winner in \val{b19.s1ctl.consensus.pooled.rcp-share}\%. The judges can thus see differences of this size. Each metric has a tie zone, a margin up to which we read its difference as a tie (\cref{app:robustness-judges-margins}). Where both metrics' differences exceed their tie zones, the judges side with \rcpndcg{} in \val{b19.mtie.consensus.pooled.both-decisive.rcp-share}\%.

\begin{table}[t]
  \centering
  \caption{Share (\%) of decided cases favoring the order of \rcpndcg{} over Count-nDCG [95\% CI] ($n$). 50 means no preference. For controls, the share counts the common winner. Pooled over NanoBEIR, BRIGHT, TREC-DL, and ViDoRe v3.}
  \label{tab:count}
  \footnotesize
  \setlength{\tabcolsep}{4pt}
  \begin{tabular}{@{}ll@{}}
    \toprule
    Test (unit) & LLM judges \\
    \midrule
    Ideal lists (queries) & \val{b19.s2.consensus.pooled.rcp-share} \val{b19.s2.consensus.pooled.rcp-share.ci} (\val{b19.s2.consensus.pooled.n-decided}) \\
    Minimal pairs (reranker lists) & \val{b19.s3.consensus.pooled.rcp-share} \val{b19.s3.consensus.pooled.rcp-share.ci} (\val{b19.s3.consensus.pooled.n-decided}) \\
    Per-query disagreements (reranker pairs) & \val{b19.s1.consensus.pooled.rcp-share} \val{b19.s1.consensus.pooled.rcp-share.ci} (\val{b19.s1.consensus.pooled.n-decided}) \\
    \quad both differences beyond their tie zones & \val{b19.mtie.consensus.pooled.both-decisive.rcp-share} \val{b19.mtie.consensus.pooled.both-decisive.rcp-share.ci} (\val{b19.mtie.consensus.pooled.both-decisive.n-decided}) \\
    \quad only the difference in \rcpndcg{} beyond & \val{b19.mtie.consensus.pooled.rcp-only.rcp-share} \val{b19.mtie.consensus.pooled.rcp-only.rcp-share.ci} (\val{b19.mtie.consensus.pooled.rcp-only.n-decided}) \\
    \quad only the difference in Count-nDCG beyond & \val{b19.mtie.consensus.pooled.count-only.rcp-share} \val{b19.mtie.consensus.pooled.count-only.rcp-share.ci} (\val{b19.mtie.consensus.pooled.count-only.n-decided}) \\
    Controls: common winner (reranker pairs) & \val{b19.s1ctl.consensus.pooled.rcp-share} \val{b19.s1ctl.consensus.pooled.rcp-share.ci} (\val{b19.s1ctl.consensus.pooled.n-decided}) \\
    \bottomrule
  \end{tabular}
\end{table}

\begin{table}[t]
\centering
\small
\caption{Share (\%) of decided queries on which each LLM judge prefers the ideal top-10 list of \rcpndcg{} over that of Count-nDCG, with 95\% CIs and the number of decided queries ($n$); 50\% means no preference. These are the values of \cref{fig:ranking}b. On BRIGHT, the lists can contain documents that BRIGHT's evaluation excludes (\cref{app:nano-bright-scope}).}
\label{tab:app-judge-suite}
\begin{tabular}{@{}lcccc@{}}
\toprule
Judge & NanoBEIR & BRIGHT & ViDoRe v3 & TREC-DL \\
\midrule
\model{GLM-5.3-flash} & \val{b19.s2.glm-5-3-flash.nanomteb.rcp-share} \val{b19.s2.glm-5-3-flash.nanomteb.rcp-share.ci} & \val{b19.s2.glm-5-3-flash.bright.rcp-share} \val{b19.s2.glm-5-3-flash.bright.rcp-share.ci} & \val{b19.s2.glm-5-3-flash.vidore.rcp-share} \val{b19.s2.glm-5-3-flash.vidore.rcp-share.ci} & \val{b19.s2.glm-5-3-flash.trecdl.rcp-share} \val{b19.s2.glm-5-3-flash.trecdl.rcp-share.ci} \\
 & $n$\,=\,\val{b19.s2.glm-5-3-flash.nanomteb.n-decided} & $n$\,=\,\val{b19.s2.glm-5-3-flash.bright.n-decided} & $n$\,=\,\val{b19.s2.glm-5-3-flash.vidore.n-decided} & $n$\,=\,\val{b19.s2.glm-5-3-flash.trecdl.n-decided} \\
\model{DeepSeek-4.1-flash} & \val{b19.s2.deepseek-4-1-flash.nanomteb.rcp-share} \val{b19.s2.deepseek-4-1-flash.nanomteb.rcp-share.ci} & \val{b19.s2.deepseek-4-1-flash.bright.rcp-share} \val{b19.s2.deepseek-4-1-flash.bright.rcp-share.ci} & \val{b19.s2.deepseek-4-1-flash.vidore.rcp-share} \val{b19.s2.deepseek-4-1-flash.vidore.rcp-share.ci} & \val{b19.s2.deepseek-4-1-flash.trecdl.rcp-share} \val{b19.s2.deepseek-4-1-flash.trecdl.rcp-share.ci} \\
 & $n$\,=\,\val{b19.s2.deepseek-4-1-flash.nanomteb.n-decided} & $n$\,=\,\val{b19.s2.deepseek-4-1-flash.bright.n-decided} & $n$\,=\,\val{b19.s2.deepseek-4-1-flash.vidore.n-decided} & $n$\,=\,\val{b19.s2.deepseek-4-1-flash.trecdl.n-decided} \\
\model{Kimi K3} & \val{b19.s2.kimi-k3.nanomteb.rcp-share} \val{b19.s2.kimi-k3.nanomteb.rcp-share.ci} & \val{b19.s2.kimi-k3.bright.rcp-share} \val{b19.s2.kimi-k3.bright.rcp-share.ci} & \val{b19.s2.kimi-k3.vidore.rcp-share} \val{b19.s2.kimi-k3.vidore.rcp-share.ci} & \val{b19.s2.kimi-k3.trecdl.rcp-share} \val{b19.s2.kimi-k3.trecdl.rcp-share.ci} \\
 & $n$\,=\,\val{b19.s2.kimi-k3.nanomteb.n-decided} & $n$\,=\,\val{b19.s2.kimi-k3.bright.n-decided} & $n$\,=\,\val{b19.s2.kimi-k3.vidore.n-decided} & $n$\,=\,\val{b19.s2.kimi-k3.trecdl.n-decided} \\
\bottomrule
\end{tabular}
\end{table}

\paragraph{Tie resolution.}
We take same-query document pairs that Count-nDCG ties exactly and that a reference orders, keeping pairs that can change a top-10 list. The tournament orders \val{b8.nist.qwen.pooled.tielev.acc.tournament} of these pairs like NIST, \val{b8.human.docs.all.tielev.acc.tournament} like the human panel, and \val{b8.vidore.tielev.acc.tournament} like the ViDoRe v3 qrels. All three accuracies are well above one half (\cref{tab:app-ties}).

\begin{table}[t]
  \centering
  \footnotesize
  \caption{Tie resolution: accuracy on same-query document pairs that Count ties exactly and a reference orders, restricted to pairs that can change a top-10 list. Ties count one half, chance 0.5. The NIST row uses the TREC-DL run of the named judge, the other rows \model{Qwen3.5-397B}. Brackets are 95\% query-cluster bootstrap intervals.}
  \label{tab:app-ties}
  \begin{tabular}{@{}lrc@{}}
    \toprule
    Reference & Pairs & Tournament \\
    \midrule
    NIST (\model{Qwen3.6-27B}) & \val{b8.nist.qwen.pooled.tielev.n} & \val{b8.nist.qwen.pooled.tielev.acc.tournament} [\val{b8.nist.qwen.pooled.tielev.acc.tournament.lo}, \val{b8.nist.qwen.pooled.tielev.acc.tournament.hi}] \\
    Human panel & \val{b8.human.docs.all.tielev.n} & \val{b8.human.docs.all.tielev.acc.tournament} [\val{b8.human.docs.all.tielev.acc.tournament.lo}, \val{b8.human.docs.all.tielev.acc.tournament.hi}] \\
    ViDoRe v3 qrels & \val{b8.vidore.tielev.n} & \val{b8.vidore.tielev.acc.tournament} [\val{b8.vidore.tielev.acc.tournament.lo}, \val{b8.vidore.tielev.acc.tournament.hi}] \\
    \bottomrule
  \end{tabular}
\end{table}

\paragraph{DL19 references.}
\citet{parry2025reannotation} re-assessed the \val{a2.setup.dl19.queries} queries of DL19, and we use their grades for analysis only (\cref{app:assets}). The two re-assessors graded \val{b16.nl.1a.cover.rcp}\% of the slots of \rcp{}'s ideal top 10 and \val{b16.nl.1a.cover.count-id}\% of those of Count's. An ungraded passage counts as 0. Under nDCG@10 on its own grades, NIST alone prefers neither ideal list. The consensus, the mean grade of NIST and both re-assessors, prefers the tournament's list on \val{b16.nl.1a.cons3.k10.count-id.winrate}\% of the queries where the two lists differ (95\% CI \val{b16.nl.1a.cons3.k10.count-id.winrate.ci}\%, \cref{tab:app-trecdl-parry}). On the \val{b9.nist.dl19.tie.all.n-pairs} DL19 pairs that NIST separates, Count ties, and both re-assessors graded, the tournament orders \val{b9.nist.dl19.tie.all.tournament.vs-nist} like NIST. On the pairs that a re-assessor separates, it agrees with that re-assessor on \val{b9.nist.dl19.tie.all.tournament.vs-parry1}, whereas NIST agrees with the re-assessors on \val{b9.nist.dl19.tie.all.nist-vs-parry-both}. Single NIST grades thus appear noisy for such fine distinctions.

\begin{table}[t]
  \centering
  \footnotesize
  \caption{DL19 references built from NIST and the re-assessments of \citet{parry2025reannotation} (\val{a2.setup.dl19.queries} queries). Ideal lists: share of queries, among those where the reference scores the two differently, where it scores the top 10 of \rcp{} above that of Count. Ties break by passage identifier. $\Delta$nDCG@10 is the mean difference in nDCG@10. Minimal pairs: share of reranker top-10 lists where the reference prefers the tournament's order of Count-tied passages. Rerankers gives the number (of 14) whose paired $t$-test over queries favors that order ($p<.05$). Brackets are 95\% CIs from a query bootstrap.}
  \label{tab:app-trecdl-parry}
  \setlength{\tabcolsep}{4pt}
  \begin{tabular}{@{}l ll cc@{}}
    \toprule
    & \multicolumn{2}{c}{Ideal lists} & \multicolumn{2}{c}{Minimal pairs} \\
    \cmidrule(lr){2-3}\cmidrule(lr){4-5}
    Reference & \rcp{} wins (\%) & $\Delta$nDCG@10 & tournament (\%) & rerankers \\
    \midrule
    NIST                       & \val{b16.nl.1a.nist.k10.count-id.winrate} \val{b16.nl.1a.nist.k10.count-id.winrate.ci} & \efv{b16.nl.1a.nist.k10.count-id.delta} & \val{b16.nl.mp.dl19.tie.nist.winrate-R} & \val{b16.nl.mp.dl19.tie.nist.sys-R} \\
    Re-assessor 1              & \val{b16.nl.1a.p1.k10.count-id.winrate} \val{b16.nl.1a.p1.k10.count-id.winrate.ci} & \efv{b16.nl.1a.p1.k10.count-id.delta} & \val{b16.nl.mp.dl19.tie.p1.winrate-R} & \val{b16.nl.mp.dl19.tie.p1.sys-R} \\
    Re-assessor 2              & \val{b16.nl.1a.p2.k10.count-id.winrate} \val{b16.nl.1a.p2.k10.count-id.winrate.ci} & \efv{b16.nl.1a.p2.k10.count-id.delta} & \val{b16.nl.mp.dl19.tie.p2.winrate-R} & \val{b16.nl.mp.dl19.tie.p2.sys-R} \\
    Consensus (mean of all three) & \val{b16.nl.1a.cons3.k10.count-id.winrate} \val{b16.nl.1a.cons3.k10.count-id.winrate.ci} & \efv{b16.nl.1a.cons3.k10.count-id.delta} & \val{b16.nl.mp.dl19.tie.cons3.winrate-R} & \val{b16.nl.mp.dl19.tie.cons3.sys-R} \\
    \bottomrule
  \end{tabular}
\end{table}

\paragraph{Minimal pairs.}
\label{tab:app-minimal}
For each reranker and query, we build two versions of its top 10 that differ only in the order of exactly Count-tied documents. One version follows the tournament's order and one reverses it. Unlike the minimal pairs of the LLM judges (\cref{app:robustness-judges}), these pairs use the reversed order, not the reranker's own order, as the alternative. Count-nDCG scores both alike. The qrels of each suite, NIST, and both re-assessors all prefer the tournament's version. Each reference's paired test over queries favors that order for all 14 rerankers of each suite, except on DL19 under NIST alone (\val{b16.nl.mp.dl19.tie.nist.sys-R} of 14) and under one re-assessor (\val{b16.nl.mp.dl19.tie.p2.sys-R} of 14).

\paragraph{Saturation and per-dataset leaderboards.}
Moving each reranker's top 10 toward Count's ideal list by greedy swaps raises Count's share of exact ties among the five best rerankers. On TREC-DL it rises from \val{b16.nl.sat.dl.pathC.s0.tie5C}\% to \val{b16.nl.sat.dl.pathC.s5.tie5C}\%. Elsewhere it reaches \val{b16.nl.sat.vidore.pathC.s5.tie5C}--\val{b16.nl.sat.nanomteb.pathC.s5.tie5C}\%, whereas \rcpndcg{} ties at most \val{b16.nl.sat.vidore.pathC.s5.tie5R}\%. Per dataset, the top-ranked reranker differs between the metrics in \val{b20.new.draft.app.hi.b14.all.top1diff} of the \val{draft.app.hi.b14.all.ndatasets} datasets.

\FloatBarrier
\section{External Preference Validation}
\label{app:external}

\subsection{NIST-decisive reranker pairs}
\label{app:trecdl-pairs}
\label{app:trecdl-leaderboards}

We compare each metric with the NIST grades of TREC-DL in the pool and top-150 panels of \cref{app:trecdl-protocol}. A pair is NIST-decisive when qrel-nDCG@10 with the NIST grades as linear gains separates it (\cref{sec:external}). \Cref{tab:app-trecdl-pairs} counts, for each year and panel, the reranker pairs that each metric separates and the NIST-decisive pairs on which it agrees with NIST. The share of separated pairs is a metric's sensitivity~\citep{sakai2006bootstrap,clarke2020compatibility} (\cref{app:notation-stats}). A metric without signal would separate about \val{draft.2.siglevel}\% of pairs. In the pool panel, qrel-nDCG@10 separates \val{a2.dec.pool.dl19.nistlin} of the \val{draft.5.dl.pairs} pairs in 2019 and \val{a2.dec.pool.dl20.nistlin} in 2020. Count-nDCG and \rcpndcg{} separate many more pairs in both years. Since Count-nDCG scores documents by the rubric alone, the rubric supplies most of the added resolution.

In the pool panel with linear gain, \rcpndcg{} agrees with NIST on every NIST-decisive pair in both years. The agreement holds under each of its three calibration fits, CAL-A, CAL-B, and CAL-C (\cref{app:trecdl-protocol}), of which \cref{tab:app-trecdl-pairs} shows the first two. Under exponential gain, in the top-150 panel, and in every other setting we test, the few NIST-decisive pairs that go the other way all involve \reranker{zerank-2} or \reranker{zerank-1} (\cref{app:trecdl-zerank}). When we test each pair once over the queries of both years, in the pool panel with linear gain, \val{a2.dec.pool.pooled.nistlin} pairs are NIST-decisive. \rcpndcg{} agrees with NIST on \val{draft.app.ef.dl.pooled.rcpa.agree.k} of them. \rcpndcg{} separates only \val{draft.app.ef.dl.pooled.rcpa.contra} pair in the direction opposite to NIST, \reranker{CTXL-RR-2B} against \reranker{zerank-2}. We change one setting at a time, from linear to exponential NIST gain or from CAL-A to CAL-B. Each change alone leaves the agreement of \rcpndcg{} at worst one pair short in 2019 and complete in 2020. When tied reranker scores share their mean gain, as on ViDoRe v3, \val{n17.ties.cala.dl19.decisive} pairs are NIST-decisive in 2019 and \val{n17.ties.cala.dl20.decisive} in 2020. With either TREC-DL judge, \rcpndcg{} favors NIST's winner on all of them.

Over the 14 rerankers of the pool panel, the leaderboards of \rcpndcg{} and NIST agree weakly in 2019 and moderately in 2020. The Kendall correlations are \val{a2.kendall.pool.dl19.rcpa.nistlin.n14} and \val{a2.kendall.pool.dl20.rcpa.nistlin.n14}, where 1 means identical orders and 0 means unrelated orders. With the four first-stage runs added, the correlations of \rcpndcg{} with NIST are \val{a2.kendall.pool.dl19.rcpa.nistlin.n18} and \val{draft.app.ef.dl20.kendall18}. On DL19, NIST separates none of the ten pairs among the five rerankers that \rcpndcg{} ranks highest. The re-assessment grades of \citet{parry2025reannotation} (\cref{app:trecdl-audit}) separate none either. No human reference that we use thus orders the five leading rerankers on DL19.

\begin{table}[t]
  \centering
  \footnotesize
  \caption{Reranker pairs on TREC-DL (14 rerankers, \val{draft.5.dl.pairs} pairs per year). Upper block: pairs that each metric separates by a paired $t$-test over the year's queries ($p<.05$). Lower block: NIST-decisive pairs on which a metric's mean difference has the sign of NIST's under linear gain. \rcpndcg{} uses \model{Qwen3.6-27B} and the calibration named in each row (\cref{app:trecdl-protocol}). Under exponential gain and in the top-150 panel, every disagreement with NIST involves \reranker{zerank-2}.}
  \label{tab:app-trecdl-pairs}
  \setlength{\tabcolsep}{6pt}
  \begin{tabular}{@{}l cc cc@{}}
    \toprule
    & \multicolumn{2}{c}{DL19} & \multicolumn{2}{c}{DL20} \\
    \cmidrule(lr){2-3}\cmidrule(l){4-5}
    & pool & top-150 & pool & top-150 \\
    \midrule
    \multicolumn{5}{@{}l}{\emph{Pairs separated}} \\
    qrel-nDCG (NIST), linear gain      & \val{a2.dec.pool.dl19.nistlin} & \val{a2.dec.top150.dl19.nistlin} & \val{a2.dec.pool.dl20.nistlin} & \val{a2.dec.top150.dl20.nistlin} \\
    qrel-nDCG (NIST), exponential gain & \val{a2.dec.pool.dl19.nistexp} & \val{a2.dec.top150.dl19.nistexp} & \val{a2.dec.pool.dl20.nistexp} & \val{a2.dec.top150.dl20.nistexp} \\
    Count-nDCG                         & \val{a2.dec.pool.dl19.count}   & \val{a2.dec.top150.dl19.count}   & \val{a2.dec.pool.dl20.count}   & \val{a2.dec.top150.dl20.count} \\
    \rcpndcg{}, CAL-A                  & \val{a2.dec.pool.dl19.rcpa}    & \val{a2.dec.top150.dl19.rcpa}    & \val{a2.dec.pool.dl20.rcpa}    & \val{a2.dec.top150.dl20.rcpa} \\
    \rcpndcg{}, CAL-B                  & \val{a2.dec.pool.dl19.rcpb}    & \val{a2.dec.top150.dl19.rcpb}    & \val{a2.dec.pool.dl20.rcpb}    & \val{a2.dec.top150.dl20.rcpb} \\
    \midrule
    \multicolumn{5}{@{}l}{\emph{Same winner as NIST, linear gain}} \\
    \rcpndcg{}, CAL-A & \val{a2.agree.pool.dl19.rcpa.atnistlin}  & \val{a2.agree.top150.dl19.rcpa.atnistlin}  & \val{a2.agree.pool.dl20.rcpa.atnistlin}  & \val{a2.agree.top150.dl20.rcpa.atnistlin} \\
    \rcpndcg{}, CAL-B & \val{a2.agree.pool.dl19.rcpb.atnistlin}  & \val{a2.agree.top150.dl19.rcpb.atnistlin}  & \val{a2.agree.pool.dl20.rcpb.atnistlin}  & \val{a2.agree.top150.dl20.rcpb.atnistlin} \\
    \bottomrule
  \end{tabular}
\end{table}

\subsection{NIST positives that \texorpdfstring{\rcp{}}{RCP} scores lowest}\label{app:trecdl-audit}
We ask whether the NIST positives that \rcp{} scores lowest point to errors of \rcp{} or of the NIST grades. In the spirit of the label audit of \citet{land2026auditing}, we examine the NIST positives (grade 2 or 3) whose calibrated score lies in the bottom quartile of their year's positives. These passages are \val{n1.counts.pooled.suspects} of the \val{n1.counts.pooled.g2} positives of both years. We compare each with the other positives of its query.

All 14 rerankers rank these passages below the other positives of their query. The within-query AUC is the probability that a reranker ranks another positive of the query above such a passage, where 0.5 means chance. Brackets give 95\% CIs from a bootstrap over queries. The mean AUC over the 14 rerankers is \val{n1.rr.auc.pooled}, and every reranker's interval lies above 0.5. Rerankers still rank them above most passages that NIST grades 0 or 1. Against those passages, the same AUC is \val{n1.rr.auc01.pooled}, below the chance level of 0.5. Stronger rerankers, measured by nDCG@10 on the NIST grades, push the passages further down. Over the 14 rerankers, the Spearman correlation between this nDCG@10 and a reranker's AUC is \val{n1.rr.rho.pooled}. Without the zerank models, which learn from LLM preferences, it is \val{n1.nz.rho.pooled}. The first-stage retrievers also rank these passages below the other positives, but less consistently (AUC \val{n1.fs.auc.pooled}).

On DL19, the re-assessment of \citet{parry2025reannotation} grades all of these passages, blind to the NIST grades. Its re-assessors rate \val{n1.parry.ge2_pmean.sus}\% of them highly relevant or better, against \val{n1.parry.ge2_pmean.ref}\% of the other positives. The mean re-assessment grades, on the scale from 0 to 3, are \val{n1.parry.mean.pmean.sus} for these passages and \val{n1.parry.mean.pmean.ref} for the other positives. Among NIST grade-3 passages alone, re-assessors rate \val{n1.parry.g3.ge2_pmean.sus}\% of the \val{n1.parry.g3.ge2_pmean.sus.n} low-scored ones highly relevant or better, against \val{n1.parry.g3.ge2_pmean.ref}\% of the other \val{n1.parry.g3.ge2_pmean.ref.n}. The 2019 low-gain tail at grade 3 in \cref{fig:extcal}c thus consists largely of passages that the re-assessors also rate low.

The re-assessment also tests the calibrated scores beyond the low-scored positives. Within-query concordance is the share of same-query passage pairs with different grades that a score orders like the grades, where one half is chance. On all passages that both references grade, the concordance of $\dpar{\tilde{\theta}}$ is \val{n2.doc.conc.theta.nist} with NIST and \val{n2.doc.conc.theta.pmean} with the re-assessors. The re-assessors' own concordance with NIST is \val{n2.doc.conc.nist_pmean}. On DL19, nDCG@10 on the re-assessment grades separates \val{n2.sys.dec.pmean} of the \val{draft.5.dl.pairs} reranker pairs. \rcpndcg{}@10 picks the winner of each. The rerankers may share preferences with our judge. The quartile cut also selects candidate label errors rather than estimating an error rate.

\subsection{Human verdicts and controls}
\label{app:human-verdicts}
\label{app:human-controls}

This subsection gives the full results and controls behind the human verdicts of \cref{sec:external}. \Cref{tab:app-human-verdicts} collects the results of the verdict layer (\cref{app:human}). The main measure and its slices count decided contests (\cref{app:human-instrument}). Because a metric tie counts only for the other metric, chance for the main measure is not one half. Suppose each verdict were a coin flip. Decided disagreements then favor either metric with probability one half. A contest in which only one metric names a winner counts, for that metric, only when the verdict names that winner, which happens with probability one half. The study has \val{a1a.desc.decided-flips} disagreements with a verdict, \val{a1a.desc.decided-ndcg-tie} contests with a verdict in which only \rcpndcg{} names a winner, and \val{a1a.desc.decided-theta-tie} in which only qrel-nDCG does. The expected share is then $(\val{a1a.desc.decided-flips} + \val{a1a.desc.decided-ndcg-tie}) / (2 \cdot \val{a1a.desc.decided-flips} + \val{a1a.desc.decided-ndcg-tie} + \val{a1a.desc.decided-theta-tie}) = \val{draft.ff.chance.exact}\%$. On disagreements alone, chance is one half.

\begin{table}[!htbp]
\footnotesize
\caption{Verdict layer. Rates are percentages of decided contests won by \rcpndcg{}. Disagreements with a verdict are always decided. Chance is about \val{draft.ff.chance}\% for the main measure and one half on disagreements. CIs are clustered by query for rates and come from a delete-one-query jackknife for margin AUCs.}
\label{tab:app-human-verdicts}
\setlength{\tabcolsep}{3pt}
\begin{tabular*}{\textwidth}{@{}l@{\extracolsep{\fill}}rcl@{}}
\toprule
Statistic & Est. & 95\% CI & $k/n$ \\
\midrule
Main measure & \val{a1a.primary} & [\val{a1a.primary.lo}, \val{a1a.primary.hi}] & \val{a1a.primary.k}/\val{a1a.primary.n} \\
\quad NanoBEIR & \val{nr.a2.main.nanomteb} & [\val{nr.a2.main.nanomteb.lo}, \val{nr.a2.main.nanomteb.hi}] & \val{nr.a2.main.nanomteb.k}/\val{nr.a2.main.nanomteb.n} \\
\quad BRIGHT & \val{nr.a2.main.bright} & [\val{nr.a2.main.bright.lo}, \val{nr.a2.main.bright.hi}] & \val{nr.a2.main.bright.k}/\val{nr.a2.main.bright.n} \\
\quad ViDoRe v3 & \val{nr.a2.main.vidore} & [\val{nr.a2.main.vidore.lo}, \val{nr.a2.main.vidore.hi}] & \val{nr.a2.main.vidore.k}/\val{nr.a2.main.vidore.n} \\
\quad without zerank & \val{nr.a2.main.nn} & [\val{nr.a2.main.nn.lo}, \val{nr.a2.main.nn.hi}] & \val{nr.a2.main.nn.k}/\val{nr.a2.main.nn.n} \\
\quad one annotator left out (range) & -- & \val{nr.a2.main.loao.min}--\val{nr.a2.main.loao.max} & \val{nr.a2.main.loao.runs} runs \\
\quad without contests that qrel-nDCG cannot score & \val{nr.a2.main.ndcg-seen} & [\val{nr.a2.main.ndcg-seen.lo}, \val{nr.a2.main.ndcg-seen.hi}] & \val{nr.a2.main.ndcg-seen.k}/\val{nr.a2.main.ndcg-seen.n} \\
\midrule
Disagreements only & \val{a1a.h2h} & [\val{a1a.h2h.lo}, \val{a1a.h2h.hi}] & \val{a1a.h2h.k}/\val{a1a.h2h.n} \\
\quad without zerank & \val{a1a.h2h.nn} & [\val{a1a.h2h.nn.lo}, \val{a1a.h2h.nn.hi}] & \val{a1a.h2h.nn.k}/\val{a1a.h2h.nn.n} \\
\quad one annotator left out (range) & -- & \val{a1a.loao.h2h.min}--\val{a1a.loao.h2h.max} & \val{a1a.loao.h2h.min.n} runs \\
Disagreements, both margins $\geq$ \val{draft.tab.verdicts.decisive-margin} & \val{a1a.decisive} & [\val{a1a.decisive.lo}, \val{a1a.decisive.hi}] & \val{a1a.decisive.k}/\val{a1a.decisive.n} \\
Disagreements, a margin $<$ \val{draft.tab.verdicts.decisive-margin} & \val{a1a.marginal} & [\val{a1a.marginal.lo}, \val{a1a.marginal.hi}] & \val{a1a.marginal.k}/\val{a1a.marginal.n} \\
\midrule
Margin AUC: \rcpndcg{} & \val{a1a.auc.theta} & [\val{a1a.auc.theta.lo}, \val{a1a.auc.theta.hi}] & \val{a1a.auc.theta.n} picks \\
\quad qrel-nDCG & \val{a1a.auc.ndcg} & [\val{a1a.auc.ndcg.lo}, \val{a1a.auc.ndcg.hi}] & \val{a1a.auc.ndcg.n} picks \\
\quad \rcpndcg{} vs.\ qrel-nDCG, NanoBEIR & \multicolumn{2}{c}{\val{a1a.auc.theta.nanomteb} vs.\ \val{a1a.auc.ndcg.nanomteb}} & \val{a1a.auc.theta.nanomteb.n} picks \\
\quad \rcpndcg{} vs.\ qrel-nDCG, BRIGHT & \multicolumn{2}{c}{\val{a1a.auc.theta.bright} vs.\ \val{a1a.auc.ndcg.bright}} & \val{a1a.auc.theta.bright.n} picks \\
\quad \rcpndcg{} vs.\ qrel-nDCG, ViDoRe v3 & \multicolumn{2}{c}{\val{a1a.auc.theta.vidore} vs.\ \val{a1a.auc.ndcg.vidore}} & \val{a1a.auc.theta.vidore.n} picks \\
\quad \rcpndcg{} vs.\ qrel-nDCG, without zerank & \multicolumn{2}{c}{\val{a1a.auc.theta.nn} vs.\ \val{a1a.auc.ndcg.nn}} & \val{a1a.auc.theta.nn.n} picks \\
\bottomrule
\end{tabular*}
\end{table}

The main measure's 95\% CI lies above chance in each suite and without zerank. A metric built on an LLM judge could favor the three zerank rerankers, which learn from LLM preferences~\citep{pipitone2025zelo}. Without them, the main measure is \val{nr.a2.main.nn}\% (\val{nr.a2.main.nn.k} of \val{nr.a2.main.nn.n}). When we leave out each annotator in turn and recompute the main measure, it stays between \val{nr.a2.main.loao.min}\% and \val{nr.a2.main.loao.max}\% over \val{nr.a2.main.loao.runs} runs, inside its CI. In \val{b21.human.arguana.contests-own} NanoArguAna contests, the displayed documents include the query's own argument, which the NanoBEIR leaderboards leave out (\cref{app:nano-bright}). The contests and their metric margins keep it as displayed. Annotators graded it \val{b21.human.arguana.own-reviews} times with a mean grade of \val{b21.human.arguana.own-grade-mean} and never Useful or better. This low grade matches the low \rcp{} gain of the argument better than the higher credit that Count-nDCG gives it. Without these contests, the main measure is \val{nr.a2.b21.main-wo-own}\% (\val{nr.a2.b21.main-wo-own.k} of \val{nr.a2.b21.main-wo-own.n}).

Two secondary measures also favor \rcpndcg{}. Where a metric names a winner, the verdict agrees with \rcpndcg{} in \val{a1a.hit-theta}\% (\val{a1a.hit-theta.k} of \val{a1a.hit-theta.n}) of contests and with qrel-nDCG in \val{a1a.hit-ndcg}\%. Chance is one half. Where both margins of a disagreement are also at least \val{draft.tab.verdicts.decisive-margin}, verdicts side with \rcpndcg{} in \val{a1a.decisive}\% of contests.

The margin AUC asks whether a metric's margin, its score difference between the two lists, predicts each annotator's pick. One half is chance, and 1 is perfect. The margin AUC never classifies a winner. The tie asymmetry of the main measure therefore does not enter. \rcpndcg{} reaches \val{a1a.auc.theta}, against \val{a1a.auc.ndcg} for qrel-nDCG, and leads in every suite and without zerank. On disagreements alone, where the two margins point in opposite directions by construction, \rcpndcg{} reaches \val{a1a.auc.theta.flips} and qrel-nDCG \val{a1a.auc.ndcg.flips}.

If a metric measures what annotators see, its winner should match the verdict more often when its margin is large. \Cref{fig:dose}a bins each metric's absolute margins into five bands. Since bands can differ in how many disagreements they hold, we also bin the margins into fifths over the \val{a1a.qstd.n} contests in which both metrics name a winner. Each fifth is reweighted to the study's share of disagreements among all contests with a verdict (\val{a1a.qstd.share-disagree}\%). From the lowest to the highest fifth, the reweighted rate of \rcpndcg{} rises by \val{a1a.qstd.rcp.d51}\,pp (95\% CI [\val{a1a.qstd.rcp.d51.lo}, \val{a1a.qstd.rcp.d51.hi}]\,pp). The reweighted rate of qrel-nDCG rises by \val{a1a.qstd.ndcg.d51}\,pp (95\% CI [\val{a1a.qstd.ndcg.d51.lo}, \val{a1a.qstd.ndcg.d51.hi}]\,pp). Both intervals come from a bootstrap over queries.

In the lowest band, below \val{draft.5.margin.decisive}, the verdict matches the winner of \rcpndcg{} in \val{a1a.band.rcp.1}\% of \val{a1a.band.rcp.1.n} contests (95\% CI [\val{a1a.band.rcp.1.lo}, \val{a1a.band.rcp.1.hi}]\%). This share is at chance level. The marginal row of \cref{tab:app-human-verdicts} instead counts only the decided disagreements in which both metrics name a winner. Where qrel-nDCG names the same list in this band, \rcpndcg{} matches \val{nr.a2.neartie.same.k} of \val{nr.a2.neartie.same.n} verdicts. Where qrel-nDCG names the other list, \rcpndcg{} matches only \val{nr.a2.neartie.other.k} of \val{nr.a2.neartie.other.n} verdicts. The margin of qrel-nDCG in these contests is large, with a median of \val{nr.a2.neartie.other.median-abs-dndcg}. Above \val{draft.5.margin.decisive}, \rcpndcg{} matches \val{a1a.sweep.p02.above}\% of verdicts. We therefore read differences of \rcpndcg{} below \val{draft.5.margin.decisive} between two lists of one query as ties.

The controls test for artifacts of the instrument and of the annotators. We replay every verdict with a metric's own gains in place of the annotators' grades. The replayed verdicts match that metric's winner in \val{a1a.ctrl.ceiling.qrel-perfect}\% of contests for qrel-nDCG and \val{a1a.ctrl.ceiling.rcp-perfect}\% for \rcpndcg{}. The instrument thus handicaps \rcpndcg{}. A query-blind baseline tests whether verdicts follow overall reranker strength alone. The baseline fits a Bradley--Terry model of reranker strength to the verdicts of all other queries. This model predicts the held-out query's verdicts and matches \val{a1a.ctrl.query-blind-loqo}\% of them. Split pairs meet in several contests and win on some queries and lose on others. Always picking the pair's more frequent winner matches \val{a1a.ctrl.withinpair.split.constant}\% of verdicts, whereas \rcpndcg{} matches \val{a1a.ctrl.withinpair.split.rcp}\% and qrel-nDCG \val{a1a.ctrl.withinpair.split.ndcg}\%. Annotators prefer the longer displayed list in \val{a1a.len.human-prefers-longer}\% of contests. We therefore fit a logistic regression that predicts each annotator's pick from both standardized margins and the log ratio of the two lists' lengths. The length weight is $\val{a1a.len.logit.b-loglen}$ (95\% CI [$\val{a1a.len.logit.b-loglen.lo}$, $\val{a1a.len.logit.b-loglen.hi}$]), and the two margin weights stay at \val{a1a.len.logit.b-theta} for \rcpndcg{} and \val{a1a.len.logit.b-ndcg} for qrel-nDCG. List position cannot drive a verdict, because annotators saw the documents in random order and never saw the lists or their labels. We recompute the verdicts under eight other gain maps. \rcpndcg{} leads qrel-nDCG under every map, and the main measure under the study's map lies in the middle of their range. Under the four binary maps, which count relevance from one question upward, the main measure ranges from \val{nr.a2.gainmap.binary-ge4.main}\% (from Complete) to \val{nr.a2.gainmap.binary-ge2.main}\% (from Useful).

\FloatBarrier

\subsection{LLM judge panel and Study 2}\label{app:meta}

This subsection describes the three LLM judges of Study~1 (\cref{app:meta-study1}) and Study~2 and reports Study~2, which compares the reranker lists that qrel-nDCG and \rcpndcg{} prefer.

\paragraph{Judges and verdicts.}
Studies 1 and 2 use the three judges of \cref{app:robustness-judges}, \model{GLM-5.3-flash}, \model{DeepSeek-4.1-flash}, and \model{Kimi K3}. The three come from model families other than those of our judges. Each judge compares two candidates without knowing their source and never produces a relevance gain. Each case shows a query and two candidates. Except in the controls, one is favored by \rcp{} and the other by the qrels. Every case appears twice, with the labels and the display order swapped. A judge decides a case when both presentations name the same candidate. Under the majority rule, a case counts for a candidate when at least two of the three judges decide it for that candidate. We also report the share among cases that all three judges decide for the same candidate. Under either rule, a share of 50\% means no preference. \Cref{app:prompts-meta} prints the prompts.

\paragraph{Study 2.}\label{app:meta-study2}
Study 2 asks which of two rerankers' top-10 lists a judge prefers when qrel-nDCG@10 and \rcpndcg{}@10 name different winners. On each NanoBEIR dataset, we pair the four best rerankers under qrel-nDCG with the four best under \rcpndcg{}. A sign-flip is a query and one of these reranker pairs for which the two metrics name different winners, each by a margin above \val{draft.app.d.tie-margin}. We sampled \val{n10.s2.nano.signflips} sign-flips and \val{n10.s2.nano.controls} controls, on which both metrics name the same winner.

By majority, the judges prefer the list of \rcpndcg{} in \val{n10.s2.nano.maj}\% of decided sign-flips (\val{n10.s2.nano.maj.k} of \val{n10.s2.nano.maj.n}), with a 95\% CI of [\val{n10.s2.nano.maj.lo}, \val{n10.s2.nano.maj.hi}]\% by a query-cluster bootstrap. On the \val{n10.s2.nano.unan.n} sign-flips that all three judges decide the same way, the share is \val{n10.s2.nano.unan}\%. Each judge alone also prefers this list, with shares from \val{n10.s2.nano.deepseek}\% to \val{n10.s2.nano.glm}\%. The majority favors it on \val{n10.s2.nano.ds.maj-favors-rcp} of the \val{n10.s2.nano.ds.count} datasets, and HotpotQA is the exception (\cref{app:meta-hotpotqa}). On the controls, the judges pick the list that both metrics prefer in \val{n10.s2.nano.ctl.maj}\% of decided cases (\val{n10.s2.nano.ctl.maj.k} of \val{n10.s2.nano.ctl.maj.n}).

Neither display order nor list length explains this preference. A decided case needs the same verdict in both orders. In single presentations, each judge sides with \rcpndcg{} in either display order. \model{GLM-5.3-flash}, for example, sides with \rcpndcg{} in \val{n10.s2.nano.pos.glm.rcp-first}\% of its non-tie verdicts when that list comes first and in \val{n10.s2.nano.pos.glm.rcp-second}\% when it comes second. The list of \rcpndcg{} has more characters in \val{n10.s2.nano.len.rcp-longer-share}\% of sign-flips. The judges pick it in \val{n10.s2.nano.len.maj.rcp-shorter}\% of decided cases where it is shorter and in \val{n10.s2.nano.len.maj.rcp-longer}\% where longer.

On BRIGHT, the share for the list of \rcpndcg{} is \val{n10.s2.bright.all.maj}\% over all twelve tasks (\val{n10.s2.bright.all.maj.k} of \val{n10.s2.bright.all.maj.n}, 95\% CI [\val{n10.s2.bright.all.maj.lo}, \val{n10.s2.bright.all.maj.hi}]\%), and \val{n10.s2.bright.all.unan}\% when all three judges must agree. Unlike Study~1, Study~2 keeps its original lists. On NanoArguAna, both lists of 17 of the 800 sign-flips show the query's own argument (\cref{app:nano-bright}). On AoPS, LeetCode, and TheoremQA Questions, these lists still contain BRIGHT's excluded documents (\cref{app:nano-bright-scope}). Both lists show the query's own problem in \val{n10.s2.bright.aff.self-both} of the \val{n10.s2.bright.aff.signflips} sign-flips on these tasks, and only the list of \rcpndcg{} shows it in \val{n10.s2.bright.aff.self-only-rcp}. On these three tasks, the share for the list of \rcpndcg{} is \val{n10.s2.bright.aff.maj}\%, against \val{n10.s2.bright.unaff.maj}\% on the other nine. On the BRIGHT controls, the judges pick the list that both metrics prefer in \val{n10.s2.bright.all.ctl.maj}\% of decided cases.

\paragraph{Judgments by two of the authors.}
The authors who judged Study~1 (\cref{app:meta-study1}) also judged a subset of \val{draft.app.jm.s2.human.signflips} sign-flips and \val{draft.app.jm.s2.human.controls} controls as in Study~1. With their judgments pooled, \val{a56.old.s2.human.rate}\% of \val{a56.old.s2.human.rate.n} non-tie judgments on sign-flips prefer the list of \rcpndcg{} (95\% Wilson CI [\val{a56.old.s2.human.rate.lo}, \val{a56.old.s2.human.rate.hi}]\%). List length confounds this result. Their judgments favor the list of \rcpndcg{} in \val{a56.old.s2.len.human.when-theta-longer}\% of the non-tie cases where its list is the longer one. Where it is the shorter one, they favor it in \val{a56.old.s2.len.human.when-theta-shorter}\% (95\% CI [\val{a56.old.s2.len.human.when-theta-shorter.lo}, \val{a56.old.s2.len.human.when-theta-shorter.hi}]\%), an interval that includes one half. Unlike these judgments by two of the authors, the external human study of \cref{sec:external,app:human} uses paid annotators who are not among the authors.

\subsection{External LLM judges}\label{app:robustness-judges}\label{app:robustness-judges-results}\label{app:robustness-judges-margins}

\paragraph{Setup.}
To test the tournament's order of Count-tied documents (\cref{app:robustness-count}) directly, we asked LLM judges from model families other than those of our judges to compare two ranked lists. Each case contrasts the \rcp{} side, the list that the tournament's order favors, with the Count side, the list that Count's order favors. The comparison concerns the order within a query and tests neither the calibration across queries nor the leaderboards.

We built four sets of cases, namely per-query disagreements, controls, ideal lists, and minimal pairs, on NanoBEIR, BRIGHT, ViDoRe v3, and TREC-DL, with the paper's 14 rerankers. A per-query disagreement is one query and two rerankers for which Count-nDCG@10 and \rcpndcg{}@10 name opposite winners. We sample these cases in equal thirds by the difference in \rcpndcg{}@10. Controls are cases in which both metrics name the same winner, ten per third and suite. An ideal-list case sets \rcp{}'s ideal top 10 of a query against Count's ideal top 10, whose ties we break at random. A minimal pair sets a reranker's top 10 against the same list in which each group of exactly Count-tied documents takes the tournament's order on its own positions. Count-nDCG@10 scores the two lists alike. Unlike the minimal pairs of \cref{tab:app-minimal}, which reverse the tournament's order, this contrast keeps the reranker's own order as the alternative. The four sets hold \val{nr.a5.b19.cases} cases.

The judges are \model{GLM-5.3-flash}, \model{DeepSeek-4.1-flash}, and \model{Kimi K3}, and \cref{app:prompts-external} gives the prompt. The label swap, the decision rule, and the majority rule are those of \cref{app:meta}. Our measure is the share of these decided cases that favor the \rcp{} side, pooled over suites, with 95\% bootstrap intervals over queries per suite. For controls, the share counts the list that both metrics prefer. A share of 50\% means no preference.

\paragraph{Results.}
\Cref{tab:app-b19-judges} gives the shares pooled over suites, and \cref{tab:app-judge-suite} the per-suite ideal-list shares. The majority favors the \rcp{} side in all three sets other than the controls. The share is \val{b19.s1.consensus.pooled.rcp-share}\% on per-query disagreements (95\% CI \val{b19.s1.consensus.pooled.rcp-share.ci}\%), \val{b19.s2.consensus.pooled.rcp-share}\% on ideal lists, and \val{b19.s3.consensus.pooled.rcp-share}\% on minimal pairs. Each judge alone also favors the \rcp{} side in every set. Each judge's pooled 95\% interval lies above 50\% in every set. No suite favors the Count side in any set.

\paragraph{Margin split and restricted sets.}
Opposite winners can also come from a near tie. We therefore split the per-query disagreements by tie zones. Each metric receives a tie zone, a margin up to which we read a difference between the two lists as a tie. For \rcpndcg{}, the zone is \val{b19.mtie.tiezone.rcp}, the margin of \rcpndcg{}@5 below which the human verdicts are near chance (\cref{sec:external}). For Count-nDCG, we choose the zone that holds the same share of all per-query reranker pairs as the zone of \rcpndcg{}. This share is \val{b19.mtie.tiezone.nanomteb.percentile}\% on NanoBEIR, \val{b19.mtie.tiezone.bright.percentile}\% on BRIGHT, \val{b19.mtie.tiezone.vidore.percentile}\% on ViDoRe v3, and \val{b19.mtie.tiezone.trecdl.percentile}\% on TREC-DL. Where both metrics are beyond their zones, the judges side with the winner of \rcpndcg{} in \val{b19.mtie.consensus.pooled.both-decisive.rcp-share}\% of decided cases (95\% CI \val{b19.mtie.consensus.pooled.both-decisive.rcp-share.ci}\%). Where only Count-nDCG is beyond its zone, they favor neither list, with \val{b19.mtie.consensus.pooled.count-only.rcp-share}\% (95\% CI \val{b19.mtie.consensus.pooled.count-only.rcp-share.ci}\%), as \rcpndcg{} itself is inside its tie zone there. The cells of this split are tied to the thirds of the sampling. The split is therefore not independent evidence. Dropping BRIGHT cases that show an excluded document and NanoArguAna cases that show the query's own argument changes the shares little. The majority's shares become \val{b19.restr.excl-arguana.s1.consensus.pooled.rcp-share}\% on per-query disagreements, \val{b19.restr.excl-arguana.s2.consensus.pooled.rcp-share}\% on ideal lists, and \val{b19.restr.excl-arguana.s3.consensus.pooled.rcp-share}\% on minimal pairs, each with a 95\% interval above 50\%.

\begin{table}[t]
  \centering
  \footnotesize
  \caption{External LLM judges, pooled over the four suites: share of decided cases that favor the \rcp{} side (\%), with the decided cases in parentheses. The majority column adds 95\% query-bootstrap intervals and the number of all cases. For controls, the share counts the list that both metrics prefer. 50\% means no preference.}
  \label{tab:app-b19-judges}
  \setlength{\tabcolsep}{2pt}
  \begin{tabular*}{\textwidth}{@{}l@{\extracolsep{\fill}}cccc@{}}
    \toprule
    Set & Majority & DeepSeek & GLM & Kimi \\
    \midrule
    Disagreements & \val{b19.s1.consensus.pooled.rcp-share} \val{b19.s1.consensus.pooled.rcp-share.ci} (\val{b19.s1.consensus.pooled.n-decided}/\val{b19.s1.consensus.pooled.n-cases}) & \val{b19.s1.deepseek-4-1-flash.pooled.rcp-share} (\val{b19.s1.deepseek-4-1-flash.pooled.n-decided}) & \val{b19.s1.glm-5-3-flash.pooled.rcp-share} (\val{b19.s1.glm-5-3-flash.pooled.n-decided}) & \val{b19.s1.kimi-k3.pooled.rcp-share} (\val{b19.s1.kimi-k3.pooled.n-decided}) \\
    \addlinespace
    Controls & \val{b19.s1ctl.consensus.pooled.rcp-share} \val{b19.s1ctl.consensus.pooled.rcp-share.ci} (\val{b19.s1ctl.consensus.pooled.n-decided}/\val{b19.s1ctl.consensus.pooled.n-cases}) & \val{b19.s1ctl.deepseek-4-1-flash.pooled.rcp-share} (\val{b19.s1ctl.deepseek-4-1-flash.pooled.n-decided}) & \val{b19.s1ctl.glm-5-3-flash.pooled.rcp-share} (\val{b19.s1ctl.glm-5-3-flash.pooled.n-decided}) & \val{b19.s1ctl.kimi-k3.pooled.rcp-share} (\val{b19.s1ctl.kimi-k3.pooled.n-decided}) \\
    \addlinespace
    Ideal lists & \val{b19.s2.consensus.pooled.rcp-share} \val{b19.s2.consensus.pooled.rcp-share.ci} (\val{b19.s2.consensus.pooled.n-decided}/\val{b19.s2.consensus.pooled.n-cases}) & \val{b19.s2.deepseek-4-1-flash.pooled.rcp-share} (\val{b19.s2.deepseek-4-1-flash.pooled.n-decided}) & \val{b19.s2.glm-5-3-flash.pooled.rcp-share} (\val{b19.s2.glm-5-3-flash.pooled.n-decided}) & \val{b19.s2.kimi-k3.pooled.rcp-share} (\val{b19.s2.kimi-k3.pooled.n-decided}) \\
    \addlinespace
    Minimal pairs & \val{b19.s3.consensus.pooled.rcp-share} \val{b19.s3.consensus.pooled.rcp-share.ci} (\val{b19.s3.consensus.pooled.n-decided}/\val{b19.s3.consensus.pooled.n-cases}) & \val{b19.s3.deepseek-4-1-flash.pooled.rcp-share} (\val{b19.s3.deepseek-4-1-flash.pooled.n-decided}) & \val{b19.s3.glm-5-3-flash.pooled.rcp-share} (\val{b19.s3.glm-5-3-flash.pooled.n-decided}) & \val{b19.s3.kimi-k3.pooled.rcp-share} (\val{b19.s3.kimi-k3.pooled.n-decided}) \\
    \bottomrule
  \end{tabular*}
\end{table}

\FloatBarrier

\FloatBarrier
\section{Theory}\label{app:theory}
This section proves the properties of the gain and of the calibration that \cref{sec:method} states. The last subsection relates \rcpndcg{} to its rubric-only special case, Count-nDCG.
Throughout, all discriminations are positive, $\cpar{\gamma_c} > 0$, and $P_{\cpar{c}}(\theta) = \sigma\bigl(\cpar{\gamma_c}(\theta - \cpar{\beta_c})\bigr)$.

\subsection{The gain, its slope, and order preservation}\label{app:theory-gain}\label{app:theory-orderscale}
IRT measures the precision of the criteria with the test information
\begin{equation}
  I(\theta) = \sum\nolimits_{\cpar{c}} \cpar{\gamma_c^2}\, P_{\cpar{c}}(\theta)\bigl(1 - P_{\cpar{c}}(\theta)\bigr),
  \label{eq:app-info}
\end{equation}
the Fisher information that one yes/no answer to every criterion carries about the ability~\citep{lord1980applications}. Higher test information means that the answers locate the ability more precisely.

\begin{lemma}[Properties of the gain]\label{lem:theory-gain}\label{prop:monotone-gain}\label{prop:gain-fisher}
  For any number of criteria $C \ge 1$, discriminations $\cpar{\gamma_c} > 0$, and difficulties $\cpar{\beta_c}$, the gain $g$ of \cref{eq:theta-gain} has the following properties.
  \begin{enumerate}[nosep,label=(\roman*)]
    \item $g$ is smooth, $0 < g(\theta) < 1$ for every $\theta$, $g(\theta) \to 0$ as $\theta \to -\infty$, and $g(\theta) \to 1$ as $\theta \to \infty$.
    \item Its slope is $g'(\theta) = I(\theta) / \sum_{\cpar{c}} \cpar{\gamma_c} > 0$. The gain $g$ is therefore strictly increasing. Under the constraint $\sum_{\cpar{c}} \cpar{\gamma_c} = C$, the slope is $I(\theta)/C$.
  \end{enumerate}
\end{lemma}
\begin{proof}[Proof sketch]
  Each $P_{\cpar{c}}$ is smooth, lies in $(0, 1)$, and tends to 0 and 1. The gain is a convex combination of the $P_{\cpar{c}}$ with positive weights, which gives (i). Since $\sigma' = \sigma(1 - \sigma)$, we have $P_{\cpar{c}}' = \cpar{\gamma_c} P_{\cpar{c}}(1 - P_{\cpar{c}})$. Hence $g' = I / \sum_{\cpar{c}} \cpar{\gamma_c}$ with every term positive, which gives (ii).
\end{proof}

\Cref{prop:order} states that sorting a query's documents by $\dpar{\tilde{\theta}}$, or by any strictly increasing function of it, reproduces their Stage~A order.
\begin{proof}[Proof sketch]
  The fit sets $\qpar{\tau_j}$ to the softplus of an unconstrained parameter. Hence $\qpar{\tau_j} > 0$ (\cref{app:method-fit}). Then $x \mapsto f(\qpar{\tau_j} x + \qpar{\alpha_j})$ is strictly increasing for every strictly increasing $f$, such as $g$ (\cref{lem:theory-gain}(ii)). Hence $\dpar{\hat{\theta}_i^{\mathrm{BT}}} > \dpar{\hat{\theta}_{i'}^{\mathrm{BT}}}$ holds if and only if $\dpar{\tilde{\theta}_i} > \dpar{\tilde{\theta}_{i'}}$ and if and only if $f(\dpar{\tilde{\theta}_i}) > f(\dpar{\tilde{\theta}_{i'}})$. The ideal ranking sorts the pool by decreasing gain and follows the Stage~A order.
\end{proof}
Raw-BT-nDCG applies the gain $g$ to the BT scores centered within each query, without calibration (\cref{app:notation-metrics}). This baseline shares the Stage~A order and the ideal ranking of \rcpndcg{} and differs only in the gain values (\cref{app:theory-rubriconly}).

\subsection{Identification of the calibration}\label{app:theory-identification}\label{app:bt-2pl-identifiability}
The BT model depends only on score differences, and the judge compares documents only within a query. Each query's pool is thus its own comparison set. The 2PL likelihood of \cref{eq:2pl} stays the same if we stretch the axis and shrink every discrimination by the same factor. The likelihood also stays the same if we shift all offsets and difficulties together. \Cref{prop:theory-identification} shows that these are the only symmetries and that the hard constraints of \cref{sec:calibration-merge} remove them.

\begin{proposition}[Identification of the link]\label{prop:theory-identification}\label{prop:app-identification}
  Let every query have two documents with distinct BT scores. Consider two parameter settings $(\qpar{\tau}, \qpar{\alpha}, \cpar{\gamma}, \cpar{\beta})$ and $(\qpar{\tau}', \qpar{\alpha}', \cpar{\gamma}', \cpar{\beta}')$ with all scales and discriminations positive.
  \begin{enumerate}[nosep,label=(\roman*)]
    \item The two settings give every document the same pass probability on every criterion if and only if, for some $a > 0$ and $b \in \mathbb{R}$, $\qpar{\tau_j'} = a\qpar{\tau_j}$ and $\qpar{\alpha_j'} = a\qpar{\alpha_j} + b$ for every query, and $\cpar{\gamma_c'} = \cpar{\gamma_c}/a$ and $\cpar{\beta_c'} = a\cpar{\beta_c} + b$ for every criterion.
    \item Two settings related as in (i) give every document the same gain $g(\dpar{\tilde{\theta}_i})$.
    \item Among the settings related as in (i), exactly one satisfies $\sum_{\cpar{c}} \cpar{\gamma_c} = C$ and $\operatorname{mean}_{\cpar{c}} \cpar{\beta_c} = 0$.
    \item With the criterion parameters fixed, two settings of $(\qpar{\tau_j}, \qpar{\alpha_j})$ give the documents of query $\qpar{j}$ the same pass probabilities only if they are equal.
  \end{enumerate}
\end{proposition}
\begin{proof}[Proof sketch]
  The ``if'' direction of (i) follows by substituting the transformed parameters into \cref{eq:2pl}. For the ``only if'' direction, equal pass probabilities mean equal logits because the logistic function is injective. Two affine functions of the BT score that agree at two distinct values have equal slopes and intercepts. Part (ii) uses that the gain is a convex combination of the pass probabilities. Part (iii) fixes $a$ by the constraint on $\sum_{\cpar{c}} \cpar{\gamma_c}$ and $b$ by the centered difficulties. Part (iv) follows from equal logits at two distinct BT scores.
\end{proof}
Under the likelihood alone, other constraints would give the same gains (part (ii)). The hard constraints identify the parameters (part (iii)), and the priors play no part in identification. The constraints only decide where the weak priors on $\qpar{\tau_j}$ and $\qpar{\alpha_j}$ act (\cref{app:method-fit}).

\subsection{The rubric-only special case and the resolution of discrete gains}\label{app:theory-rubriconly}\label{app:theory-resolution}\label{app:sensitivity-loss}\label{app:tournament-value}\label{app:calibration-zeropass}
The Count-nDCG baseline uses Stage~B alone. Its gain is a document's pass share (\cref{tab:notation-metrics}). Let document $\dpar{i}$ appear in $n_{\dpar{i}}$ Stage~B windows and pass criterion $\cpar{c}$ in $S_{\dpar{i}\cpar{c}}$ of them, and let $T_{\dpar{i}} = \sum_{\cpar{c}} \cpar{\gamma_c} S_{\dpar{i}\cpar{c}}$. With the criterion parameters fixed, the rubric-only likelihood of an ability $\theta$ is
\begin{equation*}
  L_{\dpar{i}}(\theta) = \prod\nolimits_{\cpar{c}} P_{\cpar{c}}(\theta)^{S_{\dpar{i}\cpar{c}}} \bigl(1 - P_{\cpar{c}}(\theta)\bigr)^{n_{\dpar{i}} - S_{\dpar{i}\cpar{c}}}.
\end{equation*}

\begin{proposition}[Counting is rubric-only IRT scoring]\label{prop:theory-count-identity}\label{prop:app-count-identity}
  Fix the criterion parameters. Let $\dpar{\hat{\theta}_i}$ maximize $L_{\dpar{i}}$ over the extended real line $[-\infty, \infty]$. Then $g(\dpar{\hat{\theta}_i}) = \mathrm{Count}^{\gamma}_{\dpar{i}} = T_{\dpar{i}} / (n_{\dpar{i}} \sum_{\cpar{c}} \cpar{\gamma_c})$, the discrimination-weighted pass share. With equal discriminations (a Rasch model), the right-hand side is $\mathrm{Count}_{\dpar{i}}$, the gain of Count-nDCG.
\end{proposition}
\begin{proof}[Proof sketch]
  Since $P_{\cpar{c}}' = \cpar{\gamma_c} P_{\cpar{c}}(1 - P_{\cpar{c}})$, the log-likelihood derivative is $(\log L_{\dpar{i}})'(\theta) = T_{\dpar{i}} - n_{\dpar{i}} \sum_{\cpar{c}} \cpar{\gamma_c}\, g(\theta)$. The strictly concave $\log L_{\dpar{i}}$ has a finite maximizer exactly when $0 < T_{\dpar{i}} < n_{\dpar{i}} \sum_{\cpar{c}} \cpar{\gamma_c}$, where $(\log L_{\dpar{i}})' = 0$ gives the claim. Otherwise the maximum lies at $\pm\infty$, where $g$ equals the pass share. Equal discriminations turn the weighted share into the unweighted one.
\end{proof}
Count-nDCG with weighted shares, raw-BT-nDCG, and \rcpndcg{} evaluate the same gain curve $g$ (\cref{prop:theory-count-identity}). The three metrics differ only in the input to $g$: the rubric's maximum-likelihood ability, the raw tournament score, and the calibrated score $\dpar{\tilde{\theta}}$. Differences between \rcpndcg{} and weighted Count-nDCG therefore come from the tournament and the calibration.

\begin{proposition}[Invisible swaps]\label{prop:theory-invisible-swaps}\label{prop:invisible-swaps}
  Let a query's pool hold $K$ documents whose gains take $L$ distinct values, with $n_{\ell}$ documents at value $\ell$.
  \begin{enumerate}[nosep,label=(\roman*)]
    \item An exchange of two documents with equal gains leaves DCG@$k$ of every ranking unchanged. An exchange of two documents with different gains changes DCG@$k$ whenever at least one of them is in the top $k$.
    \item The share of document pairs with equal gains is $s = \sum_{\ell} n_{\ell}(n_{\ell} - 1) / \bigl(K(K - 1)\bigr) \ge (K - L) / \bigl(L(K - 1)\bigr)$.
    \item With binary gains, $s \ge (K - 2) / \bigl(2(K - 1)\bigr)$, which approaches one half as $K$ grows.
  \end{enumerate}
\end{proposition}
\begin{proof}[Proof sketch]
  Part (i) follows because the discount $1/\log_2(r + 1)$ strictly decreases in the rank, (ii) follows by Cauchy--Schwarz, and (iii) sets $L = 2$ in (ii).
\end{proof}
Within a query, \rcp{} gains are a strictly increasing function of the BT score (\cref{prop:order}). Two documents of a query therefore tie in gain only if they tie in BT score.

\begin{proposition}[Resolution of rubric-only scores]\label{prop:theory-criterion-classes}\label{prop:criterion-classes}
  Let the judge answer the $C$ criteria for every document of a pool of $K$ documents.
  \begin{enumerate}[nosep,label=(\roman*)]
    \item If the judge answers each criterion once per document, at most $2^C$ response patterns occur. If, in addition, every document that passes a criterion also passes all easier criteria, at most $C + 1$ patterns occur.
    \item Any gain that depends only on the response pattern gives equal gains to documents with equal patterns. Under the ordered ladder of (i), at least $\lceil K/(C + 1) \rceil$ documents of a pool of $K$ share one pattern.
  \end{enumerate}
\end{proposition}
\begin{proof}[Proof sketch]
  Under the ladder of (i), the passed criteria are the $m$ easiest ones for some $m \in \{0, \dots, C\}$, and (ii) follows by the pigeonhole principle.
\end{proof}
If the judge answered each criterion once per document, then with $K = \val{draft.2.poolsize}$ and $C = 5$ the ladder bound would put at least \val{draft.app.ab.pigeonhole} documents on one level. Even the unrestricted $2^5 = 32$ patterns would fall below $K$. On NanoBEIR, Count-nDCG separates \val{a4.nano.sens.count-qwen}\% of the reranker pairs, close to the \val{a4.nano.sens.rcp-qwen}\% of \rcpndcg{} (\cref{sec:ranking-results}). Count-nDCG cannot see how a reranker orders documents with equal pass shares, whereas \rcp{} gains inherit the order of Stage~A. \Cref{sec:count} tests whether the tournament does so correctly.

A document that fails every criterion in every window has no finite rubric-only ability. Count gives all such documents gain 0, whereas \rcp{} orders them by the tournament. Such documents make up \val{b12.zp.trecdl.allfail.share}\% of TREC-DL pool documents. In the three suites of the human study, such documents have a mean \rcp{} gain of \val{b12.zp.human.all.p0.g-rcp}\%. Of the annotators' grades for these documents, \val{b12.zp.human.all.p0.ge2}\% are Useful or better. Relative to this rate, the gain differs by \val{b12.zp.human.all.p0.gap-rcp}\,pp for \rcp{} and by \val{b12.zp.human.all.p0.gap-count}\,pp for Count, whose gain is 0. Only the interval of the \rcp{} difference, \val{b12.zp.human.all.p0.gap-rcp.ci}\,pp, includes 0 (Count: \val{b12.zp.human.all.p0.gap-count.ci}\,pp).

Setting the \rcp{} gain of these documents to 0 changes the paired $t$-test's decision in \val{b12.tt.trecdl.rcp-h0.dchg} of the \val{b10.t1.dl.cells-n} (dataset, reranker pair) cells on TREC-DL. The change affects \val{b12.tt.nanomteb.rcp-h0.dchg} of \val{b10.t1.nano.cells-n} cells on NanoBEIR. On BRIGHT, whose lists here still contain each query's excluded documents, it affects \val{b12.tt.bright.rcp-h0.dchg} of \val{b10.t1.bright.cells-n} (\cref{app:nano-bright-scope}). Setting the gain to 0 reverses no separated pair.

\FloatBarrier
\section{Prompts}\label{app:prompts}
\begingroup
\frenchspacing
\makeatletter\newcolumntype{P}{>{\@minipagetrue}p{\linewidth}}\makeatother
\makeatletter\newcolumntype{J}{>{\@minipagetrue\raggedright\arraybackslash}p{\linewidth}}\makeatother
\tcbset{coltitle=black, colbacktitle=black!10}

This appendix prints the full Stage~A and Stage~B prompts and summarizes the prompts of the LLM judge studies. At run time, the query replaces \texttt{\{query\_placeholder\}}, and the documents of the current window replace \texttt{\{passages\_placeholder\}}. We wrote and refined the criteria on \val{n14.rubric.queries} queries from each of three NanoBEIR datasets (MS~MARCO, FiQA, HotpotQA), reading \val{n14.rubric.docs} documents per query. A coding agent edited the wording of the criteria to raise the agreement between the answers of \model{Qwen3.5-397B} and the judgments of two of the authors on these documents. These \val{n14.rubric.docs.total} documents make up about \val{n14.rubric.docshare}\% of the NanoBEIR pools. No grade of our external annotators or of NIST entered this process. BRIGHT, ViDoRe v3, and TREC-DL are fully held out.

\subsection{Stage A prompt}\label{app:prompts-stagea}
The TREC-DL runs of \model{Qwen3.6-27B} used exactly this prompt. Together with the shipped rubric, this prompt forms the reference cell of the factorial in \cref{app:robustness-factorial}. The stored runs of NanoBEIR, BRIGHT, and ViDoRe v3 do not record their Stage~A prompt file. We therefore do not claim that this text produced their scores unchanged. The prompt describes the score range as a logit scale, which the soft preferences of \cref{sec:stage-a} read as confidence gaps.

\begin{promptboxclean}[Stage A prompt (tournament)]
  \fontsize{8pt}{9.6pt}\selectfont

  \noindent\begin{tabular}{@{}P@{}}
    You are an \textbf{expert relevance ranker} for an information-retrieval system. Your job is to reorder a given list of documents so that the most useful ones for answering the user's search query appear first, and to assign a calibrated relevance score to each document.
  \end{tabular}
  \noindent\begin{tabular}{@{}P@{}}
    \textbf{Input} \\
    \begin{itemize}[leftmargin=1em, itemsep=0pt, topsep=0pt, parsep=0pt]
      \item \textbf{Query}: a natural-language question or search statement.
      \item \textbf{Documents}: entries labelled \texttt{doc\_1}, \texttt{doc\_2}, \ldots\ \texttt{doc\_\{num\_documents\_placeholder\}}, presented in arbitrary order.
    \end{itemize}
  \end{tabular}
  \noindent\begin{tabular}{@{}P@{}}
    \textbf{How to decide relevance} \\
    \begin{enumerate}[leftmargin=1em, itemsep=0pt, topsep=0pt, parsep=0pt]
      \item \textbf{Understand the query}
        \begin{itemize}[leftmargin=1em, itemsep=0pt, topsep=0pt, parsep=0pt]
          \item Identify the main topic, entities, intent, and any required answer type (factual, navigational, how-to, comparison, causal, opinion, transactional, list/aggregation, etc.).
          \item If the query is \textbf{genuinely ambiguous} (e.g., ``apple'' could mean the fruit or the company), commit to the \textbf{single most probable interpretation} given the documents and note it in your reasoning. Score every document against that one interpretation -- do not mix interpretations across documents.
        \end{itemize}
    \end{enumerate}
  \end{tabular}

  \noindent\begin{tabular}{@{}P@{}}
    \begin{itemize}[leftmargin=2em, itemsep=0pt, topsep=0pt, parsep=0pt]
      \item Surface \textbf{explicit constraints} (e.g., ``free'', ``Python'', ``2024'', ``without X'') and \textbf{implicit constraints} (e.g., language, recency, platform) that a relevant document must satisfy. Documents that violate a hard constraint should be scored substantially lower than those that satisfy it, even when the violating document has strong topical overlap.
      \item Expand the query with obvious synonyms or related terms (e.g., ``clear'' $\rightarrow$ ``clearing services'', ``condensate pump'' $\rightarrow$ ``dehumidifier pump'').
    \end{itemize}
  \end{tabular}

  \noindent\begin{tabular}{@{}P@{}}
    \begin{enumerate}[leftmargin=1em, itemsep=0pt, topsep=1pt, parsep=0pt, start=2]
      \item \textbf{Score each document} (criteria listed in order of importance)
        \begin{itemize}[leftmargin=1em, itemsep=0pt, topsep=0pt, parsep=0pt]
          \item \textbf{Answer directness} -- sentences that directly answer the query are the strongest signal.
          \item \textbf{Specificity} -- mentioning the exact entity or detail asked for (a product name, a service, a concrete fact) outranks vague or generic mentions.
          \item \textbf{Semantic relevance} -- discussing the same concept even without exact keyword overlap deserves a high score.
          \item \textbf{Comprehensiveness} -- addressing multiple facets of a multi-part query is better than addressing only one.
        \end{itemize}
    \end{enumerate}
  \end{tabular}

  \noindent\begin{tabular}{@{}P@{}}
    \begin{itemize}[leftmargin=2em, itemsep=0pt, topsep=0pt, parsep=0pt]
      \item \textbf{Constraint satisfaction} -- documents that meet both explicit and implicit query constraints rank higher; those that violate a constraint are penalised proportionally to its importance.
      \item \textbf{Keyword presence} -- exact query terms (or synonyms) in context are a supporting signal, not a primary one.
      \item \textbf{Accuracy \& credibility} -- correct information from a trustworthy source.
      \item \textbf{Penalize irrelevance} -- unrelated topics, wrong entities, or generic filler should receive negative scores. Truncated snippets: judge only visible content. Near-duplicates: rank the more complete version higher. Boilerplate or login walls: score $\leq -4$.
    \end{itemize}
  \end{tabular}
  \noindent\begin{tabular}{@{}P@{}}
    \textbf{Bias checks} (actively counteract these): \\
    \begin{itemize}[leftmargin=1em, itemsep=0pt, topsep=0pt, parsep=0pt]
      \item The document order is arbitrary -- your ranking must not be influenced by it.
      \item Longer documents are not inherently more relevant -- judge by information density.
    \end{itemize}
  \end{tabular}

  \noindent\begin{tabular}{@{}P@{}}
    \begin{itemize}[leftmargin=1em, itemsep=0pt, topsep=0pt, parsep=0pt]
      \item \textbf{Context independence} -- score each document against the query on the absolute scale above, regardless of what other documents appear in this batch. Anchor by imagining a hypothetical document that perfectly answers the query ($+5$) and one that is completely off-topic ($-5$); place each real document on that fixed ruler. A document's score must not change because a stronger or weaker document is present.
    \end{itemize}
  \end{tabular}
  \noindent\begin{tabular}{@{}P@{}}
    \begin{enumerate}[leftmargin=1em, itemsep=0pt, topsep=0pt, parsep=0pt, start=3]
      \item \textbf{Relative ordering}
        \begin{itemize}[leftmargin=1em, itemsep=0pt, topsep=0pt, parsep=0pt]
          \item All high-relevance documents must appear before any low-relevance ones.
        \end{itemize}
    \end{enumerate}
  \end{tabular}
  \noindent\begin{tabular}{@{}P@{}}
    \textbf{Output format} \\
    Return \textbf{only} a JSON object with three fields -- no other text. \\[2pt]
    \texttt{"reasoning"}: a brief string summarising how you interpreted the query and the key factors that determined your top-ranked and bottom-ranked documents. This is your self-check -- use it to verify that your ranking and scores are internally consistent before finalising.
  \end{tabular}
  \noindent\begin{tabular}{@{}P@{}}
    \texttt{"ranking"}: document identifiers (\texttt{doc\_1}, \texttt{doc\_2}, \ldots) from most to least relevant. The array must be sorted by strictly descending score.
  \end{tabular}
  \noindent\begin{tabular}{@{}P@{}}
    \texttt{"scores"}: for \textbf{every} document, a calibrated relevance score on a logit scale from $-5.0$ (completely irrelevant) to $+5.0$ (perfect answer). Differences on this scale should reflect your confidence:
    \begin{itemize}[leftmargin=1em, itemsep=0pt, topsep=0pt, parsep=0pt]
      \item A $1$-point gap $\approx 73\,\%$ sure the higher doc is better (moderate)
      \item A $2$-point gap $\approx 88\,\%$ sure (strong)
      \item A $3$-point gap $\approx 95\,\%$ sure (overwhelming)
      \item A $0$-point gap $\approx 50\,\%$ -- a coin flip
    \end{itemize}
  \end{tabular}

  \noindent\begin{tabular}{@{}P@{}}
    Score anchors:
    \begin{itemize}[leftmargin=1em, itemsep=0pt, topsep=0pt, parsep=0pt]
      \item[] $+4$ to $+5$: directly and completely answers the query; could be the sole result
      \item[] $+2$ to $+4$: highly relevant, key facts present, may be incomplete
      \item[] $\phantom{+}0$ to $+2$: partially relevant, right topic, doesn\textquotesingle{}t clearly answer
      \item[] $-2$ to $\phantom{+}0$: tangentially related, shares keyword/domain, mostly off-topic
      \item[] $-5$ to $-2$: irrelevant, wrong topic/entity, or noise
    \end{itemize}
  \end{tabular}

  \noindent\begin{tabular}{@{}P@{}}
    Use tenths freely (e.g., $2.3$, $-0.7$). \textbf{Always assign distinct scores.} The \textbf{size of each gap} must match how confident you are that the higher-scored document is truly better (see the gap-confidence table above). When documents are nearly indistinguishable, keep gaps small ($0.1$--$0.2 \approx$ a coin flip); use gaps of $\geq 1.0$ only when you are genuinely confident. Never give every document the same score.
  \end{tabular}
  \noindent\begin{tabular}{@{}J@{}}
    Example for 4 documents: \\
    \texttt{\{} \\
    \hangindent=4em\hangafter=1\texttt{\ \ "reasoning": "The query asks for X. doc\_3 directly answers with a specific fact; doc\_2 is about an unrelated topic.",} \\
    \hangindent=4em\hangafter=1\texttt{\ \ "ranking": ["doc\_3", "doc\_1", "doc\_4", "doc\_2"],} \\
    \hangindent=4em\hangafter=1\texttt{\ \ "scores": \{"doc\_3": 3.5, "doc\_1": 1.5, "doc\_4": 0.5, "doc\_2": -2.5\}} \\
    \texttt{\}}
  \end{tabular}
  \noindent\begin{tabular}{@{}P@{}}
    \textbf{Recap} \\
    \begin{enumerate}[leftmargin=1em, itemsep=0pt, topsep=0pt, parsep=0pt]
      \item Read the query -- extract core concepts, intent type, and both explicit and implicit constraints.
      \item For each document, judge how well it addresses those concepts and satisfies those constraints.
      \item Assign a calibrated logit-scale score to every document.
      \item Sort document identifiers from highest to lowest score.
      \item Write a brief reasoning string to verify your ranking is consistent.
      \item Return the JSON object and nothing else.
    \end{enumerate}
  \end{tabular}
  \noindent\begin{tabular}{@{}P@{}}
    \textbf{Query} \\
    \texttt{\{query\_placeholder\}} \\[2pt]
    \textbf{Documents} \\
    \texttt{\{passages\_placeholder\}}
  \end{tabular}
\end{promptboxclean}

\subsection{Stage B prompt}\label{app:prompts-stageb}
The Stage~B prompt asks five yes/no questions about each document of a window. The judge received exactly this prompt on ViDoRe v3 and TREC-DL. The NanoBEIR and BRIGHT runs used an earlier version of the prompt, with the same wording of the criteria C1 to C5. We therefore do not claim that the printed text produced their scores unchanged.

\begin{promptboxclean}[Stage B prompt (rubric)]
  \fontsize{8pt}{9.6pt}\selectfont

  \noindent\begin{tabular}{@{}P@{}}
    \textbf{Role} \\
    You are an \textbf{expert relevance assessor} for an information-retrieval system. For each document, answer five binary diagnostic questions that measure how well the document serves the user's query.
  \end{tabular}
  \noindent\begin{tabular}{@{}P@{}}
    \textbf{Instructions} \\
    \begin{enumerate}[leftmargin=1em, itemsep=0pt, topsep=0pt, parsep=0pt]
      \item \textbf{Understand the query}
        \begin{itemize}[leftmargin=1em, itemsep=0pt, topsep=0pt, parsep=0pt]
          \item Identify the core information need, intent type, and any required answer type (factual, navigational, how-to, comparison, causal, opinion, transactional, list/aggregation, etc.).
          \item If the query is \textbf{genuinely ambiguous}, commit to the \textbf{single most probable interpretation} and evaluate every document against it consistently.
        \end{itemize}
    \end{enumerate}
  \end{tabular}

  \noindent\begin{tabular}{@{}P@{}}
    \begin{itemize}[leftmargin=2em, itemsep=0pt, topsep=0pt, parsep=0pt]
      \item Surface \textbf{explicit constraints} (e.g., ``free'', ``Python'', ``2024'', ``without X'') and \textbf{implicit constraints} (e.g., language, recency, platform).
      \item For \textbf{argumentative or stance queries} (e.g., ``Should X be allowed?'', claims to be verified), a document is relevant (C1) if it discusses the same debate topic -- even from the opposite side. Information utility (C2) includes evidence or reasoning that bears on the claim, whether supporting or opposing.
    \end{itemize}
  \end{tabular}

  \noindent\begin{tabular}{@{}P@{}}
    \begin{enumerate}[leftmargin=1em, itemsep=1pt, topsep=1pt, parsep=0pt, start=2]
      \item \textbf{Evaluate each document independently} against the five criteria below. Each criterion is a standalone yes/no question -- answer 1 (yes) or 0 (no). Criteria are \textbf{not hierarchically dependent}: a document may pass any criterion regardless of whether it passes the others.
      \item \textbf{Coherence review} -- after scoring all criteria, briefly verify that your ratings are internally consistent. If you notice any apparent contradiction, re-read the relevant criterion definition and the document before finalizing.
    \end{enumerate}
  \end{tabular}
  \noindent\begin{tabular}{@{}P@{}}
    \textbf{Bias checks} -- actively counteract these:
    \begin{itemize}[leftmargin=1em, itemsep=0pt, topsep=0pt, parsep=0pt]
      \item Document order is arbitrary -- do not let it influence your judgments.
      \item Longer documents are not inherently more relevant -- judge by substance, not length.
      \item Evaluate each document against the query on an absolute standard, regardless of what other documents appear in this batch.
      \item C5 is a \textbf{high bar} -- most documents should fail it. Do not let a document's length, apparent effort, or broad coverage bias you toward awarding it.
    \end{itemize}
  \end{tabular}
  \noindent\begin{tabular}{@{}P@{}}
    \textbf{Criteria} \\[2pt]
    \textbf{C1 -- Topical Relevance} \\
    ``Does this document address a topic that is clearly related to the query's information need -- not just loosely related or in the same broad domain?'' \\
    Award 1 if a two-sentence summary of the document would mention the same topic as the query -- i.e., the connection is immediate and obvious. Award 0 only if the connection requires reasoning through shared methods, broader categories, or domain overlap.
  \end{tabular}
  \noindent\begin{tabular}{@{}P@{}}
    \textbf{C2 -- Information Utility} \\
    ``Does this document contain specific, useful information -- facts, data, procedures, or explanations -- that would help address the query?'' \\
    Award 1 if a user could extract at least one concrete fact, data point, procedure step, or example that directly bears on the query. Award 0 if the document only provides general background, definitions, or context without specific applicable information.
  \end{tabular}
  \noindent\begin{tabular}{@{}P@{}}
    \textbf{C3 -- Entity/Detail Match} \\
    ``Does this document specifically name and discuss the particular entity, product, concept, or detail mentioned in the query -- not merely a related alternative or a broader category that includes it?'' \\
    Award 1 if the exact entity or detail from the query appears by name and the document provides at least some specific information about it -- not just a mention in a list, heading, or aside. Award 0 if the document covers only a broader category, a competing or related entity, or never names the specific entity.
  \end{tabular}
  \noindent\begin{tabular}{@{}P@{}}
    \textbf{C4 -- Direct Answer} \\
    ``Does this document explicitly and directly answer or resolve the query's primary question or information need?'' \\
    Award 1 if the document contains a clear, direct response to the main question -- even if it doesn\textquotesingle{}t cover every secondary aspect. Award 0 if the user would still need to search elsewhere for the core answer.
  \end{tabular}
  \noindent\begin{tabular}{@{}P@{}}
    \textbf{C5 -- Thorough Treatment} \\
    ``Does this document treat the query's topic in depth -- providing detailed explanations, thorough coverage, or comprehensive detail rather than a superficial or cursory mention?'' \\
    Award 1 only if the document provides genuinely deep, substantive coverage that goes well beyond a basic answer -- e.g., explaining mechanisms, providing supporting evidence, addressing edge cases, or offering practical detail that a knowledgeable reader would expect. Award 0 if the treatment is surface-level, merely lists aspects without elaboration, restates the same point in different ways, or is long without adding real depth. A lengthy or broad document is not automatically thorough.
  \end{tabular}
  \noindent\begin{tabular}{@{}P@{}}
    \textbf{Query} \\
    \texttt{\{query\_placeholder\}} \\[2pt]
    \textbf{Documents} \\
    \texttt{\{passages\_placeholder\}}
  \end{tabular}
  \noindent\begin{tabular}{@{}P@{}}
    \textbf{Output Format} \\
    Return a JSON object. For each document, provide:
    \begin{itemize}[leftmargin=1em, itemsep=0pt, topsep=0pt, parsep=0pt]
      \item \texttt{"reasoning"}: a brief one-line justification
      \item \texttt{"criteria"}: an object mapping C1--C5 to 1 or 0
    \end{itemize}
  \end{tabular}

  \noindent\begin{tabular}{@{}J@{}}
    Example: \\
    \texttt{\{} \\
    \texttt{\ \ "doc\_1": \{} \\
    \hangindent=6em\hangafter=1\texttt{\ \ \ \ "reasoning": "Official Python 3.11 docs on sorted(); exact match, comprehensive, authoritative.",} \\
    \texttt{\ \ \ \ "criteria": \{"C1":1,"C2":1,"C3":1,"C4":1,"C5":1\}} \\
    \texttt{\ \ \},} \\
    \texttt{\ \ "doc\_2": \{} \\
    \hangindent=6em\hangafter=1\texttt{\ \ \ \ "reasoning": "Python sorting tutorial for 3.9. Answers the question but explanations stay surface-level and omit edge cases.",} \\
    \texttt{\ \ \ \ "criteria": \{"C1":1,"C2":1,"C3":1,"C4":1,"C5":0\}} \\
    \texttt{\ \ \},}
  \end{tabular}

  \noindent\begin{tabular}{@{}J@{}}
    \texttt{\ \ "doc\_3": \{} \\
    \hangindent=6em\hangafter=1\texttt{\ \ \ \ "reasoning": "Sorting algorithms theory. Relevant topic but no Python code or practical answer.",} \\
    \texttt{\ \ \ \ "criteria": \{"C1":1,"C2":0,"C3":0,"C4":0,"C5":0\}} \\
    \texttt{\ \ \},} \\
    \texttt{\ \ "doc\_4": \{} \\
    \hangindent=6em\hangafter=1\texttt{\ \ \ \ "reasoning": "Java sorting tutorial. Sorting concepts transfer, but does not cover Python's sorted().",} \\
    \texttt{\ \ \ \ "criteria": \{"C1":1,"C2":1,"C3":0,"C4":0,"C5":0\}} \\
    \texttt{\ \ \},} \\
    \texttt{\ \ "doc\_5": \{} \\
    \hangindent=6em\hangafter=1\texttt{\ \ \ \ "reasoning": "Unrelated article about wildlife conservation.",} \\
    \texttt{\ \ \ \ "criteria": \{"C1":0,"C2":0,"C3":0,"C4":0,"C5":0\}} \\
    \texttt{\ \ \}} \\
    \texttt{\}}
  \end{tabular}
\end{promptboxclean}

\subsection{Meta-judge and external-judge prompts}\label{app:prompts-meta}\label{app:prompts-external}
Studies~1 and~2 send their prompts to the three judges of \cref{app:meta}, each at its maximum reasoning effort, with no output cap and the provider's default temperature. The two system prompts are those of the original study scripts. The judges answer in the shared response format printed below. Each prompt also shows the dataset's MTEB retrieval-task instruction. Study~1 asks which of two candidate sets of relevant documents better serves as the query's annotation, weighing relevance, precision, and coverage. Study~2 compares two rerankers' top-10 lists, adds ranking quality to these criteria, and shortens each document to \val{draft.app.ckl.metajudge.doccap} characters. Two of the authors judged part of Study~2 and saw the same shortened lists (\cref{app:meta-study2}).

\begin{promptboxclean}[Shared response format]
  \fontsize{8pt}{9.6pt}\selectfont
  \noindent\begin{tabular}{@{}P@{}}
    Respond with exactly the following three fields, in this order: \\[2pt]
    JUDGMENT: A concise explanation (1--3 sentences) naming a concrete, content-level reason for your choice. \\
    CONFIDENCE: Exactly one of \texttt{"high"}, \texttt{"medium"}, or \texttt{"low"}. \\
    VERDICT: Exactly one of \texttt{"A"}, \texttt{"B"}, or \texttt{"Tie"}.
  \end{tabular}
\end{promptboxclean}

In the comparison of \cref{app:robustness-judges}, each case also shows a task description of what counts as relevant in its dataset. The NanoBEIR descriptions ask, for example, for documents that refute a claim or for duplicate questions. BRIGHT asks for passages that answer a question in the task's discipline. TREC-DL cases use the MSMARCO web-search description, and ViDoRe v3 cases ask for the pages with the needed information.

\Cref{sec:count} and \cref{app:robustness-judges} use the same three judges, \model{GLM-5.3-flash}, \model{DeepSeek-4.1-flash}, and \model{Kimi~K3}, with a prompt of their own. This prompt asks which of two ranked lists better serves the query and weighs top ranks more heavily than lower ones. The prompt tells the judges to decide by content match, to treat near-duplicates as equivalent, and to answer Tie only when neither list is defensibly better. Each response is a JSON object with the fields \texttt{choice} (A, B, or Tie), \texttt{confidence}, and a one-sentence \texttt{reason}. The user prompt shows the task description, the query, and the two lists without document identifiers or metric hints. A document longer than \val{nr.a6.b19.textcap} characters is cut at its last sentence boundary before this limit, and the marker \texttt{[...]} shows the cut. Each case is sent twice, with the labels and the display order of the two lists swapped. The full text of every prompt in this appendix is part of the release (\cref{app:assets}).

\endgroup

\section{Assets and Release}\label{app:assets}
\Cref{tab:app-assets} lists the datasets with their sources and licenses.

\begin{table}[!ht]
  \centering
  \caption{Datasets and third-party relevance data, with each provider's stated license or terms of use.}
  \label{tab:app-assets}
  \fontsize{8pt}{9.6pt}\selectfont
  \begin{tabularx}{\linewidth}{@{}l>{\raggedright\arraybackslash}X>{\raggedright\arraybackslash}p{0.28\linewidth}@{}}
    \toprule
    Asset & Source and use & License or terms \\
    \midrule
    NanoBEIR & Nano versions of 13 BEIR datasets in MTEB~\citep{thakur2021beir,muennighoff2023mteb,zetaalpha2024nanobeir}, \texttt{mteb/Nano*Retrieval}; \val{draft.4.nano.queries} queries & CC BY 4.0 on the dataset cards; the BEIR source datasets state their own terms \\
    BRIGHT & 12 tasks~\citep{su2024bright}, \nolinkurl{mteb/BrightRetrieval}; \val{a4.bright.queries.evaluated} evaluated queries & CC BY 4.0 \\
    ViDoRe v3 & 8 corpora~\citep{vidore3}, \texttt{vidore/vidore\_v3\_*}; \val{a3.n.questions} questions in their original language, judged as OCR text & CC BY 4.0 \\
    TREC-DL & MS~MARCO passages~\citep{nguyen2016msmarco} with the NIST grades of the 2019 and 2020 tracks~\citep{craswell2020overview,craswell2021overview}; \val{a2.setup.dl19.queries} and \val{a2.setup.dl20.queries} queries & MS~MARCO terms, non-commercial research only; NIST grades public from TREC \\
    DL19 re-assessments & \citet{parry2025reannotation}, public repository at commit \val{draft.app.ckl.parry.commit}; analysis only, not redistributed & Paper CC BY 4.0; no license file in the repository \\
    \bottomrule
  \end{tabularx}
\end{table}

\paragraph{Models and software.}
Our three judges are open-weight models under Apache 2.0: the NVFP4 release of \model{Qwen3.5-397B}, \model{Qwen3.6-27B}, and \model{gpt-oss-120b}. We serve them with SGLang. Of the three first-stage retrievers, \baseline{Octen-Embedding-8B} is also under Apache 2.0, and BM25 runs in the MIT-licensed \texttt{bm25s}. We used the three panel judges (\cref{app:meta}), \model{GLM-5.2} (\cref{app:trecdl-judge2}), \baseline{Cohere Embed~v4}, and Voyage~AI's rerankers under their providers' terms. \Cref{tab:app-rerankers} gives the licenses of the rerankers.
\paragraph{Human annotations.}
The human grades of \cref{sec:calibration,sec:blind-eval,sec:external} come from an annotation platform (\cref{app:human-annotators}). The platform stored the annotators' names and email addresses. In the released data, one-way hashes replace them, and no table maps the hashes back. We have not screened the annotators' free-text comments for personal information beyond replacing email addresses. We therefore do not release these comments.

\paragraph{Release.}
 We release the code of both stages, the calibration, the metric, and the full text of every judge prompt of \cref{app:prompts} at \releaseurl{}. The release also holds the calibrated scores of every judged query-document pair, the parameters that turn them into gains, and the \val{a1a.desc.grades} human grades on \val{a1a.desc.distinct-docs} documents. The BRIGHT and NanoArguAna labels also cover the documents that our rankings remove for some queries (\cref{app:nano-bright-scope}). Our rankings and tables remove these documents, except in the few follow-up analyses that say so.
\section{Extended Related Work}\label{app:related}
We compare \rcp{} in more detail with three lines of work that the main text covers only briefly.

\paragraph{Relative and absolute judgments together.}
\citet{yan2024consolidating} observe that an LLM ranks well from pairwise prompts, whereas its pointwise relevance labels rank documents less well and its pairwise scores carry no absolute meaning. They also find that pairwise scores are not calibrated across queries. They perturb each document's pointwise label as little as possible until the labels respect the pairwise order. \rcp{} instead keeps the tournament order fixed and fits only a scale $\qpar{\tau_j}$ and an offset $\qpar{\alpha_j}$ per query. Binary criteria shared by all queries set these two parameters, as shared anchor items place two test forms on one scale in IRT linking~\citep{lord1980applications}. Similarly, label models infer true labels from noisy annotators by estimating each annotator's expertise and each item's difficulty~\citep{whitehill2009whose}.

\paragraph{Tournament gains and per-query scores.}
Concurrently, SABER-Math builds labels for mathematical retrieval from per-query BT tournaments and rescales each query's scores to a fixed range before computing nDCG~\citep{georgiev2026sabermath}. The zELO method turns LLM pairwise preferences into per-query Elo scores and centers each query's scores at zero to train the zerank rerankers~\citep{pipitone2025zelo}. This centering fixes each query's zero point, but it makes neither the zero points nor the units comparable across queries. In \rcp{}, the rubric shared by all queries supplies both.

\paragraph{Rubrics and calibrated LLM judges.}
Rubric and checklist judges ask an LLM about several criteria and aggregate the answers, for example by counting, summing, or taking the best grade~\citep{farzi2024pencils,farzi2025criteria,lee2025checkeval,ye2024flask}. Up to the criterion weights, Count-nDCG (\cref{eq:count-gain}) is the retrieval analogue of these methods and the rubric-only special case of \rcp{} (\cref{prop:theory-count-identity}). Other work calibrates LLM judges against human labels~\citep{liu2024calibrating,hashemi2024llmrubric,takehi2025lara}. \rcp{} uses no human labels for calibration. Its rubric only sets the scale and offset of each query's tournament scores.

\FloatBarrier
\section{Author Contributions}\label{app:contributions}
Fabian David Schmidt conceived the research contributions, ran most of the analyses, coordinated the human annotation study, ran the TREC-DL experiments, and wrote the final version of the paper. Donato Crisostomi ran the initial experiments and wrote a first draft of an earlier version of the manuscript. Carlos Lassance took part in the regular research meetings and gave feedback throughout the project. Nils Reimers approved the internship and the project, and provided institutional support.

\end{appendices}

\end{document}